%% file: main.tex
\documentclass[phd,lfcs,twoside,logo,normalheadings,leftchapter,notimes]{infthesis}
\usepackage{times} 
\DeclareMathSizes{12}{12.5}{8}{8} 

\usepackage[hidelinks]{hyperref}
\usepackage{stmaryrd}
\usepackage[backend=bibtex,style=ieee,doi=false,isbn=false,url=false,eprint=false]{biblatex}
\usepackage[toc,nopostdot]{glossaries}
\usepackage{glossary-mcols}
\usepackage{listings}
\usepackage{amsmath,amsthm,mathtools,amsfonts,amssymb}
\usepackage{xcolor}
\usepackage{scalerel}
\usepackage{tikz,tikz-cd,tikzit,quiver}
\usetikzlibrary{positioning}

\allowdisplaybreaks[1]

\renewcommand{\GLS}[2]{\emph{#1}}
\renewcommand{\gls}[1]{\GLS{#1}{#1}}

\lstnewenvironment{algorithm}{
    \lstset{ 
        mathescape=true,
        escapeinside={(*@}{@*)},
        numbers=none, 
        numberstyle=\tiny,
        basicstyle=\footnotesize, 
        keywordstyle=\bfseries\em,
        keywords={,input, output, not, and, or, if, then, else, while, in, do, begin, end,}
        morecomment=[l]{//}, 
        commentstyle=\itshape
    }
}{}

\newcommand{\comment}[1]{}

\newcommand{\qwhile}{\textbf{q}\texttt{While}}
\newcommand{\dowhilePX}{\text{do }P\text{ while }X}

\newcommand{\ie}{\emph{i.e.\@}}
\newcommand{\eg}{\emph{e.g.\@}}
\newcommand{\etc}{\emph{etc.\@}}
\newcommand{\etal}{\emph{et al.\@}}
\newcommand{\st}{\mathbin{\text{\ s.t.\ }}}
\newcommand{\ifc}{\text{if\quad}}
\newcommand{\undefined}{\text{undefined}}
\newcommand{\otherwise}{\text{otherwise}}

\newcommand{\im}{\mathrm{im}}
\newcommand{\dom}{\mathrm{dom}}
\newcommand{\bigO}[1]{\ensuremath{\mathcal{O}\left({#1}\right)}}
\newcommand{\Lim}[1]{\lim_{\scaleto{#1 \mathstrut}{7pt}}}

\newcommand{\into}{\hookrightarrow}
\newcommand{\pto}{\rightharpoonup}
\newcommand{\keq}{\simeq}
\newcommand{\fml}[1]{\ensuremath{\mathbf{#1}}}

\newcommand{\floor}[1]{\lfloor {#1} \rfloor}

\newcommand{\powerset}[1]{\mathcal{P}(#1)}
\newcommand{\fpset}[1]{\mathcal{F}(#1)}

\newcommand{\cA}{\ensuremath{\mathcal{A}}}
\newcommand{\cB}{\mathcal{B}}
\newcommand{\cC}{\ensuremath{\mathcal{C}}}
\newcommand{\cF}{\mathcal{F}}
\newcommand{\cW}{\mathcal{W}_\bot}

\DeclarePairedDelimiter{\absolute}{\lvert}{\rvert}
\newcommand{\sem}[1]{\llbracket #1 \rrbracket}
\newcommand{\abs}[1]{\absolute*{#1}}
\newcommand{\norm}[1]{\abs{\!\abs{#1}\!}}
\newcommand{\opnorm}[1]{\norm{#1}_\mathrm{op}}

\newcommand{\Bset}{\mathbb{B}}
\newcommand{\Nset}{\mathbb{N}}
\newcommand{\Zset}{\mathbb{Z}}
\newcommand{\Rset}{\mathbb{R}}
\newcommand{\Cset}{\mathbb{C}}
\newcommand{\CTwo}{\mathbb{C}^2}

\newcommand{\gdir}[2]{{#1} \vert_{#2}}

\newcommand{\id}{\mathrm{id}}
\newcommand{\Ob}{\mathrm{Ob}}
\newcommand{\iso}{\cong}
\newcommand{\xto}[1]{\xrightarrow{#1}}
\newcommand{\comma}[2]{(#1 \!\Rightarrow\! #2)}

\newcommand{\ex}{\mathrm{ex}}
\newcommand{\ki}{\mathrm{ki}}
\newcommand{\Tr}{\mathrm{Tr}}
\newcommand{\sub}[1]{{\scaleto{#1 \mathstrut}{7pt}}}
\newcommand{\dsub}[1]{{\scaleto{#1 \mathstrut}{5pt}}}

\newcommand{\cat}[1]{\ensuremath{\mathbf{#1}}}
\newcommand{\SCat}[1]{\ensuremath{\mathbf{\Sigma}_\mathrm{#1}}}
\newcommand{\A}{\cat{A}}
\newcommand{\B}{\cat{B}}
\newcommand{\C}{\cat{C}}
\newcommand{\D}{\cat{D}}
\newcommand{\V}{\cat{V}}
\newcommand{\Set}{\cat{Set}}
\newcommand{\Rel}{\cat{Rel}}
\newcommand{\Top}{\cat{Top}}
\newcommand{\TopHaus}{\cat{Top}_{\mathrm{Haus}}}
\newcommand{\Mon}{\cat{Mon}}
\newcommand{\CMon}{\cat{CMon}}
\newcommand{\HausCMon}{\cat{HausCMon}}
\newcommand{\HAG}{\cat{HausAb}}
\newcommand{\Grp}{\cat{Grp}}
\newcommand{\Ab}{\cat{Ab}}
\newcommand{\Vect}{\cat{Vect}}

\newcommand{\Hilb}{\cat{Hilb}}
\newcommand{\FdHilb}{\cat{FdHilb}}
\newcommand{\PInj}{\cat{PInj}}
\newcommand{\FinPInj}{\cat{FinPInj}}
\newcommand{\Bijection}{\cat{Bijection}}
\newcommand{\SubStoch}{\cat{SubStoch}}
\newcommand{\Stoch}{\cat{Stoch}}
\newcommand{\Unitary}{\cat{Unitary}}
\newcommand{\Isometry}{\cat{Isometry}}
\newcommand{\Contraction}{\cat{Contraction}}
\newcommand{\FdContraction}{\cat{FdContraction}}
\newcommand{\FdIsometry}{\cat{FdIsometry}}
\newcommand{\FdUnitary}{\cat{FdUnitary}}
\newcommand{\LSI}{\cat{LSI}}
\newcommand{\LSIContraction}{\cat{LSI}_\leq}
\newcommand{\CPTP}{\cat{CPTP}}
\newcommand{\CPTR}{\cat{CPTR}}
\newcommand{\Int}{\cat{Int}}

\newcommand{\SU}{\mathrm{SU}}

\renewcommand{\vector}[2]{\begin{pmatrix} {#1} \\ {#2} \end{pmatrix}}
\newcommand{\conv}{\ast}
\newcommand{\ket}[1]{{\lvert{#1}\rangle}}
\newcommand{\bra}[1]{{\langle{#1}\rvert}}
\newcommand{\braket}[2]{\langle{#1} \hspace{1pt} \vert \hspace{1pt} {#2}\rangle}
\newcommand{\ketbra}[2]{\lvert{#1}\rangle \! \langle{#2}\rvert}
\newcommand{\tr}{\mathrm{tr}}
\newcommand{\ltwo}{\ell^2}
\newcommand{\Ltwo}{L^2}
\newcommand{\Cstar}{\(\mathrm{C}^\ast\)}
\DeclareMathOperator{\Span}{span}

\newcommand{\eye}{\ensuremath{\cdot\!\!(\!\!\!>}}

\newcounter{nthm}[section] 
\theoremstyle{definition}
\newtheorem{definition}[nthm]{Definition}
\newtheorem{remark}[nthm]{Remark}
\newtheorem{notation}[nthm]{Notation}
\newtheorem{example}[nthm]{Example}
\theoremstyle{plain}
\newtheorem{proposition}[nthm]{Proposition}

\newtheorem{lemma}[nthm]{Lemma}
\newtheorem{theorem}[nthm]{Theorem}
\newtheorem{corollary}[nthm]{Corollary}

\title{Unbounded loops in quantum programs: categories and weak while loops}
\author{Pablo Andr\'es-Mart\'inez}

\abstract{
    Control flow of quantum programs is often divided into two different classes: classical and quantum.
    Quantum programs with classical control flow have their conditional branching determined by the classical outcome of measurements, and these collapse quantum data.
    Conversely, quantum control flow is coherent, \ie{} it does not perturb quantum data; quantum walk-based algorithms are practical examples where coherent quantum feedback plays a major role.
    This dissertation has two main contributions: (i) a categorical study of coherent quantum iteration and (ii) the introduction of weak while loops.

    (i) The objective is to endow categories of quantum processes with a traced monoidal structure capable of modelling iterative quantum loops.
    To this end, the trace of a morphism is calculated via the execution formula, which adds up the contribution of all possible paths of the control flow.
    Haghverdi's unique decomposition categories are generalised to admit additive inverses and equipped with convergence criteria using basic topology.
    In this setting, it is possible to prove the validity of the execution formula as a categorical trace on certain categories of quantum processes.
    Among these there are categories of quantum processes over finite dimensional Hilbert spaces (as previously shown by Bartha), but also certain categories of quantum processes over infinite dimensional Hilbert spaces, such as a category of time-shift invariant quantum processes over discrete time.

    (ii) A weak while loop is a classical control flow primitive that offers a trade-off between the collapse caused on each iteration and the amount of information gained.
    The trade-off may be adjusted by tuning a parameter and, in certain situations, it is possible to set its value so that we may control the algorithm without sacrificing its quantum speed-up.
    As an example, it is shown that Grover's search problem can be implemented using a weak while loop, maintaining the same time complexity as the standard Grover's algorithm (as previously shown by Mizel).
    In a more general setting, sufficient conditions are provided that let us determine, with arbitrarily high probability, a worst-case estimate of the number of iterations the loop will run for.
}

\begin{document}

\begin{preliminary}

\maketitle

\begin{acknowledgements}
    These four years of PhD studies have been humbling; I am ultimately content with the final product --- both the experience and the result --- but it did feel grim at times. Some exceptional people helped me push through without becoming a bitter person, and this is what I am most grateful for. To those individuals, I hope you can recognise the sincerity in the words that I dedicate to you below.

    To Chris, you have given me the freedom to explore whatever I found interesting and I always had the certainty you would be willing and interested to talk about it. I am especially grateful for your friendliness, genuine curiosity and the unconditional support you have given me, both in my PhD studies and in discussions of career paths.

    To Robin, collaborating with you was a pleasure. I am inspired by your researcher style and I am grateful for the patience you (and Chris) have had with me. Your joyful spirit towards research lifted me up when I was struggling with burn out. That gave me strength to take my time to polish up my results and write a thesis I can be proud of.

    To my parents, thank you for understanding and putting up with my tendency to drift away when I am not feeling at my best. Thank you for bringing me up in a way that I felt free (and responsible) of my own choices and in an environment that gave me so many opportunities. It goes without saying I would not be here without you. On that note, I must also thank a few teachers and university professors that, along my parents, made me thirst for knowledge; thank you Enrique Meléndez, Juan María Jurado, Francisco Vico and Inmaculada P. de Guzmán.

    To the \emph{Patos Macizones}, each and every one of you. It is rare to have such a long lasting friendship that is still standing strong, it makes me feel special to be part of such an exceptional bunch of people. Spending quality time with you has always felt so genuine, refreshing and liberating --- League games do not count as quality time. A special shout out to my brother, Miguel; I have avoided so many obstacles in life by following your trail, you are a great older brother and friend.

    To Emil and Natalia, the original rad flat was a great experience and having you around made me feel instantly at home in Edinburgh. Your support during periods of self-doubt was invaluable.

    To Ali and Tim, I am really grateful to you for helping me get out and socialise every now and then. The long walks and chats with Ali by the end of the PhD helped a lot in keeping everything into perspective and maintaining a healthy attitude. I have grown as a person by simply being around you both and I have regained confidence in who I am in great deal thanks to you.

    To Hono, by far you are the person I have had the largest number of deep conversations with. Talking with you about almost anything tends to be a wonderful exercise in creative thinking. I derive a great deal of joy from our conversations and they remind me that the thirst for knowledge also asks for things beyond academic matters.

    To Vicky, you are a wonderful partner. I have made some selfish choices to get where I am and you have welcomed each and every one of them because you genuinely want me to grow and enjoy my life at its fullest. As cheesy as it is, your support keeps me afloat and reminds me that, ultimately, whoever I end up being or doing will be alright. There is so much to thank you for that literally everything I have said in the previous four paragraphs applies almost verbatim to you too.

    I am truly grateful.
\end{acknowledgements}

\begin{laysummary}
    If you have ever followed a recipe from a cookbook, it is likely you have encountered an unbounded loop: these have the form ``bake the soufflé until it turns golden". It is a \emph{loop} since you keep doing the same action --- baking --- for an extended duration of time and it is \emph{unbounded} since the actual amount of time you will bake for has not been specified; it may take longer on less powerful ovens. Writing a program is very similar to telling a computer how to follow a recipe: you specify the steps to solve a problem so that the computer ``bakes" the solution for you. Quantum computers are no different, you still need to give them a recipe; however, quantum computer scientists have generally been shy in their use of unbounded loops in quantum recipes. The main motivation behind this thesis is to argue that unbounded loops in quantum computer science may be more interesting than they are usually credited for.

    When following the instruction ``bake the soufflé until it turns golden" you are repeatedly peeking at the soufflé  to decide whether or not it is time to turn off the oven; a classical computer would do the same. However, quantum computers deal with quantum data which is very delicate in the sense that the simple act of measuring it alters it: as if peeking at your quantum soufflé would spoil it. Then, how can you possibly know when the soufflé is ready? It turns out that how ``badly" the soufflé is spoiled depends on how ``intensively" you peeked at it. One of the contributions of this thesis is the discussion of a general procedure that let us play with this trade-off and come out on top: we can keep an eye on the soufflé so we detect when it is ready without spoiling it badly.

    The soufflé in the previous metaphor was meant to represent quantum data, but the nature of our decision making was still classical: we either turned off the oven or let it bake for a little longer. Quantum physics allows us to do something crazier: we may \emph{entangle} the state of the soufflé with the state of the oven. If we engineer the right entanglement, the scenario where the oven is still \emph{on} while the soufflé is burning will simply be unfeasible. Then, we no longer need to control the oven ourselves, avoiding the risk of spoiling the soufflé by peeking at it. Previous works in the literature have attempted to formalise these ideas in mathematical terms that would help us better understand these ``fully quantum" unbounded loops. In this thesis I make contributions to this line of research, discussing a mathematical framework where these ``fully quantum" unbounded loops are proven to be well-defined and consistent with the theory.
\end{laysummary}

\standarddeclaration


\tableofcontents


\end{preliminary}


\include{Chapters/1} 
\include{Chapters/2} 
\include{Chapters/3} 
\include{Chapters/4} 
\include{Chapters/5} 

\appendix


\printbibliography[heading=bibintoc, title={Bibliography}]
\printglossary[style=mcolindex]

\end{document}

%% file: Chapters/1.tex
\chapter{Introduction}

This thesis studies quantum iteration using multiple technical results on traced monoidal categories, Hausdorff topology, the Fourier transform and weak quantum measurements. For the sake of completeness, a brief introduction to each of these fields of study is provided where relevant throughout the text.
Thus, the thesis is self-contained in the sense that no prior knowledge is required from the reader.
This stylistic choice results in a rather long text with multiple sections on introductory material.
These sections have been labelled as (\emph{Preamble}) in their title (also appearing in the index) so that readers experienced in the subject may skip them without risk of missing out on novel contributions.

This first chapter serves as the introduction to the thesis. Section~\ref{sec:QC} provides a brief introduction to quantum computing and, in particular, a discussion of the state of the art in the field of control flow of quantum programs --- the field of study that this thesis contributes to. Sections~\ref{sec:contributions} and~\ref{sec:publications} summarise the contributions of this thesis along with its structure and enumerate previous works we build upon. Sections~\ref{sec:CT} and~\ref{sec:QM_cats} provide a brief introduction to category theory and examples of categories relevant to quantum computing.

\section{Brief introduction to quantum computing}
\label{sec:QC}

The information held in the memory of a computer at any given time is known as its \emph{state}.
A programmer provides a computer with instructions on how it must manipulate its state so that, at the end of the process, the resulting state encodes the solution to the problem the programmer is interested in.
Thus, there are three essential aspects that need to be captured by any mathematical framework used to discuss computation:
\begin{itemize}
  \item the collection of possible states the computer can be in,
  \item the collection of transformations that we may apply to such states and
  \item how information may be retrieved from the computer's state.
\end{itemize}
The collection of possible states of a \emph{quantum computer} is described as a Hilbert space \(H\). 
Information may be retrieved from a quantum state by applying a \emph{measurement}; however, unlike in classical computing, measuring perturbs the quantum state being observed.
To apply a measurement, we  must first choose a decomposition of \(H\) into orthogonal subspaces:
\begin{equation*}
  H \iso \bigoplus_{i \in I} H_i
\end{equation*}
with projections \(\pi_i \colon H \to H_i\).
When the measurement is applied to a state \(v \in H\) it becomes one of the states in the collection \(\{\pi_i(v)\}_{i \in I}\), up to normalisation; with the probability of the \(i\)-th outcome determined by \(\norm{\pi_i(v)}^2\).
Consequently, vectors \(v \in H\) describing quantum states are required to be unit vectors \(\norm{v} = 1\), so that the probabilities of all possible measurement outcomes add up to one.
The rest of the vectors in \(H\) are understood as `unnormalised' states.

\begin{notation}
  Given two vectors \(v,u \in H\), their inner product is denoted \(\braket{u}{v}\); the widespread convention in quantum computing is to define the inner product to be linear in the second argument (and anti-linear in the first one).
  In the bra-ket notation, vectors are denoted as `kets' \(\ket{x} \in H\) where \(x\) may be any arbitrary label, \eg{} \(\ket{0},\ket{1},\ket{\psi} \in H\).
  For each vector \(\ket{\psi} \in H\) we define its `bra' to be the linear map \(\bra{\psi} \colon H \to \Cset\) defined as follows:
  \begin{equation*}
    \bra{\psi}(v) = \braket{\psi}{v}.
  \end{equation*}
  For any two vectors \(\ket{\psi},\ket{\varphi} \in H\) we define the linear map \(\ketbra{\psi}{\varphi} \colon H \to H\) to be:
  \begin{equation*}
    \ketbra{\psi}{\varphi}(v) = \braket{\varphi}{v} \cdot \ket{\psi}.
  \end{equation*}
\end{notation}

To manipulate their state, quantum computers may apply \emph{coherent} operations on their state space \(H\).
The term `coherent' means that no loss of quantum information occurs when applying the operation; formally, this is captured by the requirement that coherent operations are invertible.
The coherent operations that may be applied on a quantum state space \(H\) are unitary linear maps \(f \colon H \to H\), whose inverse is given by their adjoint.
In some situations --- for instance, after applying a measurement --- we may be certain that the computer's state is in a closed subspace \(H_0 \subset H\); then, it is reasonable to consider the restriction of a unitary \(f \colon H \to H\) to \(f \rvert_{H_0} \colon H_0 \to H\).
Such a linear map \(f \rvert_{H_0}\) is no longer surjective and, consequently, it cannot be unitary; however, it still is an isometry (and, hence, it has a left inverse).
The requirement that the computer's operations are isometries guarantees that \(\norm{f(v)} = \norm{v}\) for any \(v \in H\), so that unit vectors (\ie{} normalised states) are mapped to unit vectors.

Recall that, for any \(n \in \Nset\), all Hilbert spaces of dimension \(n\) are isomorphic to each other, the zero-dimensional space is just \(\{0\}\) and the one dimensional space is the field \(\Cset\).
Thus, \(\Cset^2\) can be seen as the smallest nontrivial Hilbert space and, as such, it plays a special role in quantum computing: that of representing the state space of a \gls{qubit}, the smallest unit of quantum memory.
It is customary to fix an orthonormal basis \(\ket{0},\ket{1} \in \Cset^2\) and refer to these as the `classical states' of the qubit, whereas any linear combination:
\begin{equation*}
  \ket{\psi} = \alpha \ket{0} + \beta \ket{1}
\end{equation*}
is referred to as a \emph{superposition} of \(\ket{0}\) and \(\ket{1}\).
Given two quantum memories whose state spaces are \(H\) and \(H'\), the space of states that the two together can store is the tensor product \(H \otimes H'\).
For instance, the state space provided by two qubits is \(\Cset^2 \otimes \Cset^2 \iso \Cset^4\) whose classical states are now denoted by the orthonormal basis 
\begin{equation*}
  \{\ket{0} \!\otimes\! \ket{0},\ \ket{0} \!\otimes\! \ket{1},\ \ket{1} \!\otimes\! \ket{0},\ \ket{1} \!\otimes\! \ket{1}\}
\end{equation*}
whose vectors are often denoted using the shorthand \(\ket{01} := \ket{0} \!\otimes\! \ket{1}\).
All unit vectors in \(H \otimes H'\) represent physically realisable states, but recall that not all vectors in a tensor product are of the form \(\ket{\varphi} \otimes \ket{\phi}\), \ie{} the canonical bilinear map \(p \colon H \times H' \to H \otimes H'\) is not surjective.
States \(\ket{\psi} \in H \otimes H'\) that are not in the image of \(p\) are known as \emph{entangled} states; the term is a direct reference to the fact that \(\ket{\psi} \not \in \im(p)\) implies that the state of the two quantum memories cannot be represented by describing the state of each one in isolation.

Instead of dealing with arbitrary Hilbert spaces \(H\) and arbitrary unitary operators \(H \to H\), quantum computer scientists often work at a level of abstraction closer to the real-world quantum computer.
The memory of a standard quantum computer is comprised of a finite number \(n\) of qubits that it can maintain and manipulate at any given time and, consequently, the state space of a real-world quantum computer is a finite-dimensional Hilbert space \(\Cset^{2^n}\).
Quantum computers support the direct application of only a selection of unitary maps often comprised of:
\begin{itemize}
  \item a finite set of unitary maps acting on single qubits \(\Cset^2 \to \Cset^2\) that, under composition, generates a group \(G_2\) that is a dense subset of \(\SU(2)\) --- in the sense that, for any \(f \in \SU(2)\) and any \(\epsilon > 0\) there is an element \(g \in G_2\) such that \(\opnorm{f - g} < \epsilon\) --- and
  \item a finite set (often a singleton set) of unitary maps acting on two qubits \(\Cset^4 \to \Cset^4\) that, along with \(G_2\) and under composition and tensor product generates a dense subset of \(\SU(2^n)\), so that any arbitrary unitary operation on the state space of \(n\) qubits may be implemented up to arbitrary precision.
\end{itemize}
Moreover, quantum computers support measurement of qubits in their \(\{\ket{0}, \ket{1}\}\) basis which, paired with the ability to approximate any unitary map, enables us to perform any arbitrary measurement.

\subsection{Control flow of quantum programs}
\label{sec:control_flow}

A program specifies a list of operations that a computer is instructed to perform.
The program's \emph{control flow} dictates the order in which such operations ought to be applied; the ordering is often made to depend on the computer's state during the program's execution.
For instance, the control flow may branch according to a decision making process within the program --- thus, applying a sequence of operations or another --- or it may loop and repeat a segment of the program until a condition is met.
Programming languages offer specialised instructions, known as control primitives, that let the programmer specify the control flow of the program; common examples are \emph{if-then-else} statements, \emph{while loops}, recursion, \etc{}
When designing programming languages for quantum computers, two fundamentally different notions of control flow may be considered.

\paragraph*{Classical control flow.} All decision making is dictated by the outcome of measurements, \ie{} classical information.
This paradigm is simple to realise on physical devices --- it may be realised by a quantum chip controlled by a classical computer.
Its mathematical formalisation is well understood, since it shares much in common with formal methods for programming languages in classical computer science (see Section~\ref{sec:classical_loop}).
Its main drawback is that, due to the use of measurements, quantum information is irreversibly lost whenever a control primitive needs to be evaluated.

\paragraph*{Quantum control flow.} The motivation behind this paradigm comes from the realisation that conditional control primitives may be performed \emph{coherently}, \ie{} in a reversible manner --- and, consequently, without loss of quantum information.
The canonical example is the construction of a quantumly controlled unitary \(\Lambda U \colon \Cset^4 \to \Cset^4\) defined by linear extension of:
\begin{equation*}
  \Lambda U(\ket{0} \otimes \ket{\psi}) = \ket{0} \otimes \ket{\psi} \quad\quad\quad \Lambda U(\ket{1} \otimes \ket{\psi}) = \ket{1} \otimes U \ket{\psi}
\end{equation*}
where \(U \colon \Cset^2 \to \Cset^2\) is an arbitrary unitary and \(\ket{\psi} \in \Cset^2\) an arbitrary state.
The unitary \(\Lambda U\) applies \(U\) on the second qubit if the first qubit is in state \(\ket{1}\) and does nothing if it is in state \(\ket{0}\); furthermore, when the first qubit is in a superposition \(\ket{\phi} = \alpha \ket{0} + \beta \ket{1}\), we obtain the following behaviour:
\begin{equation*}
  \Lambda U(\ket{\phi} \otimes \ket{\psi}) = \alpha (\ket{0} \otimes \ket{\psi}) + \beta (\ket{1} \otimes U \ket{\psi})
\end{equation*}
so that both options --- either applying \(U\) or not --- coexist in a superposition.
Under reasonable assumptions (\eg \(P\) describes a unitary process), we may define a conditional control primitive in quantum programming languages:
\begin{algorithm}
                        if $q \in K$ then do $P$
\end{algorithm}
that applies the subprogram \(P\) on every state \(q\) in the closed subspace \(K\), does nothing on the orthogonal complement of \(K\) and, on any other vector, acts accordingly to the linear extension --- for instance, \(\Lambda U\) would be described by setting \(P\) to apply \(U\) on the second qubit and define \(K\) as the subspace of vectors of the form \(\ket{1} \otimes -\).
The next natural step is to consider quantumly controlled iterative loops; unfortunately, the requirement that the computation is reversible creates some major obstacles.
Assume a programming language offers a control primitive
\begin{algorithm}
                        do $P$ while $q \in K$
\end{algorithm}
whose behaviour is described by the following recursive definition:\footnote{The choice of defining a `do-while' primitive instead of the more common `while-do' is not accidental and will be made clear in Section~\ref{sec:semantics} of the concluding chapter.}
\begin{algorithm}
                        do $P$
                        if $q \in K$ then
                          do $P$ while $q \in K$
\end{algorithm}  
so that subprogram \(P\) will be repeatedly applied on any linear component of the state that lands in \(K\).
There are three important issues with this iterative control primitive, which have been noted in multiple previous works.
\begin{itemize}
  \item \emph{A formal issue.} As pointed out by B{\u{a}}descu and Panangaden~\cite{AlternationPanangaden}, it is not clear whether the map \(H \to H\) induced by the iterative primitive would converge at all. Moreover, for the formalism to be considered physically sound, we would require that the solution is unitary and that the map giving the denotational semantics of the `do-while' primitive be continuous in a topological sense --- \ie{} that small changes to \(P\) would not cause abrupt changes in the behaviour of the loop.
  \item \emph{A physical issue.} Some linear components of the input state will have \(P\) applied to them more times than others.
  Evidently, applying subprogram \(P\) twice will take more time than applying it once, which means that different terms in the linear combination may not be `synchronised' in time. However, vectors in the Hilbert space \(H\) are understood as the state of the quantum computer \emph{at a particular time-step} and, hence, the physical intuition of this iterative loop is not properly captured in the mathematical formalism.
  \item \emph{A functionality issue.} To amend the physical issue, Ying, Yu and Feng~\cite{YingFlow} propose to introduce an auxiliary qubit on each iteration, generating a chain of qubits of unbounded length that records the `history' of each term in the superposition. However, as Linden and Popescu discuss in~\cite{HaltingPopescu} this causes linear components with different histories to no longer interact with each other, as their distinct histories make them orthogonal, and unitary operations send orthogonal states to orthogonal states.
  This prevents any interference between different branches of the control flow, making it effectively classical.
\end{itemize}
  
The formal issue was partially resolved in~\cite{Bartha}: Bartha shows that the category of unitary maps is traced with respect to the execution formula; we may interpret this as saying that, if the body of the do-while loop (subprogram \(P\)) describes a unitary \(H \to H\), then the result of aggregating all possible control flow branches of the quantum loop yields a unitary. However, Bartha's approach makes use of pseudoinverses, which complicates the study of whether the iterative primitive is continuous.\footnote{In this work we do not manage to prove that it is in fact continuous, but our framework is more amenable to conduct such a study, as discussed in Remark~\ref{rmk:continuous_trace}.}
The physical issue may be resolved by adding a notion of time to the formalism, in the manner discussed in Section~\ref{sec:LSI}; the crucial distinction with the proposal from Ying \etal{}~\cite{YingFlow} is that their locally generated `history' is replaced with a global notion of `time-line'.
States at different time-steps may not interfere with each other, but the same time-step may be reached via multiple paths in the control flow, and such paths will interact with each other; thus, the functionality issue is prevented.

One of the main objectives of the thesis is to formalise this argument, laying out the groundwork for formal semantics of programming languages with quantum iteration.
In particular, Chapters~\ref{chap:Sigma} and~\ref{chap:trace} will introduce a categorical framework that Section~\ref{sec:semantics} will argue to be appropriate for this task. 
The discussion in Chapters~\ref{chap:Sigma} and~\ref{chap:trace} is highly abstract and technical, but the ideas are rooted in physical intuitions drawn from the computational model of quantum walks, which is briefly introduced in the following section.

\subsection{Quantum walks: an example of coherent quantum feedback}
\label{sec:QW}

Let \(G = (V,E)\) describe a graph on a finite set of vertices \(V\) and a set of edges \(E \subseteq V \times V\) such that if \((u,v) \in E\) then \((v,u) \in E\), \ie{} the graph is undirected.
For each \(v \in V\) let \(E_v\) be the following set of edges:
\begin{equation*}
  E_v = \{(u,v) \in E \mid u \in V\},
\end{equation*}
let \(\Delta \colon V \to \Nset\) be the function that yields the degree of a vertex, \(\Delta(v) = \abs{E_v}\), and let \(H_v\) be the finite-dimensional Hilbert space spanned by taking \(E_v\) as an orthonormal basis.
Fix a collection of unitary maps \(\{f_v \colon H_v \to H_v\}_{v \in V}\) to be the set of \emph{coins} of the walk.
A \gls{quantum walk} is determined by the pair \((G,\{f_v\}_V)\) and it describes the evolution of a quantum state in the Hilbert space \(H = \oplus_{v \in V} H_v\) as follows:
\begin{itemize}
  \item let \(C \colon H \to H\) be the \emph{coin operator} defined as the direct sum
  \begin{equation*}
    C = \bigoplus_{v \in V} f_v;
  \end{equation*}
  \item let \(S \colon H \to H\) be the \emph{shift operator} defined as the linear extension of
  \begin{equation*}
    S \ket{u,v} = \ket{v,u}
  \end{equation*}
  for each \((u,v) \in E\), with \(\ket{u,v}\) being its corresponding unit vector in \(H\);
  \item for any state \(\ket{\psi} \in H\), after \(t \in \Nset\) time-steps the state becomes:
  \begin{equation*}
    (SC)^t \ket{\psi}.
  \end{equation*}
\end{itemize}

An illustrative physical interpretation of quantum walks is to view the graph \(G\) as the layout of an experiment, where each vertex determines a linear optics device (\eg{} lenses, mirrors, beamsplitters, \etc{}) and edges describe the path a photon can take from one device to another.
For each \((u,v) \in E\), the state \(\ket{u,v} \in H\) is interpreted to indicate the presence of a photon located somewhere between the device at vertex \(u\) and the one at vertex \(v\).
A linear combination
\begin{equation*}
  \ket{\psi} = \sum_{(u,v) \in E} \alpha_{u,v} \ket{u,v}
\end{equation*}
describes \emph{a single photon} in a superposition of locations.
The action of the device at vertex \(v\) on an incoming photon is described by the unitary operator \(f_v\) and the coin operator \(C\) simply arranges all of the coins \(f_v\) into a single linear map.
The shift operator \(S\) conveys the structure of the graph, ensuring that the output that \(f_v\) leaves at edge \((u,v)\) is, on the next time-step, an input to \(f_u\).
Quantum walks are not limited to describe experiments on linear optics and, in general, may be understood to be describe arbitrary quantum processes, similarly to how Markov chains and random walks describe stochastic processes~\cite{DTQWOrig}.
In the abstract case, instead of a photon we use the term \emph{walker} to refer to the entity traversing the graph.

The control flow of any program may be described as a graph, where vertices correspond to the different instructions in the program and edges indicate the possible control flow paths between them.
The state of the control flow during execution may be described as a quantum walk on such a graph, where the position of the walker indicates the branch of the control flow (or superposition of branches) that is currently being computed.
Crucially, the graph describing a quantum walk may have cycles, and these make it so that the walker (or a linear term within its superposition) may visit the same vertex any arbitrary number of times, thus representing a loop in the control flow.

Notice that both of the operators \(C\) and \(S\) describing the quantum walk's evolution are unitary and, therefore, reversible.
Consequently, quantum walks are a physically realisable framework that supports coherent feedback of quantum processes and, in particular, it can be understood to model quantum (coherent) control flow of programs.
However, there are two caveats: first, the graph describing the control flow of a program is directed while the graph of a quantum walk is undirected; fortunately, the requirement that \(G\) is undirected is only a matter of convention (which simplifies the quantum walk's definition) and directed quantum walks may be defined, as long as we guarantee that the in-degree and out-degree of each vertex coincides, so that its coin may be a unitary.
The second caveat is that programs are supposed to have inputs and outputs so that they may be composed, but quantum walks are defined on closed graphs and there is no immediate notion of composition for them.
A way around this is to define open quantum walks: quantum walks on graphs with special vertices marked as either `input' or `output' so that the walker may enter or leave the graph through them.
But then, some linear term of the walker's state may leave the graph at time-step \(t\), while other terms may not leave until a later time-step \(t' > t\) (see Figure~\ref{fig:time_domain_loop}).
Thus, to discuss the input-output behaviour of open quantum walks we must extend the Hilbert space \(H\) to describe states over time; a way to do so for arbitrary quantum states and processes is described in Section~\ref{sec:LSI}.

\begin{figure}
  \centering
  \input{Figures/1/time_domain_loop}
  \caption{A state \(\ket{\psi}\) is given as input to a quantum iterative process; two time-steps later, the state is in a superposition \(\ket{\psi_\sub{ba}} + \ket{\psi_\sub{bua}} + \ket{\psi_\sub{uua}}\) where \(\ket{\psi_\sub{ba}} = f_\sub{BA}\ket{\psi}\), \(\ket{\psi_\sub{bua}} = f_\sub{BU}f_\sub{UA}\ket{\psi}\) and \(\ket{\psi_\sub{uua}} = f_\sub{UU}f_\sub{UA}\ket{\psi}\). The term \(\ket{\psi_\sub{ba}}\) `leaves' the system a step earlier than \(\ket{\psi_\sub{bua}}\).}
  \label{fig:time_domain_loop}
\end{figure}

Open quantum walks may also be defined over continuous time~\cite{GraphTails} and, in order to describe their composition with ease, their input-output behaviour is described in the frequency domain; this insight motivates much of Section~\ref{sec:LSI}.
Both continuous-time and discrete-time open quantum walks have been shown to be universal models of quantum computation~\cite{UnivCTQW,UnivDTQW}, in the sense that any unitary operation may be approximated to arbitrary precision by combining a finite variety of resources.
Multiple quantum algorithms have been defined in terms of quantum walks (see~\cite{ElemDistinctness} and the survey from~\cite{QWAlgSurvey}), which may be seen as evidence of the power of coherent feedback in quantum computing.
In practice, a quantum algorithm described via a quantum walk is implemented in a quantum computer using simulation techniques such as the ones described in~\cite{SimQW}.

\section{Contributions and structure of the thesis}
\label{sec:contributions}

This thesis has two main contributions, both relevant in quantum computing; the first one in the field of quantum control flow and the second one in the field of classical control flow of quantum programs.
\begin{itemize}
  \item The first contribution is on the formal study of iterative quantum loops by means of categorical traces on categories of quantum processes.
  In particular, the thesis focuses on the execution formula, which aggregates all of the possible paths of the quantum control flow.
  It is shown that both the category of finite-dimensional Hilbert spaces and contractions and the category of linear shift invariant contractions (as introduced in Section~\ref{sec:LSI}) are totally traced with respect to the execution formula.
  To do so, Haghverdi's unique decomposition categories~\cite{Haghverdi} are generalised and equipped with convergence criteria based on basic topology.
  These contributions are presented in Chapters~\ref{chap:Sigma} and~\ref{chap:trace}.
  \item The second contribution is the proposal of a classical control flow primitive --- the weakly measured while loop --- that offers a trade-off between decoherence due to measurement and information gained.
  It is argued that this primitive offers a restricted ability to monitor a quantum state as it evolves; sufficient conditions for such a monitoring to be compatible with quantum speed-up are discussed.
  These contributions are presented in Chapter~\ref{chap:while}, which reproduces the contents of the published work~\cite{WeakWhileLoop}.
\end{itemize}
More details on the motivation and structure of these chapters is given in the following subsections.
Chapter~\ref{chap:remarks} concludes the thesis and proposes further work in both of these lines of research.

\subsection{Categorical study of iterative quantum loops}

Haghverdi's unique decomposition categories~\cite{Haghverdi} have been used to model iterative loops in classical computer science.
These categories rely on a notion of addition of morphisms that does not admit additive inverses.
Unfortunately, additive inverses are a fundamental aspect of categories of quantum processes, as they capture destructive interference, which is one of the key aspects of quantum theory.
Chapter~\ref{chap:Sigma} proposes a generalisation of the additive structure on Haghverdi's categories, so that additive inverses are allowed.
A connection to topological groups is made, which will be essential to prove the main results in Chapter~\ref{chap:trace}.
Multiple categories capturing subtle differences in the notion of infinitary addition are presented, and their relationship is exhibited in terms of adjunctions.

Chapter~\ref{chap:trace} builds upon the results of the previous chapter to provide a version of Haghverdi's unique decomposition categories equipped with convergence criteria using basic topology.
In this setting, it is possible to prove the validity of the execution formula as a categorical trace on certain categories of quantum processes.
Among these there are categories of quantum processes over finite-dimensional Hilbert spaces (as previously shown by Bartha~\cite{Bartha} for the case of isometries and unitaries), but also certain categories of quantum processes over infinite-dimensional Hilbert spaces, such as the category of time-shift invariant quantum processes over discrete time (introduced in Section~\ref{sec:LSI}).
The latter result lays the groundwork for the design of categorical semantics for quantum programming languages supporting (coherent) quantum iterative loops.
What such a language and its semantics may look like is sketched in Section~\ref{sec:semantics} of the concluding chapter; in essence, the discrete time line needs to be formalised in the semantics and managed as transparently as possible.

\subsection{Weakly measured while loops}

Weakly measured while loops are proposed and discussed in Chapter~\ref{chap:while}.
A weakly measured while loop is a classical control flow primitive that offers a trade-off between the collapse of the quantum state caused on each iteration and the amount of information gained.
The trade-off may be adjusted by tuning a parameter and, in certain situations, it is possible to set its value so that quantum speed-up is achieved.
As an example, it is shown that Grover's search problem can be implemented using a weak while loop, maintaining the same time complexity as the standard Grover's algorithm.
In a more general setting, sufficient conditions are provided that let us determine, with arbitrarily high probability, a worst-case estimate of the number of iterations a weakly measured while loop will run for, thus allowing us to study the worst-case time (and query) complexity of quantum algorithms that use this primitive.

\section{Publications and previous work}
\label{sec:publications}

In this section we enumerate the main works from the literature this thesis builds upon; these will be discussed in more detail when relevant throughout the thesis and, in particular, in the last section of each chapter.
In this section we also point out the publications that have come out from this PhD thesis and discuss where novel results presented here may lead to new publications.

\paragraph*{Categories for infinitary addition.} This topic is explored in Chapter~\ref{chap:Sigma} and builds upon the notion of \(\Sigma\)-monoids proposed by Haghverdi~\cite{Haghverdi} and results and proof strategies introduced by Hoshino~\cite{RTUDC} regarding the category of \(\Sigma\)-monoids. We also consider Hausdorff commutative monoids but only use concepts and results appearing in standard introductory texts on general topology. During the process of connecting \(\Sigma\)-monoids with Hausdorff commutative monoids we took inspiration from Higgs' work on \(\Sigma\)-groups~\cite{Higgs}.

\paragraph*{Categorical study of quantum iteration.} This topic is explored in Chapter~\ref{chap:trace} and builds upon the unique decomposition categories of Haghverdi~\cite{Haghverdi} and their refinement due to Hoshino~\cite{RTUDC}. We provide a new proof for a result shown by Bartha~\cite{Bartha} --- that \(\FdIsometry\) (and, more generally, \(\FdContraction\)) is totally traced using the execution formula --- and, to do so, we use some lemmas from Bartha's work along with the kernel-image trace introduced by Malherbe, Scott and Selinger~\cite{Malherbe}. The transition from \(\FdContraction\) to the category of time-shift invariant quantum processes over discrete time \(\LSIContraction\) is achieved using well-known results on the Fourier transform that can be found in introductory texts on signal processing and engineering. 

Our result that the execution formula in \(\LSIContraction\) is a valid categorical trace is novel. We believe this result may be used to provide semantics for programming languages with quantum control flow (see Section~\ref{sec:semantics} for a toy language). In the coming months we will work on preparing these results in a format suitable for publication.

\paragraph*{Weakly measured while loops.} This topic is explored in Chapter~\ref{chap:while} and builds upon the notion of weak measurement (see~\cite{ToddTutorial} for an introduction). Previous work by Mizel~\cite{Mizel} proposed an algorithm that is fundamentally the same as the weakly measured Grover algorithm we present in Chapter~\ref{chap:while}. What distinguishes our work from Mizel's is that we propose a programming primitive --- the \(\kappa\)-while loop --- and lay the foundations of weakly controlled quantum iteration for the purpose of quantum algorithms, presenting the case of Grover's search as a proof of concept. The results discussed in Chapter~\ref{chap:while} were accepted as a journal publication~\cite{WeakWhileLoop}.

\paragraph*{Universal properties of partial quantum maps.} This topic is tangential to the thesis and, as such, it is not discussed in this text. However, it is worth mentioning that a collaboration with Kaarsgaard and Heunen took place during the final year of my PhD studies, yielding a manuscript~\cite{Robin} that was accepted at the International Workshop on Quantum Physics and Logic (QPL 2022). In this project we succeeded in describing how the category of finite-dimensional \Cstar-algebras and completely positive trace non-increasing maps can be obtained after applying a sequence of constructions characterised by universal properties which start from the category of finite-dimensional Hilbert spaces and unitary maps \(\FdUnitary\).

\section{Preliminaries on category theory (\emph{Preamble})}
\label{sec:CT}

This section is included for the sake of completeness and to introduce the notation used.
The reader is not assumed to have prior knowledge of category theory; however, the section may be too dense in concepts for an uninitiated reader.
There are multiple books providing a pedagogical introduction to category theory and, among them, the book by Heunen and Vicary~\cite{HeunenVicaryBook} is a particularly good match for this thesis, considering both its focus on monoidal categories and the discussion of their relevance in quantum computer science.
For the concepts of (co)limits and adjunctions, the book by Leinster~\cite{Leinster} is recommended.

Composition (of processes, devices, programs, functions, \etc{}) is a fundamental operation often taken for granted.
Given two functions \(f \colon A \to B\) and \(g \colon B \to C\), their composite \(g \circ f\) is obtained by applying \(f\) to the input, then \(g\) to the result:
\begin{equation}
  (g \circ f)(a) = g(f(a)).
\end{equation}
If we see \(- \circ -\) as a binary operation over the set of functions, we find that it is not a total function: not all pairs of functions can be composed together.
At its most basic level, a category is an axiomatisation of composition.

\begin{definition} \label{def:category}
  A \gls{category} \(\C\) is comprised of a collection of \emph{objects} \(\Ob(\C)\) and, for each pair of objects \(A,B \in \Ob(\C)\), a collection of \emph{morphisms} \(\C(A,B)\), together with a \emph{composition} operation 
  \begin{equation*}
    \circ \colon \C(B,C) \times \C(A,B) \to \C(A,C)
  \end{equation*}
  for all \(A,B,C \in \Ob(\C)\) satisfying the following axioms.
  \begin{itemize}
    \item \emph{Identities}: For every object \(A\), there is a special morphism \(\id_A \in \C(A,A)\) so that for all \(A,B \in \Ob(\C)\) and for all \(f \in \C(A,B)\),
    \begin{equation}
      f \circ \id_A = f = \id_B \circ f.
    \end{equation}
    \item \emph{Associativity}: For all \(A,B,C,D \in \Ob(\C)\) and all \(f \in \C(A,B)\), \(g \in \C(B,C)\)and \(h \in \C(C,D)\),
    \begin{equation}
      (h \circ g) \circ f = h \circ (g \circ f).
    \end{equation}
  \end{itemize}
  For any morphism \(f \in \C(A,B)\) we refer to \(f \colon A \to B\) as its \emph{type}.
  To reduce clutter, \(A \in \C\) is often used to indicate that \(A\) is an object of \(\C\).
\end{definition}

\begin{example}
  The category \GLS{\(\Set\)}{Set} has sets as objects and each \(\Set(A,B)\) is the set of all functions of domain \(A\) and codomain \(B\). Composition of functions \((g \circ f)(a) = g(f(a))\) is clearly associative and its identities \(\id_A \colon A \to A\) are the usual identity functions.
\end{example}

\begin{example}
  The category \GLS{\Rel}{Rel} has sets as objects and each \(\Rel(A,B)\) is the set of all relations between set \(A\) and set \(B\), \ie{} \(\Rel(A,B)\) is the powerset of \(A \times B\). The composite of \(\mathcal{R} \colon A \to B\) and \(\mathcal{S} \colon B \to C\) is given as follows:
  \begin{equation*}
    \mathcal{S} \circ \mathcal{R} = \{(a,c) \in A \times C \mid \exists b \in B,\ a\mathcal{R}b \ \text{and}\  b\mathcal{S}c\}.
  \end{equation*}
  Composition is clearly associative and identities are relations \(\id_A = \{(a,a) \mid a \in A\}\).
\end{example}

\begin{example}
  The category \GLS{\(\Mon\)}{Mon} has monoids as objects and each \(\Mon(A,B)\) is the set of all monoid homomorphisms from \(A\) to \(B\). Composition and identities are inherited from \(\Set\).
  Similarly, there is a category \GLS{\(\CMon\)}{CMon} of commutative monoids, a category \GLS{\(\Grp\)}{Grp} of groups and a category \GLS{\(\Ab\)}{Ab} of abelian groups, all of which have monoid homomorphisms as their morphisms.
\end{example}

\begin{example}
  Fix a field \(K\); the category \(\Vect_K\) has vector spaces over the field \(K\) as objects and each \(\Vect_K(V,U)\) is the set all linear maps from the vector space \(V\) to the vector space \(U\). Composition is inherited from \(\Set\), along with identities.
  Whenever \GLS{\(\Vect\)}{Vect} appears in this text without subscript it should be understood that \(K = \mathbb{C}\).
\end{example}

\begin{remark}
  There is some subtlety hidden in the word \emph{collection} in the definition of a category. If we consider the example of the category \(\Set\), we cannot say that \(\Ob(\Set)\) is a set itself, as we would run into Russell's paradox.\footnote{Loosely speaking, if it were a set, it would be the ``set containing all sets". Thus, it would need to be contained in itself, which creates a great deal of problems.}
  The way out of this conundrum is to define some sort of hierarchy of collections, where what we commonly refer to as a set is at the base level, and the collection of all sets is one step up the ladder.
  A category where the collection of objects is actually a set is called a \gls{small category}.
  Similarly, it can happen that a category \(\C\) satisfies that for all objects \(A,B \in \C\), the collection \(\C(A,B)\) is actually a set; then \(\C\) is said to be a \gls{locally small category}.
  \(\Set\), \(\Rel\) and \(\Vect\) are all locally small categories, but not small.
  All of the categories discussed in this thesis are locally small.
  When working with locally small categories, we refer to each set \(\C(A,B)\) as a \gls{hom-set}.
\end{remark}

\begin{definition} \label{def:diagram}
  Let \(\C\) be a category. A \emph{diagram} is a graph where the vertices are objects in \(\C\) and the arrows are morphisms in the direction determined by their type. For instance,
  \[\begin{tikzcd}
    A && B \\
    & C
    \arrow["f", from=1-1, to=1-3]
    \arrow["g"', from=1-1, to=2-2]
    \arrow["h"', from=2-2, to=1-3]
  \end{tikzcd}.\]
  A \gls{commuting diagram} satisfies that composing all morphisms along a path results in the same morphism obtained along any other path between the same endpoints.
  In particular, the diagram above commutes if and only if \(f = h \circ g\).
\end{definition}

From the examples above, it is reasonable to consider a notion of mapping between categories so that, for instance, we may say that all monoid homomorphisms are functions, and that all functions are relations.

\begin{definition}
  Let \(\C\) and \(\D\) be two categories. A \gls{functor} \(F \colon \C \to \D\) is comprised of a mapping between objects so that if \(A \in \C\) then \(F(A) \in \D\), and a mapping between morphisms so that if \(f \in \C(A,B)\) then \(F(f) \in \D(F(A),F(B))\). For it to be a functor, \(F\) must preserve composition and identities:
  \begin{align*}
    F(g \circ f) &= F(g) \circ F(f) \\
    F(\id_A) &= \id_{F(A)}.
  \end{align*}
  A \gls{faithful functor} satisfies that the mapping \(\C(A,B) \to \D(F(A),F(B))\) is injective for each pair of objects \(A,B \in \C\).
  A \gls{full functor} satisfies that the mapping \(\C(A,B) \to \D(F(A),F(B))\) is surjective for each pair of objects \(A,B \in \C\).
\end{definition}

\begin{example} \label{ex:Mon_to_Set}
  There is a faithful functor \(\Mon \to \Set\) that sends each monoid to its underlying set, and each homomorphism to its underlying function. Because composition in \(\Mon\) is defined in the same manner as in \(\Set\), proving that this is a functor is trivial.
\end{example}

\begin{example} \label{ex:Ab_to_CMon}
  There is a full and faithful functor \(\Ab \to \CMon\) that sends each abelian group to its underlying monoid --- which will automatically be commutative --- and acts as the identity on morphisms.
  The morphisms in \(\Ab\) are the same as those in \(\CMon\), so it is trivial to show that this is indeed a functor.
\end{example}

\begin{example}
  There is a faithful functor \(\Set \to \Rel\) that sends each set to itself and each function \(f \colon A \to B\) to its \emph{graph}, \ie{} to the relation:
  \begin{equation*}
    \mathcal{R}_f = \{(a,b) \in A \times B \mid f(a) = b\}.
  \end{equation*}
  The graph of the identity function matches the identity relation; it is easy to check that composition is preserved.
\end{example}

\begin{definition}
  Let \(\C\) and \(\D\) be categories. We say that \(\C\) is a \gls{subcategory} of \(\D\) if and only if \(\Ob(\C)\) is a subcollection of \(\Ob(\D)\) and \(\C(A,B)\) is a subcollection of \(\D(A,B)\) for each pair of objects \(A,B \in \C\) and composition and identities in \(\C\) coincide with those in \(\D\).
  If \(\C\) is a subcategory of \(\D\) then there is a canonical faithful functor \(\C \into \D\) that acts as the identity on objects and morphisms.
  If the functor is also full, we refer to \(\C\) as a \emph{full subcategory} of \(\D\).
\end{definition}

\begin{example}
  \(\Set\) is a subcategory of \(\Rel\). Both \(\CMon\) and \(\Grp\) are full subcategories of \(\Mon\). \(\Ab\) is a full subcategory both of \(\CMon\) and \(\Grp\).
\end{example}

The usefulness of functors goes beyond the simple case of embedding a category into a larger one; intuitively, a functor \(F \colon \C \to \D\) indicates that there is a way to go from \(\C\) to \(\D\) preserving the structure in the category \(\C\).
Currently, the structure of a category is only determined by its composition but, soon enough, richer flavours of categories will be discussed and, along them, there will be refined notions of functors preserving the new structure.

It will often be useful to discuss how pairs of functors \(F \colon \C \to \D\) and \(G \colon \C \to \D\) are related, this is done via natural transformations.

\begin{definition}
  Let \(\C\) and \(\D\) be categories, let \(F \colon \C \to \D\) and \(G \colon \C \to \D\) be functors. A \gls{natural transformation} \(F \xto{\alpha} G\) is a collection of morphisms \(\alpha_A \colon F(A) \to G(A)\) for each \(A \in \C\) such that the following diagram commutes for every \(f \in \C(A,B)\),
  \[\begin{tikzcd}
    {F(A)} && {G(A)} \\
    \\
    {F(B)} && {G(B)}
    \arrow["{\alpha_B}", from=3-1, to=3-3]
    \arrow["{\alpha_A}", from=1-1, to=1-3]
    \arrow["{F(f)}"', from=1-1, to=3-1]
    \arrow["{G(f)}", from=1-3, to=3-3]
  \end{tikzcd}\]
  We say \(\alpha\) is a \gls{natural isomorphism} if, additionally, all morphism \(\alpha_A\) are invertible.\footnote{A morphism \(f \colon A \to B\) is invertible if there is another morphism \(g \colon B \to A\) in the category such that \(g \circ f = \id_A\) and \(f \circ g = \id_B\). Such a morphism \(g\) is often denoted \(f^{-1}\) and referred to as the inverse of \(f\). It is straightforward to check that inverses are unique.}
\end{definition}

An important concept in category theory is the notion of limit; their definition is sketched below.
For a proper introduction to categorical limits see Chapter 5 of Leinster's book~\cite{Leinster}.

\begin{definition} \label{def:cat_limit}
  Let \(\C\) be a category. For any arbitrary diagram in \(\C\) (see Definition~\ref{def:diagram}), let \(\{A_i \in \C\}_{i \in I}\) be the set of objects in it, \eg{}
  \[\begin{tikzcd}
    & {A_0} \\
    {A_1} & {A_2.}
    \arrow["{f'}"', from=2-1, to=2-2]
    \arrow["f", from=1-2, to=2-2]
  \end{tikzcd}\]
  A \emph{cone} is an object \(X \in \C\) along with a set of morphisms \(\{g_i \colon X \to A_i\}_{i \in I}\) such that, when these are included in the diagram, \eg{}
  \[\begin{tikzcd}
    X & {A_0} \\
    {A_1} & {A_2}
    \arrow["{f'}"', from=2-1, to=2-2]
    \arrow["f", from=1-2, to=2-2]
    \arrow["{g_\dsub{1}}"', from=1-1, to=2-1]
    \arrow["{g_\dsub{0}}", from=1-1, to=1-2]
    \arrow["{g_\dsub{2}}"{description}, from=1-1, to=2-2]
  \end{tikzcd}\]
  it is satisfied for each \(i \in I\) that all paths from \(X\) to \(A_i\) yield the same morphism.
  A \gls{categorical limit} is a cone \(\{h_i \colon L \to A_i\}_{i \in I}\) such that for any other cone \(\{g_i \colon X \to A_i\}_{i \in I}\) there is a unique morphism \(m \colon X \to L\) satisfying \(g_i = h_i \circ m\) for all \(i \in I\).
\end{definition}

The limit of certain diagrams have special names; for instance, for any two objects \(A,B \in \C\) the limit of the diagram
\[\begin{tikzcd}
  A & B
\end{tikzcd}\]
is a \emph{categorical product} and, for any two morphisms \(f,g \colon A \to B\), the limit of the diagram
\[\begin{tikzcd}
  A & B
  \arrow["f", shift left=2, from=1-1, to=1-2]
  \arrow["g"', shift right=2, from=1-1, to=1-2]
\end{tikzcd}\]
is known as an \emph{equalizer}.
In general, not every diagram in a category \(\C\) has a limit; when it does, we say that \(\C\) is a \gls{complete category}.
A \emph{colimit} is the dual notion of a limit, defined by inverting the direction of all morphisms in Definition~\ref{def:cat_limit}.
The dual of a product is a coproduct and the dual of an equalizer is a coequalizer.
An object that is both a product and a coproduct and satisfies certain additional algebraic identities is known as a \emph{biproduct}.

\subsection{Monoidal categories}

We will often find that `sequential' composition of morphisms \(- \circ -\) is not the only reasonable notion of composition that may be considered.
The notion of `parallel' composition is formalised in monoidal categories.

\begin{definition}
  Let \(\C\) be a category, let \(\otimes \colon \C \times \C \to \C\) be a functor and let \(I\) be an object in \(\C\). The triple \((\C,\otimes,I)\) is a \gls{monoidal category} if there are natural isomorphisms \(\alpha\), \(\lambda\) and \(\rho\) --- known respectively as associator, left unitor and right unitor --- such that the diagrams below commute for all choices of objects \(A,B,C,D \in \C\).
  We refer to \(\otimes\) as the \gls{monoidal product} and \(I\) as the \gls{monoidal unit}.
  \[\begin{tikzcd}[column sep=small]
    {(A \otimes I) \otimes B} && {A \otimes (I \otimes B)} \\
    & {A\otimes B}
    \arrow["{\alpha_\dsub{A,I,C}}", from=1-1, to=1-3]
    \arrow["{\rho_\dsub{A} \otimes \id_\dsub{B}}"', from=1-1, to=2-2]
    \arrow["{\id_\dsub{A} \otimes \lambda_\dsub{B}}", from=1-3, to=2-2]
  \end{tikzcd}\]
  \[\begin{tikzcd}[column sep=tiny]
    & {(A\otimes B) \otimes (C \otimes D)} \\
    {((A \otimes B) \otimes C) \otimes D} && {A \otimes (B \otimes (C \otimes D))} \\
    \\
    {(A \otimes (B \otimes C)) \otimes D} && {A \otimes ((B \otimes C) \otimes D)}
    \arrow["{\alpha_\dsub{A,B \otimes C,D}}"', from=4-1, to=4-3]
    \arrow["{\alpha_\dsub{A \otimes B,C,D}}", from=2-1, to=1-2]
    \arrow["{\alpha_\dsub{A,B,C} \otimes \id_\dsub{D}}"', from=2-1, to=4-1]
    \arrow["{\alpha_\dsub{A,B,C \otimes D}}", from=1-2, to=2-3]
    \arrow["{\id_\dsub{A} \otimes \alpha_\dsub{B,C,D}}"', from=4-3, to=2-3]
  \end{tikzcd}\]
\end{definition}

The requirement that \(\otimes \colon \C \times \C \to \C\) is a functor imposes an interchange law between \(\circ\) and \(\otimes\):
\begin{equation} \label{eq:monoidal_interchange_law}
  (g \otimes k) \circ (f \otimes h) = (g \circ f) \otimes (k \circ h).
\end{equation}
The intuition behind this equation is provided in Figure~\ref{fig:monoidal_interchange_law}.
The first of the commuting diagrams above imposes that the object \(I\) acts as the unit of `parallel' composition \(- \otimes -\), whereas the second commuting diagram imposes that \(\otimes\) is associative.

\begin{figure}
\centering
\input{Figures/1/monoidal_interchange_law}
\caption{Pictorial representation of the interchange law~\eqref{eq:monoidal_interchange_law}. It does not matter whether we first compose in parallel or sequentially, the result is the same.}
\label{fig:monoidal_interchange_law}
\end{figure}
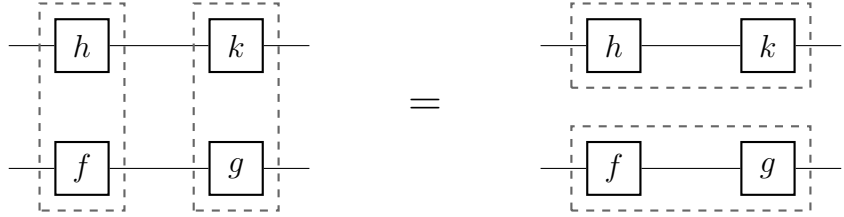

\begin{remark}
  By convention, morphisms in monoidal categories are represented pictorially as in Figure~\ref{fig:monoidal_interchange_law}. Morphisms are depicted as labelled boxes connected via wires when composed (using \(\circ\)) and, when combined using the monoidal product, the morphisms are drawn one on top of the other.
  Thanks to the coherence theorem of monoidal categories (see Section 1.3.4 from~\cite{HeunenVicaryBook}), we may prove facts about morphisms in a monoidal category \(\C\) by manipulating the pictorial representation of their morphisms.
  This is known as the \gls{graphical calculus}; its main advantage is that associators, unitors and the interchange law~\eqref{eq:monoidal_interchange_law} become trivial when represented pictorially, removing unnecessary verbosity from definitions and proofs.
\end{remark}

\begin{example}
  Let \(\times \colon \Set \times \Set \to \Set\) be the functor acting on objects as the Cartesian product of sets. On morphisms \(f \colon A \to B\) and \(g \colon C \to D\), the functor yields a function \(f \times g\) that maps each \((a,c) \in A \times C\) to \((f(a),g(c)) \in B \times D\).
  Let \(\{\bullet\}\) be an arbitrary singleton set; then, \((\Set,\times,\{\bullet\})\) is a monoidal category.
\end{example}

\begin{example}
  Let \(+ \colon \Set \times \Set \to \Set\) be the functor acting on objects as the disjoint union of sets. On morphisms \(f \colon A \to B\) and \(g \colon C \to D\), the functor yields a function \(f + g\) that maps each \(a \in A\) to \(f(a)\) and each \(c \in C\) to \(g(c)\).
  It can be shown that \((\Set,+,\varnothing)\) is a monoidal category.
\end{example}

\begin{definition}
  Given two vector spaces \(A,B \in \Vect_K\), their tensor product \(A \otimes B\) is a vector space characterised (up to isomorphism) by the property that there is a bilinear map\footnote{A bilinear map is a map \(f \colon A \times B \to C\) such that, for all \(a \in A\) and all \(b \in B\), both \(f(a,-)\) and \(f(-,b)\) are linear.} \(p \colon A \times B \to A \otimes B\) such that, for any other bilinear map \(f \colon A \times B \to C\), there is a unique linear map \(g \colon A \otimes B \to C\) making the diagram
  \[\begin{tikzcd}
    {A \times B} && {A \otimes B} \\
    && C
    \arrow["p", from=1-1, to=1-3]
    \arrow["f"', from=1-1, to=2-3]
    \arrow["g", dashed, from=1-3, to=2-3]
  \end{tikzcd}\]
  commute in \(\Set\).
  For any two linear maps \(f \colon A \to B\) and \(g \colon C \to D\), the map \(f \otimes g \colon A \otimes C \to B \otimes D\) is defined as the linear extension of the map \(a \otimes c \mapsto f(a) \otimes g(c)\).
\end{definition}

\begin{example}
  Let \(\otimes \colon \Vect_K \times \Vect_K \to \Vect_K\) be the functor acting as the tensor product on vector spaces and linear maps. 
  It can be shown that \((\Vect_K,\otimes,K)\) is a monoidal category.
\end{example}

\begin{definition}
  Given two vector spaces \(A,B \in \Vect_K\), their \gls{direct sum} \(A \oplus B\) is the vector space on the Cartesian product of their underlying sets, with coordinate-wise addition and scalar multiplication.
  For any two linear maps \(f \colon A \to B\) and \(g \colon C \to D\), the linear map \(f \oplus g \colon A \oplus C \to B \oplus D\) is defined as follows:
  \begin{equation*}
    (f \oplus g)(a,c) = (f(a),g(c)).
  \end{equation*}
\end{definition}

\begin{example}
  Let \(\oplus \colon \Vect_K \times \Vect_K \to \Vect_K\) be the functor acting as the direct sum of vector spaces and linear maps and let \(\{0\}\) be the zero-dimensional vector space. 
  It can be shown that \((\Vect_K,\oplus,\{0\})\) is a monoidal category.
\end{example}

We have seen two distinct monoidal structures both for \(\Set\) and \(\Vect\). Both monoidal structures in \(\Vect\) are relevant to this thesis, with the \(\otimes\)-monoidal structure being the usual focus in the literature of categorical quantum mechanics~\cite{HeunenVicaryBook,TheDodo}.
However, when discussing the control flow of programs, the monoidal structure induced by the direct sum is the one to take the spotlight and, hence, the one that will be most relevant to this thesis.

\begin{definition}
  Let \((\C,\otimes,I)\) be a monoidal category and let \(\sigma\) be a natural isomorphism with components
  \begin{equation*}
    \sigma_{A,B} \colon A \otimes B \to B \otimes A.
  \end{equation*}
  We say \(\C\) is a \gls{braided monoidal category} if \(\sigma\) satisfies the coherence axioms given below.
  We refer to \(\sigma\) as braiding and represent it pictorially as the crossing of wires; the coherence axioms are represented pictorially as:
  \[\input{Figures/1/hexagon_equations}\]
  We say \(\C\) is a \gls{symmetric monoidal category} if, additionally, \(\sigma_{A,B}^{-1} = \sigma_{B,A}\) for all pairs of objects.
\end{definition}

Both \(\Set\) and \(\Vect\) are symmetric monoidal categories with either of the monoidal structures discussed in the examples above.
Along each new flavour of categories comes a refined notion of structure-preserving functor.

\begin{definition}
  Let \((\C,\otimes_\C,I_\C)\) and \((\D,\otimes_\D,I_\D)\) be monoidal categories and let \(F \colon \C \to \D\) be a functor. Let \(\mu\) be a natural transformation with components
  \begin{equation*}
    \mu_{A,B} \colon F(A) \otimes_\D F(B) \to F(A \otimes_\C B)
  \end{equation*}
  and let \(\mu_I \colon I_\D \to F(I_\C)\) be a morphism in \(\D\).
  We say \((F,\mu)\) is a \emph{lax} \gls{monoidal functor} when \(\mu\) commutes with the associators and unitors of \(\C\) and \(\D\) in an appropriate way, made precise in reference texts (see equations (1.29) and (1.30) from~\cite{HeunenVicaryBook}) and omitted here for brevity.
  If \(\mu_I\) and all \(\mu_{A,B}\) are isomorphisms we say that \((F,\mu)\) is a \emph{strong} monoidal functor.
  Similarly, a \emph{braided monoidal functor} is a monoidal functor \((F,\mu) \colon \C \to \D\) between braided monoidal categories such that if \(\sigma\) is the braiding in \(\C\) then \(F(\sigma)\) is the braiding in \(\D\).
  A \emph{symmetric monoidal functor} is a braided monoidal functor between symmetric monoidal categories.
\end{definition}

There are many other flavours of categories, each of them adding more structure by introducing new operations and axioms.
Category theorists define and classify different flavours of categories, they describe how different categories are related to each other and, most importantly, they study how properties arise from the structure that is imposed.

\subsection{Enriched category theory}
\label{sec:enriched_cat}

In this thesis, it will often be required to combine the contribution of two morphisms \(f,g \colon A \to B\) to describe another morphism of type \(A \to B\).
Abstractly, for some category \(\C\), this can be framed as requiring that each collection of morphisms \(\C(A,B)\) is endowed with an operation
  \[+ \colon \C(A,B) \times \C(A,B) \to \C(A,B)\]
making each hom-set \((\C(A,B),+)\) a commutative monoid.
When hom-sets are endowed with a mathematical structure and composition interacts nicely with it, we say \(\C\) is an enriched category, formally defined below.

\begin{definition}
  Let \((\V,\otimes,I)\) be a monoidal category. A \(\V\)-\gls{enriched category} \(\C\) is comprised of a collection of objects \(\Ob(\C)\) and, for each pair of objects \(A,B \in \Ob(\C)\) a \gls{hom-object} \(\C(A,B) \in \Ob(\V)\).
  Composition and identities in \(\C\) are determined by morphisms in \(\V\),
  \begin{align*}
    \circ_{A,B,C} &\colon \C(B,C) \otimes \C(A,B) \to \C(A,C) \\
    1_A &\colon I \to \C(A,A)
  \end{align*}
  satisfying the following commuting diagrams in \(\V\), which ensure that composition in \(\C\) is associative and that the identities are its units.
  \[\begin{tikzcd}
    {(\C(C,D) \otimes \C(B,C)) \otimes \C(A,B)} && {\C(B,D) \otimes \C(A,B)} \\
    && {\C(A,D)} \\
    {\C(C,D) \otimes (\C(B,C) \otimes \C(A,B))} && {\C(C,D) \otimes \C(A,C)}
    \arrow["\alpha"', from=1-1, to=3-1]
    \arrow["{\circ_\dsub{B,C,D}\, \otimes\, \id}", from=1-1, to=1-3]
    \arrow["{\id\, \otimes\, \circ_\dsub{A,B,C}}", from=3-1, to=3-3]
    \arrow["{\circ_\dsub{A,B,D}}", from=1-3, to=2-3]
    \arrow["{\circ_\dsub{A,C,D}}"', from=3-3, to=2-3]
  \end{tikzcd}\]
  \[\begin{tikzcd}
    {\C(A,B) \otimes \C(A,A)} & {\C(A,B)} & {\C(B,B) \otimes \C(A,B)} \\
    {\C(A,B) \otimes I} && {I \otimes \C(A,B)}
    \arrow["{\circ_\dsub{A,A,B}}", from=1-1, to=1-2]
    \arrow["{\circ_\dsub{A,B,B}}"', from=1-3, to=1-2]
    \arrow["{\id\, \otimes\, 1_\dsub{A}}", from=2-1, to=1-1]
    \arrow["\rho"', from=2-1, to=1-2]
    \arrow["\lambda", from=2-3, to=1-2]
    \arrow["{1_\dsub{B}\, \otimes\, \id}"', from=2-3, to=1-3]
  \end{tikzcd}\]
\end{definition}

An immediate example of an \(\Ab\)-enriched category is \(\Vect\); the proof is sketched in Example~\ref{ex:vect_ab-enriched} below.
The monoidal structure in \(\Ab\) must be chosen with care: even though the cartesian product of groups provides a valid monoidal structure, it is \emph{not} the appropriate one to be used for the enrichment.\footnote{Given two abelian groups \(A\) and \(B\), their cartesian product \(A \times B\) has its group operation defined as \((a,b) + (a',b') = (a+a',b+b')\). If we used this monoidal structure for enrichment, we would have that \[h \circ (f+g) = \circ(h,f+g) = \circ((h,f) + (0,g)) = h \circ f + 0\] which is not the correct way composition and addition of linear maps interact.}
Instead, the monoidal category to be used is \((\Ab,\otimes,\Zset)\) where \(\otimes\) is the tensor product of abelian groups and \(\Zset\) is the group of integers with addition.
The explicit definition of this tensor product can be found in~\cite{AbTensor}; for the purposes of this thesis, it suffices to say that \(A \otimes B\) is special in that the set of linear maps in \(\Ab(A \otimes B, C)\) is in one-to-one correspondence with the set of functions \(f \colon A \times B \to C\) satisfying
\begin{equation} \label{eq:Ab_bihom}
  \begin{aligned}
    f(a,b+b') = f(a,b) + f(a,b'), &\quad\quad f(a,0) = 0,\\
    f(a+a',b) = f(a,b) + f(a',b), &\quad\quad f(0,b) = 0.
  \end{aligned}
\end{equation}

\begin{example} \label{ex:vect_ab-enriched}
  \(\Vect\) is an \(\Ab\)-enriched category where for any two linear maps \(f,g \colon U \to V\) their addition is defined pointwise, \ie{} for all \(u \in U\):
  \begin{equation*}
    (f + g)(u) = f(u) + f(v).
  \end{equation*}
  Indeed, this operation is associative and commutative, there is a linear map sending every vector to the zero vector --- thus acting as the addition's neutral element --- and every linear map \(f\) has an additive inverse \(u \mapsto -f(u)\), making the set of all linear maps \(U \to V\) an abelian group.
  Importantly, asking that composition in \(\Vect\) is a morphism in \(\Ab\) along with~\eqref{eq:Ab_bihom} implies that composition distributes over addition, as expected:
  \begin{align*}
    h \circ (f + g) &= h \circ f + h \circ g, \\
    (f + g) \circ h &= f \circ h + g \circ h.
  \end{align*}
\end{example}

Chapter~\ref{chap:Sigma} will define multiple \(\SCat{*}\) categories capturing the notion of infinitary sum; an appropriate notion of tensor product of objects in these categories will be given in Section~\ref{sec:SCat_tensor}.
Then, Chapter~\ref{chap:trace} will define \(\SCat{*}\)-enriched categories and use their notion of infinitary sum of morphisms to formalise iterative loops.

\subsection{Adjunctions}

Example~\ref{ex:Mon_to_Set} discussed the existence of a functor \(\Mon \to \Set\) that `forgets' the monoid structure, sending monoids and homomorphisms to their underlying sets and functions.
A similar forgetful functor \(\CMon \to \Set\) exists as well; we may entertain the idea of a functor on the opposite direction \(\Set \to \CMon\) as described in the following example.

\begin{example} \label{ex:Set_to_CMon}
  Let \(A\) be an arbitrary set and let \(F(A)\) be the collection of all finite multisets of elements in \(A\).
  There is a functor \(F \colon \Set \to \CMon\) that maps each set \(A\) to the commutative monoid on the set \(F(A)\) with disjoint union as its monoid operation: disjoint union is associative and commutative, with \(\varnothing \in M(A)\) acting as its neutral element.
  Each function \(f \colon A \to B\) is lifted to a monoid homomorphism \(F(f) \colon F(A) \to F(B)\) that maps each multiset \(\{a,b, \ldots\} \in F(A)\) to \(\{f(a),f(b), \ldots\} \in F(B)\).
  It is trivial to check that \(F(g \circ f) = F(g) \circ F(f)\) and \(F(\id_A) = \id_{F(A)}\) so \(F\) is indeed a functor.
\end{example}

This functor \(F \colon \Set \to \CMon\) satisfies a powerful property: it is left adjoint to the canonical forgetful functor \(G \colon \CMon \to \Set\).
The concept of categorical adjunction is defined below in terms of the initial object of a comma category.

\begin{definition} \label{def:comma}
  Let \(\A\), \(\B\) and \(\C\) be arbitrary categories and let \(F \colon \A \to \C\) and \(G \colon \B \to \C\) be arbitrary functors.
  The \gls{comma category} \(\comma{F}{G}\) is defined as follows:
  \begin{itemize}
    \item objects are triples \((A,h,B)\) where \(A\) and \(B\) are objects in \(\A\) and \(\B\) respectively and \(h \colon F(A) \to G(B)\) is a morphism in \(\C\);
    \item morphisms \((A,h,B) \to (A',h',B')\) are pairs \((f,g)\) where \(f \colon A \to A'\) is a morphism in \(\A\) and \(g \colon B \to B'\) is a morphism in \(\B\) such that the diagram
    \[\begin{tikzcd}
      {F(A)} & {F(A')} \\
      {G(B)} & {G(B')}
      \arrow["{F(f)}", from=1-1, to=1-2]
      \arrow["{G(g)}"', from=2-1, to=2-2]
      \arrow["h"', from=1-1, to=2-1]
      \arrow["{h'}", from=1-2, to=2-2]
    \end{tikzcd}\]
  \end{itemize}
  commutes.
\end{definition} 

Let \(\cat{1}\) be the category with a single object and a single morphism; for any arbitrary object \(A \in \C\) we may define a functor \(A \colon \cat{1} \to \C\) that maps the unique object in \(\cat{1}\) to \(A\) and maps its unique morphism to \(\id_A\).
Let \(G \colon \D \to \C\) be a functor; the comma category \(\comma{A}{G}\) is used often throughout the thesis.
To reduce clutter, objects of \(\comma{A}{G}\) are denoted by pairs \((h,B)\) --- since the domain of \(h\) is guaranteed to be \(A\) --- and morphisms \((\id_A,g)\) are simply denoted \(g\).

\begin{definition} \label{def:adjunction}
  Let \(\C\) and \(\D\) be categories, let \(1_\C \colon \C \to \C\) be the identity functor and let \(F \colon \C \to \D\) and \(G \colon \D \to \C\) be arbitrary functors.
  We say that \(F\) is \gls{left adjoint} to \(G\), denoted \(F \dashv G\), if there is a natural transformation \(1_\C \xto{\eta} GF\) such that, for every \(A \in \C\), the object \((\eta_A,F(A))\) of the comma category \(\comma{A}{G}\) is initial --- \ie{} for each object \((h,B) \in \comma{A}{G}\) there is a unique morphism \((\eta_A,F(A)) \to (h,B)\).
  The triple \((F,G,\eta)\) is known as an \gls{adjunction}, with \(\eta\) being its \emph{unit}; equivalently, we say that \(G\) is \emph{right adjoint} to \(F\).
\end{definition}

Unwrapping the definition we find that a unique morphism \((\eta_A,F(A)) \to (h,B)\) exists if and only if there is a unique morphism \(\bar{h} \colon F(A) \to B\) in \(\D\) that makes the diagram
\[\begin{tikzcd}
  A && {G(B)} \\
  && {GF(A)}
  \arrow["h", from=1-1, to=1-3]
  \arrow["{\eta_A}"', from=1-1, to=2-3]
  \arrow["{G(\bar{h})}"', dashed, from=2-3, to=1-3]
\end{tikzcd}\]
commute.

\begin{example} \label{ex:CMon_Set_adj}
  Let \(F \colon \Set \to \CMon\) be the functor from Example~\ref{ex:Set_to_CMon} and let \(G \colon \CMon \to \Set\) be the canonical forgetful functor; there is an adjunction \(F \dashv G\).
  For every set \(A\), the function \(\eta_A \colon A \to GF(A)\) maps each \(a \in A\) to \(\{a\} \in GF(A)\); it is straightforward to check that \(1_\Set \xto{\eta} GF\) is a natural transformation.
  Let \((B,+)\) be an arbitrary commutative monoid; for any function \(f \colon A \to B\) the requirement that the diagram
  \[\begin{tikzcd}
    A && {G(B)} \\
    && {GF(A)}
    \arrow["f", from=1-1, to=1-3]
    \arrow["{\eta_A}"', from=1-1, to=2-3]
    \arrow["{G(\bar{f})}"', dashed, from=2-3, to=1-3]
  \end{tikzcd}\]
  commutes implies that \(\bar{f}\) must map each \(\eta_A(a) = \{a\}\) to \(f(a)\).
  Moreover, the requirement that \(\bar{f}\) is a monoid homomorphism implies that
  \begin{equation*}
    \bar{f}(\{a,a'\}) = \bar{f}(\{a\} \uplus \{a'\}) = \bar{f}(\{a\}) + \bar{f}(\{a'\}) = f(a) + f(a').
  \end{equation*}
  This may be extended to any element in \(F(A)\), \ie{} any finite multiset of elements in \(A\).
  Such a monoid homomorphism \(\bar{f}\) is well-defined and it is the unique one satisfying the commuting diagram above so, indeed, \((F,G,\eta)\) is an adjunction.
\end{example}

\begin{example} \label{ex:Ab_CMon_adj}
  The canonical embedding functor \(G \colon \Ab \into \CMon\) from Example~\ref{ex:Ab_to_CMon} has a left adjoint \(F \colon \CMon \to \Ab\) explicitly defined below.
  \begin{itemize}
    \item Let \((A,+)\) be a commutative monoid and define an equivalence relation \(\sim\) such that for all pairs \((a_+,a_-),(b_+,b_-) \in A \times A\):
    \begin{equation*}
      (a_+,a_-) \sim (b_+,b_-) \iff \exists z \in A \st a_+ + b_- + z = b_+ + a_- + z.
    \end{equation*}
    \item \(F(A)\) is the abelian group \(((A \times A)/{\sim},+')\) where
    \begin{equation*}
      [(a_+,a_-)] +' [(b_+,b_-)] = [(a_+ + b_+,a_- + b_-)].
    \end{equation*}
    It can be shown that the result of \(+'\) is independent of the choice of representatives and, hence, it is well-defined.
    Then, \(F(A)\) is an abelian group whose neutral element is \([(0,0)]\), while the inverse of any \([(a_+,a_-)]\) is \([(a_-,a_+)]\).
    \item Let \(f \colon A \to B\) be a monoid homomorphism, then \(F(f)\) is the monoid homomorphism that maps each \([(a_+,a_-)] \in F(A)\) to \([(f(a_+),f(a_-))] \in F(B)\).
  \end{itemize}
  For each \(A \in \CMon\) the morphism \(\eta_A \colon A \to GF(A)\) maps each \(a \in A\) to \([(a,0)] \in F(A)\).
  Let \((B,+)\) be a group; for any monoid homomorphism \(f \colon A \to G(B)\), the corresponding group homomorphism \(\bar{f} \colon F(A) \to B\) is
  \begin{equation}
    \bar{f}[(a_+,a_-)] = g(a_+) - g(a_-).
  \end{equation}
  It is straightforward to check that \(\eta\) is a natural transformation and that \(\bar{f}\) is the unique group homomorphism such that \(f = G(\bar{f}) \circ \eta_A\), implying that \(F \dashv G\).
\end{example}

The functor \(FG \colon \Ab \to \Ab\) obtained by `forgetting' the existence of inverses --- via \(G \colon \Ab \into \CMon\) --- then recovering them using \(F \colon \CMon \to \Ab\) yields the original abelian group up to isomorphism.
This is the case whenever a full and faithful functor has a left adjoint.
In particular, the embedding functor of every full subcategory is full and faithful; thus, full subcategories whose embedding has a left adjoint are of particular interest.

\begin{definition}
  Let \(\C\) be a category and let \(\D\) be a subcategory of \(\C\) with \(G \colon \D \into \C\) being its corresponding embedding.
  We say \(\D\) is a \gls{reflective subcategory} of \(\C\) if \(G\) is full and it has a left adjoint.  
\end{definition}

Thus, \(\Ab\) is a reflective subcategory of \(\CMon\).
Multiple examples of reflective subcategories will be discussed in Chapter~\ref{chap:Sigma}.

In Examples~\ref{ex:CMon_Set_adj} and~\ref{ex:Ab_CMon_adj} the existence of a left adjoint has been proven by defining it explicitly.
However, it is sometimes convenient to establish the existence of an adjunction via non-constructive means since, when working at a higher level of abstraction, cumbersome details may be avoided.
A powerful tool to this end is the general adjoint functor theorem, which is introduced below.
The statement of this result (and its proof) make use of the concept of categorical limit and complete categories; these were introduced in Definition~\ref{def:cat_limit} and are covered in depth in most books on basic category theory (for instance, see Chapter 5 from~\cite{Leinster}).

\begin{definition} \label{def:weakly_initial_set}
  Let \(\C\) be a category. A \gls{weakly initial set} is a set \(S\) of objects in \(\C\) satisfying that for every object \(B \in \C\) there is at least one morphism \(A \to B\) such that \(A \in S\).\footnote{It is essential that it is an actual set, \ie{} a \emph{small} collection.}
\end{definition}

\begin{lemma} \label{lem:weakly_initial_set}
  Let \(\C\) be a complete locally small category with a weakly initial set. Then \(\C\) has an initial object.
\end{lemma} \begin{proof}
  See Lemma A.1 from the book on basic category theory by Leinster~\cite{Leinster}.
\end{proof}

The previous lemma is the main technical result used in the proof of the general adjoint functor theorem. Moreover, the lemma will be relevant on its own in Section~\ref{sec:SCat_tensor}.

\begin{theorem} \label{thm:GAFT}
  Let \(G \colon \D \to \C\) be a functor. Assume that \(\D\) is complete and locally small and that for each object \(A \in \C\) the comma category \(\comma{A}{G}\) has a weakly initial set. Then,
  \begin{equation*}
    G \text{ has a left adjoint} \iff G \text{ preserves limits}.
  \end{equation*}
\end{theorem} \begin{proof}[Proof. (Sketch)]
  A detailed proof can be found in the appendix of the book on basic category theory by Leinster~\cite{Leinster}.
  One direction is immediate since every right adjoint functor preserves limits (see Theorem 6.3.1 from~\cite{Leinster}).
  To prove the other direction, we assume that \(G\) preserves limits; then, \(\D\) being complete implies that the comma category \(\comma{A}{G}\) is complete (see Lemma A.2 from~\cite{Leinster}) and,
  since \(\D\) is locally small, it follows that \(\comma{A}{G}\) is locally small.
  Then, for all \(A \in \C\) the existence of a weakly initial set in \(\comma{A}{G}\) implies that \(\comma{A}{G}\) has an initial object, as established by Lemma~\ref{lem:weakly_initial_set}.
  According to Corollary 2.3.7 from~\cite{Leinster} if for all \(A \in \C\) the comma category \(\comma{A}{G}\) has an initial object then a functor \(F \colon \C \to \D\) and a natural transformation \(1_\C \xto{\eta} GF\) exist such that \((F,G,\eta)\) is an adjunction.
\end{proof}

In general, it is preferable to provide constructive proofs when attainable since they tend to be more illustrative.
Unfortunately, for some of the left adjoint functors whose existence is established in Chapter~\ref{chap:Sigma} an explicit definition could not be achieved and their existence is instead proven via the general adjoint functor theorem.

\section{Categories in quantum computer science (\emph{Preamble})}
\label{sec:QM_cats}

As discussed in Section~\ref{sec:QC}, the state of a quantum computer is described by a vector in a Hilbert space and operators are described by linear maps.
Therefore, the following categories are of great importance in quantum computing.

\begin{definition}
  Let \GLS{\(\Hilb\)}{Hilb} be the category whose objects are complex Hilbert spaces and whose morphisms are bounded linear maps.
  Let \GLS{\(\Contraction\)}{Contraction} be the subcategory of \(\Hilb\) obtained by restricting morphisms to (weak) contractions; a linear map \(f \colon A \to B\) between normed vector spaces is a (weak) \gls{contraction} if it satisfies \(\norm{f(a)}_B \leq \norm{a}_A\) for all \(a \in A\).\footnote{A \emph{strict} contraction \(f\) satisfies \(\norm{f(a)}_B < \norm{a}_A\) for all nonzero \(a \in A\). In this thesis, the term `contraction' refers to weak contractions unless stated otherwise.}
  Let \GLS{\(\Isometry\)}{Isometry} be the subcategory of \(\Contraction\) obtained by restricting morphisms to isometries; a linear map \(f \colon A \to B\) between normed vector spaces is an isometry if for all \(a \in A\) it satisfies \(\norm{f(a)}_B = \norm{a}_A\).
  Let \GLS{\(\Unitary\)}{Unitary} be the subcategory of \(\Isometry\) obtained by restricting morphisms to unitaries; a linear map is unitary if it is a surjective isometry.
\end{definition}

Evidently, there is a chain of faithful functors:
\begin{equation*}
  \Unitary \into \Isometry \into \Contraction \into \Hilb \to \Vect.
\end{equation*}
Considering that real-world quantum computers have finite memory, quantum computer science tends to deal with categories whose objects are finite-dimensional vector spaces.

\begin{definition}
  Let \(\FdHilb\) be the full subcategory of \(\Hilb\) obtained by restricting objects to finite-dimensional Hilbert spaces.
  Similarly, let \(\FdContraction\), \(\FdIsometry\) and \(\FdUnitary\) be the corresponding full subcategories of \(\Contraction\), \(\Isometry\) and \(\Unitary\), respectively.
\end{definition}

All of these categories can be given a monoidal structure in terms of the (orthogonal) direct sum of Hilbert spaces.

\begin{definition}
  For any two Hilbert spaces \(A,B \in \Hilb\), let \(A \oplus B\) be their direct sum, made into a Hilbert space by the inner product \(\braket{(a,b)}{(a',b')} = \braket{a}{a'} + \braket{b}{b'}\).
\end{definition}

\begin{proposition}
  Let \(\C\) be either \(\Hilb\), \(\Contraction\), \(\Isometry\), \(\Unitary\) or the finite-dimensional version of any of them.
  Then, \((\C,\oplus,\{0\})\) is a symmetric monoidal category.
\end{proposition} \begin{proof}
  It is straightforward to check that \(\oplus\) is a valid monoidal product and that the zero-dimensional Hilbert space \(\{0\}\) is its monoidal unit so that \((\Hilb,\oplus,\{0\})\) is a symmetric monoidal category.
  Let \(f \colon A \to B\) and \(g \colon C \to D\) be morphisms in \(\Contraction\); then, for all \(a \in A\) and \(c \in C\):
  \begin{equation*}
    \norm{(f \oplus g)(a,c)}^2_{B \oplus D} = \norm{f(a)}^2_B + \norm{g(c)}^2_D \leq \norm{a}^2_A + \norm{c}^2_C = \norm{(a,c)}^2_{A \oplus C}
  \end{equation*}
  implying that \(f \oplus g\) is a morphism in \(\Contraction\).
  Thus, \((\Contraction,\oplus,\{0\})\) is a symmetric monoidal subcategory of \((\Hilb,\oplus,\{0\})\).
  Similarly, if \(f\) and \(g\) are isometries then \(f \oplus g\) is an isometry, and if both \(f\) and \(g\) are surjective (so that they are unitary) then \(f \oplus g\) is also surjective and, hence, unitary.
  Consequently, both \((\Isometry,\oplus,\{0\})\) and \((\Unitary,\oplus,\{0\})\) are symmetric monoidal subcategories of \((\Contraction,\oplus,\{0\})\).
  Evidently, each of the subcategories where objects are finite-dimensional are symmetric monoidal as well.
\end{proof}

For all objects \(A,B \in \Hilb\) the object \(A \oplus B \in \Hilb\) is a biproduct.
In contrast, \(A \oplus B\) is not a biproduct in \(\Contraction\) (nor in its subcategories), since the diagonal morphism \(a \mapsto (a,a)\) is not a contraction; but \(A \oplus B\) is still a coproduct both in \(\Contraction\) and \(\Isometry\).
In \(\Unitary\), \(A \oplus B\) is not even a coproduct, since injections \(A \to A \oplus B\) are not surjective and, hence, not unitary.
Even though \(A \oplus B\) is not a product in \(\Contraction\), the projection \(A \oplus B \to A\) from \(\Hilb\) is still a morphism in \(\Contraction\); Haghverdi~\cite{Haghverdi} and other authors have noticed that this is sufficient for morphisms in \(\Contraction\) to have a unique matrix decomposition (more on this in Chapter~\ref{chap:trace}).
This hints at the importance of \(\Contraction\) in quantum computing: even though strict contractions are not physical operations, every isometry and unitary is a morphism in \(\Contraction\) and the entries in their block matrix decomposition are (possibly strict) contractions.

Apart from the monoidal structure induced by direct sum, \(\Hilb\) can be given a monoidal structure using tensor product of Hilbert spaces.

\begin{proposition}
  Let \(A,B \in \Hilb\) and let \(A \boxtimes B\) be the tensor product of their underlying vector spaces made into an inner product space by defining
  \begin{equation*}
    \braket{a \boxtimes b}{a' \boxtimes b'} = \braket{a}{a'} \cdot \braket{b}{b'}.
  \end{equation*}
  The inner product in \(A \boxtimes B\) induces a norm, making it a metric space. The metric completion of \(A \boxtimes B\) (see Theorem 43.7 from~\cite{Munkres}) is a Hilbert space denoted \(A \otimes B\) whose inner product is the canonical extension of that of \(A \boxtimes B\).\footnote{If \(A\) and \(B\) are finite-dimensional, \(A \otimes B = A \boxtimes B\).}
\end{proposition}

The monoidal category \((\Hilb,\otimes,\Cset)\) captures the notion of combining the state space \(A,B \in \Hilb\) of two distinct quantum memories into a single state space \(A \otimes B \in \Hilb\) (see Section~\ref{sec:QC}).
In contrast, \((\Hilb,\oplus,\{0\})\) is the monoidal structure that captures the notion of parallel composition arising in quantum walks (see Section~\ref{sec:QW}) and the orthogonal decomposition of a state space with respect to a predicate (see Section~\ref{sec:QC}).
Consequently, \((\Hilb,\oplus,\{0\})\) is the appropriate monoidal category to use when discussing control flow of quantum programs.

\subsection{Beyond coherent quantum operations}
\label{sec:CPTR}

Measurements are not coherent operations --- projections are generally not invertible --- and the outcome they yield is probabilistic.
A general mathematical framework where measurements and coherent operations are described on equal footing needs to capture not only `pure' quantum states \(\ket{\psi} \in H\) but, more generally, probability distributions of quantum states.

Let \(H\) be a finite-dimensional complex Hilbert space and let \(B(H)\) be the vector space of bounded linear maps \(H \to H\); in the general framework, a state is represented by a \gls{density operator} \(\rho \in B(H)\): a positive semi-definite map that satisfies \(\tr(\rho) = 1\).\footnote{Recall that, if \(H\) is a \emph{complex} Hilbert space then positive semi-definiteness implies self-adjointness.}
Let \(\ket{\psi} \in H\) be a unit vector and let \(\rho \in B(H)\) be a positive density operator; then it is guaranteed that:
\begin{equation*}
  0 \leq \bra{\psi} \rho \ket{\psi} \leq 1.
\end{equation*} 
Moreover, for any orthonormal basis of \(H\), say \(\{\psi_i\}_{i \in I}\), the trace of \(\rho\) satisfies:
\begin{equation*}
  \tr(\rho) = \sum_{i \in I} \bra{\psi_i} \rho \ket{\psi_i} = 1.
\end{equation*}
The value \(\bra{\psi_i} \rho \ket{\psi_i} \in [0,1]\) for each \(i \in I\) describes the probability of obtaining the pure quantum state \(\ket{\psi_i}\) as the result of measuring \(\rho\) in the chosen basis.
In this framework:
\begin{itemize}
  \item all pure quantum states \(\ket{\varphi} \in H\) may be represented as density operators \(\ketbra{\varphi}{\varphi} \in B(H)\) and, more generally,
  \item for any finite set of pure states \(\{\ket{\varphi_i} \in H\}_{i \in I}\) and any probability distribution \(p \colon I \to [0,1]\) assigned to them,
  \begin{equation*}
    \rho = \sum_{i \in I} p(i) \cdot \ketbra{\varphi_i}{\varphi_i}
  \end{equation*}
  is a density operator.
  We refer to these as \emph{mixed} quantum states.
\end{itemize} 

Given two Hilbert spaces \(H\) and \(K\), any density operator in the vector space \(B(H) \oplus B(K)\) is a pair \((\rho,\rho')\) where \(\rho \in B(H)\) and \(\rho' \in B(H)\) are positive semi-definite operators satisfying \(\tr(\rho) + \tr(\rho') = 1\).
Thus, density operators \((\rho,\rho') \in B(H) \oplus B(K)\) describe mixed states that are known to be either in \(B(H)\) with probability \(\tr(\rho)\) or in \(B(K)\) with probability \(\tr(\rho')\).
Notice that \(B(H) \oplus B(K)\) is strictly contained in \(B(H \oplus K)\); in particular, pure states \(\ketbra{\phi}{\phi} \in B(H \oplus K)\) are not in \(B(H) \oplus B(K)\).
Thus, in certain situations --- such as after applying a measurement that projects a state to either \(H\) or \(K\) --- we will use \(B(H) \oplus B(K)\) to indicate that we are certain that the state is a probabilistic mixture of quantum states.

\begin{remark} \label{rmk:Cstar_algebra}
  The state spaces we are interested in are direct sums \(\oplus_{i \in I} B(H_i)\) where \(\{H_i\}_{i \in I}\) is a finite collection of finite-dimensional Hilbert spaces.
  Any such space can be made into a \Cstar-algebra by defining its multiplication to be composition of maps and defining its involution to yield the adjoint of a map.
  In fact, any finite-dimensional \Cstar-algebra is isomorphic to the one defined from some \(\oplus_{i \in I} B(H_i)\) as above.
  Thus, when managing non-coherent operations, the convention in quantum computing is to use \Cstar-algebras as the state spaces.
  However, the extra structure that a \Cstar-algebra carries is not relevant to this thesis and, hence, no knowledge of \Cstar-algebras is required to follow it.
  For the purposes of this thesis, it suffices to think of a finite-dimensional \Cstar-algebra as a direct sum \(\oplus_{i \in I} B(H_i)\).
\end{remark}

Given a pair of finite-dimensional \Cstar-algebras \(A\) and \(B\), the linear maps \(A \to B\) that are physically meaningful are those that map states to states, \ie{} density operators to density operators.
These are the so-called \gls{completely positive trace-preserving (CPTP)} maps:
\begin{itemize}
  \item a positive map \(A \to B\) sends positive semi-definite operators in \(A\) to positive semi-definite operators in \(B\);
  \item a positive map \(f \colon A \to B\) is completely positive if for every finite-dimensional \Cstar-algebra \(C\) it is satisfied that \(f \otimes \id_C \colon A \otimes C \to B \otimes C\) is a positive map;
  \item a trace-preserving map \(f \colon A \to B\) satisfies \(\tr(f(\rho)) = \tr(\rho)\) for every \(\rho \in A\).
\end{itemize}
A common example of a positive map that is not completely positive is the map \(B(H) \to B(H)\) sending each operator to its matrix transpose.

\begin{definition}
  The category \(\CPTP\) has finite-dimensional \Cstar-algebras as objects and CPTP maps as morphisms.
  The triple \((\CPTP,\oplus,Z)\) is a symmetric monoidal category, where \(\oplus\) is the direct sum of \Cstar-algebras and CPTP maps and \(Z\) is the zero-dimensional \Cstar-algebra (\ie{} \(Z = B(\{0\}) = \{0\}\)).
\end{definition}

Similarly to how the category of (weak) contractions will be useful to discuss the entries of `block matrices' describing morphisms in \(\Isometry\) and \(\Unitary\), it is useful to define a category of completely positive (weak) trace-reducing maps so that \(\CPTP\) embeds in it.

\begin{definition} \label{def:CPTR}
  The category \(\CPTR\) has finite-dimensional \Cstar-algebras as objects and completely positive (weak) trace-reducing maps as morphisms; a map \(f \colon A \to B\) is (weak) trace-reducing if it satisfies that \(\tr(f(\rho)) \leq \tr(\rho)\) for all \(\rho \in A\).\footnote{A \emph{strict} trace-reducing map \(f\) satisfies \(\tr(f(\rho)) < \tr(\rho)\) for all \(\rho \in A\). In this thesis, the term `trace-reducing' refers to weak trace-reducing maps unless stated otherwise.}
  The triple \((\CPTR,\oplus,\{0\})\) is a symmetric monoidal category.
\end{definition}

%% file: Figures/1/time_domain_loop.tex
\begin{tikzpicture}
  \node[rectangle,draw=black,thick] (f) {\(\begin{pmatrix} f_\sub{BA} & f_\sub{BU} \\ f_\sub{UA} & f_\sub{UU}  \end{pmatrix}\)};
  \coordinate[below=5mm of f.west] (A);
  \node[circle,fill=black,minimum size=2.5mm,inner sep=0,left=7mm of A] (w) {};
  \coordinate[left=3mm of w] (Ad);
  \draw (Ad) -- (w);
  \draw (w) -- node[above] {\tiny \(A\)} (A);
  \coordinate[below=5mm of f.east] (B);
  \coordinate[right=12mm of B] (Bd);
  \draw (Bd) -- node[above] {\tiny \(B\)} (B);
  \coordinate[above=5mm of f.west] (Ui);
  \coordinate[left=6mm of Ui] (Uid);
  \draw (Uid) -- node[above] {\tiny \(U\)} (Ui);
  \coordinate[above=5mm of f.east] (Uo);
  \coordinate[right=6mm of Uo] (Uod);
  \draw (Uod) -- node[above] {\tiny \(U\)} (Uo);
  \coordinate[above=9mm of Uid] (Uiu);
  \coordinate[above=9mm of Uod] (Uou);
  \draw (Uiu) -- (Uou);
  \draw (Uid) edge[out=180,in=180,looseness=1.5] (Uiu);
  \draw (Uod) edge[out=0,in=0,looseness=1.5] (Uou);
  \node[below=0mm of w] (wt) {\scriptsize \(\ket{\psi}\)};
\end{tikzpicture}
\quad {\Large \(\xmapsto{\text{\scriptsize \ \ two steps later \ \ }}\)} \quad
\begin{tikzpicture}
  \node[rectangle,draw=black,thick] (f) {\(\begin{pmatrix} f_\sub{BA} & f_\sub{BU} \\ f_\sub{UA} & f_\sub{UU}  \end{pmatrix}\)};
  \coordinate[below=5mm of f.west] (A);
  \coordinate[left=13mm of A] (Ad);
  \draw (Ad) -- node[below] {\tiny \(A\)} (A);
  \coordinate[below=5mm of f.east] (B);
  \node[circle,fill=black,minimum size=2.5mm,inner sep=0,right=3mm of B] (wBUA) {};
  \node[circle,fill=black,minimum size=2.5mm,inner sep=0,right=6mm of wBUA] (wBA) {};
  \coordinate[right=3mm of wBA] (Bd);
  \draw (B) -- (wBUA);
  \draw (wBUA) -- node[above] {\tiny \(B\)} (wBA);
  \draw (wBA) -- (Bd);
  \coordinate[above=5mm of f.west] (Ui);
  \node[circle,fill=black,minimum size=2.5mm,inner sep=0,left=4mm of Ui] (wUUA) {};
  \coordinate[left=1mm of wUUA] (Uid);
  \draw (Uid) -- (wUUA);
  \draw (wUUA) -- node[above] {\tiny \(U\)} (Ui);
  \coordinate[above=5mm of f.east] (Uo);
  \coordinate[right=6mm of Uo] (Uod);
  \draw (Uod) -- node[above] {\tiny \(U\)} (Uo);
  \coordinate[above=9mm of Uid] (Uiu);
  \coordinate[above=9mm of Uod] (Uou);
  \draw (Uiu) -- (Uou);
  \draw (Uid) edge[out=180,in=180,looseness=1.5] (Uiu);
  \draw (Uod) edge[out=0,in=0,looseness=1.5] (Uou);
  \node[below=0mm of wBA] (wBAt) {\scriptsize \(\ket{\psi_\dsub{ba}}\)};
  \node[below=0mm of wBUA] (wBUAt) {\scriptsize \(\ket{\psi_\dsub{bua}}\)};
  \node[below=0mm of wUUA] (wUUAt) {\scriptsize \(\ket{\psi_\dsub{uua}}\)};
\end{tikzpicture}

%% file: Figures/1/monoidal_interchange_law.tex
\begin{tikzpicture}
  \node (equals) {\Large \(=\)};
  \node[left=10mm of equals] (left) {
    \begin{tikzpicture}
      \node[draw=black,thick,minimum size=7mm] (f) {\(f\)};
      \node[draw=black,thick,minimum size=7mm,above=9mm of f] (h) {\(h\)};
      \node[draw=black,thick,minimum size=7mm,right=13mm of f] (g) {\(g\)};
      \node[draw=black,thick,minimum size=7mm,above=9mm of g] (k) {\(k\)};
      \coordinate[left=6mm of f] (fp);
      \coordinate[right=6mm of g] (gp);
      \coordinate[left=6mm of h] (hp);
      \coordinate[right=6mm of k] (kp);
      \draw (gp) -- (g) -- (f) -- (fp);
      \draw (kp) -- (k) -- (h) -- (hp);
      \coordinate[below left=2.75mm of g] (x1);
      \coordinate[above right=2.75mm of k] (x2);
      \draw[draw=black!60,thick,dashed] (x1) rectangle (x2);
      \coordinate[below left=2.75mm of f] (y1);
      \coordinate[above right=2.75mm of h] (y2);
      \draw[draw=black!60,thick,dashed] (y1) rectangle (y2);
    \end{tikzpicture}
  };
  \node[right=10mm of equals] (right) {
    \begin{tikzpicture}
      \node[draw=black,thick,minimum size=7mm] (f) {\(f\)};
      \node[draw=black,thick,minimum size=7mm,above=9mm of f] (h) {\(h\)};
      \node[draw=black,thick,minimum size=7mm,right=13mm of f] (g) {\(g\)};
      \node[draw=black,thick,minimum size=7mm,above=9mm of g] (k) {\(k\)};
      \coordinate[left=6mm of f] (fp);
      \coordinate[right=6mm of g] (gp);
      \coordinate[left=6mm of h] (hp);
      \coordinate[right=6mm of k] (kp);
      \draw (gp) -- (g) -- (f) -- (fp);
      \draw (kp) -- (k) -- (h) -- (hp);
      \coordinate[below left=2.75mm of f] (x1);
      \coordinate[above right=2.75mm of g] (x2);
      \draw[draw=black!60,thick,dashed] (x1) rectangle (x2);
      \coordinate[below left=2.75mm of h] (y1);
      \coordinate[above right=2.75mm of k] (y2);
      \draw[draw=black!60,thick,dashed] (y1) rectangle (y2);
    \end{tikzpicture}
  };
\end{tikzpicture}

%% file: Figures/1/hexagon_equations.tex
\begin{tikzpicture}
  \node (equals1) {\Large \(=\)};
  \node[right=70mm of equals1] (equals2) {\Large \(=\)};
  \node[left=-1mm of equals1] (LHS1) {
    \begin{tikzpicture}
      \coordinate (Xu);
      \coordinate[above=5mm of Xu] (Yu);
      \coordinate[above=5mm of Yu] (Zu);
      \coordinate[right=10mm of Xu] (Xc);
      \coordinate[above=5mm of Xc] (Yc);
      \coordinate[above=5mm of Yc] (Zc);
      \coordinate[right=10mm of Xc] (Xd);
      \coordinate[above=5mm of Xd] (Yd);
      \coordinate[above=5mm of Yd] (Zd);
      \draw[thick] (Xu) edge[in=180,out=0] (Yc);
      \draw[thick] (Yc) edge[in=180,out=0] (Zd);
      \coordinate[below right=2.5mm and 5mm of Zc] (O1);
      \filldraw [white] (O1) circle (1.5mm);
      \coordinate[below right=2.5mm and 5mm of Yu] (O2);
      \filldraw [white] (O2) circle (1.5mm);
      \draw[thick] (Yu) edge[in=180,out=0] (Xc);
      \draw[thick] (Xc) edge[in=180,out=0] (Xd);
      \draw[thick] (Zu) edge[in=180,out=0] (Zc);
      \draw[thick] (Zc) edge[in=180,out=0] (Yd);
      \node[left=-0.75mm of Xu] {\scriptsize \(A\)};
      \node[left=-0.75mm of Yu] {\scriptsize \(B\)};
      \node[left=-0.75mm of Zu] {\scriptsize \(C\)};
      \node[right=-0.75mm of Xd] {\scriptsize \(B\)};
      \node[right=-0.75mm of Yd] {\scriptsize \(C\)};
      \node[right=-0.75mm of Zd] {\scriptsize \(A\)};
      \coordinate[above left=4mm and 4mm of O1] (box1ul);
      \coordinate[below right=4mm and 4mm of O1] (box1dr);
      \draw[semithick,black!60,dashed] (box1ul) rectangle (box1dr);
      \node[above=4mm of O1, black!75] {\scriptsize \(\sigma_{A,C}\)};
      \coordinate[above left=4mm and 4mm of O2] (box2ul);
      \coordinate[below right=4mm and 4mm of O2] (box2dr);
      \draw[semithick,black!60,dashed] (box2ul) rectangle (box2dr);
      \node[below=4mm of O2, black!75] {\scriptsize \(\sigma_{A,B}\)};
    \end{tikzpicture}
  };
  \node[right=-1mm of equals1] (RHS1) {
    \begin{tikzpicture}
      \coordinate (Yuc);
      \coordinate[below=1mm of Yuc] (Yul);
      \coordinate[above=1mm of Yuc] (Yur);
      \coordinate[right=20mm of Yuc] (Yd);
      \coordinate[below=10mm of Yuc] (Xu);
      \coordinate[right=20mm of Xu] (Xdc);
      \coordinate[below=1mm of Xdc] (Xdl);
      \coordinate[above=1mm of Xdc] (Xdr);
      \draw[thick] (Xu) edge[in=180,out=0] (Yd);
      \coordinate[below right=4.5mm and 10.5mm of Yuc] (Ol);
      \filldraw [white] (Ol) circle (1.5mm);
      \coordinate[below right=5.5mm and 9.5mm of Yuc] (Or);
      \filldraw [white] (Or) circle (1.5mm);
      \draw[thick] (Yul) edge[in=180,out=0] (Xdl);
      \draw[thick] (Yur) edge[in=180,out=0] (Xdr);
      \node[left=-0.75mm of Xu] {\scriptsize \(A\)};
      \node[below left=-1.5mm and -0.75mm of Yul] {\scriptsize \(B\)};
      \node[above left=-1.5mm and -0.75mm of Yur] {\scriptsize \(C\)};
      \node[below right=-1.5mm and -0.75mm of Xdl] {\scriptsize \(B\)};
      \node[above right=-1.5mm and -0.75mm of Xdr] {\scriptsize \(C\)};
      \node[right=-0.75mm of Yd] {\scriptsize \(A\)};
      \coordinate[above left=6mm and 6mm of Ol] (boxul);
      \coordinate[below right=6mm and 6mm of Or] (boxdr);
      \draw[semithick,black!60,dashed] (boxul) rectangle (boxdr);
      \node[above=7mm of Or, black!75] {\scriptsize \(\sigma_{A, B \otimes C}\)};
    \end{tikzpicture}
  };
  \node[left=-1mm of equals2] (LHS2) {
    \begin{tikzpicture}
      \coordinate (Xu);
      \coordinate[above=5mm of Xu] (Yu);
      \coordinate[above=5mm of Yu] (Zu);
      \coordinate[right=10mm of Xu] (Xc);
      \coordinate[above=5mm of Xc] (Yc);
      \coordinate[above=5mm of Yc] (Zc);
      \coordinate[right=10mm of Xc] (Xd);
      \coordinate[above=5mm of Xd] (Yd);
      \coordinate[above=5mm of Yd] (Zd);
      \draw[thick] (Xu) edge[in=180,out=0] (Xc);
      \draw[thick] (Yu) edge[in=180,out=0] (Zc);
      \coordinate[above right=2.5mm and 5mm of Yu] (O2);
      \filldraw [white] (O2) circle (1.5mm);
      \draw[thick] (Zu) edge[in=180,out=0] (Yc);
      \draw[thick] (Xc) edge[in=180,out=0] (Yd);
      \coordinate[above right=2.5mm and 5mm of Xc] (O1);
      \filldraw [white] (O1) circle (1.5mm);
      \draw[thick] (Yc) edge[in=180,out=0] (Xd);
      \draw[thick] (Zc) edge[in=180,out=0] (Zd);
      \node[left=-0.75mm of Xu] {\scriptsize \(A\)};
      \node[left=-0.75mm of Yu] {\scriptsize \(B\)};
      \node[left=-0.75mm of Zu] {\scriptsize \(C\)};
      \node[right=-0.75mm of Xd] {\scriptsize \(C\)};
      \node[right=-0.75mm of Yd] {\scriptsize \(A\)};
      \node[right=-0.75mm of Zd] {\scriptsize \(B\)};
      \coordinate[above left=4mm and 4mm of O1] (box1ul);
      \coordinate[below right=4mm and 4mm of O1] (box1dr);
      \draw[semithick,black!60,dashed] (box1ul) rectangle (box1dr);
      \node[below=4mm of O1, black!75] {\scriptsize \(\sigma_{A,C}\)};
      \coordinate[above left=4mm and 4mm of O2] (box2ul);
      \coordinate[below right=4mm and 4mm of O2] (box2dr);
      \draw[semithick,black!60,dashed] (box2ul) rectangle (box2dr);
      \node[above=4mm of O2, black!75] {\scriptsize \(\sigma_{B,C}\)};
    \end{tikzpicture}
  };
  \node[right=-1mm of equals2] (RHS2) {
    \begin{tikzpicture}
      \coordinate (Xuc);
      \coordinate[below=1mm of Xuc] (Xul);
      \coordinate[above=1mm of Xuc] (Xur);
      \coordinate[right=20mm of Xuc] (Xd);
      \coordinate[above=10mm of Xuc] (Yu);
      \coordinate[right=20mm of Yu] (Ydc);
      \coordinate[below=1mm of Ydc] (Ydl);
      \coordinate[above=1mm of Ydc] (Ydr);
      \draw[thick] (Xul) edge[in=180,out=0] (Ydl);
      \draw[thick] (Xur) edge[in=180,out=0] (Ydr);
      \coordinate[above right=4.5mm and 10.5mm of Xuc] (Ol);
      \filldraw [white] (Ol) circle (1.5mm);
      \coordinate[above right=5.5mm and 9.5mm of Xuc] (Or);
      \filldraw [white] (Or) circle (1.5mm);
      \draw[thick] (Yu) edge[in=180,out=0] (Xd);
      \node[below left=-1.5mm and -0.75mm of Xul] {\scriptsize \(A\)};
      \node[above left=-1.5mm and -0.75mm of Xur] {\scriptsize \(B\)};
      \node[left=-0.75mm of Yu] {\scriptsize \(C\)};
      \node[right=-0.75mm of Xd] {\scriptsize \(C\)};
      \node[below right=-1.5mm and -0.75mm of Ydl] {\scriptsize \(A\)};
      \node[above right=-1.5mm and -0.75mm of Ydr] {\scriptsize \(B\)};
      \coordinate[above left=6mm and 6mm of Ol] (boxul);
      \coordinate[below right=6mm and 6mm of Or] (boxdr);
      \draw[semithick,black!60,dashed] (boxul) rectangle (boxdr);
      \node[below=7mm of Or, black!75] {\scriptsize \(\sigma_{A \otimes B, C}\)};
    \end{tikzpicture}
  };
\end{tikzpicture}

%% file: Chapters/2.tex
\chapter{Categories for infinitary addition}
\label{chap:Sigma}

Let \((X,+)\) be an arbitrary commutative monoid. By applying \(+\) multiple times, we can add up any finite number of elements and, thanks to associativity and commutativity, the result will be unambiguous. However, standard monoids fall short if we are interested in adding up an infinite collection of elements. For instance, not every infinite sequence of real numbers is summable: to assign a result to a series we must verify that its sequence of partial sums converges. 
The first section of this chapter introduces \(\Sigma\)-monoids, a generalisation of monoids whose monoid operation is partially defined and aggregates (possibly infinite) collections of elements.
The second section provides a brief introduction to topological monoids: these are monoids endowed with extra structure that lets us define a general notion of convergence and, consequently, formalise infinite sums.
It will be shown that certain class of topological monoids --- Hausdorff commutative monoids --- provides an abundant number of examples of \(\Sigma\)-monoids.
 
The notion of \emph{strong} \(\Sigma\)-monoid presented in this chapter was originally proposed by Haghverdi~\cite{Haghverdi}; other flavours of \(\Sigma\)-monoids are proposed in this chapter, and their relation to each other is described in terms of adjunctions between their categories.
Moreover, it is shown that the tensor product of \(\Sigma\)-monoids is well-defined, thus allowing us to define categories enriched over \(\Sigma\)-monoids; this will be essential in Chapter~\ref{chap:trace}.

\begin{definition} \label{def:families}
  A \gls{family} of elements in \(X\) is an indexed multiset \(\{x_i \in X\}_{i \in I}\), where \(I\) is an arbitrary countable set.\footnote{Being an indexed multiset, there may be distinct indices \(i,j \in I\) with \(x_i = x_j\).}
  Two families \(\{x_i\}_{i \in I}\) and \(\{x_j\}_{j \in J}\) are the same if there is a bijection \(\phi \colon I \to J\) such that for all \(i \in I\) it satisfies \(x_i = x_{\phi(i)}\).
  The collection of all families of elements in \(X\) is denoted \GLS{\(X^*\)}{X_star}.
\end{definition}

\begin{proposition} \label{prop:family_set}
  For any set \(X\), the collection \(X^*\) is a (small) set.
\end{proposition} \begin{proof}
  We can represent a multiset \(\{x_i \in X\}_{i \in I}\) as a set \(\{(x_i,i) \mid i \in I\} \subseteq X \times I\). Since the indexing set \(I\) of a family is countable by definition, there is an injection from \(X^*\) to the collection of all countable sets. The collection of all countable sets is small, hence, \(X^*\) is small.
\end{proof}

\begin{notation} \label{not:families}
  The shorthand notation \(\{x_i\}_I\) identifies indices with the lowercase of the letter denoting the indexing set.
  When the indexing set is not specified, all elements of the family are explicitly given within curly brackets, for instance, \(\{x\}\) is the singleton family and \(\{a,b,c\}\) is a finite family containing three elements.
  When explicit reference to the elements of the family is not necessary, a family may be denoted with a bold font variable \(\fml{x} \in X^*\).
  A family \(\{x_j\}_J\) is a \gls{subfamily} of another \(\{x_i\}_I\) if there is an injection \(\phi \colon J \to I\) such that \(x_j = x_{\phi(j)}\) for each \(j \in J\).
  Hence, there is a partial order \(\subseteq\) on \(X^*\) so that \(\fml{x'} \subseteq \fml{x}\) iff \(\fml{x'}\) is a subfamily of \(\fml{x}\).
  Union, disjoint union and intersection of families are defined with respect to the corresponding operations on their indexed sets.\footnote{Whenever an index \(i\) appears in both \(\fml{x}\) and \(\fml{y}\), their corresponding elements \(x_i\) and \(y_i\) must be equal for union and intersection to be well-defined.}
  Let \(f \colon X \to Y\) be a function and \(\fml{x} = \{x_i\}_I\) a family in \(X^*\); the shorthand \(f\fml{x}\) refers to the family \(\{f(x_i)\}_I\) in \(Y^*\).
\end{notation}

\begin{notation} \label{not:keq}
  For partial functions \(f \colon A \pto B\), \(g \colon A' \pto B\), the Kleene equality \(f(a) \keq g(a')\) indicates that \(f(a)\) is defined if and only if \(g(a')\) is defined and, when they both are, their results agree.
  In particular, \(f(a) \keq b\) implies that \(f(a)\) is defined and is equal to \(b\).
  The equals sign \(f(a) = b\) is only used in the presence of partial functions when the context has already established that \(f(a)\) is defined.
\end{notation}

\section{\(\Sigma\)-monoids}
\label{sec:SCat}

The notion of \emph{strong} \(\Sigma\)-monoid that will be presented in this section was originally proposed by Haghverdi~\cite{Haghverdi}, although similar mathematical structures predate it: for instance, the partially additive monoids from Manes and Arbib~\cite{ManesArbib} or the \(\Sigma\)-groups of Higgs~\cite{Higgs}.
Haghverdi~\cite{Haghverdi} and Hoshino~\cite{RTUDC} have discussed how certain categories enriched in strong \(\Sigma\)-monoids may be used to capture iteration; this line of work is the main focus of Chapter~\ref{chap:trace}.
However, strong \(\Sigma\)-monoids are too restrictive for our application: Proposition~\ref{prop:no_inverses} shows that these forbid the existence of additive inverses, which are essential in quantum computing since they capture destructive interference.
In this section, a weaker version of \(\Sigma\)-monoids is proposed which admits additive inverses.

A \(\Sigma\)-monoid is a set \(X\) together with a partial function \(\Sigma \colon X^* \pto X\) describing an infinitary operation which satisfies certain axioms to ensure commutativity, associativity and existence of a neutral element.

\begin{definition} \label{def:wSm}
  Let \(X\) be a set and let \(\Sigma \colon X^* \pto X\) be a partial function. We say \(\fml{x} \in X^*\) is a \gls{summable family} if \(\Sigma \fml{x}\) is defined. The pair \((X,\Sigma)\) is a \GLS{weak $\Sigma$-monoid}{wSm} if the following axioms are satisfied.
  \begin{itemize}
    \item \emph{Singleton.} For all \(x \in X\) the family \(\{x\}\) is summable with \(\Sigma \{x\} \keq x\).
    \item \emph{Neutral element.} (\emph{i}) There is an element \(0 \in X\) such that \(\Sigma \varnothing \keq 0\). \\
    (\emph{ii}) Let \(\fml{x} \in X^*\) be a family and let \(\fml{x}_\emptyset\) be its subfamily of elements other than \(0\); if \(\fml{x}\) is summable then \(\fml{x}_\emptyset\) is summable as well.
    \item \emph{Bracketing.} Let \(\{\fml{x}_j\}_J\) be an indexed set of \emph{finite} summable families \(\fml{x}_j \in X^*\) and let \(\fml{x} = \uplus_J \fml{x}_j\); then:
    \begin{equation*} \label{eq:bracketing}
      \Sigma \fml{x} \keq x \implies \Sigma \{\Sigma \fml{x}_j\}_J \keq x.
    \end{equation*}
    \item \emph{Flattening.} Let \(\{\fml{x}_j\}_J\) be an indexed set of summable families \(\fml{x}_j \in X^*\) whose indexing set \(J\) is \emph{finite} and let \(\fml{x} = \uplus_J \fml{x}_j\); then:
    \begin{equation*} \label{eq:fin_flattening}
      \Sigma \{\Sigma \fml{x}_j\}_J \keq x \implies \Sigma \fml{x} \keq x.
    \end{equation*}
  \end{itemize}
\end{definition}

In this definition, the axioms of commutative monoids have been lifted to the infinitary case. 
Bracketing and flattening together capture associativity and commutativity.
Notice that these are not exactly dual to each other: bracketing assumes that each family \(\fml{x}_j\) is finite, whereas flattening admits infinite families \(\fml{x}_j\) but assumes there is only a finite number of them.
As will be established in Proposition~\ref{prop:no_inverses}, the flattening axiom's requirement that \(J\) is finite is necessary for weak \(\Sigma\)-monoids to admit additive inverses.
On the other hand, the bracketing axiom's requirement that each family \(\fml{x}_j\) is finite is necessary so that every Hausdorff commutative monoid is a weak \(\Sigma\)-monoid whose summable families are precisely those that converge according to the topology (see Proposition~\ref{prop:wSm_Hausdorff}).
The following proposition establishes that \(0 \in X\) acts as the neutral element of the \(\Sigma\)-monoid.

\begin{proposition} \label{prop:neutral_element}
  Let \((X,\Sigma)\) be a weak \(\Sigma\)-monoid with \(\Sigma \varnothing = 0\). Let \(\fml{x} \in X^*\) be a family such that \(\Sigma \fml{x} \keq x\) for some \(x \in X\).
  Let \(\fml{0} \in X^*\) be a (possibly infinite) family whose elements are all \(0\); then \(\Sigma (\fml{x} \uplus \fml{0}) \keq x\).
  Conversely, let \(\fml{x}_\emptyset \subseteq \fml{x}\) be the subfamily of elements other than \(0\); then \(\Sigma \fml{x}_\emptyset \keq x\).
\end{proposition} \begin{proof}
  Since \(\Sigma \varnothing = 0\), the family \(\fml{0}\) may be equivalently defined as
  \begin{equation*}
    \fml{0} = \{\Sigma \varnothing, \Sigma \varnothing, \ldots\}
  \end{equation*}
  and, due to bracketing, \(\Sigma \fml{0} \keq 0\).
  Moreover, for any element \(x \in X\), the singleton axiom imposes that \(\Sigma \{x\} \keq x\) and, clearly, \(\{x\} = \{x\} \uplus \varnothing\) so according to bracketing:
  \begin{equation*}
    \Sigma \{x\} \keq x \implies \Sigma \{\Sigma \{x\}, \Sigma \varnothing\} \keq x.
  \end{equation*}
  Using the singleton axiom and the definition of the neutral element, the right hand side can be rewritten as \(\Sigma \{x,0\} \keq x\).
  Then, it follows that \(\Sigma \{\Sigma \fml{x}, \Sigma \fml{0}\} \keq x\) and the flattening axiom implies that \(\Sigma (\fml{x} \uplus \fml{0}) \keq x\), proving the first part of the claim.

  On the other hand, we may partition any summable family \(\fml{x} \in X^*\) into its subfamily \(\fml{x}_\emptyset\) and its subfamily \(\fml{0}\) so that \(\fml{x} = \fml{x}_\emptyset \uplus \fml{0}\).
  Since the neutral element axiom imposes that \(\fml{x}_\emptyset\) is summable, there is some \(x' \in X\) such that \(\Sigma \fml{x}_\emptyset \keq x'\) and, since \(\Sigma \fml{0} \keq 0\) and \(\Sigma \{x',0\} \keq x'\) have been established above, we may apply flattening to obtain \(\Sigma \fml{x} \keq x'\).
  But \(\Sigma \fml{x} \keq x\) by assumption, so it follows that \(x = x'\) and \(\Sigma \fml{x}_\emptyset \keq x\), as claimed.
\end{proof}

\begin{example} \label{ex:wSm_pm}
  Let \(\pm\) be the set \(\{0,+,-\}\) and, for each family \(\fml{x} \in \pm^*\), let \(n_+(\fml{x})\) (and \(n_-(\fml{x})\)) be the number of occurrences of \(+\) (respectively, of \(-\)). Define a partial function \(\Sigma \colon \pm^* \pto \pm\) as follows:
  \begin{equation}
    \Sigma \fml{x} = \begin{cases}
      0 &\ifc n_+(\fml{x}) < \infty \text{ and } n_+(\fml{x}) = n_-(\fml{x}) \\
      + &\ifc n_+(\fml{x}) < \infty \text{ and } n_+(\fml{x}) = n_-(\fml{x})+1 \\
      - &\ifc n_+(\fml{x}) < \infty \text{ and } n_+(\fml{x}) = n_-(\fml{x})-1 \\
      \undefined &\otherwise.
    \end{cases}
  \end{equation}
  Then, \((\pm,\Sigma)\) is a weak \(\Sigma\)-monoid.
\end{example} \begin{proof}
  Both neutral element and singleton axioms are trivial to check; flattening is also straightforward: if each \(\fml{x}_j\) is summable, then each \(n_+(\fml{x}_j)\) is finite and
  \begin{equation*}
    n_+(\uplus_J \fml{x}_j) = \sum_J n_+(\fml{x}_j)
  \end{equation*}
  holds, where the right hand side is a finite sum of natural numbers since \(J\) is finite.
  The same can be said of \(n_-\), so the difference between \(n_+\) and \(n_-\) stays the same after flattening.
  To see how bracketing is satisfied, recall that both the family \(\fml{x}\) and each \(\fml{x}_j\) are assumed to be summable.
  For all \(j \in J\), the sum \(\Sigma \fml{x}_j\) is either \(0\), \(+\) or \(-\), so we may partition \(J\) into \(J_0\), \(J_+\) and \(J_-\) depending on this result.
  Notice that \(J_+\) and \(J_-\) must be finite sets, otherwise \(\fml{x} = \uplus_J \fml{x}_j\) would have an infinite number of nonzero elements and would not be summable.
  Then, a simple counting argument determines that: if \(\Sigma \fml{x} = 0\) then \(\abs{J_+} = \abs{J_-}\), if \(\Sigma \fml{x} = +\) then \(\abs{J_+} = \abs{J_-} + 1\) and if \(\Sigma \fml{x} = -\) then \(\abs{J_+} = \abs{J_-} - 1\).
  Therefore, it follows that the family \(\{\Sigma \fml{x}_j\}_J\) is summable, and the result agrees with that of \(\Sigma \fml{x}\), thus satisfying the bracketing axiom.
\end{proof}

\begin{example} \label{ex:wSm_ominus}
  Let \(X\) be a set and let \(P = \mathcal{P}(X)\) be its power set. For any \(x \in X\), define \(n_x \colon P^* \to \Nset \cup \{\infty\}\) as follows:
  \begin{equation}
    n_x(\{A_i\}_I) = \abs{\{i \in I \mid x \in A_i\}}.
  \end{equation}
  Define a partial function \(\Sigma \colon P^* \pto P\) as follows:
  \begin{equation}
    \Sigma \{A_i\}_I = \begin{cases}
      \{x \in X \mid n_x(\{A_i\}_I) \text{ odd}\} &\ifc \forall x \in X,\, n_x(\{A_i\}_I) < \infty \\
      \undefined &\otherwise.
    \end{cases}
  \end{equation}
  Then, \((P,\Sigma)\) is a weak \(\Sigma\)-monoid.
\end{example} \begin{proof}
  Notice that on finite families \(\Sigma\) corresponds to symmetric difference of sets.
  Singleton and neutral element axioms are trivial and flattening follows from associativity and commutativity of the symmetric difference of sets. Bracketing can be verified by realising that \(n_x(\{A_i\}_I)\) is odd if and only if, for every partition \(I = \uplus_J I_j\), there is an odd number of odd \(n_x(\{A_i\}_{I_j})\).
\end{proof}

Particularly relevant to this thesis is the fact that every Hausdorff commutative monoid is a weak \(\Sigma\)-monoid. This is established in Section~\ref{sec:topology} along with a brief introduction to general topology and Hausdorff commutative monoids. 
Many examples arise from this result; one of the simplest among them is given below.

\begin{example} \label{ex:wSm_Rset}
  Let \((\Rset,+)\) be the standard Hausdorff commutative monoid of real numbers. According to Proposition~\ref{prop:wSm_Hausdorff}, we may define a partial function \(\Sigma \colon \Rset^* \pto \Rset\) so that \((\Rset,\Sigma)\) is a weak \(\Sigma\)-monoid.
  Indeed, such \(\Sigma\) corresponds to the usual notion of absolute convergence of series.
\end{example}

In the same spirit that monoid homomorphisms are functions preserving the monoid structure, \(\Sigma\)-homomorphisms are functions that preserve summable families.

\begin{definition}
  Let \((X,\Sigma)\) and \((Y,\Sigma')\) be weak \(\Sigma\)-monoids. A function \(f \colon X \to Y\) is a \(\Sigma\)-homomorphism if for every family \(\fml{x} \in X^*\) and every \(x \in X\):
  \begin{equation*}
    \Sigma \fml{x} \keq x \implies \Sigma' f\fml{x} \keq f(x).
  \end{equation*}
\end{definition}

\begin{example}
  Let \((\pm,\Sigma)\) and \((\Rset,\Sigma)\) be the weak \(\Sigma\)-monoids from Examples~\ref{ex:wSm_pm} and~\ref{ex:wSm_Rset} respectively.
  The function \(f \colon \pm \to \Rset\) given by:
  \begin{equation*}
    f(0) = 0 \quad\quad\quad f(+) = 1 \quad\quad\quad f(-) = -1
  \end{equation*}
  is a \(\Sigma\)-homomorphism.
\end{example}

Notice that in the weak \(\Sigma\)-monoid \((\pm,\Sigma)\) from Example~\ref{ex:wSm_pm} the family \(\{+,+,-\}\) is summable whereas its subfamily \(\{+,+\}\) is not.
Similarly, we may define a weak \(\Sigma\)-monoid \(([0,1],\Sigma)\) whose \(\Sigma\) function is the same as in \((\Rset,\Sigma)\) when the result is in the interval \([0,1]\) and it is undefined otherwise.
Once again, \(\{0.75,0.5,-0.25\}\) is summable whereas its subfamily \(\{0.75,0.5\}\) is not.
This idea of defining a weak \(\Sigma\)-monoid from an existing one by restricting the underlying set is generalised by the following lemma.

\begin{lemma} \label{lem:wSm_restriction}
  Let \((Y,\Sigma)\) be a weak \(\Sigma\)-monoid and let \(X\) be a set. Let \(f \colon X \to Y\) be an injective function such that \(0 \in \im(f)\). Define a partial function \(\Sigma^{f} \colon X^* \pto X\) as follows:
  \begin{equation*}
    \Sigma^f \fml{x} = \begin{cases}
      x &\ifc \exists x \in X \st \Sigma f\fml{x} \keq f(x) \\
      \undefined &\otherwise.
    \end{cases}
  \end{equation*}
  Then, \((X,\Sigma^f)\) is a weak \(\Sigma\)-monoid and \(f \colon X \to Y\) is a \(\Sigma\)-homomorphism.
\end{lemma} \begin{proof}
  The definition of \(\Sigma^f\) is unambiguous thanks to the requirement that \(f\) is injective.
  Neutral element and singleton axioms in \((X,\Sigma^f)\) are trivially derived from those in \((Y,\Sigma)\).
  Let \(\{\fml{x}_j\}_J\) be an indexed set where each \(\fml{x}_j \in X^*\) is a summable family in \((X,\Sigma^f)\) and let \(\fml{x} = \uplus_J \fml{x}_j\), evidently, \(f\fml{x} = \uplus_J f\fml{x}_j\) (see Notation~\ref{not:families}).
  Assume that each family \(\fml{x}_j\) is finite, then the following sequence of implications proves bracketing:
  \begin{align*}
    \Sigma^f \fml{x} \keq x &\iff \Sigma f\fml{x} \keq f(x) &&\text{(definition of \(\Sigma^f\))} \\
        &\implies \Sigma \{\Sigma f\fml{x}_j\}_J \keq f(x) &&\text{(bracketing in \(Y\))} \\
        &\iff \Sigma \{f(\Sigma^f \fml{x}_j)\}_J \keq f(x) &&\text{(\(\fml{x}_j\) summable and def. of \(\Sigma^f\))} \\
        &\iff \Sigma^f \{\Sigma^f \fml{x}_j\}_J \keq x &&\text{(definition of \(\Sigma^f\))}.
  \end{align*}
  Finally, flattening of \((X,\Sigma^f)\) can be proven following a similar argument, this time using flattening in \((Y,\Sigma)\).
  Thus, it has been shown that \((X,\Sigma^f)\) is a weak \(\Sigma\)-monoid; to check that \(f \colon X \to Y\) is a \(\Sigma\)-homomorphism realise that, by definition, \(\fml{x}\) is summable if and only if \(f\fml{x}\) is summable, so it is immediate that \(\Sigma^f \fml{x} \keq x\) implies \(\Sigma f\fml{x} \keq f(x)\).
\end{proof}

The remainder of this section will study the category of weak \(\Sigma\)-monoids and certain important full subcategories of it, including that of the strong \(\Sigma\)-monoids discussed by Haghverdi~\cite{Haghverdi}.

\begin{definition}
  Let \GLS{\(\SCat{w}\)}{SCatw} be the category whose objects are weak \(\Sigma\)-monoids and whose morphisms are \(\Sigma\)-homomorphisms.
\end{definition}

\begin{proposition} \label{prop:SCatw_complete}
  The category \(\SCat{w}\) is complete.
\end{proposition} \begin{proof}
  First, we show that \(\SCat{w}\) has all equalizers.
  Let \(X\) and \(Y\) be weak \(\Sigma\)-monoids and let \(f,g \colon X \to Y\) be \(\Sigma\)-homomorphisms.
  Define the subset \(E = \{x \in X \mid f(x) = g(x)\}\); the inclusion \(e \colon E \to X\) is injective and \(0 \in E\) due to \(f(\Sigma \varnothing) = g(\Sigma \varnothing)\).
  Then, according to Lemma~\ref{lem:wSm_restriction}, \(E\) can be endowed with a weak \(\Sigma\)-monoid structure:
  \begin{equation} \label{eq:SCat_equalizer}
    \Sigma^e \fml{x} = \begin{cases}
      x &\ifc \exists x \in E \st \Sigma e\fml{x} \keq e(x) \\
      \undefined &\text{otherwise}
    \end{cases}
  \end{equation}
  so that \(e \colon E \to X\) is a \(\Sigma\)-homomorphism.
  Consequently, \((E,\Sigma^e)\) is a cone in \(\SCat{w}\) and, since \(E\) is an equalizer in \(\Set\), for any other cone \((A,\Sigma')\) with a \(\Sigma\)-homomorphism \(h \colon A \to X\) there is a unique function \(m \colon A \to E\) such that \(h = e \circ m\).
  Therefore, to prove that \((E,\Sigma^e)\) is an equalizer in \(\SCat{w}\) we only need to show that \(m\) is a \(\Sigma\)-homomorphism.
  For any summable family \(\fml{a} \in A^*\) we have that \(\Sigma h\fml{a} \keq h(\Sigma' \fml{a})\) due to \(h\) being a \(\Sigma\)-homomorphism and since \((A,\Sigma')\) is a cone \(fh(\Sigma' \fml{a}) = gh(\Sigma' \fml{a})\) so \(h(\Sigma' \fml{a}) \in E\) or, more precisely, \(m(\Sigma' \fml{a}) \in E\).
  It then follows from \(h = em\) that \(\Sigma em\fml{a} \keq em(\Sigma' \fml{a})\) so, by definition of \(\Sigma^e\), the family \(m\fml{a}\) is summable in \(E\) and \(m\) is a \(\Sigma\)-homomorphism.

  Next, we show that \(\SCat{w}\) has all small products.
  Let \(\{(X_i,\Sigma^i)\}_I\) be a small collection of objects in \(\SCat{w}\).
  Endow the set \(\times_I X_i\) with the following partial function:
  \begin{equation} \label{eq:SCat_small_product}
    \Sigma^\times (\fml{x}_i)_{i \in I} = \begin{cases}
      (\Sigma^i \fml{x}_i)_{i \in I} &\ifc \forall i \in I,\, \fml{x}_i \text{ is summable} \\
      \undefined &\otherwise.
    \end{cases}
  \end{equation}
  Then, \((\times_I X_i,\Sigma^\times)\) is a weak \(\Sigma\)-monoid; the proof is straightforward.
  It is immediate that projections are \(\Sigma\)-homomorphisms, so this is a cone for the discrete diagram \(\{(X_i,\Sigma^i)\}_I\).
  Moreover, \(\times_I X_i\) is the categorical product in \(\Set\) so, for any other cone \((A,\Sigma')\) with \(\Sigma\)-homomorphisms \(h_i \colon A \to X_i\), there is a unique function \(m \colon A \to \times_I X_i\) so that each \(h_i\) factors through \(m\).
  To prove that \((\times_I X_i,\Sigma^\times)\) is a categorical product in \(\SCat{w}\) it only remains to show that this unique function \(m\) is a \(\Sigma\)-homomorphism.
  If \(\fml{a}\) is summable in \(A\), then each \(h_i\fml{a}\) is summable in \(X_i\) and, by definition, \(\Sigma^\times (h_i\fml{a})_{i \in I}\) is summable in \(\times_I X_i\).
  Each \(h_i\) is equal to \(\pi_i \circ m\) where \(\pi_i\) is the corresponding projection, so it is immediate that \(\Sigma^\times m\fml{a} \keq m(\Sigma' \fml{a})\) and \(m\) is a \(\Sigma\)-homomorphism.

  A category with all equalizers and small products is complete, so \(\SCat{w}\) is complete as claimed.
\end{proof}

\subsection{Important full subcategories of \(\SCat{w}\)}
\label{sec:SCat*}

Different flavours of \(\Sigma\)-monoids are introduced in this subsection.
The differences between them are with respect to the partiality of \(\Sigma\), \ie{} new axioms are introduced that require more families to be summable.
The first to be introduced is the notion of \(\Sigma\)-monoid proposed by Haghverdi~\cite{Haghverdi}.

\begin{definition} \label{def:sSm}
  A \GLS{strong \(\Sigma\)-monoid}{sSm} is a weak \(\Sigma\)-monoid \((X,\Sigma)\) that satisfies the following extra axioms.
  \begin{itemize}
    \item \emph{Subsummability.} Let \(\{x_i\}_I\) be a summable family; for every subset \(J \subset I\) the subfamily \(\{x_j\}_J\) is also summable.
    \item \emph{Strong bracketing.} Let \(\{\fml{x}_j\}_J\) be an indexed set where each \(\fml{x}_j \in X^*\) is a summable family and let \(\fml{x} = \uplus_J \fml{x}_j\); then:
    \begin{equation*}
       \Sigma \fml{x} \keq x \implies \Sigma \{\Sigma \fml{x}_j\}_J \keq x.
    \end{equation*}
    \item \emph{Strong flattening.} Let \(\{\fml{x}_j\}_J\) be an indexed set where each \(\fml{x}_j \in X^*\) is a summable family and let \(\fml{x} = \uplus_J \fml{x}_j\); then:
    \begin{equation*}
      \Sigma \{\Sigma \fml{x}_j\}_J \keq x \implies \Sigma \fml{x} \keq x.
    \end{equation*}
  \end{itemize}
\end{definition}

Notice that the preconditions of finiteness required in the standard bracketing and flattening axioms are removed, thus making the axioms stronger.
It is straightforward to check that this definition is equivalent to that of Haghverdi~\cite{Haghverdi}.
In Haghverdi's definition, strong bracketing, strong flattening and subsummability are condensed into a single axiom, and the neutral element axiom of weak \(\Sigma\)-monoids is omitted since it follows immediately from subsummability --- notice that the empty family is a trivial subfamily.
Importantly, the following proposition establishes that strong flattening forbids the existence of inverse elements.

\begin{proposition}[Haghverdi~\cite{Haghverdi}] \label{prop:no_inverses}
Let \((X,\Sigma)\) be a strong \(\Sigma\)-monoid.
If a family \(\fml{x} \in X^*\) satisfies \(\Sigma \fml{x} \keq 0\), then every element \(a \in \fml{x}\) is the neutral element \(a = 0\).
\end{proposition} \begin{proof}
  Let \(\fml{x}\) be a family such that \(\Sigma \fml{x} \keq 0\) and choose an arbitrary element \(a \in \fml{x}\). Define the following infinite families:
  \begin{equation*}
    \begin{aligned}
      \fml{z} &= \fml{x} \uplus \fml{x} \uplus \dots \\
      \fml{a} &= \{a\} \uplus \fml{z}
    \end{aligned}
  \end{equation*}
  Due to strong flattening, neutral element and singleton axioms:
  \begin{equation*}
    \begin{aligned}
      \Sigma \fml{z} &\keq \Sigma \{\Sigma \fml{x}, \Sigma \fml{x}, \dots \} \keq \Sigma \{0,0,\dots\} \keq 0\\
      \Sigma \fml{a} &\keq \Sigma \{\Sigma \{a\}, \Sigma \fml{z}\} \keq \Sigma \{a,0\} \keq a
    \end{aligned}
  \end{equation*}
  However, both families \(\fml{z}\) and \(\fml{a}\) contain the same elements: a countably infinite number of copies of each element in \(\fml{x}\). Being the same family, it is immediate that \(\Sigma \fml{a} \keq \Sigma \fml{z}\) which implies \(a = 0\).
  This arguments holds for each element \(a \in \fml{x}\), proving the claim.
\end{proof}

Since weak \(\Sigma\)-monoids do admit inverses, it is relevant to consider the notion of \(\Sigma\)-groups.
It is straightforward to check that the axioms of \(\Sigma\)-groups defined below are equivalent to those proposed by Higgs~\cite{Higgs}. However, Higgs' \(\Sigma\)-groups are more general since Higgs' notion of families are not restricted to be countable, unlike ours (see Definition~\ref{def:families}).

\begin{definition}
  A \GLS{finitely total \(\Sigma\)-monoid}{ftSigma-monoid} is a weak \(\Sigma\)-monoid where every finite family is summable.
  A \GLS{\(\Sigma\)-group}{Sigma-group} is a finitely total \(\Sigma\)-monoid where, for every \(x \in X\), there is an element \(-x \in X\) satisfying \(\Sigma \{x,-x\} = 0\) and where the function that maps each \(x \in X\) to \(-x\) is a \(\Sigma\)-homomorphism.
\end{definition}

\begin{remark} \label{rmk:SCat_inverse}
  Notice that every finitely total \(\Sigma\)-monoid \((X,\Sigma)\) is trivially a commutative monoid \((X,+,0)\) where \(x + y\) is defined to be \(\Sigma \{x,y\}\) for every \(x, y \in X\) and \(0 = \Sigma \varnothing\); both associativity and commutativity follow from the combination of bracketing and flattening.
  Similarly, every \(\Sigma\)-group is trivially an abelian group and, hence, the inverse of an element in a \(\Sigma\)-group is unique.
  Moreover, thanks to the inverse mapping being a \(\Sigma\)-homomorphism, any summable family \(\fml{x} \in X\) (even infinite ones) have an `inverse' \(\fml{-x}\) obtained by applying the inverse mapping to each of its elements so that:
\begin{equation*}
  0 = \Sigma \{\Sigma \fml{x}, -(\Sigma \fml{x})\} = \Sigma \{\Sigma \fml{x}, \Sigma \fml{-x}\}.
\end{equation*}
\end{remark}

For each of these flavours of \(\Sigma\)-monoids we may consider a full subcategory of \(\SCat{w}\) as defined below.

\begin{definition} \label{def:SCat_subcats}
  Let \GLS{\(\SCat{s}\)}{SCats}, \GLS{\(\SCat{ft}\)}{SCatft} and \GLS{\(\SCat{g}\)}{SCatg} be the full subcategories of \(\SCat{w}\) obtained by restricting the class of objects to strong \(\Sigma\)-monoids, finitely total \(\Sigma\)-monoids and \(\Sigma\)-groups, respectively.
  For brevity, these subcategories together with \(\SCat{w}\) are referred to as \GLS{\(\SCat{*}\)}{SCatstar} categories.
\end{definition}

Since strong \(\Sigma\)-monoids do not admit inverses (see Proposition~\ref{prop:no_inverses}) it is clear that \(\Sigma\)-groups and strong \(\Sigma\)-monoids are disjoint subclasses of weak \(\Sigma\)-monoids.
The hierarchy of inclusions between \(\SCat{*}\) categories can be summarised by the existence of the following full and faithful embedding functors:
\begin{equation} \label{eq:SCat_inclusions}
  \SCat{s} \into \SCat{w} \quad\quad\quad\quad \SCat{g} \into \SCat{ft} \into \SCat{w}.
\end{equation}

\begin{proposition} \label{prop:SCat*_complete}
  All of the \(\SCat{*}\) categories are complete.
  All of the embedding functors from~\eqref{eq:SCat_inclusions} create limits.
\end{proposition} \begin{proof}
  Proposition~\ref{prop:SCatw_complete} already established the claim for \(\SCat{w}\).
  Whenever the objects in the diagram are in \(\SCat{s}\), it is straightforward to check that the construction from Proposition~\ref{prop:SCatw_complete} --- both for equalizers and small products --- yields a strong \(\Sigma\)-monoid and, hence, provides a cone in \(\SCat{s}\).
  The inclusion functor \(\SCat{s} \into \SCat{w}\) reflects limits, by virtue of being full and faithful; therefore, \(\SCat{s}\) is complete.
  The same argument applies to \(\SCat{ft}\) and \(\SCat{g}\) so they are all complete, as claimed.
  Since all \(\SCat{*}\) categories are complete and the embedding functors reflect limits, it is immediate that they also preserve them and, hence, all of the embedding functors from~\eqref{eq:SCat_inclusions} create limits.
\end{proof}

The following subsections establish that \(\SCat{s}\), \(\SCat{ft}\) and \(\SCat{g}\) are not only full subcategories of \(\SCat{w}\) but reflective ones --- \ie{} there is a left adjoint to each of the embedding functors in~\eqref{eq:SCat_inclusions}.

\subsection{\(\SCat{s}\) is a reflective subcategory of \(\SCat{w}\)}
\label{sec:s_w_adj}

The left adjoint to the canonical embedding \(\SCat{s} \into \SCat{w}\) is explicitly constructed in Proposition~\ref{prop:s_w_adj}.
The key insight is that the strong versions of bracketing and flattening can be recovered by defining a  congruence \(\sim\) and quotienting the set of all families with respect to it.
Then, we ought to find the `smallest' strong \(\Sigma\)-monoid defined on this quotient set that satisfies the universal property of the adjunction.
The following lemma is instrumental to find such a `smallest' strong \(\Sigma\)-monoid.

\begin{lemma} \label{lem:intersection_of_sSm}
  Let \(\{(X_i,\Sigma^i)\}_I\) be a collection of strong \(\Sigma\)-monoids.
  Define a partial function \(\Sigma^\cap \colon (\cap_I X_i)^* \pto \cap_I X_i\) as follows:
  \begin{equation*}
    \Sigma^\cap \fml{x} = \begin{cases}
      x &\ifc \forall i \in I,\ \Sigma^i \fml{x} \keq x \\
      \undefined &\otherwise.
    \end{cases}
  \end{equation*}
  Then, \((\cap_I X_i, \Sigma^\cap)\) is a strong \(\Sigma\)-monoid.
\end{lemma} \begin{proof}
  Whenever a family \(\fml{x}\) is summable in \(\Sigma^\cap\), it must be summable in \(\Sigma^i\) for all \(i \in I\) and, considering that each \((X_i,\Sigma^i)\) is a strong \(\Sigma\)-monoid, it follows that every subfamily of \(\fml{x}\) is summable in every \(\Sigma^i\) and, hence, it is summable in \(\Sigma^\cap\) so that subsummability is satisfied.
  Similarly, strong bracketing and strong flattening in \(\Sigma^\cap\) follow from those in each \(\Sigma^i\) and the neutral element and singleton axioms are trivially satisfied.
  Consequently, \((\cap_I X_i, \Sigma^\cap)\) is a strong \(\Sigma\)-monoid.
\end{proof}

\begin{proposition} \label{prop:s_w_adj}
  There is a left adjoint functor to the embedding \(\SCat{s} \into \SCat{w}\).
\end{proposition} \begin{proof}
  Let \(G \colon \SCat{s} \into \SCat{w}\) denote the canonical embedding; its left adjoint \(F \colon \SCat{w} \to \SCat{s}\) is defined explicitly.
  Let \(X \in \SCat{w}\) be an arbitrary weak \(\Sigma\)-monoid and define a relation \(\leadsto\) on \(X^*\) such that for any two families \(\fml{x},\fml{x'} \in X^*\), we have that \(\fml{x} \leadsto \fml{x'}\) iff there is a partition of \(\fml{x}\) into subfamilies \(\fml{x} = \uplus_J \fml{x}_j\) such that each (possibly infinite) subfamily \(\fml{x}_j\) is summable and such that \(\fml{x'}\) is the family of their sums \(\fml{x'} = \{\Sigma \fml{x}_j\}_J\).
  The relation \(\leadsto\) is reflexive since any family may be partitioned into its singleton subfamilies.
  Let \(\sim\) be the equivalence closure of \(\leadsto\); then, \(\fml{x} \sim \fml{x'}\) iff there is a zig-zag chain of \(\leadsto\) relations such as:
  \[\begin{tikzcd}[column sep=tiny]
    & {\fml{a_1}} && {\fml{a_3}} & \ldots & {\fml{a_n}} \\
    {\fml{x}} && {\fml{a}_2} && \ldots && {\fml{x'}}
    \arrow[squiggly, from=1-2, to=2-1]
    \arrow[squiggly, from=1-2, to=2-3]
    \arrow[squiggly, from=1-4, to=2-3]
    \arrow[squiggly, from=1-4, to=2-5]
    \arrow[squiggly, from=1-6, to=2-5]
    \arrow[squiggly, from=1-6, to=2-7]
  \end{tikzcd}\]
  Notice that \(\sim\) is a congruence with respect to arbitrary disjoint union, \ie{} it satisfies:
  \begin{equation} \label{eq:weak_strong_sim_cong}
    \forall j \in J,\, \fml{x}_j \sim \fml{x'}_j \implies \uplus_J \fml{x}_j \sim \uplus_J \fml{x'}_j
  \end{equation}
  since a zig-zag chain from \(\uplus_J \fml{x}_j\) to \(\uplus_J \fml{x'}_j\) may be obtained by composing the chains relating each subfamily \(\fml{x}_j\) to \(\fml{x'}_j\).
  Let \(A\) be the quotient set \(X^*/{\sim}\) and let \(q \colon X \to A\) be the function that maps each \(x \in X\) to the equivalence class \([\{x\}] \in A\).
  
  Let \(Y \in \SCat{s}\) be an arbitrary strong \(\Sigma\)-monoid and let \(f \colon X \to G(Y)\) be an arbitrary \(\Sigma\)-homomorphism.
  Define \(\leadsto\) on \(Y^*\) as above; since \(f\) is a \(\Sigma\)-homomorphism we have that \(\fml{x} \leadsto \fml{x'}\) implies \(f\fml{x} \leadsto f\fml{x'}\) and, consequently,
  \begin{equation} \label{eq:sim_hom}
    \fml{x} \sim \fml{x'} \implies f\fml{x} \sim f\fml{x'}.
  \end{equation}
  Moreover, in a strong \(\Sigma\)-monoid \(\fml{y} \leadsto \fml{y'}\) implies \(\Sigma \fml{y} \keq \Sigma \fml{y'}\), since if \(\fml{y}\) is summable then \(\fml{y'}\) is summable due to strong bracketing whereas if \(\fml{y'}\) is summable then \(\fml{y}\) is summable due to strong flattening.
  Consequently, we have that for any two families \(\fml{x},\fml{x'} \in X^*\):
  \begin{equation} \label{eq:s_w_sim_implies}
    \fml{x} \sim \fml{x'} \implies \Sigma f\fml{x} \keq \Sigma f\fml{x'}.
  \end{equation}
  Let \(A_f\) be the subset of \(A\) defined as follows:
  \begin{equation*}
    A_f = \{[\fml{x}] \in A \mid f\fml{x} \text{ is summable in } Y\}.
  \end{equation*}
  Notice that, for any equivalence class \([\fml{x}] \in A_f\), each of its families \(\fml{x'} \in [\fml{x}]\) satisfies that \(f\fml{x'}\) is summable, due to the implication~\eqref{eq:s_w_sim_implies} and the fact that \(f\fml{x}\) is summable.
  Let \(\Sigma^f \colon A_f^* \pto A_f\) be the following partial function:
  \begin{equation*}
    \Sigma^f \{[\fml{x}_i]\}_I = \begin{cases}
      [\uplus_I \fml{x}_i] &\ifc [\uplus_I \fml{x}_i] \in A_f \\
      \undefined &\otherwise.
    \end{cases}
  \end{equation*}
  Notice that this definition is independent of the choice of representatives thanks to \(\sim\) being a congruence~\eqref{eq:weak_strong_sim_cong}.
  We must check that \((A_f,\Sigma^f)\) is a strong \(\Sigma\)-monoid.
  \begin{itemize}
    \item \emph{Singleton.} Singleton families are trivially summable.
    \item \emph{Subsummability.} If \(\{[\fml{x}_i]\}_I \in A_f^*\) is summable then \(\uplus_I f\fml{x}_i\) is summable in \(Y\) and so is \(\uplus_J f\fml{x}_j\) for any \(J \subseteq I\) due to subsummability in \(Y\). Thus, \(\{[\fml{x}_j]\}_J\) is summable in \(A_f\) and \(\Sigma^f\) satisfies subsummability.
    \item \emph{Neutral element.} Follows immediately from subsummability.
    \item \emph{Strong flattening.} Let \(I\) be an arbitrary set partitioned into \(I = \uplus_J I_j\) and, for each \(j \in J\), let \(\{[\fml{x}_i]\}_{I_j}\) be a summable family in \(A_f\), \ie{} \(\Sigma^f \{[\fml{x}_i]\}_{I_j} \keq [\uplus_{I_j} \fml{x}_i]\). Assume that the family \(\{[\uplus_{I_j} \fml{x}_i]\}_J \in A_f^*\) is summable, \ie{} \(\Sigma^f \{[\uplus_{I_j} \fml{x}_i]\}_J \keq [\uplus_J (\uplus_{I_j} \fml{x}_i)]\); then, due to associativity and commutativity of \(\uplus\) we have that \(\uplus_I \fml{x}_i = \uplus_J (\uplus_{I_j} \fml{x}_i)\) and, hence, \([\uplus_I \fml{x}_i] \in A_f\) so that \(\{[\fml{x}_i]\}_I\) is summable in \(A_f\).
    \item \emph{Strong bracketing.} As in the case of flattening, this axiom follows from associativity and commutativity of disjoint union.
  \end{itemize}

  Let \(S\) be the set of all such strong \(\Sigma\)-monoids \((A_f,\Sigma^f)\).\footnote{Even though \(f \colon X \to G(Y)\) ranges over all \(Y \in \SCat{s}\), each \(A_f\) is by definition a subset of \(A\) which is small since \(X^*\) is small (see Proposition~\ref{prop:family_set}), hence, \(S\) is a (small) set.}
  Let \(F(X)\) be the strong \(\Sigma\)-monoid obtained as the intersection of all the members of \(S\) as per Lemma~\ref{lem:intersection_of_sSm}.
  Let \(\eta_X \colon X \to GF(X)\) be the function that maps each \(x \in X\) to \([\{x\}]\); such an equivalence class is present in every \(A_f\), so it is present in \(F(X)\).
  Notice that \(\eta_X\) is a \(\Sigma\)-homomorphism: if \(\fml{x} = \{x_i\}_I \in X^*\) is summable, then \(\fml{x} \leadsto \{\Sigma \fml{x}\}\) and \([\fml{x}] = [\{\Sigma \fml{x}\}]\), hence,
  \begin{equation*}
    \Sigma^\cap \{\eta_X(x_i)\}_I \keq [\uplus_I \{x_i\}] = [\fml{x}] = [\{\Sigma \fml{x}\}] = \eta_X(\Sigma \fml{x}).
  \end{equation*}
  On morphisms \(f \colon X \to X'\), let \(F(f)\) be the function that maps \([\fml{x}] \in F(X)\) to \([f\fml{x}] \in F(X')\); this is a \(\Sigma\)-homomorphism thanks to~\eqref{eq:sim_hom}.
  It is straightforward to check that this construction yields a functor \(F \colon \SCat{w} \to \SCat{s}\) and a natural transformation \(\eta\) whose components \(\eta_X\) were defined above.

  It remains to check that \(F\) is left adjoint to \(G\).
  Fix some arbitrary \(X \in \SCat{w}\), \(Y \in \SCat{s}\) and \(f \in \SCat{w}(X,G(Y))\).
  There is a unique \(\Sigma\)-homomorphism \(\bar{f} \colon F(X) \to Y\) making the diagram
  \[\begin{tikzcd}
    X && {G(Y)} \\
    && {GF(X)}
    \arrow["f", from=1-1, to=1-3]
    \arrow["{\eta_X}"', from=1-1, to=2-3]
    \arrow["{G(\bar{f})}"', dashed, from=2-3, to=1-3]
  \end{tikzcd}\]
  commute: the requirement that \(f = G(\bar{f}) \circ \eta_X\) imposes that \(\bar{f}\) maps each \(\eta_X(x) = [\{x\}]\) to \(f(x)\); the requirement that \(\bar{f}\) is a \(\Sigma\)-homomorphism imposes that, for any arbitrary equivalence class \([\{x_i\}_I] \in F(X)\),
  \begin{equation*}
    \bar{f} [\{x_i\}_I] = \bar{f} [\uplus_I \{x_i\}] = \bar{f} (\Sigma^\cap \{\eta_X(x_i)\}_I) = \Sigma \{f(x_i)\}_I.
  \end{equation*}
  The value of such a \(\Sigma\)-homomorphism \(\bar{f}\) is uniquely determined and, thus, it has been shown that \(F\) is left adjoint to \(G\), as claimed.
\end{proof}

Since \(\SCat{s}\) is a full subcategory of \(\SCat{w}\) it follows that \(\SCat{s}\) is a reflective subcategory of \(\SCat{w}\).
Let \(X\) be a weak \(\Sigma\)-monoid containing a pair of elements \(a,b \in X\) such that \(\Sigma \{a,b\} \keq 0\).
The family \(\fml{x} = \{a\}_\Nset \uplus \{b\}_\Nset\) that contains an infinite number of copies of \(a\) and \(b\) satisfies that \(\fml{x} \leadsto \{a\}\), \(\fml{x} \leadsto \{0\}\) and \(\fml{x} \leadsto \{b\}\) depending on how we partition it.
Then, \([\{a\}] = [\{0\}] = [\{b\}]\), implying \(\eta_X\) is not injective and, in fact, every element of \(X\) that has an additive inverse is mapped to the neutral element \([\{0\}]\) of \(F(X)\).
This is to be expected since, according to Proposition~\ref{prop:no_inverses}, strong \(\Sigma\)-monoids cannot have additive inverses.

\subsection{\(\SCat{ft}\) is a reflective subcategory of \(\SCat{w}\)}
\label{sec:ft_w_adj}

The canonical embedding \(\SCat{ft} \into \SCat{w}\) has a left adjoint.
In this case, an explicit construction of the free finitely total \(\Sigma\)-monoid has not been achieved; instead, the proof uses the general adjoint functor theorem.

\begin{proposition} \label{prop:ft_w_adj}
  There is a left adjoint functor to the embedding \(\SCat{ft} \into \SCat{w}\).
\end{proposition} \begin{proof}
  Let \(G \colon \SCat{ft} \into \SCat{w}\) denote the canonical embedding.
  It is immediate that \(\SCat{ft}\) is locally small since there is a faithful forgetful functor \(\SCat{ft} \to \Set\).
  According to Proposition~\ref{prop:SCat*_complete}, the category \(\SCat{ft}\) is complete and the embedding functor \(G\) preserves limits.
  It remains to show that, for any \(X \in \SCat{w}\), the comma category \(\comma{X}{G}\) has a weakly initial set; then, the existence of a left adjoint will follow from the general adjoint functor theorem (see Theorem~\ref{thm:GAFT}).
  Let \(Y \in \SCat{ft}\) be an arbitrary finitely total \(\Sigma\)-monoid and let \(f \colon X \to G(Y)\) be an arbitrary \(\Sigma\)-homomorphism.
  Let \(Z\) be the following subset of \(X^*\):
  \begin{equation*}
    Z = \{\fml{x} \in X^* \mid f\fml{x} \text{ is summable in } Y\}
  \end{equation*}
  and let \(\sim\) be the equivalence relation on \(Z\) where:
  \begin{equation*}
    \fml{x} \sim \fml{x'} \iff \Sigma f\fml{x} = \Sigma f\fml{x'}.
  \end{equation*}
  Let \(A_f\) be the quotient set \(Z/{\sim}\) and let \(\bar{f} \colon A_f \to Y\) be the injective function that maps each \([\fml{x}] \in A_f\) to \(\Sigma f\fml{x}\).
  According to Lemma~\ref{lem:wSm_restriction}, there is a partial function \(\Sigma^{\bar{f}} \colon A_f^* \pto A_f\) induced by \(\bar{f}\):
  \begin{equation*}
    \Sigma^{\bar{f}} \{[\fml{x}_j]\}_J = \begin{cases}
      [\fml{x}] &\ifc \exists [\fml{x}] \in A_f \st \Sigma \{\bar{f}[\fml{x}_j]\}_J \keq \bar{f}[\fml{x}] \\
      \undefined &\otherwise.
    \end{cases}
  \end{equation*}
  which, by virtue of \(\bar{f}\) being injective, satisfies that \((A_f,\Sigma^{\bar{f}})\) is a weak \(\Sigma\)-monoid.
  Using the definition of \(\bar{f}\) we obtain the following equivalent definition of \(\Sigma^{\bar{f}}\):
  \begin{equation*}
    \Sigma^{\bar{f}} \{[\fml{x}_j]\}_J = \begin{cases}
      [\fml{x}] &\ifc \exists \fml{x} \in Z \st \Sigma \{\Sigma f\fml{x}_j\}_J \keq \Sigma f\fml{x} \\
      \undefined &\otherwise.
    \end{cases}
  \end{equation*}
  The definition is independent of the choice of representatives since \(Z\) is comprised of families whose image under \(f\) is summable and \(\sim\) relates those that sum up to the same value in \(Y\).
  Recall that each \(\fml{x}_j\) is in \(Z\) by definition and, consequently, each \(f\fml{x}_j\) is summable; then, when \(J\) is a finite set the family \(\{\Sigma f\fml{x}_j\}_J\) is summable thanks to \(Y\) being finitely total and \(\Sigma \{\Sigma f\fml{x}_j\}_J \keq \Sigma (\uplus_J f\fml{x}_j)\) due to flattening.
  Therefore, whenever \(J\) is finite we have that:
  \begin{equation*}
    \Sigma^{\bar{f}} \{[\fml{x}_j]\}_J \keq [\uplus_J \fml{x}_j]
  \end{equation*}
  implying that every finite family is summable in \(\Sigma^{\bar{f}}\) so that \((A_f,\Sigma^{\bar{f}}) \in \SCat{ft}\).

  Let \(q \colon X \to G(A_f)\) be the function that maps each \(x \in X\) to \([\{x\}]\).
  Assume that \(\fml{x} = \{x_i\}_I \in X^*\) is summable, then \(f\fml{x}\) is summable in \(Y\) due to \(f\) being a \(\Sigma\)-homomorphism; thus, \(\fml{x} \in Z\) and \(q\) is a \(\Sigma\)-homomorphism:
  \begin{align*}
    q(\Sigma \fml{x}) &= [\{\Sigma \fml{x}\}] && \text{(def\@. of \(q\))} \\
      &= [\fml{x}] && \text{(\(f(\Sigma \fml{x}) = \Sigma f\fml{x}\))} \\
      &\keq \Sigma^{\bar{f}} \{[\{x_i\}]\}_I && \text{(\(\Sigma f\fml{x} = \Sigma \{\Sigma \{f(x_i)\}\}_I\))} \\
      &= \Sigma^{\bar{f}} \{q(x_i)\}_I. && \text{(def\@. of \(q\))}
  \end{align*}
  According to Lemma~\ref{lem:wSm_restriction}, the function \(\bar{f} \colon A_f \to Y\) is also a \(\Sigma\)-homomorphism and it is straightforward to check that \(f = G(\bar{f}) \circ q\).
  Thus, it has been shown that for every object \((f,Y)\) in the comma category \(\comma{X}{G}\) there is an object \((q,A_f)\) such that a morphism \(\bar{f} \colon (q,A_f) \to (f,Y)\) exists.
  Let \(S\) be the collection of all such \((q,A_f)\) objects; notice that \(S\) is a (small) set since each \(A_f\) is the quotient of a subset of \(X^*\) (which is itself small, Proposition~\ref{prop:family_set}), so the collection of all such \(A_f\) is small and, for each of them, there is only a small collection of partial functions \(A_f^* \pto A_f\).
  Consequently, \(S\) is a weakly initial set in the comma category \(\comma{X}{G}\) and the claim that \(G\) has a left adjoint follows from the general adjoint functor theorem (see Theorem~\ref{thm:GAFT}).
\end{proof}

Since \(\SCat{ft}\) is a full subcategory of \(\SCat{w}\) it follows that \(\SCat{ft}\) is a reflective subcategory of \(\SCat{w}\).
The proof given above follows the same strategy that Hoshino~\cite{RTUDC} used to show that the category of totally defined strong \(\Sigma\)-monoids is a reflective subcategory of \(\SCat{s}\).
It is reasonable to think that a constructive proof may be achieved using the following strategy: given any weak \(\Sigma\)-monoid \((X,\Sigma)\), let \((X,+)\) be the partially commutative monoid obtained by defining \(x+x' \keq \Sigma \{x,x'\}\) and let \((Y,+')\) be its free commutative monoid constructed via the left adjoint to \(\CMon \into \cat{PCM}\) (see~\cite{JacobsPCM}); then, we may attempt to extend \((Y,+')\) to a \(\Sigma\)-monoid where the only infinite families that are summable are those required for \(\eta_X \colon (X,\Sigma) \to (Y,\Sigma')\) to be a \(\Sigma\)-homomorphism, plus the ones imposed by bracketing and flattening.
Such a strategy was attempted but, unfortunately, it was not clear whether quotienting the underlying set \(Y\) would be necessary to guarantee that different bracketings of the same infinite family would agree in their sum and, if quotienting is required, verifying that the definition of \(\Sigma'\) is independent of the choice of representatives becomes exceedingly subtle and it was not pursued further.

\subsection{\(\SCat{g}\) is a reflective subcategory of \(\SCat{ft}\)}

The canonical embedding \(\SCat{g} \into \SCat{ft}\) has a left adjoint.
An explicit construction of the free \(\Sigma\)-group has not been achieved; instead, the proof uses the general adjoint functor theorem.
The proof is a combination of the Grothendieck group construction (Example~\ref{ex:Ab_CMon_adj}) and the proof of Proposition~\ref{prop:ft_w_adj} from the previous subsection.

\begin{proposition} \label{prop:g_ft_adj}
  There is a left adjoint functor to the embedding \(\SCat{g} \into \SCat{ft}\).
\end{proposition} \begin{proof}
  Let \(G \colon \SCat{g} \into \SCat{ft}\) denote the canonical embedding.
  It is immediate that \(\SCat{g}\) is locally small, since there is a faithful forgetful functor \(\SCat{g} \to \Set\).
  According to Proposition~\ref{prop:SCat*_complete}, the category \(\SCat{g}\) is complete and the embedding functor \(G\) preserves limits.
  It remains to show that, for any \(X \in \SCat{ft}\), the comma category \(\comma{X}{G}\) has a weakly initial set; then, the existence of a left adjoint will follow from the general adjoint functor theorem (see Theorem~\ref{thm:GAFT}).
  Let \(Y \in \SCat{g}\) be an arbitrary \(\Sigma\)-group and let \(f \colon X \to G(Y)\) be an arbitrary \(\Sigma\)-homomorphism.
  Let \(Z\) be the following subset of \(X^* \times X^*\):
  \begin{equation*}
    Z = \{(\fml{p},\fml{n}) \in X^* \times X^* \mid f\fml{p} \text{ and } f\fml{n} \text{ are summable in } Y\}
  \end{equation*}
  and let \(\sim\) be the equivalence relation on \(Z\) where:
  \begin{equation*}
    (\fml{p},\fml{n}) \sim (\fml{p'},\fml{n'}) \iff \Sigma \{\Sigma f\fml{p}, -(\Sigma f\fml{n})\} = \Sigma \{\Sigma f\fml{p'}, -(\Sigma f\fml{n'})\}.
  \end{equation*}
  Let \(A_f\) be the quotient set \(Z/{\sim}\) and let \(\bar{f} \colon A_f \to Y\) be the injective function defined as follows for every equivalence class \([(\fml{p},\fml{n})] \in A_f\):
  \begin{equation*}
    \bar{f}[(\fml{p},\fml{n})] = \Sigma \{\Sigma f\fml{p}, -\Sigma f\fml{n}\}.
  \end{equation*}
  According to Lemma~\ref{lem:wSm_restriction}, there is a partial function \(\Sigma^{\bar{f}} \colon A_f^* \pto A_f\) induced by \(\bar{f}\):
  \begin{equation*}
    \Sigma^{\bar{f}} \{[(\fml{p}_j,\fml{n}_j)]\}_J = \begin{cases}
      [(\fml{p},\fml{n})] &\ifc \exists (\fml{p},\fml{n}) \in Z \st \Sigma \{\bar{f}[(\fml{p}_j,\fml{n}_j)]\}_J \keq \bar{f}[(\fml{p},\fml{n})] \\
      \undefined &\otherwise.
    \end{cases}
  \end{equation*}
  which, by virtue of \(\bar{f}\) being injective, satisfies that \((A_f,\Sigma^{\bar{f}})\) is a weak \(\Sigma\)-monoid.
  Next, we show that every family \(\{(\fml{p}_j,\fml{n}_j) \in A_f\}_J\) such that \(J\) is finite is a summable family.
  Notice that, by definition of \(A_f\), if \((\fml{p}_j,\fml{n}_j) \in A_f\) for all \(j \in J\) then \(f\fml{p}_j\) and \(f\fml{n}_j\) are summable in \(Y\) and, due to \(J\) being finite and \(Y\) finitely total, \(\{f\fml{p}_j\}_J\) and \(\{f\fml{n}_j\}_J\) are summable as well.
  It follows that:
  \begin{align*}
    \Sigma \{\bar{f}[(\fml{p}_j,\fml{n}_j)]\}_J &= \Sigma \{\Sigma \{\Sigma f\fml{p}_j, -\Sigma f\fml{n}_j\}\}_J && \text{(def\@. \(\bar{f}\))} \\
    &\keq \Sigma \{ \Sigma \{\Sigma f\fml{p}_j\}_J, \Sigma \{-\Sigma f\fml{n}_j\}_J\} && \text{(flattening, then bracketing)} \\
    &\keq \Sigma \{ \Sigma (\uplus_J f\fml{p}_j), -\Sigma (\uplus_J f\fml{n}_j) \} && \text{(inversion \(\Sigma\)-hom\@. and flattening)} \\
    &= \bar{f}[(\uplus_J \fml{p}_j, \uplus_J \fml{n}_j)] && \text{(def\@. \(\bar{f}\))}
  \end{align*}
  where the first \(\keq\) follows from the fact that \(J\) is finite so that \(\Sigma \{\Sigma \{a_j,-b_j\}\}_J\) can be first flattened into a sum of a finite family and then bracketed into \(\Sigma \{\Sigma \{a_j\}_J, \Sigma \{-b_j\}_J\}\).
  According to the definition of \(\Sigma^{\bar{f}}\), the above implies that
  \begin{equation*}
    \Sigma^{\bar{f}} \{[(\fml{p}_j,\fml{n}_j)]\}_J \keq [(\uplus_J \fml{p}_j, \uplus_J \fml{n}_j)]
  \end{equation*}
  so that every finite family in \(A_f^*\) is summable.
  Finally, for every element \([(\fml{p},\fml{n})] \in A_f\) the element \([(\fml{n},\fml{p})] \in A_f\) acts as its inverse; to check this, notice that the neutral element of \(A_f\) is \([(\varnothing,\varnothing)]\) which satisfies \(\bar{f}[(\varnothing,\varnothing)] = 0\) and let \(a = \Sigma f\fml{p}\) and \(b = \Sigma f\fml{n}\) so that:
  \begin{align*}
    \Sigma \{\bar{f}[(\fml{p},\fml{n})],\bar{f}[(\fml{n},\fml{p})]\} &= \Sigma \{\Sigma \{a,-b\},\Sigma \{b,-a\}\} &&\text{(def. of \(\bar{f}\))} \\
      &= \Sigma \{a,-b,b,-a\} &&\text{(flattening in \(Y\))} \\
      &= \Sigma \{\Sigma \{a,-a\},\Sigma \{b,-b\}\} &&\text{(bracketing in \(Y\))} \\
      &= \Sigma \{0,0\} = 0 &&\text{(inverses in \(Y\))}
  \end{align*}
  thus, \(\Sigma^{\bar{f}} \{[(\fml{p},\fml{n})],[(\fml{n},\fml{p})]\} = [(\varnothing,\varnothing)]\).
  It is straightforward to check that:
  \begin{equation*}
    \Sigma^{\bar{f}} \{[(\fml{p}_j,\fml{n}_j)]\}_J \keq [(\fml{p},\fml{n})] \iff \Sigma^{\bar{f}} \{[(\fml{n}_j,\fml{p}_j)]\}_J \keq [(\fml{n},\fml{p})]
  \end{equation*}
  so that the inverse mapping is a \(\Sigma\)-homomorphism.
  Thus, it has been established that \((A_f,\Sigma^{\bar{f}})\) is a \(\Sigma\)-group.

  Let \(q \colon X \to G(A_f)\) be the function that maps each \(x \in X\) to \([(\{x\},\varnothing)]\).
  It is straightforward to check that \(q\) is a \(\Sigma\)-homomorphism since for every summable family \(\fml{x} = \{x_i\}_I \in X^*\) it follows that:
  \begin{equation*}
    q(\Sigma \fml{x}) = [(\{\Sigma \fml{x}\},\varnothing)] = [(\fml{x},\varnothing)] \keq \Sigma^{\bar{f}} \{q(x_i)\}_I.
  \end{equation*}
  According to Lemma~\ref{lem:wSm_restriction}, the function \(\bar{f} \colon A_f \to Y\) is also a \(\Sigma\)-homomorphism and it is straightforward to check that \(f = G(\bar{f}) \circ q\).
  Thus, it has been shown that for every object \((f,Y)\) in the comma category \(\comma{X}{G}\) there is an object \((q,A_f)\) such that a morphism \(\bar{f} \colon (q,A_f) \to (f,Y)\) exists.
  Let \(S\) be the collection of all such \((q,A_f)\) objects; notice that \(S\) is a (small) set since each \(A_f\) is a quotient of a subset of \(X^*\) (which is itself small, Proposition~\ref{prop:family_set}), so the collection of all such \(A_f\) is small and, for each of them, there is only a small collection of partial functions \(A_f^* \pto A_f\).
  Consequently, \(S\) is a weakly initial set in the comma category \(\comma{X}{G}\) and the claim that \(G\) has a left adjoint follows from the general adjoint functor theorem (see Theorem~\ref{thm:GAFT}).
\end{proof}

Since \(\SCat{g}\) is a full subcategory of \(\SCat{ft}\) it follows that \(\SCat{g}\) is a reflective subcategory of \(\SCat{ft}\).
A constructive proof may perhaps be achieved by adapting the construction of the Grothendieck group (see Example~\ref{ex:Ab_CMon_adj}) to infinitary addition \(\Sigma\).
Unfortunately, as in the case of the left adjoint to \(\SCat{ft} \into \SCat{w}\), it is unclear whether quotienting would be necessary to guarantee that the bracketing axiom is satisfied and, if quotienting is required, proving well-definedness of \(\Sigma\) becomes exceedingly subtle.

\subsection{Tensor product of \(\Sigma\)-monoids}
\label{sec:SCat_tensor}

The results in this section are adapted from those in the appendix of Hoshino's work on the category of strong \(\Sigma\)-monoids~\cite{RTUDC}.
The proofs presented here are more detailed and apply to all of the \(\SCat{*}\) categories but, in essence, the strategy they follow is due to Hoshino.

Proposition~\ref{prop:SCat*_complete} already established that each \(\SCat{*}\) category has small products; denote the terminal object as \((\{0\},\Sigma) \in \SCat{*}\) whose \(\Sigma\) is the unique total function of its type. Consequently, \(\SCat{*}\) can be given a Cartesian monoidal structure.

\begin{proposition}
  For each \(\SCat{*}\) category, let \(\times \colon \SCat{*} \times \SCat{*} \to \SCat{*}\) be the functor mapping each pair of objects to their categorical product and each pair of morphisms to the universal morphism that makes the following diagram
  \[\begin{tikzcd}[row sep=small, column sep=small]
    X && {X \times Y} && Y \\
    \\
    Z && {Z \times W} && W
    \arrow["{\pi_l}"', from=1-3, to=1-1]
    \arrow["{\pi_l}", from=3-3, to=3-1]
    \arrow["{\pi_r}", from=1-3, to=1-5]
    \arrow["{\pi_r}"', from=3-3, to=3-5]
    \arrow["g", from=1-5, to=3-5]
    \arrow["f"', from=1-1, to=3-1]
    \arrow[dashed, from=1-3, to=3-3]
\end{tikzcd}\]
  commute.
  Then, each \((\SCat{*},\times,\{0\})\) is a symmetric monoidal category.
\end{proposition} \begin{proof}
  This construction is well-known for categories with finitary products.
  Unitors are given by projections and associators are given by the fact that both \((X \times Y) \times Z\) and \(X \times (Y \times Z)\) are categorical products of the collection \(\{X,Y,Z\}\), so they are isomorphic.
  The (symmetric) braiding is given by the fact that \(X \times Y\) and \(Y \times X\) are both products and, therefore, isomorphic.
  All coherence diagrams follow from the universal properties of categorical products and terminal object.
\end{proof}

Unfortunately such a Cartesian monoidal structure is not suitable for our purposes: we are interested in defining a monoidal structure on each \(\SCat{*}\) category so that categories enriched over \(\SCat{*}\) satisfy that their composition distributes over addition.
As discussed in Section~\ref{sec:enriched_cat}, to do so we need that each of the \(\SCat{*}\) categories has tensor products and use these to define their monoidal structure.
Such is the goal of this section, which begins by defining the notion of \(\Sigma\)-bilinear function.

\begin{definition} \label{def:bihom}
  For each \(\SCat{*}\) category, let \(X,Y,Z \in \SCat{*}\) be arbitrary objects and let \(f \colon X \times Y \to Z\) be a function.
  We say \(f\) is a \emph{\(\Sigma\)-bilinear} function if \(f(x,-)\) and \(f(-,y)\) are \(\Sigma\)-homomorphisms for all \(x \in X\) and all \(y \in Y\).
  Let \(\SCat{*}^{X,Y}(Z)\) be the set of all \(\Sigma\)-bilinear functions of type \(X \times Y \to Z\).
\end{definition}

\begin{remark} \label{rmk:bihom_inverse}
  It immediately follows from the definition of a \(\Sigma\)-bilinear function \(f \colon X \times Y \to Z\) that \(f(x,y) = 0\) if \(x = 0\) or \(y = 0\):
  \begin{equation*}
    f(x,0) = f(x,\Sigma \varnothing) = \Sigma \varnothing = 0.
  \end{equation*}
  Furthermore, if \(x \in X\) has an inverse then \(f(\Sigma \{x,-x\}, y) = f(0,y) = 0\) and, due to \(f\) being \(\Sigma\)-bilinear:
  \begin{equation*}
    \Sigma \{f(x,y), f(-x,y)\} \keq 0.
  \end{equation*}
  When the codomain \(Z\) is a \(\Sigma\)-group the uniqueness of the inverse element implies \(-f(x,y) = f(-x,y)\) and similarly \(-f(x,y) = f(x,-y)\) if \(y \in Y\) has an inverse.
\end{remark}

For each \(\SCat{*}\) category and every pair of objects \(X,Y \in \SCat{*}\) there is a functor \(\SCat{*}^{X,Y} \colon \SCat{*} \to \Set\) that maps each object \(Z \in \SCat{*}\) to the set of \(\Sigma\)-bilinear functions \(\SCat{*}^{X,Y}(Z)\) and maps each \(\Sigma\)-homomorphism \(f \colon Z \to W\) to a function \(\SCat{*}^{X,Y}(Z) \to \SCat{*}^{X,Y}(W)\) that maps each \(h \in \SCat{*}^{X,Y}(Z)\) to the \(\Sigma\)-bilinear function \(f \circ h\).

\begin{lemma} \label{lem:bihom_set_functor_preserves_limits}
  For all \(X,Y \in \SCat{w}\), the functor \(\SCat{w}^{X,Y} \colon \SCat{w} \to \Set\) preserves limits.
\end{lemma} \begin{proof}
  Let \(Z,W \in \SCat{w}\) be arbitrary objects, let \(f,g \in \SCat{w}(Z,W)\) be arbitrary \(\Sigma\)-homomorphisms and let \(E \in \SCat{w}\) be their equalizer constructed as in Proposition~\ref{prop:SCatw_complete}. Recall that \(E \subseteq Z\) and \(z \in E\) implies \(f(z) = g(z)\) so it is clear that any \(h \in \SCat{w}^{X,Y}(E)\) satisfies \(f \circ h = g \circ h\).
  Conversely, if a \(\Sigma\)-bilinear function \(h \colon X \times Y \to Z\) satisfies \(f \circ h = g \circ h\) then, for every element \(z \in \im(h)\), it holds that \(f(z) = g(z)\), so \(\im(h) \subseteq E\).
  Therefore, \(\SCat{w}^{X,Y}(E)\) is precisely the subset of \(\SCat{w}^{X,Y}(Z)\) such that for all \(h \in \SCat{w}^{X,Y}(E)\),
  \begin{equation*}
    \SCat{w}^{X,Y}(f)(h) = f \circ h = g \circ h = \SCat{w}^{X,Y}(g)(h)
  \end{equation*}
  \ie{} \(\SCat{w}^{X,Y}(E)\) is an equalizer of the diagram in \(\Set\).
  Thus, \(\SCat{w}^{X,Y}\) preserves equalizers.

  Let \(Z \times W \in \SCat{w}\) be the categorical product of two weak \(\Sigma\)-monoids and let \(\pi_l\) and \(\pi_r\) be the corresponding projections.
  Since the categorical product in \(\Set\) is given by the Cartesian product, there is a unique function \(m\) making the following diagram
  \[\begin{tikzcd}[column sep=small]
    && {\SCat{w}^{X,Y}(Z \times W)} \\
    \\
    {\SCat{w}^{X,Y}(Z)} && {\SCat{w}^{X,Y}(Z) \times \SCat{w}^{X,Y}(W)} && {\SCat{w}^{X,Y}(W)}
    \arrow["{\SCat{w}^{X,Y}(\pi_l)}"', from=1-3, to=3-1]
    \arrow["{\pi_l}", from=3-3, to=3-1]
    \arrow["m", dashed, from=1-3, to=3-3]
    \arrow["{\SCat{w}^{X,Y}(\pi_r)}", from=1-3, to=3-5]
    \arrow["{\pi_r}"', from=3-3, to=3-5]
  \end{tikzcd}\]
  commute in \(\Set\).
  Such a function \(m\) maps every \(\Sigma\)-bilinear function \(f \in \SCat{w}^{X,Y}(Z \times W)\) to the pair of \(\Sigma\)-bilinear functions \((\pi_l \circ f, \pi_r \circ f)\).
  Conversely, there is a function \(u \colon \SCat{w}^{X,Y}(Z) \times \SCat{w}^{X,Y}(W) \to \SCat{w}^{X,Y}(Z \times W)\) mapping each pair of \(\Sigma\)-bilinear functions \((g,h)\) to the function given below:
  \begin{equation*}
    k(x,y) = (g(x,y),h(x,y))
  \end{equation*}
  for every \(x \in X\) and \(y \in Y\).
  Notice that \(u\) is well-defined since \(k\) is a \(\Sigma\)-bilinear function:
  \begin{align*}
    k(\Sigma \fml{x},y) &= (g(\Sigma \fml{x},y),h(\Sigma \fml{x},y)) = (\Sigma g(\fml{x},y),\Sigma h(\fml{x},y)) \\
      &= \Sigma^\times (g(\fml{x},y), h(\fml{x},y)) = \Sigma^\times k(\fml{x},y)
  \end{align*}
  where \(g(\fml{x},y)\) is a shorthand for the family obtained after applying \(g(-,y)\) to each element in \(\fml{x}\).
  It is straightforward to check that \(u \circ m = \id\) and \(m \circ u = \id\) so that \(\SCat{w}^{X,Y}(Z \times W)\) is isomorphic to \(\SCat{w}^{X,Y}(Z) \times \SCat{w}^{X,Y}(W)\).
  Thus, it follows that \(\SCat{w}^{X,Y}(Z \times W)\) is a categorical product and \(\SCat{w}^{X,Y}\) preserves binary products.
  It is straightforward to generalise this argument to small products and, hence, \(\SCat{w}^{X,Y}\) preserves small products.

  A functor that preserves equalizers and small products and whose domain is a complete category automatically preserves all limits. Therefore, \(\SCat{w}^{X,Y}\) preserves limits, as claimed.
\end{proof}

\begin{corollary} \label{cor:bihom_set_functor_preserves_limits}
  For every \(\SCat{*}\) category and every pair of objects \(X,Y \in \SCat{*}\) the functor \(\SCat{*}^{X,Y} \colon \SCat{*} \to \Set\) preserves limits.
\end{corollary} \begin{proof}
  The previous lemma establishes the claim for \(\SCat{w}\).
  For the rest of the \(\SCat{*}\) categories, it is evident that the corresponding functor \(\SCat{*}^{X,Y}\) is equal to the composition \(\SCat{w}^{U(X),U(Y)} \circ U\) where \(U\) is the embedding functor \(\SCat{*} \into \SCat{w}\).
  Proposition~\ref{prop:SCat*_complete} established that \(U\) preserves limits, so it follows that \(\SCat{*}^{X,Y}\) preserves limits.
\end{proof}

\begin{lemma} \label{lem:SCatw_tensor}
 The comma category \(\comma{\{\bullet\}}{\SCat{w}^{X,Y}}\) has an initial object. Such an initial object will be denoted \((p,X \otimes Y)\).
\end{lemma} \begin{proof}
  The category \(\SCat{w}\) is complete (see Proposition~\ref{prop:SCatw_complete}) and
  Lemma~\ref{lem:bihom_set_functor_preserves_limits} establishes that the functor \(\SCat{w}^{X,Y}\) preserves limits.
  Consequently, the comma category \(\comma{\{\bullet\}}{\SCat{w}^{X,Y}}\) is complete (see Lemma A.2 from~\cite{Leinster}) and, since \(\SCat{w}\) is locally small, it follows that the comma category is locally small.
  Then, according to Lemma~\ref{lem:weakly_initial_set}, it suffices to provide a weakly initial set to prove that the comma category \(\comma{\{\bullet\}}{\SCat{w}^{X,Y}}\) has an initial object.
  The elements of such a weakly initial set ought to be functions of type \(\{\bullet\} \to \SCat{w}^{X,Y}(A)\) for some \(A \in \SCat{w}\); but a function with singleton domain simply selects an element in its codomain, so it is equivalent to think of the weakly initial set as a collection of \(\Sigma\)-bilinear functions \(X \times Y \xto{q} A\).
  Let \(S\) be the collection of all \(\Sigma\)-bilinear functions with domain \(X \times Y\) and whose codomain may be any weak \(\Sigma\)-monoid whose underlying set is a quotient of a subset of \((X \times Y)^*\).
  Notice that \(S\) is small since all quantifiers in its definition are with respect to fixed sets \(X\) and \(Y\)and \((X \times Y)^*\) is small (see Proposition~\ref{prop:family_set}).

  Let \(W \in \SCat{w}\) and let \(f \colon X \times Y \to W\) be a \(\Sigma\)-bilinear function.
  Define a subset \(Z \subseteq (X \times Y)^*\) as follows:
  \begin{equation*}
    Z = \{\fml{z} \in (X \times Y)^* \mid f\fml{z} \text{ is summable}\}
  \end{equation*}
  and define an equivalence relation \(\sim\) on \(Z\) where:
  \begin{equation*}
    \fml{z} \sim \fml{z'} \iff \Sigma f\fml{z} = \Sigma f\fml{z'}.
  \end{equation*}
  Define a function \(\bar{f} \colon Z/{\sim} \to W\) so that each equivalence class \([\fml{z}] \in Z/{\sim}\) is mapped to \(\Sigma f\fml{z}\); it is straightforward to check that \(\bar{f}\) is injective and, moreover, \([\varnothing] \in Z/{\sim}\), so the neutral element \(0 \in W\) is in the image of \(\bar{f}\).
  Therefore, according to Lemma~\ref{lem:wSm_restriction}, \(Z/{\sim}\) may be endowed with a partial function \(\Sigma^{\bar{f}}\) defined as follows for every family \(\{[\fml{z}_i]\}_I \in (Z/{\sim})^*\):
  \begin{equation*}
    \Sigma^{\bar{f}} \{[\fml{z}_i]\}_I = \begin{cases}
      [\fml{z}] &\ifc \exists [\fml{z}] \in Z/{\sim} \st \Sigma \{\bar{f}[\fml{z}_i]\}_I \keq \bar{f}[\fml{z}] \\
      \undefined &\otherwise
    \end{cases}
  \end{equation*}
  so that \((Z/{\sim},\Sigma^{\bar{f}})\) is a weak \(\Sigma\)-monoid and \(\bar{f}\) is a \(\Sigma\)-homomorphism.
  Finally, define the function \(q \colon X \times Y \to Z/{\sim}\) as the composite
  \begin{equation*}
    q = X \times Y \xto{\{-\}} Z \xto{[-]} Z/{\sim}
  \end{equation*}
  where \(\{-\}\) is the function mapping each \((x,y) \in X \times Y\) to the singleton family \((\{(x,y)\})\) and \([-]\) is the quotient map.
  It is straightforward to check that \(f = \bar{f} \circ q\) and, for any \(x \in X\) and any summable family \(\{y_i\}_I \in Y^*\),
  \begin{equation*}
    \Sigma \{\bar{f} q(x,y_i)\}_I = \Sigma \{f(x,y_i)\}_I \keq f(x,\Sigma \{y_i\}_I) = \bar{f} q (x,\Sigma \{y_i\}_I)
  \end{equation*}
  since \(f\) is \(\Sigma\)-bilinear, implying that the family \(\{q(x,y_i)\}_I\) is summable in \((Z/{\sim},\Sigma^{\bar{f}})\) with:
  \begin{equation*}
    \Sigma^{\bar{f}} \{q(x,y_i)\}_I \keq q(x,\Sigma \{y_i\}_I).
  \end{equation*}
  A similar result holds if we fix \(y \in Y\) instead and let \(\fml{x} \in X^*\) be an arbitrary summable family, thus, \(q\) is a \(\Sigma\)-bilinear function.

  In summary, it has been shown that \(Z/{\sim}\) is an object in \(\SCat{w}\) and that \(X \times Y \xto{q} Z/{\sim}\) is a \(\Sigma\)-bilinear function, hence, an element of \(S\). Moreover, \(\bar{f}\) is a \(\Sigma\)-homomorphism and \(f = \bar{f} \circ q\), implying that \(\bar{f}\) is a valid morphism from \(q\) to \(f\) in in the comma category.
  This construction may be reproduced for any \(\Sigma\)-bilinear function \(f \colon X \times Y \to W\), so it follows that \(S\) is a weakly initial set. 
  Thus, Lemma~\ref{lem:weakly_initial_set} establishes that the comma category \(\comma{\{\bullet\}}{\SCat{w}^{X,Y}}\) has an initial object, as claimed.
\end{proof}

\begin{lemma} \label{lem:SCat*_tensor}
  For each \(\SCat{*}\) category, the comma category \(\comma{\{\bullet\}}{\SCat{*}^{X,Y}}\) has an initial object. Such an initial object will be denoted \((p,X \otimes Y)\).
\end{lemma} \begin{proof}
  We must show that whenever \(X\), \(Y\) and \(W\) are objects in certain \(\SCat{*}\) category, the weak \(\Sigma\)-monoid \((Z/{\sim}, \Sigma^{\bar{f}})\) defined in the previous lemma is in fact an object in the same \(\SCat{*}\) category. Such a result has already been established for \(\SCat{w}\) in the previous lemma.

  Assume \(W \in \SCat{s}\); then, \(\Sigma^{\bar{f}}\) satisfies strong bracketing and strong flattening, the argument being the same as that for (weak) bracketing and flattening (see Lemma~\ref{lem:wSm_restriction}): these follow directly from their counterpart in \(W\) since \(\Sigma^{\bar{f}}\) is defined if and only if the sum in \(W\) is defined.
  Let \(\{[\fml{z}_i]\}_I\) be an arbitrary summable family in \((Z/{\sim})^*\); by definition of \(\Sigma^{\bar{f}}\), the fact that it is summable implies that the family \(\{\bar{f}[\fml{z}_i]\}_I \in W^*\) is summable or, equivalently, the family \(\{\Sigma f\fml{z}_i\}_I \in W^*\) is summable.
  Then, due to strong flattening in \(W\), we know that \(\uplus_I f\fml{z}_i\) is summable and, thanks to subsummability in \(W\), the subfamily \(\uplus_J f\fml{z}_j\) is summable for every \(J \subseteq I\).
  It then follows from strong bracketing in \(W\) that the subfamily \(\{\Sigma f\fml{z}_j\}_J\) is summable or, equivalently, \(\{\bar{f}[\fml{z}_j]\}_J\) is summable.
  Thus, every subfamily of \(\{[\fml{z}_i]\}_I\) is summable, implying that \(\Sigma^{\bar{f}}\) satisfies subsummability.
  Consequently, if \(W \in \SCat{s}\) then \(Z/{\sim} \in \SCat{s}\).

  Assume \(W \in \SCat{ft}\) and let \(\{[\fml{z}_i]\}_I\) be an arbitrary finite family in \((Z/{\sim})^*\).
  Notice that \(\{\Sigma f\fml{z}_i\}_I\) is a summable family due to \(W\) being finitely total.
  Since \(I\) is finite and each \(f\fml{z}_i\) is summable (due to the definition of \(Z\)), the flattening axiom in \(W\) implies that \(\uplus_I f\fml{z}_i\) is summable:
  \begin{equation*}
    \Sigma (\uplus_I f\fml{z}_i) \keq \Sigma \{\Sigma f\fml{z}_i\}_I.
  \end{equation*}
  Consequently, \(\uplus_I f\fml{z}_i \in Z\) and 
  \begin{equation*}
    \Sigma^{\bar{f}} \{[\fml{z}_i]\}_I \keq [\uplus_I \fml{z}_i]
  \end{equation*}
  implying that every finite family in \(Z/{\sim}\) is summable.
  Consequently, if \(W \in \SCat{ft}\) then \(Z/{\sim} \in \SCat{ft}\).

  Assume that both \(X\) and \(W\) are \(\Sigma\)-groups and recall that \(-f(x,y) = f(-x,y)\) (see Remark~\ref{rmk:bihom_inverse}).
  For every element \(\fml{z} = \{(x_i,y_i)\}_I\) of \(Z\) let \(\fml{-z}\) be the family \(\{(-x_i,y_i)\}_I\); notice that \(\fml{-z}\) is in \(Z\) as well since its image under \(f\) is summable in \(W\):
  \begin{equation} \label{eq:SCat_g_tensor_aux}
    \Sigma \{f(-x_i,y_i)\}_I = \Sigma \{-f(x_i,y_i)\}_I \keq - \Sigma \{f(x_i,y_i)\}_I = - \Sigma f\fml{z}
  \end{equation}
  where the \(\keq\) step follows from the inversion map being a \(\Sigma\)-homomorphism as required for \(W\) to be a \(\Sigma\)-group.
  For every \([\fml{z}] \in Z/{\sim}\) there is an inverse element \([\fml{-z}] \in Z/{\sim}\) since:
  \begin{equation*}
    \Sigma \{\bar{f}[\fml{z}],\bar{f}[\fml{-z}]\} = \Sigma \{\Sigma f\fml{z}, -\Sigma f\fml{z}\} = 0.
  \end{equation*}
  Such an inversion map \([\fml{z}] \mapsto [-\fml{z}]\) is well-defined since it is immediate from~\eqref{eq:SCat_g_tensor_aux} that \(\fml{z} \sim \fml{z'}\) implies \(\fml{-z} \sim \fml{-z'}\).
  It remains to show that the inversion map in \(Z/{\sim}\) is a \(\Sigma\)-homomorphism.
  Assume \(\{[\fml{z}_i]\}_I\) is a summable family in \((Z/{\sim})^*\); this implies the existence of a family \(\fml{z} \in Z\) such that:
  \begin{equation*}
    \Sigma \{\bar{f}[\fml{z}_i]\}_I \keq \bar{f}[\fml{z}].
  \end{equation*}
  But it is immediate from~\eqref{eq:SCat_g_tensor_aux} that \(\bar{f}[\fml{-z'}] = -\bar{f}[\fml{z'}]\) for all \(\fml{z'} \in Z\), thus:
  \begin{equation*}
    \Sigma \{\bar{f}[\fml{-z}_i]\}_I = \Sigma \{- \bar{f}[\fml{z}_i]\}_I \keq -\Sigma \{\bar{f}[\fml{z}_i]\}_I \keq -\bar{f}[\fml{z}] = \bar{f}[-\fml{z}]
  \end{equation*}
  and, hence, \(\Sigma^{\bar{f}} \{[\fml{-z}_i]\}_I \keq [-\fml{z}]\).
  Consequently, the inversion map in \(Z/{\sim}\) is a \(\Sigma\)-homomorphism and we conclude that, if \(W \in \SCat{g}\) and \(X \in \SCat{g}\) (or \(Y \in \SCat{g}\)) then \(Z/{\sim} \in \SCat{g}\).

  In summary, it has been shown that if \(X\), \(Y\) and \(W\) are all objects from the same \(\SCat{*}\) category and \(f \colon X \times Y \to W\) is a \(\Sigma\)-bilinear function, then \((Z/{\sim},\Sigma^{\bar{f}})\) as constructed in Lemma~\ref{lem:SCatw_tensor} is an object in the same \(\SCat{*}\) category.
  Then, thanks to Proposition~\ref{prop:SCat*_complete} and Corollary~\ref{cor:bihom_set_functor_preserves_limits} it is straightforward to check that the argument from the previous lemma holds for any \(\SCat{*}\) category, concluding the proof.
\end{proof}

The previous lemma establishes the existence of an object \(X \otimes Y\) known as the tensor product and a \(\Sigma\)-bilinear function \(p \colon X \times Y \to X \otimes Y\) that let us uniquely represent any \(\Sigma\)-bilinear function \(f \colon X \times Y \to Z\) as a \(\Sigma\)-homomorphism \(\bar{f} \colon X \otimes Y \to Z\) so that \(f = \bar{f} \circ p\).
In the case of \(\Sigma\)-groups, a constructive proof may perhaps be achieved by adapting the construction of the tensor product in abelian groups (see~\cite{AbTensor}) to infinitary addition \(\Sigma\).
Unfortunately, as in the case of the left adjoint to \(\SCat{ft} \into \SCat{w}\) (see Section~\ref{sec:ft_w_adj}), it is unclear whether quotienting the underlying set would be necessary to guarantee that the bracketing axiom is satisfied and, if quotienting is required, proving well-definedness of \(\Sigma\) becomes exceedingly subtle.

The goal of the rest of this section is to show that every \(\SCat{*}\) category has a (closed symmetric) monoidal structure given by this tensor product.
However, the monoidal structure is somewhat different for each of these categories; this is illustrated by the differences in the definition of their tensor unit, given below.

\begin{definition}
  Let \(S \colon \Nset^* \to \Nset \cup \{\infty\}\) be the standard sum of families of natural numbers.
  For each category \(\SCat{*}\), define the \emph{tensor unit} \(I \in \SCat{*}\) as follows:
  \begin{itemize}
    \item in \(\SCat{w}\), let \(I = \{0,1\}\) along with \(\Sigma \fml{n} = S \fml{n}\) if \(S \fml{n} \leq 1\) and otherwise undefined;
    \item in \(\SCat{s}\), let \(I\) and its \(\Sigma\) be the same as in \(\SCat{w}\);
    \item in \(\SCat{ft}\), let \(I = \Nset\) along with \(\Sigma \fml{n} = S \fml{n}\) if \(S \fml{n}\) is finite and otherwise undefined;
    \item in \(\SCat{g}\), let \(I = \Zset\) along with \(\Sigma \fml{n} = (S m_+\fml{n}) + (S m_-\fml{n})\) if \(S m_+\fml{n}\) and \(S m_-\fml{n}\) are finite and otherwise undefined; where \(m_+\) maps all negative integers to \(0\) and \(m_-\) maps all positive integers to \(0\).
  \end{itemize}
\end{definition}

It is immediate to check that each of these tensor units are well-defined \(\Sigma\)-monoids in their respective categories.
These definitions are accompanied by functions \(l_X \colon I \times X \to X\) and \(r_x \colon X \times I \to X\) for every \(X \in \SCat{*}\), defined as follows:
\begin{equation} \label{eq:lr_bihoms}
  \begin{aligned}
    l_X(n,x) &= \Sigma \{\overbrace{x,x,\ldots}^{\footnotesize n\text{ times}} \} \\
    r_X(x,n) &= \Sigma \{\overbrace{x,x,\ldots}^{\footnotesize n\text{ times}} \} \\
  \end{aligned}
\end{equation}
except for \(\SCat{g}\) where, if \(n < 0\), then the corresponding family is of \(\abs{n}\) copies of \(-x\) instead.
It is immediate that \(l_X(-,x)\) is a \(\Sigma\)-homomorphism.
To show that \(l_X(n,-)\) is also a \(\Sigma\)-homomorphism, let \(\fml{x} = \{x_i\}_I \in X^*\) be an arbitrary summable family; then, since \(n\) is finite we may apply flattening followed by bracketing to obtain:
\begin{equation*}
  l(n,\Sigma \fml{x}) = \Sigma \{\overbrace{\Sigma \fml{x},\ldots}^{\footnotesize n\text{ times}}\} \keq \Sigma \{\Sigma \{\overbrace{x_i,x_i,\ldots}^{\footnotesize n\text{ times}} \}\}_I = \Sigma \{l(n,x_i)\}_I.
\end{equation*} 
Therefore, \(l_X\) is a \(\Sigma\)-bilinear function and a similar argument holds for \(r_X\).
Moreover, it is straightforward to check that \(l\) and \(r\) are natural transformations since for any \(\Sigma\)-homomorphism \(f \colon X \to Y\),
\begin{equation*}
  l (\id \times f)(n,x) = l(n,f(x)) = \Sigma \{\overbrace{f(x),\ldots}^{\footnotesize n\text{ times}} \} = f(\Sigma \{\overbrace{x,\ldots}^{\footnotesize n\text{ times}} \}) = fl(n,x).
\end{equation*}

\begin{corollary} \label{cor:tensor_unitors}
  For any \(\SCat{*}\) category, let \(l_X\) and \(r_X\) be the \(\Sigma\)-bilinear functions defined in~\eqref{eq:lr_bihoms} and let \(a\) be the associator from the Cartesian monoidal structure on each \(\SCat{*}\). Then, the following diagram commutes:
  \[\begin{tikzcd}[column sep=tiny]
    {(X \times I) \times Y} && {X \times (I \times Y)} \\
    {X \times Y} && {X \times Y} \\
    & {X \otimes Y}
    \arrow["p", from=2-1, to=3-2]
    \arrow["p"', from=2-3, to=3-2]
    \arrow["{r \times \id}", from=1-1, to=2-1]
    \arrow["{\id \times l}"', from=1-3, to=2-3]
    \arrow["a", from=1-1, to=1-3]
  \end{tikzcd}\]
\end{corollary} \begin{proof}
  For all \(\SCat{*}\) categories, the claim follows from the following:
  \begin{equation*}
    p(\id \times l_Y)(x,n,y) = p(x,\Sigma \{\overbrace{y,\ldots}^{\footnotesize n\text{ times}} \}) = \Sigma \{\overbrace{p(x,y),\ldots}^{\footnotesize n\text{ times}} \} = p(r_X \times \id)(x,n,y)
  \end{equation*}
  due to \(p\) being a \(\Sigma\)-bilinear function.
\end{proof}

These \(\Sigma\)-bilinear functions \(l\) and \(r\) will be used to define the unitors of the monoidal structure given by \(\otimes\). The associator is trickier since the \(\Sigma\)-homomorphism representation of a `\(\Sigma\)-trilinear' function is not yet clear.
The following two lemmas deal with this, proving that the object \((X \otimes Y) \otimes Z\) along with the `\(\Sigma\)-trilinear' function \(p \circ (p \times \id)\) can take up the role of the ternary tensor product.
To do so, we must first introduce a \(\Sigma\)-monoid structure on sets of \(\Sigma\)-homomorphisms.

\begin{lemma} \label{lem:SCat_internal_hom}
  For every \(\SCat{*}\) category and objects \(X,Y \in \SCat{*}\), the hom-set \([X,Y] = \SCat{*}(X,Y)\) may be endowed with a partial function \(\Sigma^\to \colon [X,Y]^* \pto [X,Y]\) so that \([X,Y] \in \SCat{*}\).
\end{lemma} \begin{proof}
  The partial function \(\Sigma^\to\) is defined in the same manner for all of the \(\SCat{*}\) categories; it is defined pointwise, using the \(\Sigma\) from \(Y\).
  For every family \(\fml{f} = \{f_i\}_I \in [X,Y]^*\) and every \(x \in X\) let \(\fml{f}(x)\) be a shorthand for the family \(\{f_i(x)\}_I \in Y^*\) and define a partial function \(s_\fml{f} \colon X \pto Y\) as follow for every \(x \in X\):
  \begin{equation*}
    s_\fml{f}(x) = \begin{cases}
      y &\ifc \exists y \in Y \st \Sigma \fml{f}(x) \keq y \\
      \undefined &\otherwise.
    \end{cases}
  \end{equation*}
  Define the partial function \(\Sigma^\to \colon [X,Y]^* \pto [X,Y]\) as follows for every \(\fml{f} \in [X,Y]^*\):
  \begin{equation*}
    \Sigma^\to \fml{f} = \begin{cases}
      s_\fml{f} &\ifc s_\fml{f} \in [X,Y] \\
      \undefined &\otherwise.
    \end{cases}
  \end{equation*}
  Recall that all \(\Sigma\)-homomorphisms are total functions, so the condition \(s_\fml{f} \in [X,Y]\) imposes that \(\Sigma \fml{f}(x)\) is defined for all \(x \in X\).
  We now show that if \(Y\) is an object in a \(\SCat{*}\) category then \(([X,Y],\Sigma^\to)\) is an object in the same category.

  Assume \(Y \in \SCat{w}\); it is immediate that for every singleton family \(\{f\} \in [X,Y]^*\) the corresponding function \(s_{\{f\}}\) maps every \(x \in X\) to \(f(x)\), so \(\Sigma^\to \{f\} \keq f\) as required by the singleton axiom.
  Similarly, it is immediate that the function \(s_\varnothing\) maps every \(x \in X\) to the neutral element \(0 \in Y\), so the empty family is summable in \(([X,Y],\Sigma^\to)\), with the neutral element being \(s_\varnothing\).
  Moreover, let \(\fml{f} \in [X,Y]^*\) be a summable family and let \(\fml{f}_\emptyset\) be the subfamily where all occurrences of \(s_\varnothing\) have been removed; it is immediate that \(\fml{f}_\emptyset\) is summable since \(s_\varnothing\) only contributes to the sum of \(\fml{f}(x) \in Y^*\) by adding a \(0\), which may be disregarded thanks to the neutral element axiom in \(Y\).
  Consequently, the neutral element axiom is satisfied by \(\Sigma^\to\); it remains to prove the bracketing and flattening axioms.  
  Let \(\{\fml{f}_j \in [X,Y]^*\}_J\) be a collection of summable families and let \(\fml{f} = \uplus_J \fml{f}_j\).
  Assume that \(\fml{f}_j\) is a finite family for every \(j \in J\) and assume that \(\fml{f}\) is summable; to prove bracketing we must show that the family \(\fml{g} = \{\Sigma^\to \fml{f}_j\}_J \in [X,Y]^*\) is summable and \(s_\fml{g} = s_\fml{f}\).
  This is straightforward since for every \(x \in X\) we have that:
  \begin{equation*}
    s_\fml{g}(x) = \Sigma \{(\Sigma^\to \fml{f}_j)(x)\}_J = \Sigma \{\Sigma \fml{f}_j(x)\}_J \keq \Sigma \fml{f}(x) = s_\fml{f}(x)
  \end{equation*}
  where the \(\keq\) step corresponds to bracketing in \(Y\).
  Therefore, \(s_\fml{g} \in [X,Y]\) is implied by \(s_\fml{f} \in [X,Y]\) --- which holds due to \(\fml{f}\) being summable --- and, hence, \(\fml{g} = \{\Sigma^\to \fml{f}_j\}_J\) is summable with its sum matching that of \(\fml{f}\).
  Flattening is proven via the same argument, this time using flattening in \(Y\) instead of bracketing.
  Consequently, if \(Y \in \SCat{w}\) then \([X,Y] \in \SCat{w}\).

  Assume \(Y \in \SCat{s}\); then, the same argument that established bracketing and flattening in \([X,Y]\) can be used to prove strong bracketing and strong flattening in \([X,Y]\) via the corresponding axioms from \(Y\).
  Moreover, if \(\fml{f}\) is summable and \(\fml{h}\) is a subfamily of \(\fml{f}\), it follows from subsummability in \(Y\) that \(s_\fml{h}\) is a total function.
  To prove that \(\Sigma^\to\) satisfies subsummability, we must show that \(s_\fml{h}\) is a \(\Sigma\)-homomorphism; this follows from strong flattening and strong bracketing in \(Y\), as described below.
  Let \(\fml{h} = \{h_j\}_J\) and let \(\{x_i\}_I \in X^*\) be an arbitrary summable family; then:
  \begin{align*}
    s_\fml{h}(\Sigma \{x_i\}_I) &= \Sigma \{h_j(\Sigma \{x_i\}_I)\}_J && \text{(definition of \(s_\fml{h}\))} \\
      &\keq \Sigma \{\Sigma \{h_j(x_i)\}_I\}_J    && \text{(\(h_j \in [X,Y]\))} \\
      &\keq \Sigma \{h_j(x_i)\}_{I \times J}    && \text{(strong flattening in \(Y\) and \(\uplus_J I \iso I \times J\))} \\
      &\keq \Sigma \{\Sigma \{h_j(x_i)\}_J\}_I && \text{(strong bracketing in \(Y\) and \(s_\fml{h}\) total)} \\
      &= \Sigma \{s_\fml{h}(x_i)\}_I            && \text{(definition of \(s_\fml{h}\))}
  \end{align*}
  so \(s_\fml{h} \in [X,Y]\) and \(\fml{h}\) is summable, implying that subsummability is satisfied.
  Consequently, if \(Y \in \SCat{s}\) then \([X,Y] \in \SCat{s}\).

  Assume \(Y \in \SCat{ft}\); it is immediate from \(Y\) being finitely total that, for every finite family \(\fml{h} \in [X,Y]^*\), the function \(s_\fml{h}\) is total.
  Once again, to prove that \(\fml{h}\) is summable in \([X,Y]\) we must show that \(s_\fml{h}\) is a \(\Sigma\)-homomorphism; this can be achieved via the same argument used above to prove subsummability but, in this case, only (weak) flattening and bracketing are required since \(\fml{h}\) is assumed to be a finite family.
  Consequently, if \(Y \in \SCat{ft}\) then \([X,Y] \in \SCat{ft}\).

  Assume \(Y \in \SCat{g}\) and let \(f \in [X,Y]\). Due to \(Y\) having inverses, we can define a function that maps every \(x\) to \(-f(x)\), call it \(-f\).
  Since the inversion map in \(Y\) is a \(\Sigma\)-homomorphism and so is \(f\), it is trivial to check that \(-f \in [X,Y]\).
  Thanks to \(\Sigma^\to\) being defined pointwise, it is straightforward to check that \(\Sigma^\to \{f,-f\}\) is the function sending every \(x\) to \(0 \in Y\) --- \ie{} the neutral element \(s_\varnothing\) in \([X,Y]\) --- and the inversion map \(f \mapsto -f\) is a \(\Sigma\)-homomorphism.
  Consequently, if \(Y \in \SCat{g}\) then \([X,Y] \in \SCat{g}\).

  In conclusion, it has been shown that \(([X,Y],\Sigma^\to)\) is an object in the same \(\SCat{*}\) category as \(Y\), proving the claim.
\end{proof}

We may now make use of the Cartesian closed structure in \(\Set\) together with the definition of tensor product in \(\SCat{*}\) categories to lift the universal property of the binary tensor product to the ternary case.

\begin{definition} \label{def:trihom}
  For each \(\SCat{*}\) category, let \(X,Y,Z,W \in \SCat{*}\) be arbitrary objects and let \(f \colon X \times Y \times Z \to W\) be a function.
  We say \(f\) is a \emph{\(\Sigma\)-trilinear} function if \(f(x,y,-)\), \(f(x,-,z)\) and \(f(-,y,z)\) are \(\Sigma\)-homomorphisms for every choice of \(x \in X\), \(y \in Y\) and \(z \in Z\).
\end{definition}

\begin{lemma} \label{lem:SCat_3_tensor}
  For every \(\SCat{*}\) category and objects \(X,Y,Z,W \in \SCat{*}\), let \(f \colon (X \times Y) \times Z \to W\) be an arbitrary \(\Sigma\)-trilinear function.
  Then, there is a unique \(\Sigma\)-homomorphism \((X \otimes Y) \otimes Z \to W\) making the diagram
  \[\begin{tikzcd}
    {(X \times Y) \times Z} & W \\
    & {(X \otimes Y) \otimes Z}
    \arrow[dashed, from=2-2, to=1-2]
    \arrow["f", from=1-1, to=1-2]
    \arrow["{p\circ(p \times \id)}"', from=1-1, to=2-2]
  \end{tikzcd}\]
  commute in \(\Set\).
\end{lemma} \begin{proof}
  For every \(\SCat{*}\) category and objects \(Z,W \in \SCat{*}\), define a function \(\mathrm{ev} \colon [Z,W] \times Z \to W\) that maps each pair \((h,z)\) to \(h(z)\).
  Let \(h \in [Z,W]\) be an arbitrary \(\Sigma\)-homomorphism and \(\fml{z} = \{z_i\}_I \in Z\) an arbitrary summable family, then:
  \begin{equation*}
    \mathrm{ev}(h,\Sigma \fml{z}) = h(\Sigma \fml{z}) \keq \Sigma h\fml{z} = \Sigma \{\mathrm{ev}(h,z_i)\}_I.
  \end{equation*}
  Similarly, let \(\fml{h} = \{h_i\}_I \in [Z,W]^*\) be an arbitrary summable family and fix an arbitrary \(z \in Z\), then:
  \begin{equation*}
    \mathrm{ev}(\Sigma^\to \fml{h},z) = \mathrm{ev}(s_\fml{h},z) = s_\fml{h}(z) = \Sigma \{h_i(z)\}_I = \Sigma \{\mathrm{ev}(h_i,z)\}_I.
  \end{equation*}
  Therefore, \(\mathrm{ev}\) is a \(\Sigma\)-bilinear function.
  Let \(f \colon (X \times Y) \times Z \to W\) be a \(\Sigma\)-trilinear function.
  Since \(\Set\) is Cartesian closed, there is a unique function \(\Lambda{f} \colon X \times Y \to [Z,W]\) making the diagram
  \begin{equation} \label{diag:SCat_currying}
    \begin{tikzcd}
      {[Z,W] \times Z} \\
      {(X \times Y) \times Z} & W
      \arrow["f"', from=2-1, to=2-2]
      \arrow["{\Lambda{f} \times \id}", dashed, from=2-1, to=1-1]
      \arrow["{\mathrm{ev}}", from=1-1, to=2-2]
    \end{tikzcd}
  \end{equation}
  commute in \(\Set\).
  The function \(\Lambda{f}\) maps each pair \((x,y) \in X \times Y\) to the function \(f(x,y,-)\), which is a \(\Sigma\)-homomorphism since \(f\) is \(\Sigma\)-trilinear; we now show that \(\Lambda{f}\) is \(\Sigma\)-bilinear.
  Fix \(x \in X\) and let \(\fml{y} = \{y_i\}_I \in Y^*\) be an arbitrary summable family; then, for every \(z \in Z\):
  \begin{align*}
    \Sigma \{\Lambda{f}(x,y_i)(z)\}_I &= \Sigma \{f(x,y_i,z)\}_I  && \text{(definition of \(\Lambda{f}\))} \\
      &\keq f(x,\Sigma \fml{y},z)  && \text{(\(f\) is \(\Sigma\)-trilinear)} \\
      &= \Lambda{f}(x,\Sigma \fml{y})(z)  && \text{(definition of \(\Lambda{f}\))}
  \end{align*}
  with \(\Lambda{f}(x,\Sigma \fml{y})\) a \(\Sigma\)-homomorphism, as previously discussed.
  Therefore,
  \begin{equation*}
    \Sigma^\to \{\Lambda{f}(x,y_i)\}_I \keq \Lambda{f}(x,\Sigma \fml{y}).
  \end{equation*}
  The same can be said for summable families \(\fml{x} \in X^*\) if \(y \in Y\) is fixed instead, hence, \(\Lambda{f}\) is a \(\Sigma\)-bilinear function.

  Since both \(\Lambda{f}\) and \(\mathrm{ev}\) are a \(\Sigma\)-bilinear functions, the diagram
  \[\begin{tikzcd}[column sep=large]
    & {[Z,W] \times Z} \\
    {(X \otimes Y) \times Z} & {(X \times Y) \times Z} & W \\
    & {(X \otimes Y) \otimes Z}
    \arrow["f"{pos=0.4}, from=2-2, to=2-3]
    \arrow["{\Lambda{f} \times \id}"{pos=0.4}, dashed, from=2-2, to=1-2]
    \arrow["{\mathrm{ev}}", from=1-2, to=2-3]
    \arrow["{p \times \id}"', from=2-2, to=2-1]
    \arrow["{\overline{\Lambda{f}} \times \id}", dashed, from=2-1, to=1-2]
    \arrow["p"', from=2-1, to=3-2]
    \arrow[dashed, from=3-2, to=2-3]
  \end{tikzcd}\]
  commutes in \(\Set\) with each dashed arrow denoting uniqueness.
  The top right triangle is diagram~\eqref{diag:SCat_currying} which was already established to commute.
  Because \(\Lambda{f}\) is a \(\Sigma\)-bilinear function, Lemma~\ref{lem:SCat*_tensor} implies that the top left triangle commutes, where \(\overline{\Lambda{f}}\) is a \(\Sigma\)-homomorphism.
  Considering that \(\mathrm{ev}\) is a \(\Sigma\)-bilinear function, it is straightforward to check that \(\mathrm{ev} \circ (\overline{\Lambda{f}} \times \id)\) is also a \(\Sigma\)-bilinear function; then Lemma~\ref{lem:SCat*_tensor} implies that the outer edges of the diagram commute, with the corresponding unique function \((X \otimes Y) \otimes Z \to W\) being a \(\Sigma\)-homomorphism.
  Consequently, the triangle at the bottom of the diagram commutes.
  Notice that if we were to find another \(\Sigma\)-homomorphism \(g \colon (X \otimes Y) \otimes Z \to W\) making the bottom triangle commute, such \(g\) would also make the outer edges of the diagram commute but, according to Lemma~\ref{lem:SCat*_tensor}, there is only one such \(\Sigma\)-homomorphism. Therefore, the \(\Sigma\)-homomorphism \((X \otimes Y) \otimes Z \to W\) from the claim has been shown to exist and be unique, completing the proof.
\end{proof}

It is straightforward to check that the same argument allows us to factor any \(\Sigma\)-trilinear function \(f \colon X \times (Y \times Z) \to W\) through \(p \circ (\id \times p)\) so that there is a unique \(\Sigma\)-homomorphism characterising \(f\), \ie{} both \((X \otimes Y) \otimes Z\) and \(X \otimes (Y \otimes Z)\) are ternary tensor products.
Finally, it is time to define the monoidal structure on the tensor product.

\begin{definition}
  For every \(\SCat{*}\) category and \(\Sigma\)-homomorphism \(X \xto{f} Z\) and \(Y \xto{g} W\), let \(f \otimes g\), \(\lambda\), \(\rho\) and \(\alpha\) be defined as the unique \(\Sigma\)-homomorphisms making the following diagrams commute:
  \[\begin{tikzcd}[column sep=small]
    {X \times Y} & {Z \times W} && {I \times X} & X && {X \times I} & X \\
    {X \otimes Y} & {Z \otimes W} && {I \otimes X} &&& {X \otimes I}
    \arrow["{f \times g}", from=1-1, to=1-2]
    \arrow["p"', from=1-1, to=2-1]
    \arrow["p", from=1-2, to=2-2]
    \arrow["{f \otimes g}"', dashed, from=2-1, to=2-2]
    \arrow["l", from=1-4, to=1-5]
    \arrow["p"', from=1-4, to=2-4]
    \arrow["\lambda"', dashed, from=2-4, to=1-5]
    \arrow["\rho"', dashed, from=2-7, to=1-8]
    \arrow["r", from=1-7, to=1-8]
    \arrow["p"', from=1-7, to=2-7]
  \end{tikzcd}\]
  \[\begin{tikzcd}[column sep=small]
    {(X \times Y) \times Z} && {X \times(Y \times Z)} \\
    {(X \otimes Y) \otimes Z} && {X \otimes (Y \otimes Z).}
    \arrow["a", from=1-1, to=1-3]
    \arrow["{p\circ(\id \times p)}", from=1-3, to=2-3]
    \arrow["{p \circ (p \times \id)}"', from=1-1, to=2-1]
    \arrow["\alpha"', dashed, from=2-1, to=2-3]
  \end{tikzcd}\]
\end{definition}

\begin{theorem} \label{thm:SCat_tensor_monoidal}
  Every \(\SCat{*}\) category is a monoidal category with monoidal product \(\otimes\) and monoidal unit \(I\).
\end{theorem} \begin{proof}
  On objects \(X,Y \in \SCat{*}\), the functor \(- \otimes -\) yields \(X \otimes Y\) and, on morphisms \(f\) and \(g\), it yields the unique \(\Sigma\)-homomorphism \(f \otimes g\) given in the definition above; it is straightforward to check that this is indeed a functor, thanks to functoriality of the Cartesian product and the universal property of \(X \otimes Y\).
  Unitors \(\lambda\) and \(\rho\) and associators \(\alpha\) are given in the definition above.
  The triangle equation follows from Corollary~\ref{cor:tensor_unitors} after lifting all occurrences of \(\times\) to \(\otimes\) using \(p\), whereas the pentagon equation follows from the pentagon equation of the Cartesian monoidal structure on \(\SCat{*}\) along with naturality of \(a\).
  Naturality of \(\lambda\), \(\rho\) and \(\alpha\) follows from naturality of \(l\), \(r\) and \(a\).
  The inverse of \(\alpha\) is induced from the inverse of \(a\) and the universal property of \(X \otimes (Y \otimes Z)\).
  The inverse of \(\lambda\) is somewhat trickier: let \(i \colon X \to I \times X\) be the function that maps each \(x \in X\) to \((1,x)\); it is immediate that \(l \circ i = \id\), and \(\lambda \circ p \circ i = \id\) follows from the definition of \(\lambda\).
  On the other hand:
  \begin{equation*}
    p(1,l(n,x)) = p(1,\Sigma \{\overbrace{x,\ldots}^{\footnotesize n \text{ times}}\}) \keq \Sigma \{\overbrace{p(1,x),\ldots}^{\footnotesize n \text{ times}}\} \keq p(n,x)
  \end{equation*}
  so, \(p \circ i \circ l = p\) and, by definition of \(\lambda\), we have \(p \circ i \circ \lambda \circ p = p\); this means that the \(\Sigma\)-bilinear function \(p\) may be characterised by the \(\Sigma\)-homomorphism \(\bar{p} = p \circ i \circ \lambda\) but, trivially, \(\bar{p} = \id\) and due to uniqueness of \(\bar{p}\) we conclude that \((p \circ i) \circ \lambda = \id\).
  Therefore, the \(\Sigma\)-homomorphism \(p \circ i\) is the inverse of \(\lambda\) and we may define the inverse of \(\rho\) by similar means.
  This completes the proof that \((\SCat{*},\otimes,I)\) is a monoidal category for every \(\SCat{*}\) category.
\end{proof}

\begin{corollary}
  For every \(\SCat{*}\) category, \((\SCat{*},\otimes,I,[-,-])\) is a closed symmetric monoidal category.
\end{corollary} \begin{proof}
  The symmetry of the monoidal structure is derived from that of the Cartesian monoidal structure, lifted using \(p\).
  The closure follows from a similar argument to Lemma~\ref{lem:SCat_3_tensor}, but replacing each appearance of \(\times\) with \(\otimes\) and lifting through \(p\) so that all \(\Sigma\)-bilinear functions may be represented as \(\Sigma\)-homomorphisms and, hence, the diagrams can be made to commute in \(\SCat{*}\).
\end{proof}

\section{Topological monoids}
\label{sec:topology}

This section is meant to provide a brief introduction to the notions of topology that will be relevant to the thesis, namely, Hausdorff monoids, nets, limit points, continuous functions and their interaction with limits.
For an in-depth introduction to topology see the book by John L. Kelley~\cite{Topology}.

\subsection{Brief introduction to general topology (\emph{Preamble})}

Topological monoids are monoids whose underlying set has some additional structure, known as a \emph{topology}.
Sets with a topology are known as topological spaces; these are formally defined below.

\begin{definition} \label{def:topological_space}
  Let \(X\) be a set and let \(\tau\) be a subset of the power set \(\powerset{X}\). The pair \((X,\tau)\) is said to be a \gls{topological space} and \(\tau\) its topology if the following conditions are satisfied:
  \begin{itemize}
    \item the sets \(\varnothing\) and \(X\) are in \(\tau\);
    \item for any \(U,V \in \tau\), the intersection \(U \cap V\) is in \(\tau\);
    \item for any arbitrary (possibly infinite) indexed set \(\{U_i \in \tau\}_{i \in I}\),
    \begin{equation*}
      \bigcup_{i \in I} U_i \in \tau.
    \end{equation*}
  \end{itemize}
  Let \((X,\tau)\) be a topological space; an element \(U \in \tau\) is known as an \gls{open set} and, if \(x \in U\), we say that \(U\) is an \gls{open neighbourhood} of \(x\).
  Elements in \(X\) are known as \emph{points}.
\end{definition}

Although highly abstract, this definition has a down to earth interpretation which is helpful to keep in mind.
An open neighbourhood of a point \(x \in X\) is, as the name suggests, a subset of points in \(X\) that are perceived as being `near' \(x\).
Let \(U\) and \(V\) be two open neighbourhoods of a point \(x\); if \(U \subseteq V\), we say that no point in \(U\) is farther away from \(x\) than \(x\) is from the farthest point in \(V\).
Hence, a topology --- being a collection of open neighbourhoods --- provides a loose hierarchy of vicinity between points.
A common way to define a topology on a set \(X\) is to provide a \emph{base} for it: a collection of subsets of \(X\) from which the topology is unambiguously generated.

\begin{definition} \label{def:base}
  Let \(X\) be set. A \gls{base for a topology} on \(X\) is a subset \(\cB \subseteq \powerset{X}\) such that
  \begin{itemize}
    \item for each element \(x \in X\), there is at least one \(B \in \cB\) such that \(x \in B\), and
    \item for all \(B_1,B_2 \in \cB\) and all \(x \in X\),
    \begin{equation*}
      x \in B_1 \cap B_2 \implies \exists B' \in \cB \st x \in B' \text{ and } B' \subseteq B_1 \cap B_2.
    \end{equation*}
  \end{itemize}
  Each base \(\cB\) generates a topology on \(X\), defined as the set of all unions (including infinite ones) of elements in \(\cB\).
\end{definition}

Let \(\tau\) and \(\tau'\) be two topologies on the same set \(X\).
We say that \(\tau\) is a \gls{coarser topology} than \(\tau'\) when \(\tau \subseteq \tau'\), and \emph{strictly coarser} when the containment is strict.
We may show that a topology is coarser than another by studying their bases, as established below.

\begin{proposition} \label{prop:base_inclusion}
  Let \(\tau\) and \(\tau'\) be topologies on \(X\) and let \(\cB\) and \(\cB'\) be bases for them, respectively. Then, \(\tau \subseteq \tau'\) iff for each \(x \in X\) and each \(B \in \cB\):
  \begin{equation*}
    x \in B \implies \exists B' \in \cB' \st B' \subseteq B \text{ and } x \in B'.
  \end{equation*}
\end{proposition} \begin{proof}
  Assume \(\tau \subseteq \tau'\); then automatically \(B \in \tau'\), but \(\cB'\) generates \(\tau'\), so it must be that \(B\) can be obtained as the union of elements in \(\cB'\) and, in particular, if \(x \in B\) then there must be a \(B' \in \cB'\) such that \(B' \subseteq B\) and \(x \in B'\).
  To prove the other direction, let \(U \in \tau\) be an arbitrary open set and recall that, because \(\cB\) generates \(\tau\), there must be a family \(\{B_x \in \cB\}_{x \in U}\) such that each \(B_x\) satisfies \(x \in B_x\) and \(B_x \subseteq U\); consequently:
  \begin{equation*}
    U = \bigcup_{x \in U} B_x.
  \end{equation*}
  Then, by assumption, for each \(B_x\) there is a \(B'_x \in \cB'\) such that \(B'_x \subseteq B_x\) and \(x \in B'_x\).
  Then, it follows that \(B'_x \subseteq U\) and all elements \(x \in U\) are accounted for, so that:
  \begin{equation*}
    U = \bigcup_{x \in U} B'_x
  \end{equation*}
  which implies that \(U \in \tau'\).
  Since this argument holds for every \(U \in \tau\) we conclude that \(\tau \subseteq \tau'\), completing the proof.
\end{proof}

A Hausdorff space is a topological space where we may discern any two distinct points by studying their neighbourhoods.

\begin{definition}
  A \gls{Hausdorff space} is a topological space \((X,\tau)\) satisfying the Hausdorff separation condition; namely, for any two points \(x,y \in X\):
  \begin{equation}
    x \not = y \implies \exists U,V \in \tau \st x \in U,\, y \in V \text{ and } U \cap V = \varnothing.
  \end{equation}
\end{definition}

\begin{example} \label{ex:R_topology}
  The collection of all open intervals \((a,b)\) of real numbers is a base for the standard topology on \(\Rset\).
  This topological space is Hausdorff: let \(x,y \in \Rset\) and assume \(x < y\); for \(\epsilon = \tfrac{y-x}{2}\) the open neighbourhoods \((x-\epsilon,x+\epsilon)\) and \((y-\epsilon,y+\epsilon)\) are disjoint and, hence, the Hausdorff separation condition is satisfied.
\end{example}

\begin{example} \label{ex:metric_topology}
  Every metric space is a Hausdorff space.
  Let \((X,d)\) be a metric space. A ball centered at \(x \in X\) with radius \(\epsilon \in \Rset\) is a subset of \(X\) defined as follows:
  \begin{equation*}
    B_{x,\epsilon} = \{y \in X \mid d(x,y) < \epsilon\}.
  \end{equation*}
  The collection of all balls for every choice of center and radius forms a base.
  This is a Hausdorff space: if \(x \not = y\) choose \(\epsilon = \tfrac{d(x,y)}{2}\), then the open neighbourhoods \(B_{x,\epsilon}\) and \(B_{y,\epsilon}\) are disjoint.
  Moreover, every normed vector space is a metric space with its distance function induced by the norm
  \begin{equation*}
    d(u,v) = \norm{u - v}
  \end{equation*}
  and, hence, every normed vector space is a Hausdorff space.
\end{example}

As it is usual whenever a mathematical structure is defined, it is useful to provide a notion of mapping between topological spaces that, in some sense, preserves the topological structure.
Such is the role of continuous functions, defined below.

\begin{definition}
  Let \((X,\tau_X)\) and \((Y,\tau_Y)\) be two topological spaces. A function \(f \colon X \to Y\) is said to be continuous at point \(x \in X\) if for every open neighbourhood \(V\) of \(f(x)\) there is an open neighbourhood \(U\) of \(x\) such that \(f(U) \subseteq V\).\footnote{The shorthand \(f(U)\) refers to the set \(\{f(u) \mid u \in U\}\).}
  Furthermore, \(f \colon X \to Y\) is said to be a \gls{continuous function} if it is continuous at every point \(x \in X\).
\end{definition}

\begin{example}
  Endow \(\Rset\) with the standard topology, as in Example~\ref{ex:R_topology}. The function \(f \colon \Rset \to \Rset\) defined as \(f(x) = 2x\) is continuous: for any \(x\) and any open set \(V = (f(x)-a,f(x)+b)\) with \(a,b > 0\), the open set \(U = (x-\tfrac{a}{2},x+\tfrac{b}{2})\) satisfies \(f(U) \subseteq V\).
  In fact, the standard notion of a continuous function studied in analysis coincides with the topological notion when \(\Rset\) is endowed with the standard topology.
\end{example}

\begin{example} \label{ex:top_identities}
  Let \(\tau\) and \(\tau'\) be two topologies on the same set \(X\). If \(\tau \subseteq \tau'\) then the identity function \((X,\tau') \to (X,\tau)\) is continuous.
  If the containment is strict \(\tau \subset \tau'\) then the identity function on the opposite direction \((X,\tau) \to (X,\tau')\) cannot be continuous.
\end{example}

Although this is not the usual definition of a continuous function given in most introductory texts on topology, it is equivalent to the standard one, as shown in Theorem 3.1 from~\cite{Topology}.
The definition given in this text is a better fit for the arguments that will appear later in this chapter.
Thanks to the following proposition, we may prove that a function is continuous by studying the open sets in the base of the codomain.

\begin{proposition} \label{prop:base_continuous}
  Let \((X,\tau_X)\) and \((Y,\tau_Y)\) be two topological spaces and let \(\cB\) be a base for \(\tau_Y\). A function \(f \colon X \to Y\) is continuous at point \(x \in X\) iff for all \(B \in \cB\) containing \(f(x)\) there is an open neighbourhood \(U\) of \(x\) such that \(f(U) \subseteq B\).
\end{proposition} \begin{proof}
  If \(f\) is continuous at \(x\) and \(B \in \cB\) is an open neighbourhood of \(f(x)\), the existence of the required open neighbourhood \(U\) is immediate.
  To prove the other direction, assume \(U\) exists for every \(B\) as stated in the claim.
  Recall that every open set in \(\tau_Y\) is the union of a collection of elements in \(\cB\), hence, for every open neighbourhood \(V\) of \(f(x)\) there is at least one element \(B \in \cB\) such that \(B \subseteq V\) and \(f(x) \in B\).
  By assumption, there is an open neighbourhood \(U\) of \(x\) such that \(f(U) \subseteq B\) so we may conclude that \(f(U) \subseteq V\), proving that \(f\) is continuous at \(x\).
\end{proof}

\begin{proposition}
  There is a category \GLS{\(\Top\)}{Top} whose objects are topological spaces and whose morphisms are continuous functions.
  There is a full subcategory \GLS{\(\TopHaus\)}{TopHaus} of \(\Top\), obtained by restricting objects to Hausdorff spaces.
\end{proposition} \begin{proof}
  It is straightforward to prove that if \(f \colon A \to B\) and \(g \colon B \to C\) are continuous functions then \(g \circ f\) is a continuous function as well. Identity functions between the same topological spaces are continuous (see Example~\ref{ex:top_identities}). Associativity and identity axioms follow immediately from those of \(\Set\).
  These arguments are independent of whether objects are Hausdorff spaces or general topological spaces, thus, both \(\Top\) and \(\TopHaus\) are categories, and it is straightforward to check that the embedding \(\TopHaus \into \Top\) is a full functor.
\end{proof}

Whenever a set can be injectively mapped into a topological space, we may endow it with the subspace topology: the coarsest topology that makes the injection continuous.

\begin{definition} \label{def:subspace_topology}
  Let \((Y,\tau)\) be a topological space, let \(X\) be an arbitrary set and let \(f \colon X \to Y\) be an injective function. The \gls{subspace topology} on \(X\) induced by \(f\) is defined as follows:
  \begin{equation*}
    \tau_f = \{f^{-1}(V) \mid V \in \tau\}
  \end{equation*}
  where \(f^{-1}(V) = \{x \in X \mid f(x) \in V\}\).
  It is trivial to check that \(f \colon (X,\tau_f) \to (Y,\tau)\) is continuous.
\end{definition}

It is straightforward to show that, if \(Y\) is Hausdorff, the subspace topology \(\tau_f\) is also Hausdorff.

\subsection{Nets and limit points (\emph{Preamble})}

In this thesis, topology is used to discuss convergence so that, eventually, we may say that an infinite sum is defined if and only if it converges.
The notion of convergence is often linked to sequences but, in the study of general topology, it is more advantageous to discuss it in terms of nets.
Nets are a generalisation of sequences that are better behaved in the context of topological spaces, in the sense that there are important results (for instance, Proposition~\ref{prop:preserve_limit}) which hold for nets but are not necessarily true for sequences.

\begin{definition}
  Let \(D\) be a nonempty set and let \(\leq\) be a partial order on it. The pair \((D,\leq)\) is said to be a \gls{directed set} if for every pair of elements \(a,b \in D\), there is an upper bound, \ie{} there is an element \(c \in D\) such that \(a \leq c\) and \(b \leq c\).
  A \gls{net} is a function whose domain is a directed set.
  A net in \(X\) is a net whose codomain is \(X\).
\end{definition}

\begin{example} \label{ex:sequence_as_net}
  Let \(\{x_i\}_{i \in \Nset}\) be a family indexed by the natural numbers, \ie{} a sequence. The order \(\leq\) of natural numbers makes \((\Nset,\leq)\) a directed set and, hence, the function mapping each \(i \in \Nset\) to \(x_i\) is a net.
\end{example}

\begin{example} \label{ex:directed_set_product}
  Let \((A,\leq)\) and \((B,\preccurlyeq)\) be two directed sets. Define \((A \times B,\sqsubseteq)\) to be the partially ordered set where \((a,b) \sqsubset (a',b')\) iff both \(a \leq a'\) and \(b \preccurlyeq b'\). It is straightforward to check that \((A \times B,\sqsubseteq)\) is a directed set, as its upper bounds can be defined with respect to those from \((A,\leq)\) and \((B,\preccurlyeq)\).
\end{example}

\begin{example} \label{ex:fpset_directed}
  Let \(X\) be a set and let \(\fml{x} \in X^*\) be an arbitrary family. Let \(\fpset{\fml{x}}\) be the collection of all finite subfamilies of \(\fml{x}\). Evidently, \((\fpset{\fml{x}},\subseteq)\) is a partially ordered set and, for any two \(\fml{x'},\fml{x''} \in \fpset{\fml{x}}\), it is clear that \(\fml{x'} \cup \fml{x''} \in \fpset{\fml{x}}\) so their union acts as their upper bound.
\end{example}

Let \(\fml{x} \in \Rset^*\) be a family of real numbers. Let \((\fpset{\fml{x}},\subseteq)\) be the directed set defined in Example~\ref{ex:fpset_directed}. There are multiple examples of nets in \(\Rset\) with domain \(\fpset{\fml{x}}\), for instance, \(\sigma(\fml{x'}) = \sum \fml{x'}\) and \(\mu(\fml{x'}) = \prod \fml{x'}\). Importantly, every element \(\fml{x'} \in \fpset{\fml{x}}\) is finite by definition, so these sums and products are guaranteed to be well-defined.

It is time to define limit points, which formalise the notion of convergence.

\begin{definition}
  Let \((D,\leq)\) be a directed set and let \(d \in D\) be an arbitrary element. Let \((\gdir{D}{d},\leq)\) be the directed set where:
  \begin{equation*}
    \gdir{D}{d} = \{d' \in D \mid d \leq d'\}.
  \end{equation*}
  Let \(X\) be a topological space; a point \(x \in X\) is a \gls{limit point} of a net \(\alpha \colon D \to X\), written
  \begin{equation*}
    \alpha \to x
  \end{equation*}
  if and only if for each open neighbourhood \(U\) of \(x\) there is an element \(d \in D\) such that \(\alpha(\gdir{D}{d}) \subseteq U\).
\end{definition}

If we are given a base for the topology, we are only required to check the open sets in the base to prove a net has a limit point.

\begin{proposition} \label{prop:base_limit}
  Let \(\alpha \colon D \to X\) be a net and let \(\cB\) be a base for a topology in \(X\). A point \(x \in X\) is a limit point of \(\alpha\) iff for each \(B \in \cB\) such that \(x \in B\) there is an element \(d \in D\) satisfying that for all \(\alpha(\gdir{D}{d}) \subseteq B\).
\end{proposition} \begin{proof}
  If \(x\) is a limit point of \(\alpha\), the existence of the required \(d \in D\) is guaranteed by definition for every open neighbourhood of \(x\), including those in \(\cB\).
  To prove the other direction, recall that any open set in the topology is the union of a collection of elements in \(\cB\). Therefore, for any open neighbourhood \(U\) of \(x\), there must be at least one element \(B \in \cB\) such that \(x \in B\) and \(B \subseteq U\).
  Take this \(B\) and assume we can obtain a \(d \in D\) so that \(\alpha(\gdir{D}{d}) \subseteq B\); then, trivially \(\alpha(\gdir{D}{d}) \subseteq U\).
  We may do this for every open neighbourhood \(U\) of \(x\) so, indeed, \(\alpha \to x\) as claimed.
\end{proof}

The following proposition establishes that continuous functions are precisely the functions that preserve convergence.

\begin{proposition} \label{prop:preserve_limit}
  Let \(X\) and \(Y\) be topological spaces and let \(f \colon X \to Y\) be a function. The function \(f\) is continuous iff for all \(x \in X\) and every net \(\alpha \colon D \to X\):
  \begin{equation*}
    \alpha \to x \implies f \circ \alpha \to f(x).
  \end{equation*}
\end{proposition} \begin{proof}
  Assume \(f\) is continuous. Let \(V\) be an arbitrary open neighbourhood of \(f(x)\); because \(f\) is continuous, there is an open neighbourhood \(U\) of \(x\) such that \(f(U) \subseteq V\). Let \(\alpha \to x\) and take \(d \in D\) so that \(\alpha(\gdir{D}{d}) \subseteq U\); it follows that \(f(\alpha(\gdir{D}{d})) \subseteq V\) and so \(f \circ \alpha \to f(x)\).

  To prove the other direction, we show that if \(f\) is not continuous then there is a net \(\alpha\) such that \(\alpha \to x\) but \(f(x)\) is not a limit point of \(f \circ \alpha\).
  To define such a net, let \(D_x\) be the set of open neighbourhoods of \(x\); then, \((D_x,\supseteq)\) is a directed set whose upper bounds are given by intersection.
  By assumption, \(f\) is not continuous at some \(x \in X\), hence, for some open neighbourhood \(V\) of \(f(x)\) all open neighbourhoods \(U \in D_x\) contain a point \(x_U \in U\) such that \(f(x_U) \not \in V\).
  Define a net \(\alpha \colon D_x \to X\) assigning to each \(U \in D_x\) such a point \(x_U\).
  This net satisfies \(\alpha \to x\) since, for any open neighbourhood \(U\) of \(x\), it holds that \(\alpha(\gdir{D_x}{U}) \subseteq U\) since every \(U' \in \gdir{D_x}{U}\) is, by definition of \(\gdir{D_x}{U}\), contained in \(U\).
  However, for all \(U \in D_x\) the definition of \(x_U\) makes it so that \((f \circ \alpha)(U) = f(x_U)\) is not in \(V\); thus, the open neighbourhood \(V\) prevents \(f(x)\) from being a limit point of \(f \circ \alpha\).
  Therefore, it is not satisfied that \(\alpha \to x \implies f \circ \alpha \to f(x)\) and the proof by contrapositive is complete.
\end{proof}

Finally, it is useful to define the notion of subnet. 

\begin{definition} \label{def:subnet}
  Let \((A,\leq)\) and \((B,\preccurlyeq)\) be directed sets. A net \(\beta \colon B \to X\) is a \gls{subnet} of \(\alpha \colon A \to X\) if there is a function \(m \colon B \to A\) such that the following conditions are satisfied:
  \begin{itemize}
    \item the net \(\beta\) factors through \(\alpha\), \ie{} \(\beta = \alpha \circ m\);
    \item \(m\) is \emph{isotone}, \ie{} for all \(b,b' \in B\), if \(b \preccurlyeq b'\) then \(m(b) \leq m(b')\);
    \item \(m(B)\) is \emph{cofinal} in \(A\), \ie{} for every \(a \in A\) there is a \(b \in B\) satisfying \(a \leq m(b)\).
  \end{itemize}
\end{definition}

The first two conditions are straightforward: we are embedding the partial order of \(B\) into that of \(A\) while making sure that \(\alpha\) and \(\beta\) match. The third condition prevents \(B\) from being a truncated version of \(A\) that omits all of \(A\)'s larger elements; this precaution is necessary so that subnets enjoy the following property.\footnote{In Section 2.3 of~\cite{Topology} Kelley does not require the mediating function \(m\) to be isotone. However, this is a common requirement in many reference texts on topology (for instance, see Exercise 8 of the supplementary material from Chapter 3 of Munkres's book~\cite{Munkres}) since it simplifies many results on limit points of subnets.}

\begin{proposition} \label{prop:subnet_limit}
  Let \(\alpha \colon A \to X\) and \(\beta \colon B \to X\) be nets in a topological space \(X\). If \(\beta\) is a subnet of \(\alpha\) then:
  \begin{equation*}
    \alpha \to x \implies \beta \to x.
  \end{equation*}
\end{proposition} \begin{proof}
  Let \(m \colon B \to A\) be a function witnessing that \(\beta\) is a subnet of \(\alpha\). Assume \(\alpha \to x\), \ie{} for every open neighbourhood \(U\) of \(x\), there is an element \(a \in A\) such that \(\alpha(\gdir{A}{a}) \subseteq U\).
  Choose some \(b \in B\) that satisfies \(a \leq m(b)\); such an element is guaranteed to exists thanks to \(m(B)\) being cofinal in \(A\).
  Due to \(m\) being isotone we have that \(m(\gdir{B}{b}) \subseteq \gdir{A}{a}\) and, hence, \(\alpha(m(\gdir{B}{b})) \subseteq U\).
  But \(\beta = \alpha \circ m\) so \(\beta(\gdir{B}{b}) \subseteq U\), implying \(\beta \to x\) as claimed.
\end{proof}

The proposition below characterises the limits in the subspace topology. 

\begin{proposition} \label{prop:subspace_topology_limit}
  Let \((Y,\tau)\) be a topological space, let \(X\) be an arbitrary set and \(f \colon X \to Y\) an injective function. Assign to \(X\) the subspace topology induced by \(f\) (see Definition~\ref{def:subspace_topology}); then, for any net \(\alpha \colon D \to X\):
  \begin{equation*}
    \alpha \to x \iff f \circ \alpha \to f(x).
  \end{equation*}
\end{proposition} \begin{proof}
  One direction is immediate from \(f\) being continuous by definition of the subspace topology.
  The other direction follows from the fact that, for every open neighbourhood \(U \in \tau_f\) of \(x\), there is an open set \(V \in \tau\) such that \(U = f^{-1}(V)\) by definition of the subspace topology.
  Then, if \(f \circ \alpha \to f(x)\), there is an element \(d \in D\) such that \(f(\alpha(\gdir{D}{d})) \subseteq V\); thus, \(\alpha(\gdir{D}{d}) \subseteq f^{-1}(V)\), implying \(\alpha \to x\).
\end{proof}

In general, limit points need not be unique, but they are if the net is defined in a Hausdorff space, as established below.

\begin{proposition}
  Let \(\alpha\) be a net in a Hausdorff space. Then, \(\alpha\) has at most one limit point.
\end{proposition} \begin{proof}
  Let \((D,\leq)\) be a directed set and \(X\) a Hausdorff space. Assume both \(x,y \in X\) are limit points of \(\alpha \colon D \to X\).
  Let \(U_x\) and \(U_y\) be open neighbourhoods of \(x\) and \(y\) respectively.
  Take \(d_x \in D\) such that \(\alpha(\gdir{D}{d_x}) \subseteq U_x\); do the same for \(U_y\), obtaining \(d_y \in D\).
  Because \(D\) is a directed set, we may take an upper bound of \(d_x\) and \(d_y\), call it \(d\).
  It follows that \(\alpha(d)\) is both in \(U_x\) and \(U_y\) and, hence, \(U_x \cap U_y \not = \varnothing\).
  This holds for any two open neighbourhoods of \(x\) and \(y\) so, according to the Hausdorff condition, they must be equal.
\end{proof}

When working with a Hausdorff space \(X\), the notation
\begin{equation*}
  \Lim{d \in D} \alpha(d) \keq x
\end{equation*}
will be used to indicate that the net \(\alpha \colon D \to X\) has \(x\) as its (unique) limit point, and we may use \(\lim \alpha\) to refer to such a point.

\subsection{Hausdorff commutative monoids}

A Hausdorff monoid is a monoid whose underlying set comes equipped with a Hausdorff topology and whose monoid operation is continuous.
Let \((X,+)\) be a monoid; since the monoid operation has the type \(+ \colon X \times X \to X\), to discuss continuity we must first assign a topology to \(X \times X\).

\begin{definition} \label{def:product_topology}
  Let \(\{(X_i,\tau_i)\}_I\) be a collection of topological spaces. Let \(\times_I X_i\) be the Cartesian product of the underlying sets and let \(\cB_\times\) be the collection of sets of the form \(\times_I U_i\) such that \(U_i \in \tau_i\) for each \(i \in I\) and the subset \(\{i \in I \mid U_i \not = X_i\}\) is finite.
  The \gls{product topology} \(\tau_\times\) is the one generated from the base \(\cB_\times\).
\end{definition}

This definition makes the topological space \((\times_I X_i, \tau_\times)\) a categorical product in \(\Top\) and in \(\TopHaus\).
In essence, this means that \(\tau_\times\) is the coarsest topology that can be assigned to the Cartesian product while satisfying that its projections are continuous functions.
Since \(\times_I X_i\) is a categorical product in \(\Set\), every net \(\alpha \colon D \to \times_I X_i\) is uniquely determined by a collection of nets \(\{\alpha_i \colon D \to X_i\}_I\) obtained by composing \(\alpha\) with the appropriate projection.
Let \((\alpha_i)_I\) be a shorthand for a net \(\alpha\) in \(\times_I X_i\) and, similarly, let \((x_i)_I \in \times_I X_i\) denote a tuple of dimension \(\abs{I}\) where \(x_i \in X_i\) for each \(i \in I\).

\begin{proposition} \label{prop:product_topology_limit}
  Let \(\{(X_i,\tau_i)\}_I\) be an indexed set of topological spaces and let \((\alpha_i)_I \colon D \to \times_I X_i\) be an arbitrary net. Then,
  \begin{equation*}
    (\alpha_i)_I \to (x_i)_I \  \iff  \ \forall i \in I,\, \alpha_i \to x_i.
  \end{equation*}
\end{proposition} \begin{proof}
  One direction follows from Proposition~\ref{prop:preserve_limit} since projections \(\pi_i \colon \times_I X_i \to X_i\) are continuous.
  For the other direction assume that \(\alpha_i \to x_i\) for all \(i \in I\); this implies that for every open neighbourhood \(U \in \tau_i\) of \(x_i\) there is an element \(d_i \in D\) such that \(\alpha_i(\gdir{D}{d_i}) \subseteq U\).
  Let \(\times_I V_i\) be an open neighbourhood of \((x_i)_I\) in the base \(\cB_\times\) (see Definition~\ref{def:product_topology}) and define the following subset of \(I\):
  \begin{equation*}
    J = \{i \in I \mid V_i \not= X_i\}
  \end{equation*}
  which is guaranteed to be finite due to the definition of \(\cB_\times\).
  For each \(j \in J\) obtain an element \(d_j \in D\) such that \(\alpha_j(\gdir{D}{d_j}) \subseteq V_j\) and let \(\hat{d}\) be an upper bound for the collection \(\{d_j\}_J\) --- such an upper bound exists because \(J\) is finite.
  Clearly, it holds that \(\alpha_i(\gdir{D}{\hat{d}}) \subseteq V_i\) for all \(i \in I\), either because \(i \in J\) or because \(V_i = X_i\); hence, \((\alpha_i(\gdir{D}{\hat{d}}))_I \subseteq \times_I V_i\).
  This applies to every open neighbourhood of \((x_i)_I\) in the base \(\cB_\times\) so that Proposition~\ref{prop:base_limit} implies that \((\alpha_i)_I \to (x_i)_I\), concluding the proof.
\end{proof}

In particular, this allows us to characterise continuous function \(X \times Y \to Z\) as limit-preserving functions on both coordinates.

\begin{proposition} \label{prop:preserve_limit_2}
  Let \(X\), \(Y\) and \(Z\) be topological spaces and let \(f \colon X \times Y \to Z\) be a function; then, \(f\) is continuous iff for every net \((\alpha,\beta) \colon D \to X \times Y\) the following is satisfied:
  \begin{equation} \label{eq:preserve_limit_2}
    \alpha \to x \text{ and } \beta \to y \implies f(\alpha,\beta) \to f(x,y).
  \end{equation}
\end{proposition} \begin{proof}
  According to Proposition~\ref{prop:preserve_limit}, \(f\) is continuous if and only if for every net \((\alpha,\beta) \colon D \to X \times Y\) the following is satisfied: 
  \begin{equation} \label{eq:preserve_limit_2_aux}
    (\alpha,\beta) \to (x,y) \implies f \circ (\alpha,\beta) \to f(x,y).
  \end{equation}
  Assume \(\alpha \to x\) and \(\beta \to y\); then, Proposition~\ref{prop:product_topology_limit} implies that \((\alpha,\beta) \to (x,y)\) so, if \(f\) is continuous, \(f \circ (\alpha,\beta) \to f(x,y)\) and the implication~\eqref{eq:preserve_limit_2} is satisfied.
  Conversely, assume \((\alpha,\beta) \to (x,y)\) and recall that Proposition~\ref{prop:product_topology_limit} works in both directions, so it implies \(\alpha \to x\) and \(\beta \to y\). Then, if~\eqref{eq:preserve_limit_2} is satisfied, it imposes \(f \circ (\alpha,\beta) \to f(x,y)\); therefore,~\eqref{eq:preserve_limit_2_aux} holds for every net \((\alpha,\beta)\) in \(X \times Y\), implying that \(f\) is continuous.
\end{proof}

Finally, we have all the ingredients to define topological monoids and, in particular, Hausdorff commutative monoids.

\begin{definition} \label{def:topological_monoid}
  Let \((X,\tau)\) be a topological space and let \((X,+)\) be a monoid.
  We say \((X,\tau,+)\) is a \gls{topological monoid} if \(+ \colon X \times X \to X\) is continuous.
  A \gls{Hausdorff commutative monoid} is a topological monoid where \((X,\tau)\) is Hausdorff and \((X,+)\) is commutative.
  A \gls{Hausdorff abelian group} is a Hausdorff commutative monoid \((X,\tau,+)\) where \((X,+)\) is an abelian group and the inversion map \(x \mapsto -x\) is continuous.
\end{definition}

\begin{example}
  The set \(\Rset\) of real numbers with the standard topology (see Example~\ref{ex:R_topology}) and the standard addition forms a Hausdorff commutative monoid --- in fact, a Hausdorff abelian group.
\end{example}

\begin{proposition} \label{prop:normed_Hausdorff}
  Every normed vector space is a Hausdorff abelian group.
\end{proposition} \begin{proof}
  As discussed in Example~\ref{ex:metric_topology}, every normed vector space \(V\) is a Hausdorff space and, by virtue of being a vector space, it is an abelian group with respect to addition.
  For every pair of vectors \(v,u \in V\) and every open neighbourhood \(U\) of \(v+u\) there is some \(\epsilon > 0\) such that \(B_{v+u,\epsilon} \subseteq U\) where \(B_{v+u,\epsilon}\) is the ball of radius \(\epsilon\) and center \(v+u\).
  It follows from the triangle inequality of the norm that every pair of vectors \(v' \in B_{v,\tfrac{\epsilon}{2}}\) and \(u' \in B_{u,\tfrac{\epsilon}{2}}\) satisfies \(v'+u' \in B_{v+u,\epsilon}\) and, consequently, \(+(B_{v,\tfrac{\epsilon}{2}},B_{u,\tfrac{\epsilon}{2}}) \subseteq U\), implying that \(+\) is continuous.
  The fact that the inversion map is continuous can be shown by similar means, completing the proof.
\end{proof}

\begin{definition}
  Let \GLS{\(\HausCMon\)}{HausCMon} be the category whose objects are Hausdorff commutative monoids and whose morphisms are continuous monoid homomorphisms.
  Let \GLS{\(\HAG\)}{HAG} be the full subcategory of \(\HausCMon\) whose objects are restricted to Hausdorff abelian groups.
\end{definition}

It is well known that \(\CMon\), \(\Ab\) and \(\TopHaus\) are complete categories and the forgetful functor from each of these to \(\Set\) preserves limits. This can be used to show that \(\HausCMon\) and \(\HAG\) are complete categories, as established below.

\begin{proposition} \label{prop:HausCMon_complete}
  Both \(\HausCMon\) and \(\HAG\) are complete categories.
\end{proposition} \begin{proof}
  First, we show that \(\HausCMon\) has equalizers.
  For every pair of Hausdorff commutative monoids \((X,\tau_X,+_X)\) and \((Y,\tau_Y,+_Y)\) and every pair of morphisms \(f,g \in \HausCMon(X,Y)\), define the subset \(E = \{x \in X \mid f(x) = g(x)\}\) with inclusion \(e \colon E \into X\) and let \(\tau_e\) be the subspace topology:
  \begin{equation*}
    \tau_e = \{e^{-1}(U) \mid U \in \tau_X\} = \{U \cap E \mid U \in \tau_X\}.
  \end{equation*}
  It is generally known --- and a simple exercise to show --- that \((E,\tau_e)\) along with the continuous function \(e\) is an equalizer of \(f\) and \(g\) in \(\TopHaus\).
  Let \(+_E \colon E \times E \to E\) be the restriction of \(+_X\) to inputs in \(E\); again, it is general knowledge that \((E,+_E)\) along with the monoid homorphism \(e\) is an equalizer of \(f\) and \(g\) in \(\CMon\).
  Consequently, \(e \colon E \to X\) is a continuous monoid homomorphism and \((E,\tau_e,+_E)\) is both a Hausdorff space and a commutative monoid; to show that this \(E\) is a Hausdorff commutative monoid we must prove that \(+_E\) is continuous.
  Recall that \(+_X\) is continuous by definition, so for any \(x,x' \in E\) and any open neighbourhood \(V \in \tau_X\) of \(x+x'\), we may find \(U,W \in \tau_X\) such that \(+_X(U,W) \subseteq V\); notice that \(+_E(U \cap E,W \cap E) \subseteq V \cap E\) and every open set in \(\tau_e\) is of the form \(- \cap E\), so \(+_E\) is continuous and \(E \in \HausCMon\).
  Since \(E\) is the equalizer of \(f\) and \(g\) in \(\Set\), it follows that \(E\) is a cone in \(\HausCMon\).
  Moreover, for any other cone \(A \in \HausCMon\), the morphism \(h \colon A \to X\) will factor through the unique function \(m \colon A \to E\) due to the fact that \(E\) is an equalizer in \(\Set\).
  It only remains to show that such a unique function \(m\) is a morphism in \(\HausCMon\); this is immediate since \(m\) is the mediating morphism both in \(\TopHaus\) and \(\CMon\), implying that \(m\) is a continuous monoid homomorphism.
  In conclusion, we have shown that that \(\HausCMon\) has equalizers.
  To prove that \(\HAG\) has equalizers as well we simply need to show that the inversion map in \((E,\tau_e,+_E)\) is continuous; this follows immediately from the corresponding property in the Hausdorff abelian group \(X\).

  Next, we show that \(\HausCMon\) has small products.
  Let \(\{(X_i,\tau_i,+_i)\}_I\) be a small collection of objects in \(\HausCMon\).
  Define \(+_\times\) by applying the corresponding \(+_i\) to each coordinate; \((\times_I X_i, +_i)\) is a small product in \(\CMon\).
  Define the product topology \(\tau_\times\) as in Definition~\ref{def:product_topology}; it is well-known that \((\times_I X_i, \tau_\times)\) is a small product in \(\TopHaus\).
  It follows that projections are continuous monoid homomorphisms and we must show that \(+_\times\) is continuous.
  Every open neighbourhood in the base \(\cB_\times\) for the product topology is of the form \(\times_I V_i\) where \(V_i \in \tau_i\) for each \(i \in I\); consequently, we may find \(U_i,U'_i \in \tau_i\) so that \(+_i(U_i,U'_i) \subseteq V_i\) for each \(i \in I\) due to each \(+_i\) being continuous.
  Then, it is immediate that the condition from Proposition~\ref{prop:base_continuous} is satisfied, implying that \(+_\times\) is continuous.
  Thus, it has been shown that \((\times_I X_i, \tau_\times, +_\times)\) is a Hausdorff commutative monoid and, evidently, a cone of the discrete diagram in \(\HausCMon\); it remains to show that it is a limit.
  For any other cone \(A \in \HausCMon\), each morphism \(h_i \colon A \to X_i\) will factor through the unique function \(m \colon A \to \times_I X_i\) due to the fact that \(\times_I X_i\) is a categorical product in \(\Set\).
  It only remains to show that such a unique function \(m\) is a morphism in \(\HausCMon\); this is immediate since \(m\) is the mediating morphism both in \(\TopHaus\) and \(\CMon\), implying that \(m\) is a continuous monoid homomorphism.
  In conclusion, we have shown \(\HausCMon\) has small products.
  To prove that \(\HAG\) has small products as well we simply need to show that the inversion map in \((\times_I X_i, \tau_\times, +_\times)\) is continuous; this follows immediately from the corresponding property in each Hausdorff abelian group \(X_i\).

  A category with equalizers and small products is complete, so both \(\HausCMon\) and \(\HAG\) are complete categories, as claimed.
\end{proof}

As promised in Section~\ref{sec:SCat}, every Hausdorff commutative monoid \((X,\tau,+)\) can be given a partial function \(\Sigma \colon X^* \pto X\) such that \((X,\Sigma)\) is a finitely total \(\Sigma\)-monoid whose summable families are precisely those that converge according to the topology.

\begin{definition} \label{def:topological_Sigma}
  Let \((X,\tau,+)\) be a Hausdorff commutative monoid. For any family \(\fml{x} \in X^*\), let \((\fpset{\fml{x}},\subseteq)\) be the directed set of finite subfamilies of \(\fml{x}\) (see Example~\ref{ex:fpset_directed}) and let \(\sigma \colon \fpset{\fml{x}} \to X\) be the net that maps each finite subfamily \(\fml{x'} \in \fpset{\fml{x}}\) to the sum of its elements, obtained after a finite number of applications of \(+\).
  The \emph{extended monoid operation} \(\Sigma \colon X^* \pto X\) is defined as follows:
  \begin{equation*}
    \Sigma \fml{x} = \begin{cases}
      x \quad &\ifc \lim \sigma \keq x \\
      \undefined &\otherwise.
    \end{cases}
  \end{equation*}
\end{definition}

We require \((X,+)\) to be a commutative monoid so that associativity and commutativity guarantee that the net \(\sigma\) is well-defined: otherwise the result would depend on the arbitrary order in which the elements from the finite family are combined.
The space is required to be Hausdorff so that when \(\sigma\) converges the result is unique.\footnote{A topological space is Hausdorff iff each net in the space converges to at most one point; see Chapter 2, Theorem 3 from Kelley's book~\cite{Topology}.}
The following lemmas prove the axioms of finitely total \(\Sigma\)-monoids for the extended monoid operation.

\begin{lemma} \label{lem:topSigma_finite}
  Let \((X,\tau,+)\) be a Hausdorff commutative monoid and let \(\Sigma\) be its extended monoid operation.
  If \(\fml{x} \in X^*\) is a finite family, then \(\Sigma \fml{x}\) is defined and it equals the sum of its elements, obtained after a finite number of applications of \(+\).
\end{lemma} \begin{proof}
  Let \(\sigma \colon \fpset{\fml{x}} \to X\) be the net that maps each finite subfamily \(\fml{x'} \in \fpset{\fml{x}}\) to the sum of its elements, obtained after a finite number of applications of \(+\).
  When \(\fml{x}\) is finite it is the largest element in \(\fpset{\fml{x}}\). Thus, \(\gdir{\fpset{\fml{x}}}{\fml{x}} = \{\fml{x}\}\) and \(\sigma(\fml{x})\) trivially satisfies being the limit point of \(\sigma\), regardless of the topology.
  By definition, \(\Sigma \fml{x} \keq \lim \sigma\) so we conclude that \(\Sigma \fml{x} \keq \sigma(\fml{x})\), as claimed.
\end{proof}

\begin{lemma} \label{lem:topSigma_neutral}
  Let \((X,\tau,+)\) be a Hausdorff commutative monoid and let \(\Sigma\) be its extended monoid operation.
  Let \(\fml{x} \in X^*\) be an arbitrary family and let \(\fml{x}_\emptyset\) be its subfamily of elements other than \(0\); if \(\fml{x}\) is summable then \(\fml{x}_\emptyset\) is summable as well.
\end{lemma} \begin{proof}
  Let \(\sigma \colon \fpset{\fml{x}} \to X\) be the net of finite partial sums of \(\fml{x}\) and let \(\sigma_\emptyset \colon \fpset{\fml{x}_\emptyset} \to X\) be the corresponding net for the family \(\fml{x}_\emptyset\).
  If \(\Sigma \fml{x} \keq x\) it follows from the definition of \(\Sigma\) that \(\lim \sigma \keq x\), \ie{} for every open neighbourhood \(U\) of \(x\) there is a finite subfamily \(\fml{x'} \in \fpset{\fml{x}}\) such that \(\sigma(\gdir{\fpset{\fml{x}}}{\fml{x'}}) \subseteq U\).
  Then, it is straightforward to check that \(\sigma_\emptyset(\gdir{\fpset{\fml{x}_\emptyset}}{\fml{x'}_\emptyset}) \subseteq U\) as well since for all \(\fml{x'} \in \fpset{\fml{x}}\) we have that \(\sigma(\fml{x'}) = \sigma_\emptyset(\fml{x'}_\emptyset)\) due to \(0\) being the neutral element of \(+\).
\end{proof}

\begin{lemma} \label{lem:topSigma_flattening}
  Let \((X,\tau,+)\) be a Hausdorff commutative monoid and let \(\Sigma\) be its extended monoid operation.
  Let \(\{\fml{x}_j\}_J\) be a collection of summable families where \(\abs{J} < \infty\) and let \(\fml{x} = \uplus_J \fml{x}_j\); then:
  \begin{equation}~\label{eq:topSigma_flattening}
    \Sigma \{\Sigma \fml{x}_j\}_J \keq x \implies \Sigma \fml{x} \keq x.
  \end{equation}
\end{lemma} \begin{proof}
  Let \(\sigma \colon \fpset{\fml{x}} \to X\) be the net that maps each finite family \(\fml{x'} \in \fpset{\fml{x}}\) to the sum of its elements.
  Similarly, let \(\sigma_j \colon \fpset{\fml{x}_j} \to X\) be the corresponding net for each \(j \in J\); the claim assumes that each \(\fml{x}_j\) is summable, hence, the limit point \(\lim \sigma_j\) exists for each \(j \in J\).
  Assume that \(J = \{0,1\}\); the general case will follow straightforwardly.
  In this case, the implication~\eqref{eq:topSigma_flattening} simplifies to:
  \begin{equation*}
    \lim \sigma_0 + \lim \sigma_1 \keq x \implies \lim \sigma \keq x.
  \end{equation*}

  Assume that, indeed, \(\lim \sigma_0 + \lim \sigma_1 \keq x\); unfortunately, we cannot apply continuity of \(+\) since \(\sigma_0\) and \(\sigma_1\) are defined with respect to to different directed sets.
  Define a function \(m \colon \fpset{\fml{x}} \to \fpset{\fml{x}_0}\) given by \(m(\fml{x'}) = \fml{x'} \cap \fml{x}_0\); clearly, \(m\) is isotone and, for each \(\fml{y} \in \fpset{\fml{x}_0}\) it is trivially satisfied that \(\fml{y} \subseteq m(\fml{y})\).
  Thus, the net \(\hat{\sigma}_0 \colon \fpset{\fml{x}} \to X\) given by \(\hat{\sigma}_0 = \sigma_0 \circ m_0\) is a subnet of \(\sigma_0\) and Proposition~\ref{prop:subnet_limit} implies that \(\lim \hat{\sigma}_0 \keq \lim \sigma_0\).
  We may define \(\hat{\sigma}_1\) in the same manner and it is immediate that:
  \begin{equation*}
    \lim \hat{\sigma}_0 + \lim \hat{\sigma}_1 \keq x.
  \end{equation*}
  Since both \(\hat{\sigma}_0\) and \(\hat{\sigma}_1\) are now defined on the same directed set \(\fpset{\fml{x}}\), Proposition~\ref{prop:preserve_limit_2} implies that \(\lim\, (\hat{\sigma}_0 + \hat{\sigma}_1) \keq x\).
  It is straightforward to check that, due to associativity and commutativity of \(+\), the nets \(\sigma\) and \(\hat{\sigma}_0 + \hat{\sigma}_1\) are, in fact, the same function.
  Consequently, it has been shown that \(\lim \sigma \keq x\); this establishes the claim in the case that \(\abs{J} = 2\).

  If \(J\) is empty or it is a singleton set, the claim follows immediately.
  If \(\abs{J} > 2\) we may repeat the argument above as many times as necessary to aggregate all subfamilies; this process will eventually be completed since the claim assumes \(J\) is finite.
\end{proof}

\begin{lemma} \label{lem:topSigma_bracketing}
  Let \((X,\tau,+)\) be a Hausdorff commutative monoid and let \(\Sigma\) be its extended monoid operation.
  Let \(\{\fml{x}_j\}_J\) be a collection of finite families and let \(\fml{x} = \uplus_J \fml{x}_j\); then:
    \begin{equation} \label{eq:topSigma_bracketing}
      \Sigma \fml{x} \keq x \implies \Sigma \{\Sigma \fml{x}_j\}_J \keq x.
    \end{equation}
\end{lemma} \begin{proof}
  Since each family \(\fml{x}_j\) is finite, Lemma~\ref{lem:topSigma_finite} establishes that \(\Sigma \fml{x}_j \keq x_j\) for some \(x_j \in X\).
  Let \(\hat{\sigma} \colon \fpset{J} \to X\) be the net that maps each finite subset \(J' \in \fpset{J}\) to the sum of the finite family \(\{x_j\}_{J'}\).
  Notice that, by definition, \(\Sigma \{\Sigma \fml{x}_j\}_J \keq \lim \hat{\sigma}\) so our goal is to show that \(\Sigma \fml{x} \keq x\) implies \(\lim \hat{\sigma} \keq x\).

  Assume that \(\Sigma \fml{x} \keq x\) and let \(\sigma \colon \fpset{\fml{x}} \to X\) be the net that maps each finite subfamily \(\fml{x'} \in \fpset{\fml{x}}\) to the sum of its elements; then, \(\lim \sigma \keq x\).
  Define a function \(m \colon \fpset{J} \to \fpset{\fml{x}}\) as follows:
  \begin{equation*}
    m(J') = \uplus_{J'} \fml{x}_j.
  \end{equation*}
  The result of \(m(J')\) is a finite family since \(J'\) is finite and the claim assumes that each \(\fml{x}_j\) is a finite family.
  It is evident that \(m\) is isotone and, for each finite subfamily \(\fml{x'} \in \fpset{\fml{x}}\) we may define a finite subset \(\hat{J} \in \fpset{J}\) as follows:
  \begin{equation*}
    \hat{J} = \{j \in J \mid \fml{x'} \cap \fml{x}_j \not = \varnothing\}
  \end{equation*}
  so that \(\fml{x'} \subseteq m(\hat{J})\) is trivially satisfied.
  Moreover, notice that \(\hat{\sigma} = \sigma \circ m\); since all families involved are finite the identity follows from associativity and commutativity of \(+\).
  Consequently, \(\hat{\sigma}\) is a subnet of \(\sigma\) and Proposition~\ref{prop:subnet_limit} implies that \(\lim \hat{\sigma} \keq x\).
  As discussed above, this means that \(\Sigma \{\Sigma \fml{x}_j\}_J \keq x\), completing the proof.
\end{proof}

\begin{proposition} \label{prop:wSm_Hausdorff}
  Let \((X,\tau,+)\) be a Hausdorff commutative monoid and let \(\Sigma\) be its extended monoid operation; then, \((X,\Sigma)\) is a finitely total \(\Sigma\)-monoid.
\end{proposition} \begin{proof}
  Every finite family is summable according to Lemma~\ref{lem:topSigma_finite}.
  In particular, the empty family is summable and we have that \(\Sigma \{x,\Sigma \varnothing\} \keq x + \Sigma \varnothing\).
  Moreover, according to Proposition~\ref{prop:neutral_element}, we have \(\Sigma \{x,\Sigma \varnothing\} \keq x\) for all \(x \in X\).
  If we combine these two results and let \(x = 0\) we obtain \(0 + \Sigma \varnothing = 0\) which, due to \((X,+,0)\) being a monoid, results in \(\Sigma \varnothing = 0\).
  Thus, the neutral element of the monoid is the neutral element of the \(\Sigma\)-monoid and Lemma~\ref{lem:topSigma_neutral} completes the proof of the neutral element axiom.
  The singleton axiom follows trivially from Lemma~\ref{lem:topSigma_finite}, flattening is proven in Lemma~\ref{lem:topSigma_flattening} and bracketing is proven in Lemma~\ref{lem:topSigma_bracketing}.
\end{proof}

\begin{remark} \label{rmk:weak_bracketing}
  Since Hausdorff abelian groups are a subclass of Hausdorff commutative monoids that have additive inverses, it is immediate from Proposition~\ref{prop:no_inverses} that Hausdorff commutative monoids need not satisfy strong flattening.
  The question of whether or not Hausdorff commutative monoids satisfy strong bracketing is more subtle.
  The main obstacle is that, at some point, we must prove that for some net \(\sigma \colon \fpset{J} \times \fpset{\fml{x}} \to X\) the following implication holds:
  \begin{equation*}
    \Lim{\fpset{J} \times \fpset{\fml{x}}} \sigma \keq x \implies \Lim{J' \in \fpset{J}} \left( \Lim{\fpset{\fml{x}}} \sigma(J',-) \right) \keq x.
  \end{equation*}
  In essence, we need to obtain the limit of a net \(\hat{\sigma} \colon \fpset{J} \to X\) that maps each finite subset \(J' \in \fpset{J}\) to the limit point \(\lim \sigma(J',-)\).
  Unfortunately, since in the case of strong bracketing the families \(\fml{x}_j\) are allowed to be infinite, the argument via the construction of a subnet from the proof of Lemma~\ref{lem:topSigma_bracketing} cannot be used --- the function \(m(J') = \uplus_{J'} \fml{x}_j\) need not yield families in \(\fpset{\fml{x}}\).
  On the other hand, even though \(\Lim{\fpset{J} \times \fpset{\fml{x}}} \sigma \keq x\) guarantees that for every open neighbourhood \(U\) of \(x\) we may find \((J',\fml{x'})\) such that \(\sigma(\gdir{\fpset{J} \times \fpset{\fml{x}}}{(J',\fml{x'})}) \subseteq U\), this does not imply that \(\lim \sigma(J',-) \in U\) since \(U\) need not be closed.
  In fact, proving that an iterated limit \(\Lim{A} \Lim{B} \alpha\) exists and coincides with the double limit \(\Lim{A \times B} \alpha\) is generally a subtle matter. Theorem 4 from Chapter 2 in~\cite{Topology} shows that if \(\Lim{A} \Lim{B} \alpha\) exists we may define a directed set \(A \times B^A\) such that
  \begin{equation*}
    \Lim{(a,f) \in A \times B^A} \alpha(a,f(a)) \ \keq\ \Lim{a \in A} \Lim{b \in B} \alpha(a,b)
  \end{equation*}
  but, in order to get from the left hand side to \(\Lim{A \times B} \alpha\), it appears we would need that for each open neighbourhood \(U\) of \(\Lim{A} \Lim{B} \alpha\) the (possibly infinite) set \(\{f_U(a') \mid a' \geq a_U\}\) had an upper bound --- where \((a_U,f_U) \in A \times B^A\) is the point such that all \((a',f') \geq (a_U,f_U)\) satisfy \(\alpha(a',f'(a')) \in U\).
\end{remark}

Continuous monoid homomorphisms between Hausdorff commutative monoids become \(\Sigma\)-homomorphisms between the \(\Sigma\)-monoids induced by their extended monoid operation.

\begin{lemma} \label{lem:HausCMon_preserves_summability}
  Let \(X\) and \(Y\) be Hausdorff commutative monoids and let \(f \colon X \to Y\) be a continuous monoid homomorphism. Then, for every \(\fml{x} \in X^*\) and \(x \in X\):
  \begin{equation*}
    \Sigma \fml{x} \keq x \implies \Sigma f\fml{x} \keq f(x).
  \end{equation*}
\end{lemma} \begin{proof}
  Let \(\sigma \colon \fpset{\fml{x}} \to X\) be the net of finite partial sums of \(\fml{x}\) and assume that \(\Sigma \fml{x} \keq x\).
  Then \(\lim \sigma \keq x\) and, due to \(f\) being continuous, Proposition~\ref{prop:preserve_limit} implies that \(\lim\, (f \circ \sigma) \keq f(x)\).
  Considering that \(f\) is a monoid homomorphism, it is straightforward to check that \(f \circ \sigma\) is equivalent to the net assigned to \(f\fml{x}\).
  It follows that \(\Sigma f\fml{x} \keq f(x)\), proving the claim.
\end{proof}

\begin{definition} \label{def:HausCMon_ft}
  There is a faithful functor \(G \colon \HausCMon \to \SCat{ft}\) that equips each Hausdorff commutative monoid with the extended monoid operation from Definition~\ref{def:topological_Sigma} and acts as the identity on morphisms.
\end{definition} 

Clearly, \(G\) is faithful: \(f\) and \(G(f)\) are the same function, so it is obvious that \(f \not= g\) implies \(G(f) \not= G(g)\).
However, \(G\) is not full as illustrated by the following example.

\begin{example}
  Let \(A \in \HausCMon\) with underlying set \(\{0\} \cup [1,\infty)\), discrete topology and the standard addition over the real numbers.
  Let \(B \in \HausCMon\) with the same underlying set \(\{0\} \cup [1,\infty)\) along with standard addition, but whose topology is given by the base of open sets comprised of the singleton \(\{0\}\) along with every interval \([x, x+\epsilon)\) for each \(x \in [1,\infty]\) and each \(\epsilon > 0\).

  Notice that \(\Sigma\) on both \(G(A)\) and \(G(B)\) is only defined on families with a finite number of nonzero elements; the reason being that all nonzero elements are \(\geq 1\) so infinite sums of these always diverge. Since both \(G(A)\) and \(G(B)\) have the same summable families, it is immediate that \(G(A) = G(B)\). Then, if we consider the identity function on the underlying sets \(\id \colon B \to A\) we get that:
  \begin{itemize}
    \item \(\id \colon G(B) \to G(A)\) is a valid morphism in \(\SCat{ft}\) (indeed, it is the identity morphism);
    \item since \(G\) acts as the identity on morphisms, if the morphism \(\id \colon G(B) \to G(A)\) were to be in the image of the map \(\HausCMon(B,A) \to \SCat{ft}(G(B),G(A))\), we would need \(\id \colon B \to A\) to be continuous;
    \item however, \(\id \colon B \to A\) is not continuous since the topology of \(B\) is strictly coarser than that of \(A\).
  \end{itemize}
  Therefore, \(G \colon \HausCMon \to \SCat{ft}\) is not full.
\end{example}

Similarly, there is a faithful functor \(\HAG \to \SCat{g}\) since the requirement that the inversion map is continuous implies that it is a \(\Sigma\)-homomorphism; we expect this functor not to be full either.
The following subsection establishes that both this functor and \(\HausCMon \to \SCat{ft}\) have a left adjoint.
But first, the following remark briefly discusses whether \(\HAG\)-enriched categories are well-defined and what their relevance is with respect to categories in quantum quantum computer science.

\begin{remark} \label{rmk:HausCMon-enriched}
  Every Hilbert space is a Hausdorff abelian group and \((\FdHilb,\otimes,I)\) is a closed symmetric monoidal category, so its hom-sets are Hilbert spaces and, consequently, \(\FdHilb\) could potentially be perceived as a \(\HAG\)-enriched category.
  This is possible because the standard tensor product of vector spaces --- when applied to \emph{finite}-dimensional Hilbert spaces --- is already complete.
  However, the matter becomes more subtle when we consider the hom-sets of \(\Hilb\), since these are no longer Hilbert spaces, even though they are Hausdorff abelian groups using the strong operator topology from Definition~\ref{def:SOT}.
  However, tensor products in \(\HAG\) are a very subtle matter~\cite{TensorTAG} and it is well-known~\cite{TensorMonoidalTAG} that there is no tensor product \(\otimes\) for which \((\HAG,\otimes,I)\) is a monoidal category, implying that \(\HAG\)-enriched categories are not well-defined.

  Furthermore, in the case of \(\FdContraction\), hom-sets are not even (total) monoids, so \(\HausCMon\)-enriched categories (even if they were well-defined) would be too restrictive for our purposes.
  Fortunately, the fact that every \((\SCat{*},\otimes,I)\) category is monoidal (see Section~\ref{sec:SCat_tensor}) along with the functor \(\HAG \to \SCat{g}\) can be used to realise \(\Hilb\) as a \(\SCat{g}\)-enriched category.
  Then, due to Lemma~\ref{lem:wSm_restriction}, the hom-sets of \(\Contraction\) --- which are restrictions of the hom-sets of \(\Hilb\) --- are weak \(\Sigma\)-monoids, realising \(\Contraction\) as a \(\SCat{w}\)-enriched category.
  Consequently, Chapter~\ref{chap:trace} will deal with \(\SCat{*}\)-enriched categories whose \(\Sigma\) function on hom-sets is defined in terms of the extended monoid operation of a Hausdorff abelian group.
\end{remark}

\subsection{The left adjoint to \(\HausCMon \to \SCat{ft}\)}
\label{sec:HausCMon->SCatft}

The functor \(G \colon \HausCMon \to \SCat{ft}\) described in the previous section has a left adjoint.
An explicit construction has not been achieved; instead, the proof uses the general adjoint functor theorem.

\begin{lemma} \label{lem:HausCMon->SCatft_preserves_limits}
  The functor \(G \colon \HausCMon \to \SCat{ft}\) preserves limits.
\end{lemma} \begin{proof}
  First, we show that \(G\) preserves equalizers.
  Let \(X\) and \(Y\) be two Hausdorff commutative monoids and let \(f,g \in \HausCMon(X,Y)\); the equalizer of \(f\) and \(g\) in \(\HausCMon\) is defined in Proposition~\ref{prop:HausCMon_complete}.
  Its underlying set \(E = \{x \in X \mid f(x) = g(x)\}\) is the same as that of the equalizer of \(G(f)\) and \(G(g)\) in \(\SCat{ft}\) (see Proposition~\ref{prop:SCatw_complete}).
  Moreover, the extended monoid operation induced by the topology assigned to \(E \in \HausCMon\) coincides with the \(\Sigma\) function from the equalizer in \(\SCat{ft}\).
  To show this, recall that \(E\) is endowed with the subspace topology induced by \(e \colon E \into X\) so, according to Proposition~\ref{prop:subspace_topology_limit}, every family \(\fml{x} \in E^*\) whose net of finite partial sums is denoted \(\sigma \colon \fpset{\fml{x}} \to E\) satisfies:
  \begin{equation*}
    \lim \sigma \keq x \iff \lim e \circ \sigma \keq e(x)
  \end{equation*}
  Since \(e\) is a monoid homomorphism it follows from the definition of the extended monoid operation that the partial function \(\Sigma' \colon E^* \pto E\) in \(G(E)\) satisfies
  \begin{equation*}
    \Sigma' \fml{x} \keq x \iff \Sigma e\fml{x} = e(x)
  \end{equation*}
  where \(\Sigma\) is the extended monoid operation assigned to \(G(X)\).
  The latter two-way implication uniquely characterises the \(\Sigma\) function~\eqref{eq:SCat_equalizer} of the equalizer in \(\SCat{ft}\).
  Thus, both constructions yield the same \(\Sigma\)-monoid, implying \(G\) preserves equalizers.

  Similarly, the categorical product in \(\HausCMon\) of a small collection of objects \(\{X_i \in \HausCMon\}_I\) and the categorical product of \(\{G(X_i)\}_I\) in \(\SCat{ft}\) coincide in their underlying sets and projections.
  To prove that \(G\) preserves small products, it suffices to show that the extended monoid operation assigned by \(G\) to the product in \(\HausCMon\) coincides with the \(\Sigma\) function of the product in \(\SCat{ft}\).
  Let \(\{x_i \in X_i^*\}\) be a collection of families and let \(\sigma_i \colon \fpset{\fml{x}_i} \to X_i\) be their corresponding net of finite partial sums; Proposition~\ref{prop:product_topology_limit} establishes that:
  \begin{equation*}
    \lim\, (\sigma_i)_I \keq (x_i)_I \ \iff \  \forall i \in I,\, \lim \sigma_i \keq x_i.
  \end{equation*}
  Since addition in \(\times_I X_i\) is defined coordinate-wise, it follows from the definition of the extended monoid operation that the partial function \(\Sigma'\) in \(G(\times_I X_i)\) satisfies
  \begin{equation*}
    \Sigma' (\fml{x}_i)_I \keq (x_i)_I \ \iff \ \forall i \in I,\, \Sigma^i \fml{x}_i \keq x_i
  \end{equation*}
  where \(\Sigma^i\) is the extended monoid operation assigned to each \(G(X_i)\).
  The latter two-way implication uniquely characterises the \(\Sigma\) function~\eqref{eq:SCat_small_product} of the categorical product in \(\SCat{ft}\).
  Thus, both constructions yield the same \(\Sigma\)-monoid, implying \(G\) preserves small products.

  A functor that preserves equalizers and small products and whose domain is a complete category automatically preserves all limits. Therefore, \(G\) preserves limits.
\end{proof}

\begin{proposition} \label{prop:HCMon_ft_adj}
  There is a left adjoint to the functor \(G \colon \HausCMon \to \SCat{ft}\).
\end{proposition} \begin{proof}
  It is immediate that \(\HausCMon\) is locally small since there is a faithful forgetful functor \(\HausCMon \to \Set\).
  According to Proposition~\ref{prop:HausCMon_complete}, the category \(\HausCMon\) is complete and Lemma~\ref{lem:HausCMon->SCatft_preserves_limits} establishes that the functor \(G\) preserves limits.
  It remains to show that, for any \(X \in \SCat{ft}\), the comma category \(\comma{X}{G}\) has a weakly initial set; then, the existence of a left adjoint will follow from the general adjoint functor theorem (see Theorem~\ref{thm:GAFT}).
  Let \(S\) be the collection of all \(\Sigma\)-homomorphisms \(X \xto{q} G(A)\) where \(A\) may be any Hausdorff commutative monoid on a quotient of \(X\).  
  Notice that \(S\) is small since all quantifiers in its definition are with respect to a fixed set \(X\).

  Let \(Y\) be an arbitrary Hausdorff commutative monoid and let \(f \colon X \to G(Y)\) be a \(\Sigma\)-homomorphism.
  Recall that every function can be factored through its image as follows:
  \begin{equation*}
    f = X \xto{q} X/{\sim} \xto{\bar{f}} \im(f) \xto{u} G(Y)
  \end{equation*}
  where \(x \sim x'\) iff \(f(x) = f(x')\); moreover, \(q\) is surjective, \(\bar{f}\) is bijective and \(u\) is injective.
  Recall that \(G(Y)\) has the same underlying set as \(Y\), so \(\im(f) \subseteq Y\) and, thanks to \(u\) and \(\bar{f}\) being injective, we may define a subspace topology \(\tau_\Sigma\) on \(X/{\sim}\) induced by \(u\bar{f}\).
  The composite \(u \bar{f} \colon X/{\sim} \to Y\) is continuous and \((X/{\sim},\tau_\Sigma)\) is Hausdorff.
  Notice that \(\sim\) satisfies the following for every \(a,a',b,b' \in X\):
  \begin{equation*}
    a \sim a' \text{ and } b \sim b' \  \implies \ \Sigma \{a,b\} \sim \Sigma \{a',b'\}.
  \end{equation*}
  due to \(f\) being a \(\Sigma\)-homomorphism and \(X\) being finitely total:
  \begin{equation*}
    f(\Sigma \{a,b\}) = \Sigma \{f(a),f(b)\} = \Sigma \{f(a'),f(b')\} = f(\Sigma \{a',b'\}).
  \end{equation*}
  Therefore, the operation \(+_\Sigma \colon X/{\sim} \times X/{\sim} \to X/{\sim}\) defined as follows:
  \begin{equation*}
    [a] +_\Sigma [b] = [\Sigma \{a,b\}]
  \end{equation*}
  is well-defined and it is trivial to show that it forms a commutative monoid, with the equivalence class \([0]\) being its neutral element.
  Moreover, the composite \(u\bar{f} \colon X/{\sim} \to Y\) is a monoid homomorphism:
  \begin{equation*}
    u\bar{f}([a] +_\Sigma [b]) = f(\Sigma \{a,b\}) = \Sigma \{f(a),f(b)\} = f(a) + f(b) = u\bar{f}[a] + u\bar{f}[b]
  \end{equation*}
  where \(+\) is the monoid operation from \(Y\) and \(\Sigma \{f(a),f(b)\} = f(a) + f(b)\) by definition of the extended monoid operation from \(Y\).
  To prove that \(+_\Sigma\) is continuous first notice that for every two nets \(\alpha,\beta \colon D \to X\) the following is satisfied:
  \begin{align*}
    \lim\, q\alpha &\keq [a] \text{ and } \lim\, q\beta \keq [b] && \\
      &\implies \lim f\alpha \keq f(a) \text{ and } \lim f\beta \keq f(b) &&\text{(\(u\bar{f}\) continuous, \(f = u\bar{f}q\))} \\
      &\implies \lim\, (f\alpha + f\beta) \keq f(a) + f(b) &&\text{(\(+\) is continuous)} \\
      &\implies \lim f(\Sigma \{\alpha,\beta\}) \keq f(\Sigma \{a,b\}) &&\text{(\(\Sigma\)-homomorphism, \(X \in \SCat{ft}\))} \\
      &\implies \lim u\bar{f}[\Sigma \{\alpha,\beta\}] \keq u\bar{f}[\Sigma \{a,b\}] &&\text{(\(f = u\bar{f}q\))} \\
      &\implies \lim\, [\Sigma \{\alpha,\beta\}] \keq [\Sigma \{a,b\}]. &&\text{(def.\@ \(\tau_\Sigma\), Proposition~\ref{prop:subspace_topology_limit})}
  \end{align*}
  Notice that every net \(\alpha' \colon D \to X/{\sim}\) may be (non-uniquely) characterised by a net \(\alpha \colon D \to X\) such that \(\alpha' = q\alpha\) since \(q\) is surjective.
  Thus, according to Proposition~\ref{prop:preserve_limit_2} and the chain of implications above, the function \(+_\Sigma = [\Sigma \{-,-\}]\) is continuous and it follows that \((X/{\sim},\tau_\Sigma,+_\Sigma)\) is a Hausdorff commutative monoid.
  Finally, we must show that the quotient map \(q \colon X \to G(X/{\sim})\) is a \(\Sigma\)-homomorphism.
  Let \(\fml{x} \in X^*\) be a summable family and let \(\sigma_f \colon \fpset{f\fml{x}} \to Y\) and \(\sigma_q \colon \fpset{q\fml{x}} \to X/{\sim}\) be the nets of finite partial sums of \(f\fml{x} \in Y^*\) and \(q\fml{x} \in (X/{\sim})^*\) respectively.
  Let \(m \colon \fpset{q\fml{x}} \to \fpset{f\fml{x}}\) be the function that maps
  each subfamily \(q\fml{x'} \in \fpset{q\fml{x}}\) to the subfamily \(f\fml{x'} = u\bar{f}q\fml{x'}\); it is evident that \(m\) is isotone.
  For each \(f\fml{x'} \in \fpset{f\fml{x}}\) it is immediate that \(f\fml{x'} \subseteq m(q\fml{x'})\) so that \(m(\fpset{q\fml{x}})\) is cofinal in \(\fpset{f\fml{x}}\).
  Moreover, \(u\bar{f}\sigma_q = \sigma_f \circ m\) so that \(u\bar{f}\sigma_q\) is a subnet of \(\sigma_f\).
  The following chain of implications establishes that \(q\) is a \(\Sigma\)-homomorphism:
  \begin{align*}
    \Sigma \fml{x} \keq x &\implies \Sigma f\fml{x} \keq f(x)  &&\text{(\(\Sigma\)-homomorphism \(f\))} \\
      &\implies \lim \sigma_f \keq f(x)  &&\text{(def.\@ \(\Sigma\) in \(G(Y)\))} \\
      &\implies \lim\, (u\bar{f} \circ \sigma_q) \keq u\bar{f}q(x)  &&\text{(\(u\bar{f}\sigma_q\) subnet, \(f = u\bar{f}q\))} \\
      &\implies \lim \sigma_q \keq q(x)  &&\text{(def.\@ \(\tau_\Sigma\), Proposition~\ref{prop:subspace_topology_limit})} \\
      &\implies \Sigma q\fml{x} \keq q(x).  &&\text{(def.\@ \(\Sigma\) in \(G(X/{\sim})\))}
  \end{align*}

  In summary, it has been shown that \(X/{\sim}\) is a Hausdorff commutative monoid and that \(X \xto{q} G(X/{\sim})\) is an element in \(S\).
  Moreover, \(f = u \bar{f} \circ q\) and \(u\bar{f}\) is a continuous monoid homomorphism, providing a morphism \((q,X/{\sim}) \to (f,Y)\) in the comma category.
  This construction can be reproduced for every \(Y \in \HausCMon\) and every \(\Sigma\)-homomorphism \(f \colon X \to G(Y)\).
  Consequently, \(S\) is a weakly initial set in the comma category \(\comma{X}{G}\) and the claim that \(G\) has a left adjoint follows from the general adjoint functor theorem (see Theorem~\ref{thm:GAFT}).
\end{proof}

\begin{corollary} \label{cor:HAG_g_adj}
  There is a left adjoint to the functor \(G \colon \HAG \to \SCat{g}\).
\end{corollary} \begin{proof}
  The claim follows from the same argument used in the previous proposition with the exception that, in this case, we must prove that \((X/{\sim},\tau_\Sigma,+_\Sigma)\) is a Hausdorff abelian group.
  It has already been established that \(X/{\sim} \in \HausCMon\), so we only need to prove that inverses exist in \(X/{\sim}\) for each element and that the inversion map is continuous.
  Since in this case \(X\) is a \(\Sigma\)-group, the existence of inverses in \(X/{\sim}\) is immediate, where the inverse of every \([x] \in X/{\sim}\) is \([-x]\).
  To prove that the inversion map \([x] \mapsto [-x]\) is continuous, first notice that for every net \(\alpha \colon D \to X\) the following is satisfied:
  \begin{align*}
    \lim\, q\alpha \keq [a] &\implies \lim f\alpha \keq f(a) &&\text{(\(u\bar{f}\) continuous, \(f = u\bar{f}q\))} \\
      &\implies \lim\, (-f\alpha) \keq -f(a) &&\text{(continuous inversion map in \(Y\))} \\
      &\implies \lim f(-\alpha) \keq f(-a) &&\text{(group homomorphism \(f\))} \\
      &\implies \lim u\bar{f}[-\alpha] \keq u\bar{f}[-a] &&\text{(\(f = u\bar{f}q\))} \\
      &\implies \lim\, [-\alpha] \keq [-a]. &&\text{(def.\@ \(\tau_\Sigma\), Proposition~\ref{prop:subspace_topology_limit})}
  \end{align*}
  Recall that every net \(\alpha' \colon D \to X/{\sim}\) may be (non-uniquely) characterised by a net \(\alpha \colon D \to X\) such that \(\alpha' = q\alpha\) since \(q\) is surjective.
  Thus, according to Proposition~\ref{prop:preserve_limit} and the chain of implications above, the inversion map in \(X/{\sim}\) is continuous and it follows that \((X/{\sim},\tau_\Sigma,+_\Sigma)\) is a Hausdorff abelian group, completing the proof.
\end{proof}

\section{Related work}
\label{sec:Sigma_rel_work}

The main goal of this chapter was to lay the groundwork for Chapter~\ref{chap:trace} which builds upon \(\SCat{*}\)-enriched categories to formalise iteration in a categorical setting.
Such was the objective of Haghverdi's work on unique decomposition categories~\cite{Haghverdi}, which are an instance of \(\SCat{s}\)-enriched categories that will take a central role in the next chapter.
To this end, Haghverdi introduced a notion of \(\Sigma\)-monoids that has been generalised in Section~\ref{sec:SCat} of this thesis.
On the other hand, Section~\ref{sec:topology} is, for the most part, an introduction to general topology and Hausdorff commutative monoids; the main contribution in it being the construction of a functor \(\HausCMon \to \SCat{ft}\) that  extends the monoid operation to infinitary inputs using the notion of convergence of nets from general topology.
A brief discussion of a selection of works closely related to the contents of this chapter is provided below.

\paragraph*{Haghverdi~\cite{Haghverdi}.} Haghverdi's (strong) \(\Sigma\)-monoids are defined in terms of two axioms: the unary sum axiom and the partition-asociativity axiom. The former corresponds precisely to the singleton axiom from weak \(\Sigma\)-monoids (Definition~\ref{def:wSm}) whereas the latter is equivalent to the combination of the axioms of subsummability, strong bracketing and strong flattening (Definition~\ref{def:sSm}).
Separating partition-associativity into its two directions of implication (\ie{} strong bracketing and strong flattening) allow us to weaken each of them in a different way.
The motivation behind the weakening of flattening is evident from a close inspection of Proposition~\ref{prop:no_inverses}, since it is strong flattening on \emph{infinite} collections of families what prevents the existence of additive inverses.
On the other hand, the motivation of the weakening of bracketing comes from the desire to capture  Hausdorff commutative monoids as instances of weak \(\Sigma\)-monoids (see Remark~\ref{rmk:weak_bracketing}).

\paragraph*{Hoshino~\cite{RTUDC}.} Hoshino continued Haghverdi's line of work on \(\SCat{s}\)-enriched unique decomposition categories. In~\cite{RTUDC}, Hoshino proved that the category of total strong \(\Sigma\)-monoids (\ie{} those where \(\Sigma\) is a total function) is a \emph{reflective} subcategory of strong \(\Sigma\)-monoids.
This was achieved via a non-constructive argument using the general adjoint functor theorem and its strategy is reproduced in Section~\ref{sec:ft_w_adj} to prove that the embedding \(\SCat{ft} \into \SCat{w}\) has a left adjoint.
Interestingly, the proof of Proposition~\ref{prop:s_w_adj} which establishes that \(\SCat{s}\) is a reflective subcategory of \(\SCat{w}\) can be used (virtually without alteration) to give a constructive proof of Hoshino's result; in this case, the intersection construction from Lemma~\ref{lem:intersection_of_sSm} is no longer necessary.
Moreover, the proof of the existence of a monoidal structure given by the tensor product on every \(\SCat{*}\) category (Section~\ref{sec:SCat_tensor}) follows the same strategy Hoshino used in~\cite{RTUDC} to prove such a result for the particular case of \(\SCat{s}\).

\paragraph*{Higgs~\cite{Higgs}.} The \(\Sigma\)-groups defined in this thesis are equivalent to a subclass of Higgs' notion of \(\Sigma\)-groups; in particular, to those where the set of summable families is restricted to contain countable families only. Thus, the weak \(\Sigma\)-monoids introduced in Section~\ref{sec:SCat} are general enough to capture two disjoint notions of infinitary monoids: Haghverdi's strong \(\Sigma\)-monoids (which cannot have inverses) and Higgs' \(\Sigma\)-groups (which must have inverses).
To verify that the axioms in Higgs' definition of \(\Sigma\)-groups~\cite{Higgs} coincides with ours, notice that the neutral element axiom and the singleton axiom of weak \(\Sigma\)-monoids appear explicitly in Higgs' definition as axioms (\(\Sigma1\)) and (\(\Sigma2\)) respectively, whereas finite totality is imposed by Higgs by requiring that the \(\Sigma\)-group \((X,\Sigma)\) contains an abelian group.
The bracketing axiom of weak \(\Sigma\)-monoids is Higgs' (\(\Sigma4\)) axiom, whereas Higgs' (\(\Sigma3\)) axiom corresponds to flattening and is stated as follows (paraphrased using this thesis' notation):
\begin{equation*}
  \Sigma \fml{x} \keq x \text{ and } \Sigma \fml{x'} \keq x' \implies \Sigma (\fml{x} \uplus -\fml{x'}) \keq x - x'.
\end{equation*}
Since flattening of weak \(\Sigma\)-monoids only discusses addition of finite collection of families and since \(\Sigma\)-groups are required to be finitely total, Higgs' axiom implies flattening.
Moreover, when \(\fml{x} = \varnothing\), Higgs' axiom implies that every summable family \(\fml{x'}\) has a counterpart \(-\fml{x'}\) so that \(\Sigma (-\fml{x'}) \keq -(\Sigma \fml{x'})\).
Thus, Higgs combines in a single axiom both the flattening axiom of weak \(\Sigma\)-monoids and the requirement that the inverse mapping is a \(\Sigma\)-homomorphism.
With these remarks in mind, it is straightforward to check that Higgs' definition of \(\Sigma\)-groups is equivalent to the definition presented in this thesis.

On another note, the realisation that Hausdorff abelian groups can be captured as instances of \(\Sigma\)-groups is briefly mentioned in the introduction of Higgs' paper~\cite{Higgs}, although no explicit proof is provided.
It seems that such a result does not hold if \(\Sigma\)-groups are required to satisfy strong bracketing (see Remark~\ref{rmk:weak_bracketing}) and it was Higgs' definition that revealed the subtle weakening to bracketing that enables us to define a functor \(\HAG \to \SCat{g}\).
For the purposes of this thesis, it is sufficient to study the connection between \(\Sigma\)-monoids and Hausdorff commutative monoids; but Higgs was interested in the more general framework of net groups.\footnote{A net group is an abelian group \((X,+)\) equipped with a relation \(\to\) between nets \(\alpha\) in \(X\) and elements \(x \in X\), along with some axioms that make it reasonable to interpret \(\alpha \to x\) as the statement `the net \(\alpha\) converges to \(x\)' without the use of topology.}
Higgs proved that every \(\Sigma\)-group is a net group and, moreover, the corresponding embedding functor has a right adjoint.
This is somewhat dual to the result from Section~\ref{sec:HausCMon->SCatft} where the functor realising every Hausdorff abelian group as a \(\Sigma\)-group --- via a \(\Sigma\) function defined in terms of the convergence of certain net --- is shown to have a left adjoint.

%% file: Chapters/3.tex
\chapter{Categorical study of quantum iteration}
\label{chap:trace}

This chapter builds upon the previous one to introduce \(\SCat{*}\)-enriched categories and study categorical traces on them.
In particular, we are interested in proving the validity of the execution formula as a categorical trace in certain categories of quantum processes.
The execution formula is closely related to the notion of iteration in computer science and, in fact, the main motivation behind this chapter is to provide the categorical foundations to study unbounded quantum iterative loops.

This chapter is structured as follows: Section~\ref{sec:traced_categories} introduces the framework of traced monoidal categories along with standard examples of categorical traces, including the execution formula; Section~\ref{sec:UDC} introduces Haghverdi's unique decomposition categories and builds upon them to provide a framework capable of capturing quantum iteration; Section~\ref{sec:quantum_ex_trace} proves that \((\FdContraction,\oplus,\ex)\) is a totally traced category, where \(\ex\) corresponds to the execution formula; Section~\ref{sec:LSI} introduces a particular category of quantum processes over time and shows that the execution formula is a well-defined categorical trace on it; finally, Section~\ref{sec:trace_rel_work} discusses related work.

\section{Traced monoidal categories (\emph{Preamble})}
\label{sec:traced_categories}

Assume \((\C,\otimes,I)\) is a monoidal category whose morphisms \(f \colon A \to B\) may be interpreted as processes that receive input of type \(A\) and yield output of type \(B\); Section~\ref{sec:QM_cats} provides examples of such categories in the context of quantum processes.
If we wish to formalise the notion of an iterative process within the category \(\C\) we need to define an operation that connects the output of a process back to one of its inputs:
\begin{equation} \label{eq:pic_trace}
  \input{Figures/3/pic_trace}
\end{equation}
More precisely, we need a map of type \(\C(A\otimes U, B \otimes U) \to \C(A,B)\) for all \(A,B,U \in \C\) -- which we will refer to as the \emph{trace operator} --- and, as usual, we ought to impose a set of axioms that ensure this trace operator interacts nicely with composition and the monoidal structure.
Such a framework is captured by traced monoidal categories, originally proposed by Joyal, Street and Verity~\cite{TraceJoyal} and whose definition is given below.
Traced monoidal categories are not limited to iterative loops; for instance, Section~\ref{sec:CCC_trace} provides an example of a traced monoidal category that captures the notion of calculating the trace of a matrix, which has little in common with the notion of iteration.
Consequently, for the purposes of this thesis it will be necessary to find an appropriate definition of the trace operator that captures iteration; such will be the role of the execution formula, as discussed in Section~\ref{sec:ex_formula}.

\begin{definition} \label{def:traced_category}
  Let \((\C,\otimes,I)\) be a symmetric monoidal category and let there be a family of maps
  \begin{equation*}
    \Tr_{A,B}^U \colon \C(A\otimes U, B \otimes U) \to \C(A,B)
  \end{equation*}
  for all \(A,B,U \in \C\).
  The mapping \(\Tr_{A,B}^U\) may be represented pictorially as in~\eqref{eq:pic_trace} and, for the sake of brevity, the subscript is often omitted, writing \(\Tr^U\) instead.
  \(\C\) is a \gls{traced monoidal category} if the following axioms are satisfied.
  \begin{itemize}
    \item \emph{Naturality.} For all morphisms \(f \colon A \otimes U \to B \otimes U\), \(g \colon A' \to A\) and \(h \colon B \to B'\) in \(\C\),
    \begin{align*}
      \input{Figures/3/naturality} \\      
      h \circ \Tr^U(f) \circ g \quad\quad &= \quad\quad \Tr^U((h \otimes \id_U) \circ f \circ (g \otimes \id_U)).
    \end{align*}
    \item \emph{Dinaturality.} For all morphisms \(f \colon A \otimes U \to B \otimes U'\) and \(g \colon U' \to U\) in \(\C\),
    \begin{align*}
      \input{Figures/3/dinaturality} \\      
      \Tr^U((\id_B \otimes g) \circ f) \quad\quad &= \quad\quad \Tr^{U'}(f \circ (\id_A \otimes g)).
    \end{align*}
    \item \emph{Superposing.} For all morphisms \(f \colon A\otimes U \to B \otimes U\) and \(g \colon C \to D\) in \(\C\),
    \begin{align*}
      \input{Figures/3/superposing} \\
      g \otimes \Tr^U(f) \quad\quad &= \quad\quad \Tr^U(g \otimes f).
    \end{align*}
    \item \emph{Vanishing I.} For all morphisms \(f \colon A \otimes I \to B \otimes I\) in \(\C\),
    \begin{align*}
      \input{Figures/3/vanishingI} \\
      \Tr^I(f) \quad\quad &= \quad\quad \rho \circ f \circ \rho^{-1}.
    \end{align*}
    \item \emph{Vanishing II.} For all morphisms \(f \colon A \otimes U \otimes V \to B \otimes U \otimes V\) in \(\C\),
    \begin{align*}
      \input{Figures/3/vanishingII} \\
      \Tr^U(\Tr^V(f)) \quad\quad &= \quad\quad \Tr^{U \otimes V}(f).
    \end{align*}
    \item \emph{Yanking.} For all objects \(U \in \C\) where \(\sigma_{U,U}\) is the symmetric braiding,
    \begin{align*}
      \input{Figures/3/yanking} \\
      \Tr^U(\sigma_{U,U}) \quad\quad &= \quad\quad \id_U.
    \end{align*}
  \end{itemize} 
\end{definition}

In this thesis, traced monoidal categories are denoted by a triple \((\C,\otimes,\Tr)\) where \(\Tr\) is the collection of all trace operators \(\Tr_{A,B}^U\) and the monoidal unit is omitted when unambiguous.
In certain cases of interest to this thesis, the trace operator \(\Tr_{A,B}^U\) cannot be defined for all morphisms in its domain.
Haghverdi and Scott~\cite{PartialTraceOrig} proposed a generalisation of traced categories where the trace operator need not be total.
To define these, the Kleene equality is used (see Notation~\ref{not:keq}) so that, for two partial functions \(f\) and \(g\), the expression \(f(a) \keq g(a')\) indicates that \(f(a)\) is defined iff \(g(a')\) is defined and their results agree.

\begin{definition} \label{def:partially_traced_category}
  Let \((\C,\otimes,I)\) be a symmetric monoidal category and let there be a family of partial maps
  \begin{equation*}
    \Tr_{A,B}^U \colon \C(A\otimes U, B \otimes U) \pto \C(A,B)
  \end{equation*}
  for all \(A,B,U \in \C\).
  \(\C\) is a \gls{partially traced category} if the following axioms are satisfied.
  \begin{itemize}
    \item \emph{Naturality.} For all \(f \colon A \otimes U \to B \otimes U\), \(g \colon A' \to A\) and \(h \colon B \to B'\),
    \begin{equation*}     
      \Tr^U(f) \,\text{ defined} \ \implies\ h \circ \Tr^U(f) \circ g \ \keq\ \Tr^U((h \otimes \id_U) \circ f \circ (g \otimes \id_U)).
    \end{equation*}
    \item \emph{Dinaturality.} For all \(f \colon A \otimes U \to B \otimes U'\) and \(g \colon U' \to U\),
    \begin{equation*}     
      \Tr^U((\id_B \otimes g) \circ f) \ \keq \ \Tr^{U'}(f \circ (\id_A \otimes g)).
    \end{equation*}
    \item \emph{Superposing.} For all \(f \colon A\otimes U \to B \otimes U\) and \(g \colon C \to D\),
    \begin{equation*}     
      \Tr^U(f) \,\text{ defined} \ \implies\ g \otimes \Tr^U(f) \ \keq\ \Tr^U(g \otimes f).
    \end{equation*}
    \item \emph{Vanishing I.} For all \(f \colon A \otimes I \to B \otimes I\),
    \begin{equation*}     
      \Tr^I(f) \ \keq \ \rho \circ f \circ \rho^{-1}.
    \end{equation*}
    \item \emph{Vanishing II.} For all \(f \colon A \otimes U \otimes V \to B \otimes U \otimes V\),
    \begin{equation*}     
      \Tr^V(f) \,\text{ defined} \ \implies\ \Tr^U(\Tr^V(f)) \ \keq\ \Tr^{U \otimes V}(f).
    \end{equation*}
    \item \emph{Yanking.} For all objects \(U \in \C\),
    \begin{equation*}
      \Tr^U(\sigma_{U,U}) \ \keq \ \id_U.
    \end{equation*}
  \end{itemize} 
\end{definition}

Any traced monoidal category \((\C,\otimes,\Tr)\) as given in Definition~\ref{def:traced_category} is, evidently, an instance of a partially traced category where all operators \(\Tr^U_{A,B}\) are totally defined; consequently, we refer to these as \emph{totally} traced categories.
The PhD thesis of Octavio Malherbe~\cite{MalherbePhD} and its accompanying paper~\cite{Malherbe} are valuable references on partially traced categories.
The following proposition is due to Malherbe and collaborators, and it lets us induce a partial trace on \(\C\) from the partial trace of another category \(\D\).

\begin{proposition}[Malherbe~\cite{MalherbePhD}, Proposition 4.3.1] \label{prop:induced_trace}
  Let \(F \colon \C \to \D\) be a faithful strong symmetric monoidal functor with \((\D,\otimes,\Tr)\) a partially traced category and \((\C,\otimes)\) a symmetric monoidal category.
  For every morphism \(f \in \C(A \oplus U, B \oplus U)\) define:
  \begin{equation*}
    \widehat{\Tr}^U(f) = \begin{cases}
      g &\ifc \exists g \in \C(A,B) \ \text{ s.t. }\  \Tr^{F(U)}(\mu_\sub{B,U}^{-1} F(f) \mu_\sub{A,U}) \keq F(g) \\
      \undefined &\otherwise.
    \end{cases}
  \end{equation*}
  Then, \((\C,\otimes,\widehat{\Tr})\) is a partially traced category.
\end{proposition} \begin{proof}[Proof. (Sketch)]
  For simplicity, assume \(F\) is strict monoidal (\ie{} \(\mu\) is the identity); for a detailed proof in the general case see~\cite{MalherbePhD}, Proposition 4.3.1.
  First, notice that \(\widehat{\Tr}\) is well-defined: if \(g \in \C(A,B)\) exists such that \(\Tr^{F(U)}(F(f)) \keq F(g)\) then such \(g\) is unique due to \(F\) being faithful.
  To check each axiom of partially traced categories we proceed as follows:
  \begin{itemize}
    \item assume \(\widehat{\Tr}^U(f)\) is defined and appears on the left (or right) hand side of one of the axioms;
    \item by definition, \(F(\widehat{\Tr}^U(f)) = \Tr^{F(U)}(F(f))\) and we may use that \(\D\) is partially traced to derive the existence of the right (or left) hand side of the axiom in \(\D\);
    \item since \(F\) is a strict symmetric monoidal functor, the left (and right) hand side of the axiom can be rewritten as \(F\) being applied to the corresponding left (and right) hand side of the axiom in \(\C\);
    \item since \(F\) is faithful it follows that the axiom in \(\C\) is satisfied.
  \end{itemize}
  Consequently, \((\C,\otimes,\widehat{\Tr})\) is a partially traced category, as claimed.
\end{proof}

\subsection{The execution formula}
\label{sec:ex_formula}

The execution formula is a candidate for a categorical trace in multiple categories relevant in computer science~\cite{ManesArbib,Haghverdi,RTUDC}.
The primary goal of this section is to explain how the execution formula captures the notion of iterative loops and identify the requirements that a category must meet so that the execution formula may be defined on it.

In the picture below, the box labelled \(f\) may be interpreted as a physical device with input ports \(A\) and \(U\) and output ports \(B\) and \(U\) whose \(U\)-ports have been joined together; arrows have been added to clarify the direction of information flow within the wires.
\begin{equation} \label{eq:diag_loop}
  \begin{tikzpicture}
    \node[rectangle,draw=black,thick,minimum width=6mm,minimum height=10mm] (f) {\(f\)};
    \coordinate[below=3mm of f.west] (A);
    \coordinate[left=7mm of A] (Ad);
    \draw (Ad) -- node[below] {\tiny \(A\)} node {\tiny \(>\)} (A);
    \coordinate[below=3mm of f.east] (B);
    \coordinate[right=7mm of B] (Bd);
    \draw (Bd) -- node[below] {\tiny \(B\)} node {\tiny \(>\)} (B);
    \coordinate[above=3mm of f.west] (Ui);
    \coordinate[left=3mm of Ui] (Uid);
    \draw (Uid) -- node[above] {\tiny \(U\)} (Ui);
    \coordinate[above=3mm of f.east] (Uo);
    \coordinate[right=3mm of Uo] (Uod);
    \draw (Uod) -- node[above] {\tiny \(U\)} (Uo);
    \coordinate[above=6mm of Uid] (Uiu);
    \coordinate[above=6mm of Uod] (Uou);
    \draw (Uiu) -- node {\tiny \(<\)} (Uou);
    \draw (Uid) edge[out=180,in=180,looseness=1.5] (Uiu);
    \draw (Uod) edge[out=0,in=0,looseness=1.5] (Uou);
  \end{tikzpicture}
\end{equation}
Once the connection between the \(U\)-ports is made, any input on \(A\) may only leave through \(B\), but before doing so, it may traverse the loop an arbitrary number of times.
If we intend to describe the relationship between inputs on \(A\) and outputs on \(B\), we must aggregate the contribution of all of these possible paths; that is the key idea behind the execution formula.
To formalise this in a monoidal category \((\C,\oplus,Z)\) the following conditions must be satisfied:
\begin{itemize}
  \item morphisms \(f \colon A \oplus U \to B \oplus U\) must admit a characterisation in terms of components \(f_\sub{YX} \colon X \to Y\) for each \(X \in \{A,U\}\) and \(Y \in \{B,U\}\), so that we may study the different paths from \(A\) to \(B\) in isolation, for instance:
  \begin{align*}
    A &\xto{f_\sub{BA}} B, \\
    A \xto{f_\sub{UA}} U \xto{f_\sub{UU}} &\ U \xto{f_\sub{UU}} U \xto{f_\sub{BU}} B,
  \end{align*}
  \etc{} --- this can be achieved by requiring \(A \oplus B\) to be a biproduct for all objects \(A,B \in \C\) but, as we will see in Section~\ref{sec:UDC}, certain categories without biproducts may also satisfy this condition;
  \item there must be a notion of infinite summation of morphisms, so that the contribution of all paths may be aggregated --- this can be achieved by requiring that \(\C\) is a category enriched over \(\Sigma\)-monoids, as discussed in Chapter~\ref{chap:Sigma}.
\end{itemize}
In such a situation, the \gls{execution formula} may be informally presented as:
\begin{equation*}
  \ex^U(f) = f_\sub{BA} + \sum^\infty_{n = 0} f_\sub{BU} f_\sub{UU}^n f_\sub{UA}
\end{equation*}
for any morphism \(f \colon A \oplus U \to B \oplus U\).
The abstract diagram from~\eqref{eq:diag_loop} can be interpreted to depict any process \(f\) whose component \(f_\sub{UU}\) may be applied an arbitrary number of times; if these processes are described by computer programs, we recover the notion of iteration from computer science (more on this in Section~\ref{sec:classical_loop}).
When the execution formula provides a valid categorical trace some works in the literature~\cite{ParticleStyle} refer to it as a ``particle-style" trace.

Whenever a category \(\C\) is traced with respect to the execution formula, we are guaranteed that \(\C\) captures a well-behaved notion of iteration, in the sense that it satisfies the axioms of traced monoidal categories and, hence, the execution formula interacts nicely with both composition and the symmetric monoidal structure.
But, perhaps surprisingly, proving that the execution formula is a categorical trace is not a simple matter,
and the study of unique decomposition categories, initiated by Haghverdi~\cite{Haghverdi}, aims to characterise what must be required of \(\C\) for this to be the case.
This chapter builds upon the framework of unique decomposition categories to formalise the notion of iteration on categories of quantum processes.
But, before doing so, we discuss other candidates for trace operators on categories of quantum processes and how these relate to iteration.

\subsection{The kernel-image trace}
\label{sec:kernel_image}

As discussed in the previous section, the execution formula may only be defined in a category enriched over \(\Sigma\)-monoids since there is an infinite number of execution paths that need to be aggregated.
This section discusses a couple of closely related trace operators, defined on categories enriched over abelian groups.

\begin{definition} \label{def:additive_cat}
  An \emph{additive category} is an \(\Ab\)-enriched category that has all finite biproducts.
  Every additive category \(\C\) has a canonical monoidal structure \((\C,\oplus,0)\) where \(\oplus\) is the canonical extension of binary biproducts to a monoidal product and \(0\) is a zero object.
\end{definition}

\begin{notation} \label{not:matrix_decomposition}
  For a morphism \(f \colon A \oplus B \to C \oplus D\) it is useful to refer to its components using the subscript shorthand \(f_\sub{CA}\), in this case referring to the morphism
  \begin{equation*}
    f_\sub{CA} = A \xto{\iota} A \oplus B \xto{f} C \oplus D \xto{\pi} C
  \end{equation*}
  where \(\iota\) and \(\pi\) are the injection and projection morphisms of the biproduct.
  The order of the subscripts facilitates the readability of composition: for another morphism \(h \colon X \to A \oplus B\), we may compose \(f_\sub{CA} \circ h_\sub{AX}\) and it is immediate to read, from right to left, the objects that the morphism factors through.
\end{notation}

Haghverdi and Scott~\cite{PartialTraceOrig} studied the following partial trace, which may be defined on every additive category.

\begin{proposition}[Haghverdi and Scott~\cite{PartialTraceOrig}]
  Every additive category \((\C,\oplus,0)\) is partially traced, where its trace operator \(\Tr_{A,B}^U\) is defined as follows:
  \begin{equation*}
    \Tr_{A,B}^U(f) = \begin{cases}
      f_\sub{BA} + f_\sub{BU} (\id_\sub{U}\!-\!f_\sub{UU})^{-1} f_\sub{UA} &\ifc \id_\sub{U}\!-\!f_\sub{UU} \text{ is invertible} \\
      \undefined &\otherwise.
    \end{cases}
  \end{equation*}
\end{proposition}

The proof of this result has been omitted for brevity and can be found in~\cite{PartialTraceOrig}.
This Haghverdi-Scott trace is somewhat similar to the execution formula, especially in the light of the following well-known identity on the real numbers \(\abs{r} < 1\):
\begin{equation*}
  \sum_{n=0}^\infty r^n = (1-r)^{-1}.
\end{equation*}
However, there is a caveat: the Haghverdi-Scott trace of morphisms such as the identity \(f = \id_\sub{A \oplus U}\) is undefined since
\begin{equation*}
  \id_\sub{U} - f_\sub{UU} = \id_\sub{U} - \pi_\sub{U} \circ \id_\sub{A \oplus U} \circ \iota_\sub{U} = \id_\sub{U} - \id_\sub{U} = 0
\end{equation*}
is clearly not invertible.
This is somewhat odd since \(\Tr^U(\id_\sub{A \oplus U}) \keq \id_\sub{A} \oplus \Tr^U(\id_\sub{U})\) according to the superposing axiom and, even though \(\Tr^U(\id_\sub{U})\) is undefined as well, its type is \(0 \to 0\) suggests that it may naturally be defined as the zero morphism. 
In such a scenario, the superposing axiom would yield \(\Tr^U(\id_\sub{A \oplus U}) \keq \id_\sub{A}\) which matches the intuition behind iteration since tracing \(\id_\sub{A \oplus U}\) would create an isolated loop on \(U\) that no input can reach.
Indeed, Malherbe, Scott and Selinger~\cite{Malherbe} proposed a generalisation of the Haghverdi-Scott trace where identities and other morphisms that would `leave an isolated loop' can be traced.

\begin{definition}[Malherbe \etal{}~\cite{Malherbe}] \label{def:ki_trace}
  Let \(\C\) be an additive category. A morphism \(f \colon A \oplus U \to B \oplus U\) in \(\C\) is said to have a trace witnessed by morphisms \(i \colon A \to U\) and \(k \colon U \to B\) if the diagram
  \begin{equation*}
    \begin{tikzcd} [row sep=large]
      A && U \\
      & {} & U && B
      \arrow["i", dashed, from=1-1, to=1-3]
      \arrow["{f_\dsub{UA}}"', from=1-1, to=2-3]
      \arrow["{f_\dsub{BU}}", from=1-3, to=2-5]
      \arrow["k", dashed, from=2-3, to=2-5]
      \arrow["{\id - f_\dsub{UU}}"', from=1-3, to=2-3]
    \end{tikzcd}
  \end{equation*}
  commutes.
  Whenever it does, we write \((k,i) \Vdash \Tr^{U}(f)\) and define \(\Tr^U(f) \colon A \to B\) as follows:
  \begin{equation*}
    \Tr^U(f) = f_\sub{BA} + k \circ f_\sub{UA} = f_\sub{BA} + f_\sub{BU} \circ i
  \end{equation*}
  and otherwise we leave \(\Tr^U(f)\) undefined.
  The collection of all such partial functions \(\Tr^U \colon \C(A \oplus U, B \oplus U) \pto \C(A,B)\) for all objects \(A,B,U \in \C\) is known as the \gls{kernel-image trace}.
\end{definition}

\begin{proposition}[Malherbe \etal{}~\cite{Malherbe}, Proposition 3.17] \label{prop:Ab_ki_trace}
  Any additive category is partially traced using the kernel-image trace.
\end{proposition}

We write \(\Tr_\sub{\ki}\) when we wish to distinguish the kernel-image trace from other traces available on the same category.
The following proposition justifies the name given to the kernel-image trace.

\begin{proposition}[Malherbe \etal{}~\cite{Malherbe}, Remark 3.16] \label{prop:Vect_ki}
  A morphism \(f \colon A \oplus U \to B \oplus U\) in \(\Vect\) satisfies \((k,i) \Vdash \Tr^U(f)\) for some morphisms \(k \in \Vect(U,B)\) and \(i \in \Vect(A,U)\) iff the following inclusions hold:
  \begin{equation} \label{eq:ki_inclusions}
    \begin{aligned}
      \im(f_\sub{UA}) \ &\subseteq \ \im(\id\!-\!f_\sub{UU}) \\
      \ker(\id\!-\!f_\sub{UU}) \ &\subseteq \ \ker(f_\sub{BU}).
    \end{aligned}
  \end{equation}
\end{proposition} \begin{proof}
  Assume that \((k,i) \Vdash \Tr^U(f)\) is satisfied; this implies the existence of a morphism \(i \colon A \to U\) such that \(f_\sub{UA} = (\id\!-\!f_\sub{UU}) \circ i\) and, hence, \(u \in \im(f_\sub{UA})\) implies \(u \in \im(\id\!-\!f_\sub{UU})\).
  Similarly, since a morphism \(k \colon U \to B\) satisfying \(f_\sub{BU} = k \circ (\id\!-\!f_\sub{UU})\) exists, it follows that for any \(u \in U\) such that \((\id\!-\!f_\sub{UU})(u) = 0\) we must have \(f_\sub{BU}(u) = 0\).
  Therefore, \((k,i) \Vdash \Tr^U(f)\) implies both inclusions~\eqref{eq:ki_inclusions}.

  Conversely, assume that both of inclusions~\eqref{eq:ki_inclusions} are satisfied and let \(h = \id\!-\!f_\sub{UU}\). Decompose \(h\) as follows:\footnote{This decomposition is available in every abelian category, which \(\Vect\) is an example of.}
  \begin{equation*}
    h = U \xto{q} U\!/{\ker(h)} \xto{h'} \im(h) \xto{m} U
  \end{equation*}
  where \(U\!/{\ker(h)}\) is the quotient space induced by the equivalence relation \(\sim\) where:
  \begin{equation*}
    u \sim v \iff (u-v) \in \ker(h)
  \end{equation*}
  and \(q\) is its quotient map.
  The linear map \(h'\) sends equivalence classes \([u] \in U\!/{\ker(h)}\) to \(h(u)\); this is well-defined since \([u] = [v]\) implies \(h(u-v) = 0\) and hence \(h(u) = h(v)\).
  It can be shown that \(h'\) is bijective, \(q\) is surjective and \(m\) is injective.
  Consequently, \(h' \circ q\) is surjective, so there is at least one linear map \(r \colon \im(h) \to U\) acting as its right inverse, \((h' \circ q) \circ r = \id\).
  Similarly, \(m \circ h'\) is injective, so there is at least one linear map \(l \colon U \to U\!/{\ker(h)}\) acting as its left inverse, \(l \circ (m \circ h') = \id\).
  Now, since we assumed that \(\im(f_\sub{UA}) \subseteq \im(h)\) we may decompose \(f_\sub{UA}\) as follows:
  \begin{equation*}
    f_\sub{UA} = A \xto{f_\sub{UA}'} \im(h) \xto{m} U.
  \end{equation*}
  Let \(i = r \circ f_\sub{UA}'\); it is immediate that:
  \begin{equation*}
    (\id\!-\!f_\sub{UU}) \circ i = (m \circ h' \circ q) \circ (r \circ f_\sub{UA}') = m \circ f_\sub{UA}' = f_\sub{UA}.
  \end{equation*}
  On the other hand, since we assumed that \(\ker(h) \subseteq \ker(f_\sub{BU})\) we may decompose \(f_\sub{BU}\) as follows:
  \begin{equation*}
    f_\sub{BU} = U \xto{q} U\!/{\ker(h)} \xto{f_\sub{BU}'} B.
  \end{equation*}
  Let \(k = f_\sub{BU}' \circ l\); it is immediate that:
  \begin{equation*}
    k \circ (\id\!-\!f_\sub{UU}) = (f_\sub{BU}' \circ l) \circ (m \circ h' \circ q) = f_\sub{BU}' \circ q = f_\sub{BU}.
  \end{equation*}
  Therefore, the kernel-image inclusions~\eqref{eq:ki_inclusions} imply \((k,i) \Vdash \Tr^U(f)\).
\end{proof}

\begin{corollary} \label{cor:FdHilb_ki}
  A morphism \(f \colon A \oplus U \to B \oplus U\) in \(\FdHilb\) satisfies \((k,i) \Vdash \Tr^U(f)\) for some morphisms \(k \in \FdHilb(U,B)\) and \(i \in \FdHilb(A,U)\) iff the kernel-image inclusions~\eqref{eq:ki_inclusions} hold.
\end{corollary} \begin{proof}
  The claim can be proven using the argument from the previous proposition. However, we must check that the linear maps \(i\) and \(k\) defined in the second part of the proof are bounded. But this is immediate since every linear map between finite-dimensional spaces is bounded.
\end{proof}

The case of \(\Hilb\) is more complicated. On one hand, \((\Hilb,\oplus,\{0\})\) is an additive category, so it follows from Proposition~\ref{prop:Ab_ki_trace} that \((\Hilb,\oplus,\Tr_\sub{\ki})\) is a partially traced category.
On the other hand, it appears that the kernel-image inclusions~\eqref{eq:ki_inclusions} are not sufficient for a morphism \(f\) in \(\Hilb\) to be traceable.
This is because when the domain and codomain are infinite-dimensional the image of \(f\) need not be a closed subspace, which causes problems when attempting to prove the existence and boundness of \(i\) and \(k\) as constructed in the proof of Proposition~\ref{prop:Vect_ki}.
Taking this into account, the following proposition establishes sufficient conditions for \(f\) to be traceable.

\begin{proposition} \label{prop:Hilb_ki}
  For any \(f \in \Hilb(A \oplus U, B \oplus U)\), if
  \begin{itemize}
    \item it satisfies the kernel-image inclusions~\eqref{eq:ki_inclusions} and
    \item \(\im(\id\!-\!f_\sub{UU})\) is a closed subspace,
  \end{itemize}
  then \((k,i) \Vdash \Tr^U(f)\) for some morphisms \(k \in \Hilb(U,B)\) and \(i \in \Hilb(A,U)\).
\end{proposition} \begin{proof}[Proof. (Sketch)]
  As in the proof of Proposition~\ref{prop:Vect_ki}, let \(h = \id\!-\!f_\sub{UU}\) and decompose it as follows:
  \begin{equation*}
    h = U \xto{q} U\!/{\ker(h)} \xto{h'} \im(h) \xto{m} U.
  \end{equation*}
  This decomposition exists in \(\Hilb\) precisely because \(\im(h)\) is closed by assumption and, hence, it is a Hilbert space, whereas it can be shown that \(U\!/{\ker(h)}\) is isomorphic to the orthogonal complement of \(\ker(h)\) and, hence, a Hilbert space as well; since \(h\) is bounded, it follows that \(h'\) is bounded and it is simple to show that \(m\) and \(q\) are bounded as well.
  Then, it is immediate that \(h' \circ q\) is surjective, so there is at least one linear map \(r \colon \im(h) \to U\) acting as its right inverse.
  We must show that \(r\) is bounded; according to Proposition 4.6.1 from~\cite{Aubin} this is implied by \(h' \circ q\) being a surjective bounded linear map.
  Similarly, \(m \circ h'\) is injective, so there is at least one linear map \(l \colon U \to U\!/{\ker(h)}\) acting as its left inverse.
  We must show that \(l\) is bounded; according to Proposition 4.5.2 from~\cite{Aubin} this is implied by \(m \circ h'\) being an injective bounded linear map and \(\im(m \circ h') = \im(h)\) being a closed subspace --- the latter being satisfied by assumption.
  It is straightforward to check that both \(f_\sub{UA}\) and \(f_\sub{BU}\) may be decomposed into
  \begin{align*}
    f_\sub{UA} &= A \xto{f_\sub{UA}'} \im(h) \xto{m} U \\
    f_\sub{BU} &= U \xto{q} U\!/{\ker(h)} \xto{f_\sub{BU}'} B.
  \end{align*}
  where \(f_\sub{UA}'\) and \(f_\sub{BU}'\) are bounded.
  It follows that \(i = r \circ f_\sub{UA}'\) and \(k = f_\sub{BU}' \circ l\) are bounded linear maps and, hence, morphisms in \(\Hilb\). What remains of the proof follows by the same argument from Proposition~\ref{prop:Vect_ki}.
\end{proof}

The following two propositions are well-known facts about the orthogonal complement\footnote{Given a Hilbert space \(H\) and a set of vectors \(S \subseteq H\), the orthogonal complement of \(S\) is defined as follows: \(S^\perp = \{v \in H \mid \forall u \in S,\, \braket{u}{v} = 0\}\). It can be shown that \(S^\perp\) is a closed subspace of \(H\).}  of Hilbert spaces and they will be used in Section~\ref{sec:quantum_ex_trace} to prove that the kernel-image inclusions~\eqref{eq:ki_inclusions} are satisfied by every morphism in \(\FdContraction\).

\begin{proposition} \label{prop:ker_dagger_im_perp}
  Every morphism \(f \in \Hilb(A,B)\) satisfies:
  \begin{equation*}
    \ker(f^\dagger) = \im(f)^\perp
  \end{equation*}
  where \(f^\dagger\) is the adjoint of \(f\).\footnote{Recall that the adjoint of a bounded linear map \(f \in \Hilb(A,B)\) is the bounded linear map \(f^\dagger \in \Hilb(B,A)\) satisfying \(\braket{f(a)}{b} = \braket{a}{f^\dagger(b)}\) for every \(a \in A\) and \(b \in B\).}
\end{proposition} \begin{proof}
  Recall that in a Hilbert space \(A\), every vector \(a \in A\) satisfies:
  \begin{equation} \label{eq:kdip_aux}
    a = 0 \iff \forall a' \in A,\ \braket{a'}{a} = 0.
  \end{equation}
  Then, the following sequence of implications prove the claim:
  \begin{align*}
    v \in \ker(f^\dagger) &\iff \forall a \in A,\ \braket{a}{f^\dagger(v)} = 0 &&\text{(implication~\eqref{eq:kdip_aux})}\\
     &\iff \forall a \in A,\ \braket{f(a)}{v} = 0 &&\text{(def\@. of adjoint)}\\
     &\iff v \in \im(f)^\perp &&\text{(def\@. of orthogonal complement)}. \qedhere
  \end{align*}
\end{proof}

\begin{proposition} \label{prop:perp_inclusion}
  Let \(H\) be a Hilbert space, \(K\) a closed subspace of \(H\) and \(S\) a subset of \(H\). Then,
  \begin{equation*}
    S \subseteq K \iff K^\perp \subseteq S^\perp.
  \end{equation*}
\end{proposition} \begin{proof}
  Assume \(S \subseteq K\); for any \(v \in K^\perp\) we have that for all \(u \in K\), \(\braket{u}{v} = 0\). Thus, in particular, for all \(u \in S\), \(\braket{u}{v} = 0\) and \(v \in S^\perp\).
  To prove the other direction, assume \(K^\perp \subseteq S^\perp\); then the previous argument establishes that \(S^{\perp\perp} \subseteq K^{\perp\perp}\).
  It is immediate that \(S \subseteq S^{\perp\perp}\) and, for any closed subspace, \(K = K^{\perp\perp}\).
  Thus, \(S \subseteq S^{\perp\perp} \subseteq K^{\perp\perp} = K\) and the proof is complete.
\end{proof}

\subsection{The canonical trace on compact closed categories}
\label{sec:CCC_trace}

This section describes the canonical trace on the monoidal category \((\FdHilb,\otimes,\Cset)\) and, more generally, on any compact closed category.
It will be argued that the trace in \((\FdHilb,\otimes,\Cset)\) is not well suited to capture the notion of quantum iteration.
Nevertheless, this compact category is of great importance in the study of categorical quantum mechanics~\cite{HeunenVicaryBook,TheDodo} and it is worth discussing it briefly.

A monoidal category \((\C,\otimes,I)\) is said to have (right) duals if for every object \(A \in \C\) there is an object \(A^* \in \C\) and morphisms \(\eta \colon I \to A^* \otimes A\) and \(\varepsilon \colon A \otimes A^* \to I\) such that the following diagrams commute:
\begin{equation} \label{eq:diag_CCC}
  \begin{aligned}
    \begin{tikzcd}
      A & {A \otimes I} & {A \otimes (A^* \otimes A)} \\
      A & {I \otimes A} & {(A \otimes A^*) \otimes A}
      \arrow["{\rho^{-1}}", from=1-1, to=1-2]
      \arrow["{\id \otimes \eta}", from=1-2, to=1-3]
      \arrow["{\alpha^{-1}}", from=1-3, to=2-3]
      \arrow["{\varepsilon \otimes \id}", from=2-3, to=2-2]
      \arrow["\lambda", from=2-2, to=2-1]
      \arrow["\id"', from=1-1, to=2-1]
    \end{tikzcd} 
    \\
    \\
    \begin{tikzcd}
      {A^*} & {I \otimes A^*} & {(A^* \otimes A) \otimes A^*} \\
      {A^*} & {A^* \otimes I} & {A^* \otimes (A \otimes A^*)}
      \arrow["\id"', from=1-1, to=2-1]
      \arrow["{\lambda^{-1}}", from=1-1, to=1-2]
      \arrow["{\eta \otimes \id}", from=1-2, to=1-3]
      \arrow["\alpha", from=1-3, to=2-3]
      \arrow["{\id \otimes \varepsilon}", from=2-3, to=2-2]
      \arrow["\rho", from=2-2, to=2-1]
    \end{tikzcd}
  \end{aligned}
\end{equation}
A symmetric monoidal category with duals which interact nicely with the monoidal structure is known as a compact closed category (see Chapter 3, Definition 3.34 from~\cite{HeunenVicaryBook} for a proper definition).
Every compact closed category is totally traced, with the trace of a morphism \(f \colon A \otimes U \to B \otimes C\) defined as follows:
\begin{equation} \label{eq:CCC_trace}
  \Tr^U(f) = (\id_B \otimes \varepsilon) \circ (f \otimes \id_{U^*}) \circ (\id_A \otimes \sigma \eta).
\end{equation}

In particular, \(\FdHilb\) is a compact closed category:
Section~\ref{sec:QM_cats} already established that \((\FdHilb,\otimes,\Cset)\) is symmetric monoidal, and it can be shown that the dual of a finite-dimensional Hilbert space \(H\) is the space defined on the hom-set \(H^* = \FdHilb(H,\Cset)\) using pointwise addition and scalar multiplication; since \(H^*\) is finite-dimensional it can be made into a Hilbert space by defining an appropriate inner product.
Given an orthonormal basis \(\{e_i \in H\}_{1\leq i \leq \dim(H)}\) of \(H\), an orthonormal basis of \(H^*\) is given by the collection \(\{\phi_i \colon H \to \Cset\}_{1\leq i \leq \dim(H)}\) where \(\phi_i\) is defined for each \(v \in H\) as follows:
\begin{equation*}
  \phi_i(v) = \braket{e_i}{v}.
\end{equation*}
The morphisms \(\eta \colon \Cset \to H^* \otimes H\) and \(\varepsilon \colon H \otimes H^* \to \Cset\) are defined by linear extension of:
\begin{equation*}
  \eta(1) = \sum_{1\leq i \leq \dim(H)} \phi_i \otimes e_i \quad\quad \text{and} \quad\quad
  \varepsilon(e_i \otimes \phi_j) = \begin{cases}
    1 &\ifc i = j \\
    0 &\otherwise
  \end{cases}
\end{equation*}
and it is straightforward to check that the diagrams~\eqref{eq:diag_CCC} commute.
Then, for any morphism \(f \in \FdHilb(A\otimes U, B \otimes U)\) the trace defined using the compact closed structure~\eqref{eq:CCC_trace} yields:
\begin{equation} \label{eq:FdHilb_CCC_trace}
  \Tr^U(f) = \sum_{1\leq i \leq \dim(U)} (\id_B \otimes \phi_i) \circ f \circ (\id_A \otimes e_i)
\end{equation}
which corresponds to the linear algebraic trace of a block matrix \(f = (f_{i,j})_{i,j \leq \dim(U)}\) determined by the orthonormal basis \(\{e_i\}_{i \leq \dim(U)}\) of \(U\), where each block is of type \(f_{i,j} \colon A \to B\).

Such a categorical trace on \(\FdHilb\) has little in common with the notion of iterative loops: Section~\ref{sec:ex_formula} identifies the requirement that morphisms \(f \colon A \otimes U \to B \otimes U\) must be characterisable in terms of components \(f_\sub{BA}\), \(f_\sub{BU}\), \(f_\sub{UA}\) and \(f_\sub{UU}\), which is certainly not the case for the tensor product \(\otimes\).
Importantly, the fact that \((\FdHilb,\otimes,\Cset)\) cannot be given a categorical trace that captures iteration is only due to its monoidal structure and has nothing to do with its compact closedness.
In fact, there are compact closed categories whose canonical trace coincides with the execution formula; for instance, this is the case on any category obtained as the result of applying the \(\Int\) construction (see~\cite{TraceJoyal}) to a category traced with respect to the execution formula.

\section{Unique decomposition categories}
\label{sec:UDC}

Unique decomposition categories were introduced by Haghverdi~\cite{Haghverdi} as a general framework where a morphism \(f \colon A \oplus U \to B \oplus U\) may be decomposed as a matrix
\begin{equation*}
  f = \begin{pmatrix}
    f_\sub{BA} \colon A \to B \quad & f_\sub{BU} \colon U \to B \\
    f_\sub{UA} \colon A \to U \quad & f_\sub{UU} \colon U \to U \\
  \end{pmatrix}
\end{equation*}
and where a notion of infinite summation of morphisms exists, so that the execution formula may be defined:
\begin{equation} \label{eq:ex_formula_informal}
  \ex^U(f) = f_\sub{BA} + \sum^\infty_{n = 0} f_\sub{BU} f_\sub{UU}^n f_\sub{UA}.
\end{equation}
The execution formula is meant to behave as a categorical trace, and the intuition is that it aggregates all possible paths that go from \(A\) to \(B\).
However, as discussed by Hoshino in Appendix B from~\cite{RTUDC}, Haghverdi's original definition is too weak: it admits categories whose execution formula violates vanishing I.
Hoshino's proposal was to define a stronger version of unique decomposition categories; their definition is provided below.

\begin{definition} \label{def:sUDC}
  Let \(\SCat{*}\) be one of the categories of \(\Sigma\)-monoids from Definition~\ref{def:SCat_subcats} endowed with a monoidal structure \((\SCat{*},\otimes,I)\) given by the tensor product as described in Section~\ref{sec:SCat_tensor}.
  For a \(\SCat{*}\)-category \(\C\) let there be:
  \begin{itemize}
    \item a zero object \(Z \in \C\) and
    \item a functor \(\oplus \colon \C \times \C \to \C\) such that its action on hom-objects on arbitrary \(A,B,C,D \in \C\) is described by a morphism in \(\SCat{*}\)
    \begin{equation*}
      - \oplus - \colon \C(A,C) \times \C(B,D) \to \C(A\oplus B, C \oplus D)
    \end{equation*}
    whose domain is a categorical product in \(\SCat{*}\).
  \end{itemize}
  If \((\C,\oplus,Z)\) is a symmetric monoidal category we say that it is a \(\SCat{*}\)-enriched \emph{unique decomposition category} (\GLS{\SCat{*}-UDC}{SCat-UDC}).
\end{definition}

Slight changes have been made to Hoshino's original definition~\cite{RTUDC}. On one hand, Hoshino only discusses \(\SCat{s}\)-enriched UDCs; the definition has been generalised so that all flavours of \(\Sigma\)-monoids introduced in Chapter~\ref{chap:Sigma} may be used.
On the other hand, Hoshino asks that \(\id_Z\) is the neutral element of the \(\Sigma\)-monoid \(\C(Z,Z)\), instead of requiring that the monoidal unit is a zero object.
These conditions are equivalent: if \(\id_Z = \Sigma \varnothing\) any morphism \(f \colon Z \to A\) is equal to the neutral element:
\begin{equation*}
  f = f \circ \id_Z = f \circ (\Sigma \varnothing) = \Sigma \varnothing
\end{equation*}
due to composition being \(\Sigma\)-bilinear; thus, \(Z\) is initial and a similar argument shows that \(Z\) is also terminal.
Notice that the zero morphism of each hom-object \(\C(A,B)\) --- \ie{} the unique morphism that factors through \(Z\) --- is precisely the neutral element of the \(\Sigma\)-monoid \(\C(A,B)\).
Moreover, if \(Z\) is a zero object then \(\C(Z,Z)\) is a singleton hom-set, so \(\id_Z\) must be the neutral element of the \(\Sigma\)-monoid \(\C(Z,Z)\).
Furthermore, Hoshino requires that
\begin{equation*}
  \Sigma \{\id_A \oplus 0_\sub{B,B}, 0_\sub{A,A} \oplus \id_B\} \keq \id_{A \oplus B}
\end{equation*}
but this can be inferred from our definition:
\begin{align*}
  \Sigma \{\id_A \oplus 0_\sub{B,B}, 0_\sub{A,A} \oplus \id_B\} &\keq \Sigma \{\id_A, 0_\sub{A,A}\} \oplus \Sigma \{0_\sub{B,B}, \id_B\} &&\text{(product in \(\SCat{*}\))} \\
  &\keq \id_A \oplus \id_B &&\text{(\(\Sigma\)-monoid axioms)} \\
  &= \id_{A \oplus B}. &&\text{(\(\oplus\) functor)}
\end{align*}
Alternatively, it can be shown that this equality imposes that the functor \(\oplus\) acts on hom-objects as a \(\Sigma\)-homomorphism whose domain is a categorical product and, hence, our definition is equivalent to Hoshino's in the case of \(\SCat{s}\)-UDCs.

Every monoidal category whose unit is a zero object has certain morphisms that act in a similar way to projections and injections, even though they may lack their universal property.
These morphisms are essential in the discussion of UDCs and we define them below.

\begin{definition}
  For each finite collection of objects \(\{A_j \in \C\}_J\) and each \(i \in J\), define the \emph{canonical quasi-projection} \(\pi_i \colon \oplus_J A_j \to A_i\) to be the unique morphism of its type generated by \(\{\id,\alpha,\alpha^{-1},\lambda,\rho\}\) together with the zero morphism \(0 \colon A_j \to Z\) for each \(j \in J\).
  Similarly, for each \(i \in J\) define the \emph{canonical quasi-injection} \(\pi_i \colon A_i \to \oplus_J A_j\) to be the unique morphism generated by \(\{\id,\alpha,\alpha^{-1},\lambda^{-1},\rho^{-1}\}\) along with the zero morphism \(0 \colon Z \to A_j\) for each \(j \in J\).
\end{definition}

The equations Haghverdi required in the definition of a UDC~\cite{Haghverdi} are satisfied in the case of canonical quasi-injections and canonical quasi-projections, as shown in the following remark.

\begin{remark} \label{rmk:sUDC_unitors}
  It follows from the definition of the canonical quasi-projection that \(\pi_A \colon A \oplus B \to A\) is \(\pi_A = \rho (\id_A \oplus 0)\) and the canonical quasi-injection \(\iota_A \colon A \to A \oplus B\) is \(\iota_A = (\id_A \oplus 0) \rho^{-1}\); similarly for \(\pi_B\) and \(\iota_B\) using the left unitor \(\lambda\) instead.
  It is immediate that \(\pi_A \iota_A = \id_A\) whereas \(\pi_A \iota_B = 0\).
  The fact that \(\Sigma \{\id_A \oplus 0, 0 \oplus \id_B\} \keq \id_{A \oplus B}\) implies that:
  \begin{equation*}
    \Sigma \{\iota_A \pi_A, \iota_B \pi_B\} \keq \id_{A \oplus B}.
  \end{equation*}
\end{remark}

This suggests that the monoidal structure on \(\C\) is very similar to the canonical one on a category with biproducts; indeed, every additive category (see Definition~\ref{def:additive_cat}) is a \(\SCat{g}\)-UDC, as established below.

\begin{example}
  Every additive category \((\C,\oplus,0)\) is a \(\SCat{g}\)-UDC.
  To check this, recall that an additive category is enriched over \(\Ab\) and every abelian group is a \(\Sigma\)-group whose \(\Sigma\) function on infinite families is undefined.
  Moreover, the monoidal unit of an additive category is a zero object according to the canonical definition of its monoidal structure, and the requirement that \(\Sigma \{\id_A \oplus 0, 0 \oplus \id_B\} = \id_{A \oplus B}\) is immediate from the universal property of biproducts.
\end{example}

\begin{example}
  Both \(\Hilb\) and \(\FdHilb\) are additive categories and, hence, they are \(\SCat{g}\)-UDCs.
  However, as discussed above, such a trivially defined \(\SCat{g}\)-enrichment assumes no infinite family is summable.
  In contrast, Section~\ref{sec:quantum_ex_trace} will provide a Hausdorff topology on each hom-set of \(\Hilb\), thus realising each of them as a \(\Sigma\)-group via the functor \(\HAG \to \SCat{g}\) from Definition~\ref{def:HausCMon_ft} which defines summability of infinite families in terms of convergence.
  Since the definition of UDC (Definition~\ref{def:sUDC}) does not impose any extra condition on the summability of infinite families, it is immediate that \(\Hilb\) and \(\FdHilb\) are still \(\SCat{g}\)-UDCs.
  Similarly, \(\Contraction\) and \(\FdContraction\) are \(\SCat{w}\)-UDCs; the \(\Sigma\) function on each hom-set of \(\Contraction\) is induced from the inclusion \(\Contraction(A,B) \into \Hilb(A,B)\) using the construction from Lemma~\ref{lem:wSm_restriction}.
  On the other hand, \(\Isometry\) and \(\Unitary\), along with their subcategories on finite dimensions are not UDCs since they lack a terminal object (and hence, a zero object).
\end{example}

Notice that a \(\SCat{*}\)-UDC need not have biproducts; for instance, this is the case for \(\Contraction\) since the universal morphism for the diagram:
\[\begin{tikzcd}
  & A \\
  A & {A\oplus A} & A
  \arrow["{\pi_A}", from=2-2, to=2-1]
  \arrow["{\pi_B}"', from=2-2, to=2-3]
  \arrow["\id"', from=1-2, to=2-1]
  \arrow["\id", from=1-2, to=2-3]
  \arrow[dashed, from=1-2, to=2-2]
\end{tikzcd}\]
would need to be the diagonal map \(a \mapsto (a,a)\) for all \(a \in A\), which is evidently not a contraction.
However, whenever the corresponding universal morphism does exist, it is unique, as established by the following proposition.
This is perhaps the most important characteristic of unique decomposition categories, justifying their name.

\begin{proposition}[Haghverdi~\cite{Haghverdi}] \label{prop:unique_decomposition}
  Let \((\C,\oplus,Z)\) be a \(\SCat{*}\)-UDC. Let \(f \colon \oplus_I A_i \to \oplus_J B_j\) be a morphism in \(\C\), with finite \(I\) and \(J\). There is a unique summable family \(\{f_{ji} \in \C(A_i,B_j)\}_{J \times I}\) satisfying:
  \begin{equation*}
    \Sigma \{\iota_j \circ f_{ji} \circ \pi_i\}_{J \times I} \keq f.
  \end{equation*}
\end{proposition} \begin{proof}
  It straightforward to check that the following equations are satisfied for any finite collection of objects \(\{A_i \in \C\}_I\):
  \begin{equation} \label{eq:UDC}
    \pi_i \circ \iota_{i'} = \begin{cases}
      \id &\text{if \(i = i'\)} \\
      0 &\otherwise
    \end{cases}
    \quad\quad \text{and} \quad\quad
    \Sigma \{\iota_i \circ \pi_i\}_I \keq \id.
  \end{equation}
  The first equation follows trivially from the definition of the canonical quasi-injections and quasi-projections, the second one follows from the axiom \(\Sigma \{\id_A \oplus 0, 0 \oplus \id_B\} \keq \id_{A \oplus B}\) of \(\SCat{*}\)-UDCs and the flattening axiom of weak \(\Sigma\)-monoids.

  Let \(f \colon \oplus_I A_i \to \oplus_J B_j\) be a morphism in \(\C\), with finite \(I\) and \(J\). Let \(f_{ji} = \pi_j f \iota_i\); then:
  \begin{align*}
    f = \id \circ f \circ \id &= \Sigma \{\iota_j \, \pi_j\}_J \circ f \circ \Sigma \{\iota_i \, \pi_i\}_I && \text{(equation~\eqref{eq:UDC})} \\
      &\keq \Sigma \{\iota_j \, \pi_j \, f \circ \Sigma \{\iota_i \, \pi_i\}_I\}_J && \text{(\(\circ\) is \(\Sigma\)-bilinear)} \\
      &\keq \Sigma \{\Sigma \{\iota_j \, \pi_j \, f \, \iota_i \, \pi_i\}_I\}_J && \text{(\(\circ\) is \(\Sigma\)-bilinear)} \\
      &\keq \Sigma \{\iota_j \, f_{ji} \, \pi_i\}_{J \times I} && \text{(flattening axiom and def. of \(f_{ji}\)).}
  \end{align*}
  Thus, there is at least one family \(\{f_{ji}\}_{J \times I}\) satisfying the claim.
  To show uniqueness, assume there is another family \(\{g_{j'i'}\}_{J \times I}\) satisfying \(\Sigma \{\iota_{j'} \circ g_{j'i'} \circ \pi_{i'}\}_{J \times I} \keq f\); then:
  \begin{equation*}
    f_{ji} = \pi_j \, f \, \iota_i = \pi_j (\Sigma \{\iota_{j'} \, g_{j'i'} \, \pi_{i'}\}_{J \times I}) \iota_i \keq \Sigma \{\pi_j \, \iota_{j'} \, g_{j'i'} \, \pi_{i'} \, \iota_i\}_{J \times I} = \Sigma \{g_{ji}\} = g_{ji}
  \end{equation*}
  due to composition being \(\Sigma\)-bilinear.
  Therefore, both families \(\{f_{ji}\}_{J \times I}\) and \(\{g_{j'i'}\}_{J \times I}\) contain the same elements and are thus the same family.
\end{proof}

The family of components has a suggestive representation as a matrix; for instance, for \(f \colon A \oplus B \to C \oplus D\) we may write:
  \begin{equation*}
    f = \begin{pmatrix}
      f_\sub{CA} & f_\sub{CB} \\
      f_\sub{DA} & f_\sub{DB} \\
    \end{pmatrix}
  \end{equation*}
The composition of two morphisms can be given in terms of their `matrix multiplication' where multiplication of entries is replaced by their composition, as established below.

\begin{proposition} \label{prop:matrix_composition}
  Let \((\C,\oplus,Z)\) be a \(\SCat{*}\)-UDC.
  Let \(f \colon \oplus_I A_i \to \oplus_J B_j\) and \(g \colon \oplus_J B_j \to \oplus_K C_k\) be two morphisms in \(\C\), with finite \(I\), \(J\) and \(K\).
  Then, for each \((k,i) \in K \times I\) the family \(\{g_\sub{kj} \circ f_\sub{ji}\}_J\) is summable and the unique decomposition of \(g \circ f\) is given by the collection of morphisms
  \begin{equation*}
    (g \circ f)_\sub{ki} = \Sigma \{g_\sub{kj} \circ f_\sub{ji}\}_J
  \end{equation*}
  for all \((k,i) \in K \times I\).
\end{proposition} \begin{proof}
  We show that the family \(\{g_\sub{kj} \circ f_\sub{ji}\}_J\) is summable for all \((k,i) \in K \times I\):
  \begin{align*}
    \pi_k (g \circ f) \iota_i &\keq \pi_k \, g \circ (\Sigma \{\iota_j \, \pi_j\}_J) \circ f \, \iota_i &&\text{(equation~\eqref{eq:UDC})} \\
      &\keq \Sigma \{\pi_k \, g \, \iota_j \, \pi_j \, f \, \iota_i\}_J &&\text{(\(\circ\) is \(\Sigma\)-bilinear)} \\
      &= \Sigma \{g_{kj} \circ f_{ji}\}_J. &&\text{(def\@. of \(g_{kj}\) and \(f_{ji}\))}
  \end{align*}
  Recall that \((g \circ f)_\sub{ki} = \pi_k (g \circ f) \iota_i\) by uniqueness of the decomposition, thus completing the proof of the claim.
\end{proof}

Haghverdi's original definition of UDCs simply required \(\C\) to be \(\SCat{s}\)-enriched, symmetric monoidal, and satisfy~\eqref{eq:UDC}. Unlike Hoshino, Haghverdi did not impose a particular choice of morphisms to act as quasi-injections and quasi-projections, which permitted the choice of ones that altered the input in a subtle way, causing vanishing I to fail for the execution formula (see Appendix B from~\cite{RTUDC} for further details).
In the \(\SCat{*}\)-UDCs of Hoshino, the symmetric monoidal structure satisfies some useful properties discussed in the following remarks.

\begin{remark} \label{rmk:sUDC_oplus}
  Let \(f \colon A \to B\) and \(g \colon C \to D\) be morphisms in a \(\SCat{*}\)-UDC \((\C,\oplus,Z)\).
  It follows that
  \begin{align*}
    \pi_B (f \oplus g) \iota_A &= \rho (\id_B \oplus 0) (f \oplus g) (\id_A \oplus 0) \rho^{-1} && \text{(def. of \(\pi_B\) and \(\iota_A\))} \\
      &= \rho (f \oplus 0) \rho^{-1} && \text{(functoriality of \(\oplus\), zero morphism)} \\
      &= f \circ \rho (\id_A \oplus \id_Z)\rho^{-1} = f && \text{(nat. of \(\rho\), \(\id_Z = 0\))}
  \end{align*}
  and similarly \(\pi_D (f \oplus g) \iota_C = g\), whereas \(\pi_B (f \oplus g) \iota_C = 0\) and \(\pi_D (f \oplus g) \iota_A = 0\).
  Therefore, the unique decomposition of \(f \oplus g\) is:
  \begin{equation*}
    f \oplus g = \begin{pmatrix}
      f & 0 \\
      0 & g
    \end{pmatrix}.
  \end{equation*}
\end{remark}

\begin{remark} \label{rmk:sUDC_symmetry}
  Let \(\sigma \colon A \oplus B \to B \oplus A\) be the symmetric braiding of a \(\SCat{*}\)-UDC \((\C,\oplus,Z)\).
  It follows that
  \begin{align*}
    \pi_A \, \sigma \, \iota_A &= \lambda (0 \oplus \id_A) \sigma (\id_A \oplus 0) \rho^{-1} && \text{(def.\@ of \(\pi_A\) and \(\iota_A\))} \\
      &= \lambda (0 \oplus \id_A) (0 \oplus \id_A) \sigma \, \rho^{-1} && \text{(naturality of \(\sigma\))} \\
      &= \lambda (0 \oplus \id_A) \sigma \, \rho^{-1} && \text{(functoriality of \(\oplus\))} \\
      &= \lambda \, \sigma \, \rho^{-1} && \text{(\(\id_Z = 0\), nat.\@ of \(\lambda\))} \\
      &= \id_A && \text{(coherence theorem)}
  \end{align*}
  with the last step due to the coherence theorem of symmetric monoidal categories, which establishes that if two morphisms have the same type and both are built from composition and monoidal product of \(\{\id,\alpha,\lambda,\rho,\sigma\}\) and their inverses, then the two morphisms are equivalent (see Corollary 1.42 from~\cite{HeunenVicaryBook}).
  Similarly, \(\pi_B \, \sigma \, \iota_B = \id_B\), whereas \(\pi_A \, \sigma \, \iota_B = 0\) and \(\pi_B \, \sigma \, \iota_A = 0\).
  Therefore, the unique decomposition of the symmetric braiding is:
  \begin{equation*}
    \sigma = \begin{pmatrix}
      0 & \id_B \\
      \id_A & 0
    \end{pmatrix}.
  \end{equation*}
\end{remark}

These remarks together with the monoidal unit being a zero object and composition being \(\Sigma\)-bilinear are \emph{almost} sufficient for \(\SCat{*}\)-UDCs to be partially traced categories with respect to the execution formula.
This is made precise by the following lemma.

\begin{lemma} \label{lem:all_but_vII}
  Let \((\C,\oplus,Z)\) be a \(\SCat{*}\)-UDC.
  Then, for any morphism \(f \in \C(A \oplus U, B \oplus U)\), the execution formula
  \begin{equation*}
    \ex^U(f) \keq \Sigma \left(\{f_\sub{BA}\} \uplus \{f_\sub{BU} f_\sub{UU}^n f_\sub{UA}\}_\Nset \right)
  \end{equation*}
  is guaranteed to satisfy all of the axioms of partially traced categories but vanishing II.
\end{lemma} \begin{proof}
  Each of the axioms of partially traced categories (except vanishing II) are proven below.
  \begin{itemize}
    \item \emph{Naturality}. For all \(f \colon A \oplus U \to B \oplus U\), if \(\ex^U(f)\) is defined then:
    \begin{align*}
      h \circ \ex^U(f) \circ g &\keq h \circ \left(\Sigma \left(\{f_\sub{BA}\} \uplus \{f_\sub{BU} f_\sub{UU}^n f_\sub{UA}\}_\Nset \right) \right) \circ g &&\text{(def\@. of \(\ex\))} \\
        &\keq \Sigma \left(\{h f_\sub{BA} g\} \uplus \{h f_\sub{BU} f_\sub{UU}^n f_\sub{UA} g\}_\Nset \right) &&\text{(\(\circ\) is \(\Sigma\)-bilinear)} \\
        &\keq \ex^U((h \oplus \id) \circ f \circ (g \oplus \id)) &&\text{(def\@. of \(\ex\))}
    \end{align*}
    so that naturality is satisfied.
    \item \emph{Dinaturality.} For all \(f \colon A \oplus U \to B \oplus U'\) and \(g \colon U' \to U\), if \(\ex^U((\id \oplus g) f)\) is defined then:
    \begin{align*}
      \ex^U((\id \oplus g) \circ f) &\keq \Sigma \left( \{f_\sub{BA}\} \uplus \{f_\sub{BU} (g f_\sub{U'U})^n (g f_\sub{U'A})\}_\Nset \right) &&\text{(def\@. of \(\ex\))} \\
        &= \Sigma \left( \{f_\sub{BA}\} \uplus \{(f_\sub{BU} g) (f_\sub{U'U} g)^n f_\sub{U'A}\}_\Nset \right) &&\text{(associativity of \(\circ\))} \\
        &\keq \ex^{U'}(f \circ (\id \oplus g)). &&\text{(def\@. of \(\ex\))}
    \end{align*}
    Similarly, if \(\ex^{U'}(f \circ (\id \oplus g))\) is defined then \(\ex^U((\id \oplus g) \circ f)\) is defined as well and they coincide, so that dinaturality is satisfied.
    \item \emph{Superposing.} For all \(f \colon A \oplus U \to B \oplus U\) and \(g \colon C \to D\), if \(\ex^U(f)\) is defined then, according to the matrix decomposition of \(g \oplus f\) (see Remark~\ref{rmk:sUDC_oplus}):
    \begin{equation*}
      \ex^U(g \oplus f) = \ex^U \begin{pmatrix}
        g & 0 \\
        0 & f
      \end{pmatrix} \keq
      \begin{pmatrix}
        \Sigma \left( \{g\} \uplus \{0 f_\sub{UU}^n 0\}_\Nset \right) & 0 \\
        0 & \ex^U(f)
      \end{pmatrix} \keq g \oplus \ex^U(f)
    \end{equation*}
    where the neutral element and singleton axioms of weak \(\Sigma\)-monoids have been used.
    Thus, superposing is satisfied.
    \item \emph{Yanking.} The axiom follows trivially from the matrix decomposition of the symmetry \(\sigma \colon U \oplus U \to U \oplus U\) (see Remark~\ref{rmk:sUDC_symmetry}):
    \begin{equation*}
      \ex^U(\sigma) = \ex^U \begin{pmatrix}
        0 & \id_U \\
        \id_U & 0
      \end{pmatrix} \keq \Sigma \left( \{0\} \uplus \{\id_U \, 0^n \, \id_U\}_\Nset \right) \keq \id_U
    \end{equation*}
    where the neutral element and singleton axioms of weak \(\Sigma\)-monoids have been used.
    Thus, yanking is satisfied.
    \item \emph{Vanishing I.} For all \(f \colon A \oplus Z \to B \oplus Z\) notice that \(f_\sub{ZA} = 0\), \(f_\sub{BZ} = 0\) and \(f_\sub{ZZ} = 0\) due to \(Z\) being a zero object, implying:
    \begin{equation*}
      \ex^Z(f) \keq f_\sub{BA} = \pi_B \, f \, \iota_A = \rho (\id_B \oplus 0) f (\id_A \oplus 0) \rho^{-1} = \rho \, f \, \rho^{-1}
    \end{equation*}
    where the neutral element and singleton axioms of weak \(\Sigma\)-monoids have been used.
    Thus, vanishing I is satisfied. \qedhere
  \end{itemize}
\end{proof}

The proof above uses that composition is \(\Sigma\)-bilinear along with the neutral element and singleton axioms of weak \(\Sigma\)-monoids. 
However, neither the bracketing nor the flattening axioms of weak \(\Sigma\)-monoids are used.
The theorem below establishes that, in the case of \(\SCat{s}\)-UDCs, the vanishing II axiom is also satisfied.
The proof relies on the strong bracketing and the strong flattening axioms from the \(\SCat{s}\)-enrichment.
The same result was proven by Hoshino in~\cite{RTUDC} using a different strategy: Hoshino provides a representation theorem for \(\SCat{s}\)-UDCs, establishing that every such category embeds in a categories with countable biproducts, then derives a partial trace in the original category using a similar approach to Proposition~\ref{prop:induced_trace}.
In contrast, our proof illustrates the importance of strong bracketing and strong flattening explicitly and, hence, hints at the challenge of generalising this result to \(\SCat{w}\)-UDCs where these axioms are not available.
For further discussion on the connection between our work and Hoshino's, see Section~\ref{sec:trace_rel_work}.

\begin{theorem} \label{thm:classical_trace}
  Every \(\SCat{s}\)-UDC \((\C,\oplus,Z)\) is partially traced using the execution formula.
\end{theorem} \begin{proof}
  According to the previous lemma, all axioms except vanishing II are guaranteed to be satisfied.
  Thus, we only need to show that, whenever the enrichment is over \(\SCat{s}\), vanishing II is satisfied as well.
  Let \(f \colon A \oplus U \oplus V \to B \oplus U \oplus V\) be a morphism in \(\C\) such that \(\ex^V(f)\) is defined.
  For each \(n \in \Nset\), let \(\fml{d}_n \in \C(A,B)^*\) be the family of morphisms \(A \to B\) built from chains of compositions of the form:
  \begin{equation*}
    f_\sub{BW_n} \circ \ldots \circ f_\sub{W_2W_1} \circ f_\sub{W_1A}
  \end{equation*}
  where each \(W_i\) may be either \(U\) or \(V\); in the case of \(n = 0\), define \(\fml{d}_0 = \{f_\sub{BA}\}\).
  Let \(\fml{p} = \uplus_{n \in \Nset} \fml{d}_n\) and, for each \(n \in \Nset\), let \(\fml{q}_n\) be the subfamily of \(\fml{p}\) comprised of all morphisms built from chains that factor through \(U\) exactly \(n\) times; evidently, \(\fml{p} = \uplus_{n \in \Nset} \fml{q}_n\).
  Notice that for \(n > 0\) the sum of \(\fml{d}_n\) would correspond to:
  \begin{equation*}
    \Sigma \fml{d}_n \keq \begin{pmatrix}
      f_\sub{BU} & f_\sub{BV}
    \end{pmatrix} \begin{pmatrix}
      f_\sub{UU} & f_\sub{UV} \\
      f_\sub{VU} & f_\sub{VV} \\
    \end{pmatrix}^{n-1} \begin{pmatrix}
      f_\sub{UA} \\ f_\sub{VA}
    \end{pmatrix}
  \end{equation*}
  and thus \(\fml{d}_n\) is summable according to Proposition~\ref{prop:matrix_composition}.
  Since \(\Sigma \fml{d}_0 = f_\sub{BA}\), the definition of \(\ex^{U \oplus V}\) implies that:
  \begin{equation} \label{eq:vII_d}
    \ex^{U \oplus V}(f) \keq \Sigma \{\Sigma \fml{d}_n\}_\Nset.
  \end{equation}
  Similarly, \(\Sigma \fml{q}_0 \keq \pi_B \circ \ex^V(f) \circ \iota_A\) due to composition being \(\Sigma\)-bilinear and, for each \(n > 0\),
  \begin{equation*}
    \Sigma \fml{q}_n \keq (\pi_B \, \ex^V(f) \, \iota_U) \circ (\pi_U \, \ex^V(f) \, \iota_U)^{n-1} \circ (\pi_U \, \ex^V(f) \, \iota_A)
  \end{equation*}
  where each family \(\fml{q}_n\) is summable due to \(\ex^V(f)\) being defined (by assumption) and composition being \(\Sigma\)-bilinear.
  Then, according to the definition of \(\ex^U\):
  \begin{equation} \label{eq:vII_q}
    \ex^U(\ex^V(f)) \keq \Sigma \{\Sigma \fml{q}_n\}_\Nset.
  \end{equation}

  In summary, we have provided two partitions of the family \(\fml{p}\), both satisfying that for all \(n \in \Nset\) each subfamily \(\fml{d}_n\) and \(\fml{q}_n\) is summable.
  Moreover, we have shown that the family of their sums is defined iff either side of the vanishing II axiom is satisfied; it is now straightforward to prove vanishing II using the strong flattening and strong bracketing axioms from the \(\SCat{s}\)-enrichment.
  Assume \(\ex^U(\ex^V(f))\) is defined, then it follows that:
  \begin{align*}
    \ex^U(\ex^V(f)) &\keq \Sigma \{\Sigma \fml{q}_n\}_\Nset &&\text{(equation~\eqref{eq:vII_q})} \\
      &\keq \Sigma \fml{p} &&\text{(strong flattening)} \\
      &\keq \Sigma \{\Sigma \fml{d}_n\}_\Nset &&\text{(strong bracketing)} \\
      &\keq \ex^{U \oplus V}(f). &&\text{(equation~\eqref{eq:vII_d})}
  \end{align*}
  Similarly, if we assume that \(\ex^{U \oplus V}(f)\) is defined it follows that \(\ex^U(\ex^V(f))\) is defined as well and their results coincide.
  Therefore, vanishing II is satisfied and the proof is complete.
\end{proof}

An immediate question is whether this theorem can be generalised to \(\SCat{w}\)-UDCs.
Unfortunately, these are generally not partially traced since vanishing II cannot be guaranteed to hold, as shown in the following counterexample.

\begin{example} \label{ex:FdHilb_counterexample}
  Let the \(\SCat{g}\)-enrichment of \(\FdHilb\) be induced by the usual operator norm topology (see Proposition~\ref{prop:SOT_FdHilb} for details) and consider the following morphism in \(\FdHilb\):
  \begin{equation*}
    f \colon \Cset \oplus \Cset \oplus \Cset \to \Cset \oplus \Cset \oplus \Cset
  \end{equation*}
  \begin{equation*}
    f = \left( \begin{array}{ccc}
      0 & 1 & 1 \\
      1 & -\tfrac{2}{3} & 1 \\
      1 & 1 & \tfrac{1}{3}
    \end{array} \right).
  \end{equation*}
  Explicit calculation yields the following result:
  \begin{equation*}
    \ex^\Cset(f) \, = \, \ex^\Cset \left( \begin{array}{cc|c}
      0 & 1 & 1 \\
      1 & -\tfrac{2}{3} & 1 \\ \hline
      1 & 1 & \tfrac{1}{3}
    \end{array} \right) \, = \, 
    \begin{pmatrix}
      \tfrac{3}{2} & \tfrac{5}{2} \\
      \tfrac{5}{2} & \tfrac{5}{6}
    \end{pmatrix}
  \end{equation*}
  so that \(\abs{\tfrac{5}{6}} \leq 1\) and, hence, \(\ex^\Cset(\ex^\Cset(f))\) is well-defined, whereas
  \begin{equation*}
    \ex^{\Cset \oplus \Cset}(f) \, =\,  \ex^{\Cset \oplus \Cset} \left( \begin{array}{c|cc}
      0 & 1 & 1 \\ \hline
      1 & -\tfrac{2}{3} & 1 \\ 
      1 & 1 & \tfrac{1}{3}
    \end{array} \right) \, =\, 
    \sum_{n = 0}^\infty \begin{pmatrix}
      1 & 1
    \end{pmatrix} \begin{pmatrix}
      -\tfrac{2}{3} & 1 \\
      1 & \tfrac{1}{3}
    \end{pmatrix}^n \begin{pmatrix}
      1 \\ 1
    \end{pmatrix}
  \end{equation*}
  does not converge in norm and, hence, is undefined.
  Therefore, this morphism is an instance of vanishing II failing in a \(\SCat{g}\)-UDC.
\end{example} 

The construction of such a counterexample is enabled by the presence of additive inverses; in particular,  entry \(- \tfrac{2}{3}\) in the definition of \(f\).
The goal of the rest of this section is to establish sufficient conditions for vanishing II to be satisfied in \(\SCat{w}\)-UDCs.
But first, we discuss the relevance of the result for \(\SCat{s}\)-UDCs (Theorem~\ref{thm:classical_trace}) in the context of classical iterative loops.

\subsection{Classical iterative loops}
\label{sec:classical_loop}

This section provides three examples of \(\SCat{s}\)-UDCs which, according to Theorem~\ref{thm:classical_trace}, are partially traced with respect to the execution formula.
The relevance of each of these categories in classical and quantum computing is briefly discussed and, in particular, it is argued that their execution formula captures the semantics of classical iterative loops in their respective frameworks. 

\paragraph*{Classical reversible programs.} A bijection \(f \colon A \to B\) is interpreted to describe the input-output behaviour of a reversible program whose possible inputs and outputs are elements of the sets \(A\) and \(B\), respectively.

\begin{definition}
  Let \(\Bijection\) be the category whose objects are sets and whose morphisms are bijections.
  For any partial function \(f \colon A \pto B\) let \(\dom(f) \subseteq A\) be the subset where \(f\) is defined.
  Let \(\PInj\) be the category whose objects are sets and whose morphisms \(f \in \PInj(A,B)\) are partial injective functions, \ie{} partial functions whose restriction \(\dom(f) \to B\) is injective.
\end{definition}

\begin{proposition}[Haghverdi~\cite{Haghverdi}]
  Let \((\PInj,\uplus,\varnothing)\) be the symmetric monoidal category induced by disjoint union.
  A family of morphisms \(\{f_i\}_I \in \PInj(A,B)^*\) is summable iff \(\dom(f_i) \cap \dom(f_{i'}) = \varnothing\) and \(\im(f_i) \cap \im(f_{i'}) = \varnothing\) for all \(i \not= i'\); in such a case, the partial function \(\Sigma \{f_i\}_I \colon A \pto B\) is defined as follows for every \(a \in A\):
  \begin{equation*}
    (\Sigma \{f_i\}_I)(a) = \begin{cases}
      f_i(a) &\ifc \exists i \in I \text{ s.t. } a \in \dom(f_i) \\
      \undefined &\otherwise.
    \end{cases}
  \end{equation*}
  These definitions make \(\PInj\) a \(\SCat{s}\)-UDC.
\end{proposition} \begin{proof}[Proof. (Sketch)]
  We first check that each hom-set \(\PInj(A,B)\) is a strong \(\Sigma\)-monoid using the definition of \(\Sigma\) given above. 
  The neutral element of the \(\Sigma\)-monoid is the partial function that is undefined on every \(a \in A\) and the singleton axiom holds trivially.
  Strong flattening holds as well: let \(\{\fml{f}_j \in \PInj(A,B)^*\}_J\) be an indexed set of summable families and let \(\fml{f} = \uplus_J \fml{f}_j\); assume that the family \(\{\Sigma \fml{f}_j\}_J\) is summable, this implies that \(\dom(\Sigma \fml{f}_j) \cap \dom(\Sigma \fml{f}_{j'})\) for every \(j \not= j'\) and hence, necessarily, the domains of the elements of \(\fml{f}_j \uplus \fml{f}_{j'}\) are all pair-wise disjoint.
  Since this applies to every pair of indices \(j \not= j'\) in \(J\) and a similar argument holds for the images, it follows that \(\fml{f} = \uplus_J \fml{f}_j\) is summable and it is trivial to check that \(\Sigma \fml{f} = \Sigma \{\Sigma \fml{f}_j\}_J\) so that strong flattening is satisfied.
  The proof of the strong bracketing axiom is dual to this argument.
  Finally, subsummability holds trivially since the summability condition is imposed on every pair of elements from the family.
  
  To show that \(\PInj\) is a \(\SCat{s}\)-enriched category we must check that composition is \(\Sigma\)-bilinear. 
  Assume \(\{f_i\}_I \in \PInj(A,B)^*\) is summable and let \(g \in \PInj(B,C)\); then \(\{g \circ f_i\}_I\) trivially satisfies the condition on disjoint domains.
  That \(\{g \circ f_i\}_I\) also satisfies the condition on disjoint images can be shown by contradiction:
  suppose there is an element \(c \in \im(g \circ f_i) \cap \im(g \circ f_{i'})\) for \(i \not= i'\) then there must be elements \(b \in \im(f_i)\) and \(b' \in \im(f_{i'})\) such that \(g(b) \keq c \keq g(b')\), but \(b \not= b'\) is necessary for \(\{f_i\}_I\) to be summable, so this contradicts \(g\) being injective.
  A similar argument can be used to prove that composition is a \(\Sigma\)-homomorphism on the left as well, making it \(\Sigma\)-bilinear.

  It is straightforward to check that \((\PInj,\uplus,\varnothing)\) is a symmetric monoidal category and its monoidal unit \(\varnothing\) is a zero object.
  Moreover, \(\Sigma \{\id_A \uplus 0, 0 \uplus \id_B\} \keq \id_{A \uplus B}\) holds trivially.
  Thus, \(\PInj\) is a \(\SCat{s}\)-UDC, as claimed.
\end{proof} 

Notice that \(\PInj\) does not have coproducts since the codiagonal map \(\nabla \colon A \uplus A \to A\) cannot be injective because it would need to map both \(\iota_l(a),\iota_r(a) \in A \uplus A\) to \(a \in A\).
And, evidently, \(\uplus\) is not a categorical product; nevertheless, morphisms \(f \colon A \uplus B \to C \uplus D\) in \(\PInj\) may uniquely be characterised by their matrix decomposition thanks to Proposition~\ref{prop:unique_decomposition}.
According to Theorem~\ref{thm:classical_trace}, the category \((\PInj,\uplus,\ex)\) is a partially traced and, furthermore, we can show that it is totally traced.
For an arbitrary morphism \(f \in \PInj(A \uplus U, B \uplus U)\), let \(p_0 = f_\sub{BA}\) and for all \(n \in \Nset\) let \(p_{n+1} = f_\sub{BU} f_\sub{UU}^n f_\sub{UA}\).
It can be shown by induction that no element \(b \in B\) can be both in \(\im(p_n)\) and \(\im(p_m)\) for \(n \not= m\) since this would contradict \(f\) being injective; similarly, no \(a \in A\) can be both in \(\dom(p_n)\) and \(\dom(p_m)\).
Consequently, the family \(\{p_n\}_\Nset\) is always summable and \(\PInj\) is in fact totally traced (see~ Example 2.23 (i) from \cite{Haghverdi} for more details).

Even though \(\Bijection\) itself is not a UDC --- quasi-projections are necessarily partial --- it is a subcategory of \(\PInj\); therefore, every bijection \(f \colon A \uplus B \to C \uplus D\) may be decomposed into a matrix form whose entries are in \(\PInj\).
We may transport the trace of \(\PInj\) to \((\Bijection,\uplus,\widehat{\ex})\) using Proposition~\ref{prop:induced_trace} which defines:
\begin{equation*}
  \widehat{\ex}^U(f) = \begin{cases}
    g &\ifc \ex^{U}(f) \text{ is a bijection} \\
    \undefined &\otherwise.
  \end{cases}
\end{equation*}
If \(f\) denotes a reversible program then the bijection \(\widehat{\ex}^U(f) \colon A \to B\) denotes a program that repeatedly applies \(f\) whenever the output is in \(U\) --- an iterative loop.
Notice that \(\Bijection\) is only partially traced; for instance, if \(U = \Nset\) we may define a bijection so that \(f_\sub{UU}(n) = n+1\) for all \(n \in \Nset\) and, for some \(a \in A\), \(f_\sub{UA}(a) = 0\) so that \(a \not\in \dom(\ex^U(f))\) and \(\ex^U(f)\) is only a partial injection.
The latter case would correspond to an iterative loop that never halts if the input is \(a\).

\paragraph*{Classical probabilistic programs.} On each input \(a \in A\), a probabilistic program may yield an output \(b \in B\) with probability \(p_a(b)\). These probabilities may be arranged in a \(\abs{B} \times \abs{A}\) matrix \(M\) with entries \(m_{ba} = p_a(b)\) so that each of its columns add up to one:
\begin{equation*}
  \sum_{b \in B} m_{ba} = \sum_{b \in B} p_a(b) = 1.
\end{equation*}
These are known as \emph{stochastic} matrices.
If the sum of each of their columns is only required to be bounded by one \(\sum_{b \in B} m_{ba} \leq 1\) we refer to them as \emph{substochastic} matrices.
Evidently, every stochastic matrix is substochastic.

\begin{definition}
  Let \(\SubStoch\) be the category whose objects are finite sets and whose morphisms \(f \in \SubStoch(A,B)\) are substochastic matrices \(\abs{B} \times \abs{A}\).
  Composition corresponds to matrix multiplication and identities correspond to the usual identity matrices.
  Let \(\Stoch\) be the subcategory of \(\SubStoch\) whose morphisms are restricted to be stochastic matrices.
\end{definition}

Similarly to the case with \(\PInj\) and \(\Bijection\); the category we are interested in is \(\Stoch\), but we instead work on the larger category \(\SubStoch\) that happens to be a \(\SCat{s}\)-UDC, as established below.

\begin{proposition}
  Let \((\SubStoch,\oplus,\varnothing)\) be the symmetric monoidal category where \(\oplus\) acts as disjoint union of sets and direct sum of matrices.
  A family of morphisms \(\{f_i\}_I \in \SubStoch(A,B)^*\) is summable iff the standard sum of matrices\footnote{An infinite family of matrices is summable iff it is entry-wise summable according to the notion of absolute convergence on real numbers (see Example~\ref{ex:wSm_Rset}).} yields a substochastic matrix, in which case \(\Sigma \{f_i\}_I\) is defined to be the corresponding result.
  These definitions make \(\SubStoch\) a \(\SCat{s}\)-UDC.
\end{proposition} \begin{proof}[Proof. (Sketch)]
  We first check that each hom-set \(\SubStoch(A,B)\) is a strong \(\Sigma\)-monoid using the definition of \(\Sigma\) given above. 
  Its neutral element is the matrix with all zero entries and the singleton axiom is trivially satisfied. Due to the lack of negative entries, both strong bracketing and strong flattening are immediately implied by the fact that a converging series of positive real numbers is convergent no matter the ordering of the sequence; subsummability is also straightforward to check.

  To show that \(\SubStoch\) is a \(\SCat{s}\)-enriched category we must check that composition is \(\Sigma\)-bilinear. 
  If \(\{f_i\}_I \in \SubStoch(A,B)^*\) is summable then \(g \circ \Sigma \{f_i\}_I\) is a substochastic matrix for all \(g \in \SubStoch(B,C)\).
  Given that matrix multiplication distributes over matrix addition, the family \(\{g \circ f_i\}_I\) is summable, with the result being the substochastic matrix \(g \circ \Sigma \{f_i\}_I\).
  A similar argument can be used to prove that composition is a \(\Sigma\)-homomorphism on the left as well, making it \(\Sigma\)-bilinear.

  It is straightforward to check that \((\SubStoch,\oplus,\varnothing)\) is a symmetric monoidal category and its monoidal unit \(\varnothing\) is a zero object.
  Moreover, \(\Sigma \{\id_A \oplus 0, 0 \oplus \id_B\} \keq \id_{A \oplus B}\) holds trivially.
  Thus, \(\SubStoch\) is a \(\SCat{s}\)-UDC, as claimed.
\end{proof}

In the case of \(\SubStoch\) objects \(A \oplus B\) are coproducts; however, \(A \oplus B\) is not a categorical product since the diagonal morphism \(\Delta \colon A \to A \oplus A\) would need to be \[\Delta = \vector{\id_A}{\id_A}\] which is not a substochastic matrix --- each of its columns add up to \(2\).

According to Theorem~\ref{thm:classical_trace}, the category \((\SubStoch,\oplus,\ex)\) is partially traced.
Moreover, it is straightforward to check that, whenever \(U\) is a singleton set, \(\ex^U(f)\) is well-defined (due to the convergence criteria of the geometric series); then, due to objects being finite sets, we may trace out each element of any arbitrary set \(U\) one by one and, after to a finite number of applications of vanishing II:
\begin{equation*}
  \ex^U(f) \keq \ex^{\{u\}}(\ex^{\{u'\}}(\ldots(f)))
\end{equation*}
so that we conclude \((\SubStoch,\oplus,\ex)\) is totally traced.
Using Proposition~\ref{prop:induced_trace} we can induce a partial trace on \(\Stoch\) using \(\ex\) from \(\SubStoch\).
Notice that the resulting partially traced category \((\Stoch,\oplus,\widehat{\ex})\) is not total since
\begin{equation*}
  \widehat{\ex}^u \begin{pmatrix}
    0 & 0 \\ 1 & 1
  \end{pmatrix} = 0 + \sum^\infty_{n = 0} 0 \cdot (1)^n \cdot 1 = 0.
\end{equation*}

\paragraph*{Quantum programs with classical control flow.} The category \(\CPTP\) was introduced in Section~\ref{sec:CPTR} to capture both coherent quantum operations and non-coherent ones (\eg{} measurements) in a single framework.
Along it, \(\CPTR\) was introduced to take up a similar role \(\SubStoch\) plays with respect to \(\Stoch\).
The objects in \(\CPTR\) are finite-dimensional \Cstar-algebras whose density operators represent \emph{mixed} quantum states: probability distributions of \emph{pure} quantum states.
It was discussed in Section~\ref{sec:CPTR} that for any two \Cstar-algebras \(A\) and \(B\), if \((\rho,\rho') \in A \oplus B\) is a density operator, then \((\rho,\rho')\) may be interpreted as probability distributions comprised of outcome \(\tfrac{\rho}{\tr(\rho)}\) with probability \(\tr(\rho)\) and outcome \(\tfrac{\rho'}{\tr(\rho')}\) with probability \(\tr(\rho')\).
Consequently, a morphism \(f \in \CPTR(A, B \oplus C)\) may be understood to apply either \(f_\sub{BA} \colon A \to B\) or \(f_\sub{CA} \colon A \to C\) with some probability.
Intuitively, it follows that the execution formula in \(\CPTR\) will correspond to the sum of probabilistic paths and, to a certain extent, it will be similar to the case of \(\SubStoch\).

\begin{proposition}
  A family of CPTR maps \(\{f_i\}_I \in \CPTR(A,B)^*\) is summable iff its pointwise sum\footnote{A definition in terms of pointwise sums requires that the codomain \(B\) has a notion of summability for infinite families. For finite families, addition is given by the standard addition in \(B\) and, since \(B\) is a finite-dimensional \Cstar-algebra, it can be regarded as the direct sum of a finite collection of finite-dimensional Hilbert spaces. Thus, the notion of infinite summability assigned to \(B\) corresponds to the usual one derived from operator norm convergence in \(\FdHilb\) (see Proposition~\ref{prop:SOT_FdHilb} for further details).} yields a CPTR map, in which case \(\Sigma \{f_i\}_I\) is defined to be the corresponding result.
  Then, the monoidal category \((\CPTR,\oplus,\{0\})\) is a \(\SCat{s}\)-UDC.
\end{proposition} \begin{proof}[Proof. (Sketch)]
  We first check that each hom-set \(\CPTR(A,B)\) is a strong \(\Sigma\)-monoid using the definition of \(\Sigma\) given above. 
  Its neutral element is the map that sends any \(\rho \in A\) to the zero operator in \(B\); the singleton axiom is trivially satisfied. 
  Notice that the pointwise addition of CPTR maps does not have additive inverses: considering that morphisms \(f,g \in \CPTR(A,B)\) send positive elements \(\rho \in A\) to positive elements in \(B\), the pointwise addition of CPTR maps satisfies 
  \begin{equation*}
    (f+g)(\rho) = f(\rho) + g(\rho) \geq 0.
  \end{equation*}
  Then, it follows that strong flattening and strong bracketing are implied by commutativity and associativity of addition.
  Subsummability is satisfied as well since if \(\Sigma \{f_i\}_I\) is trace-reducing, summing up fewer maps \(J \subseteq I\) means that \(\Sigma \{f_j\}_J\) is guaranteed to be trace-reducing.

  It is immediate that composition distributes over pointwise addition, so that composition is \(\Sigma\)-bilinear.
  It is evident that \(\{0\}\) is a zero object and \(\Sigma \{\id_A \oplus 0, 0 \oplus \id_B\} \keq \id_{A \uplus B}\) holds:
  \begin{equation*}
    \Sigma \{\id_A \oplus 0, 0 \oplus \id_B\}(a,b) = (\id_A \oplus 0)(a,b) + (0 \oplus \id_B)(a,b) = (a,0) + (0,b) = (a,b).
  \end{equation*}
  Thus, \(\CPTR\) is a \(\SCat{s}\)-UDC, as claimed.
\end{proof}

Objects \(A \oplus B\) in \(\CPTR\) are coproducts; however, \(A \oplus B\) is not a categorical product since the diagonal morphism \(\Delta \colon A \to A \oplus A\) is not trace-reducing.
According to Theorem~\ref{thm:classical_trace}, the category \((\CPTR,\oplus,\ex)\) is partially traced.
In fact, \(\CPTR\) is totally traced, as previously shown by Selinger in~\cite{SelingerPL} (where \(\CPTR\) is denoted \(\cat{Q}\)).
The trace given in~\cite{SelingerPL} is not exactly \(\ex\), in the sense that it is not defined in terms of \(\Sigma\)-monoids, but rather in terms of recursive applications of binary addition; nevertheless, it is apparent that the two notions of trace coincide in their result.
Alternatively, it may be shown that \((\CPTR,\oplus,\ex)\) is totally traced employing a similar argument to the one used for \(\SubStoch\), relying on the fact that objects in \(\CPTR\) are finite-dimensional.
The proof is cumbersome due to the need to manage CPTR maps, but the intuition is the same.

We may transport the trace in \(\CPTR\) to \((\CPTP,\oplus,\widehat{\ex})\) using Proposition~\ref{prop:induced_trace} which defines:
\begin{equation*}
  \widehat{\ex}^U(f) = \begin{cases}
    g &\ifc \ex^{U}(f) \text{ is trace-preserving} \\
    \undefined &\otherwise.
  \end{cases}
\end{equation*}
As in the case of \(\Stoch\), this is not a totally traced category.

The fact that \(\CPTR\) is \(\SCat{s}\)-enriched tells us that there is no interference between the paths being added up in the execution formula; thus, this is a clear example of ``quantum data, classical control".
In contrast, in a quantum iterative loop we would expect that paths may cancel out with each other, allowing for additive inverses of morphisms in the category, preventing it from being \(\SCat{s}\)-enriched.
Such is the case of \(\Contraction\) as will be shown in what remains of this chapter.

\paragraph*{Some traced monoidal functors.} There is an obvious faithful functor \(M \colon \FinPInj \to \SubStoch\) where \(\FinPInj\) is the subcategory of \(\PInj\) where objects are restricted to finite sets. The functor \(M\) acts as the identity on objects and maps partial injections \(f \colon A \to B\) to the matrix \(M(f)\) of entries:
\begin{equation*}
  m_{ba}  = \begin{cases}
    1 &\ifc b = f(a) \\
    0 &\otherwise.
  \end{cases}
\end{equation*}
It is trivial to check that \(\sum_{b \in B} m_{ba} = 1\) if \(a \in \dom(f)\) whereas \(\sum_{b \in B} m_{ba} = 0\) otherwise, so that the matrix \(M(f)\) is indeed substochastic. 
It is straightforward to check that \(M\) is a strict monoidal functor and, moreover, for each family \(\fml{f} \in \PInj(A,B)^*\) we have that:
\begin{equation*}
  \Sigma \fml{f} \keq f \implies \Sigma M\fml{f} \keq M(f)
\end{equation*}
so that the action of \(M\) on hom-sets is a \(\Sigma\)-homomorphism.
Then, it is immediate to show that \(M \colon (\FinPInj,\uplus,\ex) \to (\SubStoch,\uplus,\ex)\) is a traced monoidal functor.
Intuitively, this corresponds to the fact that reversible (hence, deterministic) programs are a subclass of probabilistic programs.

On a similar note, there is a faithful functor \(P \colon \SubStoch \to \CPTR\) mapping finite sets \(A\) to \Cstar-algebras \(\oplus_{a \in A}\, B(\Cset)\). The details of how \(P\) acts on morphisms are rather verbose to give explicitly, but the concept is simple: an element \((p_a)_A \in \oplus_{a \in A}\, B(\Cset)\) that is a positive operator can be represented as a vector of positive real numbers and, if it has trace one (\ie{} if it is a density operator) then \(\sum_{a \in A} p_a = 1\). Consequently, the elements in the \Cstar-algebra \(P(A)\) correspond precisely to probability distributions over \(A\) and it is straightforward to see how multiplying a substochastic matrix \(f \in \SubStoch(A,B)\) with the vector \((p_a)_A \in \Rset^{\abs{A}}\) would result in a sub-probability\footnote{In a sub-probability distribution the probabilities are only required to sum up to a value smaller than or equal to one.} distribution over \(B\) and, consequently, a positive operator in \(P(B)\) with trace at most one.
Following this intuition, it can be shown that substochastic matrices \(A \to B\) induce CPTR maps \(P(A) \to P(B)\).
It is immediate from the definition of \(P\) on objects that \(P\) is a strict monoidal functor. Since summability in \(\CPTR\) is defined pointwise, it is straightforward to show that the action of \(P\) on hom-sets is a \(\Sigma\)-homomorphism.
Consequently, it follows that \(P \colon (\SubStoch,\uplus,\ex) \to (\CPTR,\oplus,\ex)\) is a traced monoidal functor.
Intuitively, this corresponds to the fact that classical probabilistic programs are a subclass of quantum programs with classical control flow.

\subsection{Towards an execution formula with additive inverses}
\label{sec:towards}

The goal of the following sections is to prove that \((\FdContraction,\oplus,\ex)\) is a totally traced category.
Lemma~\ref{lem:all_but_vII} established that the execution formula on every \(\SCat{*}\)-UDC automatically satisfies all axioms of partially traced categories but vanishing II.
The \(\SCat{g}\)-UDC \((\FdHilb,\oplus,\{0\})\) was shown \emph{not} to be a partially traced category with respect to the execution formula, due to vanishing II failing.
The importance of vanishing II is made apparent in the discussion from the previous section where, for instance, \(\SubStoch\) was shown to be totally traced via an induction argument using vanishing II.
In Section~\ref{sec:UDC_limit} sufficient conditions for the execution formula to satisfy vanishing II are presented.
But first, we shall explain why the proof strategy from Theorem~\ref{thm:classical_trace} --- which establishes that \(\SCat{s}\)-UDCs satisfy vanishing II and are, hence, partially traced --- cannot be used in the general case where additive inverses are allowed.
The obvious answer is that, in general, weak \(\Sigma\)-monoids do not satisfy strong flattening and strong bracketing, which are at the core of the proof of Theorem~\ref{thm:classical_trace}.
In essence, the proof made use of the fact that \(\ex^U(\ex^V(f))\) and \(\ex^{U \oplus V}(f)\) add up the same morphisms, but they do so in different arrangements.
Unfortunately, it turns out that there are examples of morphisms \(f\) in \(\FdContraction\) where both \(\ex^U(\ex^V(f))\) and \(\ex^{U \oplus V}(f)\) are well-defined and coincide but the `flattened' family is not summable and, hence, a proof strategy via the flattening and bracketing axioms will not work.
Such an example is given below.

In \(\FdContraction\), a family of morphisms \(A \to B\) will be considered summable iff the sum of their operator norms converges (see Proposition~\ref{prop:Hilb_hom-convergence} for more details); in the particular case of morphisms \(\Cset \to \Cset\) this corresponds to absolute convergence.
Let \(f \colon \Cset \oplus \Cset \oplus \Cset \to \Cset \oplus \Cset \oplus \Cset\) be the following contraction:\footnote{An easy way to check that this is indeed a contraction is to realise that it corresponds to the lower right \(3 \times 3\) block of the matrix \(H \otimes H\) where \(H\) is the unitary matrix: \[H = \frac{1}{\sqrt{2}}\begin{pmatrix}
  1 & 1 \\ 1 & -1
\end{pmatrix}.\]}
\begin{equation*}
  f = \begin{pmatrix}
    f_\sub{00} & f_\sub{01} & f_\sub{02} \\
    f_\sub{10} & f_\sub{11} & f_\sub{12} \\
    f_\sub{20} & f_\sub{21} & f_\sub{22}
  \end{pmatrix}
  = \frac{1}{2}\begin{pmatrix}
    -1 & 1 & -1 \\
    1 & -1 & -1 \\
    -1 & -1 & 1 
  \end{pmatrix}.
\end{equation*}
Multiplying the lower right \(2\!\times\!2\) block of \(f\) with itself results in \(\tfrac{1}{2} \,\id\), so we can easily calculate \(\ex^{\Cset \oplus \Cset}(f)\):
\begin{align*}
  \ex^{\Cset \oplus \Cset}(f) &= f_\sub{00} + \sum_{n=0}^\infty \begin{pmatrix}
    f_\sub{01} & f_\sub{02}
  \end{pmatrix} \begin{pmatrix}
    f_\sub{11} & f_\sub{12} \\
    f_\sub{21} & f_\sub{22} \\
  \end{pmatrix}^n \begin{pmatrix}
    f_\sub{10} \\ f_\sub{20}
  \end{pmatrix} \\
    &= -\frac{1}{2} + \frac{1}{2} \sum_{n\text{ even}} \left( \frac{1}{2} \right)^{\tfrac{n}{2}} + \frac{1}{4} \sum_{n\text{ odd}} \left( \frac{1}{2} \right)^{\tfrac{n-1}{2}} \\
    &= -\frac{1}{2} + \frac{1}{2} \sum_{n=0}^\infty \left( \frac{1}{2} \right)^n + \frac{1}{4} \sum_{n=0}^\infty \left( \frac{1}{2} \right)^n = 1.
\end{align*}  
On the other hand, we have that:
\begin{equation*}
  \ex^\Cset(f) = \frac{1}{2}\begin{pmatrix}
    -1 & 1 \\
    1 & -1
  \end{pmatrix} + \frac{1}{4} \begin{pmatrix} -1 \\ -1 \end{pmatrix} \left( \sum_{n=0}^\infty \frac{1}{2^n} \right) \begin{pmatrix} -1 & -1 \end{pmatrix} = \begin{pmatrix}
    0 & 1 \\
    1 & 0
  \end{pmatrix}
\end{equation*}
so that \(\ex^\Cset(\ex^\Cset(f)) = 0 + \sum_{n=0}^\infty 1 (0)^n 1 = 1\).
Thus, \(\ex^\Cset(\ex^\Cset(f)) \keq \ex^{\Cset \oplus \Cset}(f)\) as required by vanishing II.
In contrast, their `flattened' family contains all \(\Cset \to \Cset\) morphisms arising from chains of the form:
\begin{equation*}
  f_\sub{0w_n} \circ \ldots \circ f_\sub{w_2w_1} \circ f_\sub{w_10}
\end{equation*}
where each \(w_i\) is either subscript \(1\) or \(2\) and \(n\) is an arbitrary positive integer.
It is straightforward to check that, in general, there are \(2^{k-1}\) chains of length \(k\) and each chain of length \(k\) has absolute value \(\tfrac{1}{2^k}\), consequently, the sum of absolute values of all elements in the flattened family is \[\sum_{k=1}^\infty \frac{1}{2^k} \cdot 2^{k-1} = \sum_{k=1}^\infty \frac{1}{2}\] and, hence, the sum is not absolutely convergent and the flattened family is not summable.
Thus, such a morphism \(f\) provides an example where the correctness of vanishing II cannot be shown by means of the flattening and bracketing axioms.
Instead, the following subsections build upon the definition of the kernel-image trace and the notion of limit from Hausdorff commutative monoids to provide sufficient conditions for vanishing II to be satisfied in a \(\SCat{w}\)-UDC.

\subsection{\(\SCat{g}\)-UDCs and their kernel-image trace}
\label{sec:SCatg_UDC}

Section~\ref{sec:kernel_image} discussed the kernel-image trace from~\cite{Malherbe}: a partial trace available on every \(\Ab\)-enriched category with finite biproducts.
The proof that this is a categorical trace appears in Proposition 3.17 from~\cite{Malherbe}; in it, the existence of biproducts is only required so that morphisms of type \(A \oplus B \to C \oplus D\) may be uniquely characterised by their matrix decomposition.
However, as discussed in Proposition~\ref{prop:unique_decomposition}, the structure of a \(\SCat{*}\)-UDC is sufficient for this unique decomposition to exist.
Moreover, every \(\Sigma\)-group \((X,\Sigma)\) can be seen as an abelian group \((X,\Sigma \{-,-\})\), so we may regard any \(\SCat{g}\)-enriched category as an \(\Ab\)-enriched one by simply forgetting the action of \(\Sigma\) on infinite families.
Therefore, it is to be expected that any \(\SCat{g}\)-UDC is partially traced using the kernel-image trace.
The proof of this claim is sketched below, it follows the same strategy used by the authors of~\cite{Malherbe} to prove that additive categories are partially traced.

\begin{notation}
  To improve readability, given a \(\Sigma\)-group \((X,\Sigma)\) and elements \(x,y \in X\), we use the shorthand \(x+y\) and \(x-y\) to refer to \(\Sigma \{x,y\}\) and \(\Sigma \{x,-y\}\) respectively.
\end{notation}

Recall that, for an arbitrary morphism \(f \colon A \oplus U \to B \oplus U\) in \(\C\), we write \((k,i) \Vdash \Tr^U(f)\) iff there are morphisms \(i \colon A \to U\) and \(k \colon U \to B\) in \(\C\) such that the diagram
\begin{equation} \label{eq:ki_trace}
  \begin{tikzcd} [row sep=large]
    A && U \\
    & {} & U && B
    \arrow["i", dashed, from=1-1, to=1-3]
    \arrow["{f_\dsub{UA}}"', from=1-1, to=2-3]
    \arrow["{f_\dsub{BU}}", from=1-3, to=2-5]
    \arrow["k", dashed, from=2-3, to=2-5]
    \arrow["{\id - f_\dsub{UU}}"', from=1-3, to=2-3]
  \end{tikzcd}
\end{equation}
commutes.

\begin{proposition} \label{prop:SCat_ki_trace}
  Let \(\C\) be a \(\SCat{g}\)-UDC. Then, the kernel-image trace (as defined in Proposition~\ref{prop:Ab_ki_trace}) makes \(\C\) a partially traced category.
\end{proposition} \begin{proof}[Proof. (Sketch)]
  First, let's check that the kernel-image trace \(\Tr^{U}\) is well-defined, \ie{} it does not depend on the choice of morphisms \(k\) and \(i\). Assume \((k,i) \Vdash \Tr^U(f)\) and \((k',i') \Vdash \Tr^{U}(f)\) then,
  \begin{equation*}
    k \circ f_\sub{UA} =  k \circ (\id\!-\!f_\sub{UU}) \circ i' =  f_\sub{BU} \circ i' = k' \circ f_\sub{UA}
  \end{equation*}
  and, similarly, \(f_\sub{BU} \circ i = f_\sub{BU} \circ i'\), so \[\Tr^{U}(f) = f_\sub{BA} + f_\sub{BU} \circ i = f_\sub{BA} + k \circ f_\sub{UA}\] is well-defined.
  The proof of each axiom of partially traced categories is sketched below; for further details see Proposition 3.17 from~\cite{Malherbe}.
  \begin{itemize}
    \item \emph{Naturality.} If \((k,i) \Vdash \Tr^U(f)\), then \((hk,ig) \Vdash \Tr^U((h \oplus \id) f (g \oplus \id))\) follows trivially, and \(\Tr^U((h \oplus \id) f (g \oplus \id)) = h \circ \Tr^U(f) \circ g\) due to composition being \(\Sigma\)-bilinear.
    \item \emph{Dinaturality.} If \((k,i) \Vdash \Tr^U((\id \oplus g) f)\), then \((kg,j) \Vdash \Tr^{U'}(f (\id \oplus g))\) where \(j = f_\sub{U'A} + f_\sub{U'U'} \circ i\); this follows from the fact that \(i = gj\) thanks to the left triangle of diagram~\eqref{eq:ki_trace} commuting. That \(\Tr^U((\id \oplus g) f) = \Tr^{U'}(f (\id \oplus g))\) follows trivially from the definition of \(\Tr^U\), using the witnesses \(k\) and \(kg\) respectively. The opposite direction of implication can be proven using a similar argument.
    \item \emph{Superposing.} If \((k,i) \Vdash \Tr^U(f)\), then it is immediate from the matrix decomposition of \(g \oplus f\) (see Remark~\ref{rmk:sUDC_oplus}) that \((\iota k, i \pi) \Vdash \Tr^U(g \oplus f)\) and \(\Tr^U(g \oplus f) = g \oplus \Tr^U(f)\).
    \item \emph{Yanking.} The axiom follows trivially from the matrix decomposition of the symmetry \(\sigma\) (see Remark~\ref{rmk:sUDC_symmetry}), with \((\id,\id) \Vdash \Tr^U(\sigma)\).
    \item \emph{Vanishing I.} For every morphism \(f \colon A \oplus Z \to B \oplus Z\) each of its components \(f_\sub{ZA}\), \(f_\sub{ZZ}\) and \(f_\sub{BZ}\) are zero morphisms due to \(Z\) being a zero object. Then, \((0,0) \Vdash \Tr^Z(f)\) and \(\Tr^Z(f) = f_\sub{BA} + 0\). Moreover, \(f_\sub{BA} = \pi_\sub{B}  f  \iota_A = \rho f \rho^{-1}\) due to Remark~\ref{rmk:sUDC_unitors}; therefore, \(\Tr^Z(f) = \rho f \rho^{-1}\).
    \item \emph{Vanishing II.} Let \(f \colon A \oplus U \oplus V \to A \oplus U \oplus V\) and assume \((k,i) \Vdash \Tr^{V}(f)\); vanishing II follows from uniqueness of the matrix decomposition along with \(\SCat{g}\)-enrichment: if \((i',k') \Vdash \Tr^U(\Tr^{V}(f))\) then \((i'',k'') \Vdash \Tr^{U \oplus V}(f)\) where \(i''_\sub{UA} = i'\), \(k''_\sub{BU} = k'\), \(i''_\sub{VA} = i_\sub{VA} + i_\sub{VU} i'\) and \(k''_\sub{BV} = k_\sub{BV} + k_\sub{UV} k'\). Similarly, if \((i'',k'') \Vdash \Tr^{U \oplus V}(f)\) then \((i''_\sub{UA},k''_\sub{BU}) \Vdash \Tr^U(\Tr^{V}(f))\) and, in both cases, \(\Tr^U(\Tr^{V}(f)) = \Tr^{U \oplus V}(f)\) follows from algebraic manipulation using the \(\SCat{g}\)-enrichment. \qedhere
  \end{itemize}
\end{proof}

The left triangle of diagram~\eqref{eq:ki_trace} imposes that \((k,i) \Vdash \Tr^U(f)\) implies \(f_\sub{UA} = (\id\!-\!f_\sub{UU}) \circ i\). Thanks to the \(\SCat{g}\)-enrichment, composition distributes over addition, and additive inverses are available for every morphism. Then, the previous equation is equivalent to \(i = f_\sub{UA} + f_\sub{UU} \circ i\) and we may expand this expression recursively, obtaining:
\begin{equation} \label{eq:recursive_i}
  i = f_\sub{UA} + f_\sub{UU} \circ (f_\sub{UA} + f_\sub{UU} \circ i) = \ldots
\end{equation}
Thus, for any arbitrary \(n \in \Nset\) we have that \((k,i) \Vdash \Tr^U(f)\) implies:
\begin{align*}
  \sum_{j=0}^n f_\sub{UU}^j f_\sub{UA} = i - f_\sub{UU}^{n+1} \circ i \\
  \sum_{j=0}^n f_\sub{BU} f_\sub{UU}^j = k - k \circ f_\sub{UU}^{n+1}
\end{align*}
which, when \(n \to \infty\), resembles the infinite sum from the execution formula~\eqref{eq:ex_formula_informal}.
This hints at a strategy to prove the validity of the execution formula as a categorical trace: use the kernel-image trace as a stepping stone, and find out under which circumstances the above equation can be extended to an infinite sum.
However, as \(n \to \infty\) these equations may become ill-defined unless we can show that \(f_\sub{UU}^n \circ i\) tends to zero.
With this situation in mind, Chapter~\ref{chap:Sigma} gave a brief introduction to topological monoids, where limits may be taken, and Section~\ref{sec:HausCMon->SCatft} established the relationship between these and \(\Sigma\)-monoids.
Continuing this line of work, the following subsection presents a subclass of \(\SCat{g}\)-UDCs whose hom-sets are endowed with a topology, enabling us to use nets and their limits to discuss the convergence of infinite sums.

\subsection{A limit condition}
\label{sec:UDC_limit}

This section introduces the notion of \emph{hom-convergence UDC} which \(\Hilb\) is an example of.
In Section~\ref{sec:ex_FdContraction}, the special properties of these UDCs will be used to prove that \(\FdContraction\) is a totally traced category with respect to the execution formula.
Conceptually, the goal is to study \(\SCat{g}\)-UDCs whose hom-sets are Hausdorff abelian groups, so that the notion of limit is readily available.
The first idea that comes to mind is to define \(\HAG\)-UDCs; unfortunately, the category \(\HAG\) of Hausdorff abelian groups does not have tensor products and such an enrichment would be ill-defined (see Remark~\ref{rmk:HausCMon-enriched}).
The alternative presented below is to consider \(\SCat{g}\)-UDCs whose hom-sets are identified with Hausdorff abelian groups and whose composition is continuous.
To formalise this notion, we may use the faithful functor \(G \colon \HAG \to \SCat{g}\) that maps each \((X,\tau,+) \in \HAG\) to \((X,\Sigma) \in \SCat{g}\), where \(\Sigma\) is the extended group operation induced by the topology (see Definition~\ref{def:HausCMon_ft}).

\begin{definition}
  Let \(\C\) be a \(\SCat{g}\)-UDC.
  Let \(\Phi\) be a collection containing, a \(\Sigma\)-isomorphism\footnote{A \(\Sigma\)-isomorphism is simply an isomorphism in a \(\SCat{*}\) category, \ie{} a \(\Sigma\)-homomorphism with an inverse \(\Sigma\)-homomorphism.} \(\varphi \colon \C(A,B) \to G(X_{A,B})\) for each hom-object \(\C(A,B) \in \SCat{g}\), where \(X_{A,B}\) is some Hausdorff abelian group.
  The pair \((\C,\Phi)\) is a \gls{hom-convergence UDC} if composition is continuous in each variable; more precisely: for each triple of objects \(A,B,C \in \C\) and morphisms \(f \in \C(A,B)\) and \(g \in \C(B,C)\), the \(\Sigma\)-homomorphisms
  \begin{align*}
    \varphi (- \circ f) \varphi^{-1} &\colon G(X_{B,C}) \to G(X_{A,C}) \\
    \varphi (g \circ -) \varphi^{-1} &\colon G(X_{A,B}) \to G(X_{A,C})
  \end{align*}
  are both in the image of \(G \colon \HAG \to \SCat{g}\).
\end{definition}

Recall that composition in a \(\SCat{*}\)-enriched category is a \(\Sigma\)-homomorphism \(\C(B,C) \otimes \C(A,B) \to \C(A,C)\) and, hence, by the definition of the tensor product in \(\SCat{*}\) categories (see Section~\ref{sec:SCat_tensor}), both \(- \circ f\) and \(g \circ -\) are \(\Sigma\)-homomorphisms.
Moreover, \(G\) acts as the identity on morphisms so, if as required by the definition of hom-convergence UDCs there is some continuous function \(\hat{f}\) satisfying that \(G(\hat{f}) = \varphi (- \circ f) \varphi^{-1}\), then it is justified to say that \(- \circ f\) is continuous.
In the framework of hom-convergence UDCs we may recover the notion of limits from topology.

\begin{notation} \label{not:UDC_limit}
  Recall that \(G \colon \HAG \to \SCat{g}\) does not change the underlying set of a given Hausdorff abelian group.
  Thus, in a hom-convergence UDC \((\C,\Phi)\), every morphism \(f \in \C(A,B)\) corresponds to a point \(\varphi(f)\) in the corresponding Hausdorff abelian group \(X_{A,B}\).
  Let \(\alpha \colon D \to \C(A,B)\) be a net; we say that \(\alpha\) has a limit point \(f\) iff \(\lim\, (\varphi \circ \alpha) \keq \varphi(f)\) using the notion of limits from the Hausdorff space \(X_{A,B}\).
  Even though (strictly speaking) \(\C(A,B)\) is not a Hausdorff abelian group, we recover the notation \(\lim \alpha \keq f\) to indicate that the limit point of \(\alpha\) exists and is \(f \in \C(A,B)\).
  The beginning of this section will deal with the limit of certain sequences --- recall that a sequence is a net defined on the directed set \((\Nset,\leq)\). For a sequence \(s \colon \Nset \to \C(A,B)\) we will use the shorthand
  \begin{equation*}
    \lim_{n \to \infty} s(n)
  \end{equation*}
  to refer to its limit point (if it exists).
\end{notation}  

It is possible to define a topology on hom-sets of \(\Hilb\) so that the category is a hom-convergence UDC.
An immediate candidate is the topology induced by the operator norm, but such a topology presents problems in the infinite-dimensional case when we attempt to prove certain results, such as Lemma~\ref{lem:lim_last_term}.
Instead, we may use the strong operator topology, defined below.
Any reader only interested in the finite-dimensional case may skip to Proposition~\ref{prop:SOT_FdHilb}, where it is shown that the strong operator topology in \(\FdHilb\) coincides with the standard operator norm topology.

\begin{definition} \label{def:SOT}
  Every hom-set \(\Hilb(A,B)\) can be assigned a \gls{strong operator topology} denoted \(\tau_\sub{\mathrm{SOT}}\) whose base is comprised of the following open sets:
  \begin{equation} \label{eq:SOT_ball}
    B^f_{S,\epsilon} = \{g \in \Hilb(A,B) \mid \forall a \in S,\ \norm{f(a) - g(a)} < \epsilon\}
  \end{equation}
  where \(f \in \Hilb(A,B)\), \(S\) is a finite set of vectors in \(A\) and \(\epsilon > 0\).
\end{definition}

It is not immediate that the definition above provides a valid base; the proof is discussed below.
For any \(g \in B^f_{S,\epsilon}\) it follows that \(\epsilon - \norm{f(a) - g(a)} > 0\) for all \(a \in S\); let \(\hat{\epsilon}\) be
\begin{equation*}
  \hat{\epsilon} = \min{} \{\epsilon - \norm{f(a) - g(a)}\}_{a \in S}.
\end{equation*}
Then, for any \(h \in B^g_{S,\hat{\epsilon}}\) and for all \(a \in S\) we have that
\begin{equation*}
  \norm{g(a) - h(a)} < \hat{\epsilon} \leq \epsilon - \norm{f(a) - g(a)}
\end{equation*}
and using the triangle inequality we obtain:
\begin{equation*}
  \norm{f(a) - h(a)} \leq \norm{f(a) - g(a)} + \norm{g(a) - h(a)} < \epsilon
\end{equation*}
implying \(h \in B^f_{S,\epsilon}\).
Therefore,
\begin{equation} \label{eq:SOT_ball_in_ball}
  g \in B^f_{S,\epsilon} \implies B^g_{S,\hat{\epsilon}} \subseteq B^f_{S,\epsilon}.
\end{equation}
Assume \(g\) is a map belonging to two different open sets \(g \in B^f_{S,\epsilon}\) and \(g \in B^{f'}_{S',\epsilon'}\) and define the corresponding \(\hat{\epsilon}\) and \(\hat{\epsilon}'\) as above.
It is straightforward to check that
\begin{equation*}
  B^g_{S \cup S',\min(\hat{\epsilon},\hat{\epsilon}')} \subseteq B^g_{S,\hat{\epsilon}} \cap B^g_{S',\hat{\epsilon}'}
\end{equation*}
and, thanks to~\eqref{eq:SOT_ball_in_ball}, we may conclude that:
\begin{equation*}
  g \in B^f_{S,\epsilon} \cap B^f_{S',\epsilon'} \implies B^g_{S \cup S',\min(\hat{\epsilon},\hat{\epsilon}')} \subseteq B^f_{S,\epsilon} \cap B^f_{S',\epsilon'}
\end{equation*}
proving the base is valid (see Definition~\ref{def:base}).

\begin{proposition}
  The strong operator topology \(\tau_\sub{\mathrm{SOT}}\) on every hom-set \(\Hilb(A,B)\) is Hausdorff.
\end{proposition} \begin{proof}
  Let \(f,g \in \Hilb(A,B)\) so that \(f \not = g\) and find an element \(a \in A\) such that \(f(a) \not = g(a)\); fix \(\epsilon = \tfrac{1}{2}\norm{f(a) - g(a)}\).
  Suppose there is a map \(h\) both in \(B^f_{\{a\},\epsilon}\) and \(B^g_{\{a\},\epsilon}\) then:
  \begin{equation*}
    \norm{f(a) - g(a)} \leq \norm{f(a) - h(a)} + \norm{h(a) - g(a)} < 2\epsilon
  \end{equation*}
  and we reach a contradiction \(\norm{f(a) - g(a)} < \norm{f(a) - g(a)}\).
  Therefore, \(B^f_{\{a\},\epsilon}\) and \(B^g_{\{a\},\epsilon}\) must be disjoint, implying that the topology is Hausdorff.
\end{proof}

\begin{remark}
  In the literature, the strong operator topology is often defined in terms of a subbase: a collection \(C\) of sets that induces a base by taking all finite intersections of elements in \(C\).
  Such a subbase comprises all open sets \(B^f_{S,\epsilon}\) where \(S\) is a singleton set; let \(\tau_\sub{\mathrm{SOT}}'\) denote the topology generated by this subbase.
  Every open set in \(\tau_\sub{\mathrm{SOT}}'\) is trivially an open set in \(\tau_\sub{\mathrm{SOT}}\), since every element in the subbase for \(\tau_\sub{\mathrm{SOT}}'\) is in \(\tau_\sub{\mathrm{SOT}}\).
  Moreover, \(B^f_{S,\epsilon} = \cap_{a \in S} B^f_{\{a\},\epsilon}\) so the converse also holds, implying that \(\tau_\sub{\mathrm{SOT}}' = \tau_\sub{\mathrm{SOT}}\).
\end{remark}

\begin{proposition} \label{prop:SOT_FdHilb}
  In hom-sets \(\FdHilb(A,B)\), the strong operator topology \(\tau_\sub{\mathrm{SOT}}\) is equivalent to the operator norm\footnote{Recall that the \gls{operator norm} \(\opnorm{f}\) of a bounded linear map \(f \in \Hilb(A,B)\) is the infimum of the subset of real numbers \(c \in \Rset\) that satisfy: \[\forall v \in A,\ \norm{f(v)} \leq c\cdot \norm{v}.\]} topology \(\tau_\mathrm{op}\) whose base is comprised of the following open sets:
  \begin{equation} \label{eq:opnorm_ball}
    B^f_\epsilon = \{g \in \FdHilb(A,B) \mid \opnorm{f - g} < \epsilon\}
  \end{equation}
  for every \(f \in \FdHilb(A,B)\) and \(\epsilon > 0\).
\end{proposition} \begin{proof}
  We must show that \(\tau_\sub{\mathrm{SOT}} \subseteq \tau_\mathrm{op}\) and \(\tau_\mathrm{op} \subseteq \tau_\sub{\mathrm{SOT}}\).
  Thanks to Proposition~\ref{prop:base_inclusion}, to show that \(\tau \subseteq \tau'\) it is sufficient to prove, for their respective bases \(\cB\) and \(\cB'\), that for every \(x \in X\) and \(B \in \cB\):
  \begin{equation*}
    x \in B \implies \exists B' \in \cB' \text{ s.t. } B' \subseteq B \text{ and } x \in B'.
  \end{equation*}

  First, we show that \(\tau_\sub{\mathrm{SOT}} \subseteq \tau_\mathrm{op}\).
  For each element \(B^f_{S,\epsilon}\) of the base for \(\tau_\sub{\mathrm{SOT}}\) and each \(g \in B^f_{S,\epsilon}\) define \(\epsilon' = \epsilon - \opnorm{f-g}\).
  Consider the element \(B^g_{\epsilon'}\) of the base for \(\tau_\mathrm{op}\) so that if \(h \in B^g_{\epsilon'}\) then \(\opnorm{g - h} < \epsilon - \opnorm{f-g}\) and, consequently,
  \begin{equation*}
    \opnorm{f-h} \leq \opnorm{g-h} + \opnorm{f-g} < \epsilon
  \end{equation*}
  implying that \(h \in B^f_{S,\epsilon}\) due to the definition of the operator norm.
  Consequently, \(B^g_{\epsilon'} \subseteq B^f_{S,\epsilon}\) and, trivially, \(g \in B^g_{\epsilon'}\); thus, \(\tau_\sub{\mathrm{SOT}} \subseteq \tau_\mathrm{op}\) according to Proposition~\ref{prop:base_inclusion}.

  To show that \(\tau_\mathrm{op} \subseteq \tau_\sub{\mathrm{SOT}}\), let \(S\) be an orthonormal basis of \(A\).
  Thanks to \(A\) being finite-dimensional, we have that every \(f \in \FdHilb(A,B)\) satisfies:
  \begin{equation} \label{eq:opnorm_orthonormal}
    \opnorm{f} \leq \sqrt{\abs{S}} \cdot \max_{e \in S}\, \norm{f(e)}
  \end{equation}  
  For every open set \(B^f_\epsilon\) from the base for \(\tau_\mathrm{op}\) and every \(g \in B^f_\epsilon\) the base element \(B^g_{\epsilon'}\) where \(\epsilon' = \epsilon - \opnorm{f-g}\) satisfies that \(B^g_{\epsilon'} \subseteq B^f_\epsilon\) and \(g \in B^g_{\epsilon'}\).
  Let \(\hat{\epsilon} = \tfrac{\epsilon'}{\sqrt{\abs{S}}}\) and consider the open set \(B^g_{S,\hat{\epsilon}}\) from the base for \(\tau_\sub{\mathrm{SOT}}\) so that if \(h \in B^g_{S,\hat{\epsilon}}\) then:
  \begin{equation*}
     \forall e \in S,\ \norm{g(e) - h(e)} < \tfrac{\epsilon'}{\sqrt{\abs{S}}} 
  \end{equation*}
  which implies \(\opnorm{g-h} < \epsilon'\) due to~\eqref{eq:opnorm_orthonormal}.
  Consequently, \(B^g_{S,\hat{\epsilon}} \subseteq B^g_{\epsilon'} \subseteq B^f_\epsilon\) and, trivially, \(g \in B^g_{S,\hat{\epsilon}}\); thus, \(\tau_\mathrm{op} \subseteq \tau_\sub{\mathrm{SOT}}\) according to Proposition~\ref{prop:base_inclusion}, completing the proof.
\end{proof}

The following proposition establishes that \(\Hilb\) is a hom-convergence UDC when the strong operator topology is used.
A similar proof can be used to show that \(\Hilb\) with the operator norm topology is a hom-convergence UDC.
The reason for using the strong operator topology is that it will later enable us to prove results (such as Lemma~\ref{lem:lim_last_term}) that do not hold in \(\Hilb\) if the operator norm topology is used. 

\begin{proposition} \label{prop:Hilb_hom-convergence}
  \(\Hilb\) is a hom-convergence UDC, with the topology on each hom-set \(\Hilb(A,B)\) given by the strong operator topology.
\end{proposition} \begin{proof}
  Pointwise addition of bounded linear maps makes \(\Hilb(A,B)\) an abelian group; we now prove that the strong operator topology makes it a Hausdorff abelian group.
  Recall from Proposition~\ref{prop:base_continuous} that \(+ \colon \Hilb(A,B) \times \Hilb(A,B) \to \Hilb(A,B)\) is continuous iff for every pair \(f,g \in \Hilb(A,B)\) and every open set \(B\) from the base for \(\tau_\sub{\mathrm{SOT}}\) the following implication holds:
  \begin{equation*}
    f\!+\!g \in B \implies \exists \text{ open sets }\  U \!\ni\! f \text{ and } V \!\ni\! g \ \text{ s.t. } +(U,V) \subseteq B.
  \end{equation*}
  Thanks to~\eqref{eq:SOT_ball_in_ball} it is sufficient to prove this for open sets \(B\) of the form \(B^{f+g}_{S,\epsilon}\) with arbitrary \(S\) and \(\epsilon\) since any other \(B^h_{S',\epsilon'}\) containing \(f\!+\!g\) will contain some open set of the form \(B^{f+g}_{S',\hat{\epsilon}}\).
  For every pair \(f,g \in \Hilb(A,B)\) and every \(S\) and \(\epsilon\) we have that
  \begin{equation} \label{eq:SOT_continuous_addition}
    (h,h') \in B^f_{S,\epsilon / 2} \times B^g_{S,\epsilon / 2} \implies h+h' \in B^{f+g}_{S,\epsilon}.
  \end{equation}
  due to the following, which holds for all \(a \in S\):
  \begin{equation*}
    \norm{(f + g)(a) - (h + h')(a)} \leq \norm{f(a) - h(a)} + \norm{g(a) - h'(a)} < \tfrac{\epsilon}{2} + \tfrac{\epsilon}{2}.
  \end{equation*}
  Therefore, \(+(B^f_{S,\epsilon / 2},B^g_{S,\epsilon / 2}) \subseteq B^{f+g}_{S,\epsilon}\), implying addition is continuous.
  Moreover, every open set \(B^f_{S,\epsilon}\) satisfies \(-(B^{-f}_{S,\epsilon}) = B^f_{S,\epsilon}\), so the inversion map is continuous according to Proposition~\ref{prop:base_continuous}. Therefore, each hom-set \(\Hilb(A,B)\) endowed with the strong operator topology is a Hausdorff abelian group.

  Recall that any Hausdorff abelian group can be made into a \(\Sigma\)-group via the functor \(G \colon \HAG \to \SCat{g}\) which does not alter the underlying set; thus, every \(\Hilb(A,B)\) is a \(\Sigma\)-group.
  To conclude that the \(\SCat{g}\)-enrichment on \(\Hilb\) is well-defined, we need to verify that composition is \(\Sigma\)-bilinear.
  According to Lemma~\ref{lem:HausCMon_preserves_summability}, it is sufficient to prove that \(g \circ -\) and \(- \circ f\) are continuous monoid homomorphisms for every choice of bounded linear maps \(f\) and \(g\).
  It is immediate that composition is a group homomorphism on both variables due to addition being defined pointwise; below we show that composition is also continuous in each variable.
  \begin{itemize}
    \item For every \(f \in \Hilb(A,B)\) and \(g \in \Hilb(B,C)\) with \(g \not= 0\) it follows from \(g\) being a bounded linear map that
    \begin{equation*}
      h \in B^f_{S,\epsilon / \opnorm{g}} \implies g \circ h \in B^{g \circ f}_{S,\epsilon}
    \end{equation*}
    due to
    \begin{equation*}
      \norm{g(f(a)) - g(h(a))} \leq \opnorm{g} \cdot \norm{f(a) - h(a)} < \opnorm{g} \cdot \tfrac{\epsilon}{\opnorm{g}}
    \end{equation*}
    for all \(a \in S\) and, hence, 
    \begin{equation*}
      (g \circ -)\left(B^f_{S,\epsilon / \opnorm{g}}\right) \subseteq B^{g \circ f}_{S,\epsilon};
    \end{equation*}
    whereas if \(g = 0\) then \((0 \circ -)(B^{f}_{S,\epsilon}) = \{0\}\) so, trivially, \((g \circ -)(B^{f}_{S,\epsilon}) \subseteq B^{g \circ f}_{S,\epsilon}\) as well. Consequently, since every open set from the base for \(\tau_\sub{\mathrm{SOT}}\) containing \(g \circ f\) will contain some open set of the form \(B^{g \circ f}_{S,\epsilon}\) --- see the discussion of~\eqref{eq:SOT_ball_in_ball} --- it follows from Proposition~\ref{prop:base_continuous} that \(g \circ -\) is continuous.
    \item Alternatively,
    \begin{equation*}
      h \in B^g_{f(S),\epsilon} \implies h \circ f \in B^{g \circ f}_{S,\epsilon}
    \end{equation*}
    follows trivially, so \(- \circ f\) is continuous as well.
  \end{itemize}
  Finally, it is immediate that \((\Hilb,\oplus,\{0\})\) is a \(\SCat{g}\)-UDC since its monoidal product is a biproduct.
  Therefore, \(\Hilb\) --- or more precisely \((\Hilb,\Phi)\), where \(\Phi\) is a collection of identity \(\Sigma\)-homomorphisms --- is a hom-convergence UDC.
\end{proof}

\begin{corollary} \label{cor:FdHilb_hom-convergence}
  \(\FdHilb\) is a hom-convergence UDC, with the topology on each hom-set \(\Hilb(A,B)\) given by the operator norm topology.
\end{corollary} \begin{proof}
  The claim follows from the same argument provided in the previous proof, thanks to the equivalence in the finite-dimensional case between the strong operator topology and the operator norm topology, as established in Proposition~\ref{prop:SOT_FdHilb}.
\end{proof}

Notice that the definition of hom-convergence UDC may be reproduced for \(\SCat{ft}\)-enriched UDCs.
However, the motivation behind hom-convergence UDCs is that we may consider the limit of the recursive equation~\eqref{eq:recursive_i} that arises in categories with a kernel-image trace, and such a trace assumes the existence of additive inverses in its definition.
The following lemma illustrates how the execution formula arises from the kernel-image trace in a hom-convergence UDC.

\begin{lemma}
  Let \(\C\) be a hom-convergence UDC and let \(f \colon A \oplus U \to B \oplus U\) be a morphism in \(\C\) satisfying \((k,i) \Vdash \Tr_\sub{\ki}^U(f)\) for some morphisms \(k \in \C(U,B)\) and \(i \in \C(A,U)\). If \(f\) satisfies either
  \begin{equation*}
    \Lim{n \to \infty} f_\sub{UU}^n f_\sub{UA} \keq 0 \quad\quad \text{or} \quad\quad \Lim{n \to \infty} f_\sub{BU} f_\sub{UU}^n \keq 0
  \end{equation*}
  then the sequence of partial sums \(s \colon \Nset \to \C(A,B)\)
  \begin{equation*}
    s(n) = \sum_{j=0}^n f_\sub{BU} f_\sub{UU}^j f_\sub{UA}
  \end{equation*}
  is convergent, with
  \begin{equation*}
    \Lim{n \to \infty} s(n) \keq \Tr_\sub{\ki}^U(f) - f_\sub{BA}.
  \end{equation*}
\end{lemma} \begin{proof}
  Recall from the discussion of equation~\eqref{eq:recursive_i} that for any \(n \in \Nset\) we have:
  \begin{equation*}
    \sum_{j=0}^n f_\sub{UU}^j f_\sub{UA} = i - f_\sub{UU}^{n+1} \circ i.
  \end{equation*}
  In particular, the sequence \(\alpha \colon \Nset \to \C(A,B)\) given by
  \begin{equation*}
    \alpha(n) = s(n) + f_\sub{BU} f_\sub{UU}^{n+1} \circ i
  \end{equation*}
  is the constant function \(n \mapsto f_\sub{BU} \circ i\) and, hence,
  \begin{equation*}
    \Lim{n \to \infty} \alpha(n) \keq f_\sub{BU} \circ i = \Tr_\sub{\ki}^U(f) - f_\sub{BA}.
  \end{equation*}

  Take the claim's assumption that \(\lim f_\sub{BU} f_\sub{UU}^n \keq 0\); then, because composition in \(\C\) is continuous in each variable,
  \begin{equation*}
    \Lim{n \to \infty} f_\sub{BU} f_\sub{UU}^n \circ i \keq 0
  \end{equation*}
  and we may use that \(+\) and the inversion map are continuous to obtain:
  \begin{equation*}
    \Lim{n \to \infty} (\alpha(n) - f_\sub{BU} f_\sub{UU}^{n+1} \circ i) \keq \Lim{n \to \infty} \alpha(n) - \Lim{n \to \infty} f_\sub{BU} f_\sub{UU}^{n+1} \circ i = \left( \Tr_\sub{\ki}^U(f) - f_\sub{BA} \right) - 0
  \end{equation*}
  which simplifies to:
  \begin{equation*}
    \Lim{n \to \infty} s(n) \keq \Tr_\sub{\ki}^U(f) - f_\sub{BA}.
  \end{equation*}
  If we assume \(\lim f_\sub{UU}^n f_\sub{UA} \keq 0\) instead, we may use a similar argument along with the recursive expression of \(k\) to obtain the same result, completing the proof.
\end{proof}

Notice that the lemma above only implies the existence of the limit of a \emph{sequence}, whereas the \(\Sigma\) function of a hom-convergence UDC is defined (according to the functor \(G \colon \HAG \to \SCat{g}\)) in terms of the limit of a \emph{net}; this is meant to prevent different orderings of the summands to yield different results.
Since the execution formula is defined in terms of the \(\Sigma\) function, we need to upgrade the previous lemma from sequences to nets.
To do so, we need to impose an additional condition.

\begin{definition} \label{def:tail_convergent}
  Let \(\C\) be a hom-convergence UDC and let \(f \colon A \oplus U \to B \oplus U\) be a morphism in \(\C\). Let \((\fpset{\Nset},\subseteq)\) be the directed set of finite subsets of \(\Nset\).
  For any \(J \in \fpset{\Nset}\), let
  \begin{equation*}
    \tilde{J} = \{n \in \Nset \mid n < \max(J) \text{ and } n \not\in J\}.
  \end{equation*}
  If \(f\) satisfies both
  \begin{equation} \label{eq:sum_of_gaps}
    \Lim{J \in \fpset{\Nset}} \Sigma \{f_\sub{BU} f_\sub{UU}^j f_\sub{UA}\}_{\tilde{J}} \keq 0
  \end{equation}
  and either
  \begin{equation*}
    \Lim{n \to \infty} f_\sub{UU}^n f_\sub{UA} \keq 0 \quad\quad \text{or} \quad\quad \Lim{n \to \infty} f_\sub{BU} f_\sub{UU}^n \keq 0
  \end{equation*}
  we say \(f\) is \(U\)-\gls{tail vanishing}.
\end{definition}

For any finite set \(J \in \fpset{\Nset}\), its corresponding \(\tilde{J}\) as defined above can be seen as the set of `gaps' that are left by \(J\) when enumerating from \(0\) to \(\max(J)\).
If \(J \subseteq J'\) then \(\min(\tilde{J}) \leq \min(\tilde{J'})\) (or \(\tilde{J'} \not = \varnothing\)) so condition~\eqref{eq:sum_of_gaps} may be phrased as follows: \emph{as the first occurence of a gap tends towards infinity, the sum of the corresponding missing terms approaches zero}.

\begin{lemma} \label{lem:ex_from_ki}
  Let \(\C\) be a hom-convergence UDC and let \(f \colon A \oplus U \to B \oplus U\) be a \(U\)-tail vanishing morphism in \(\C\) satisfying \((k,i) \Vdash \Tr^U(f)\) for some morphisms \(k \in \C(U,B)\) and \(i \in \C(A,U)\).
  Then, the execution formula on \(f\):
  \begin{equation*} 
    \ex^U(f) = \Sigma \left( \{f_\sub{BA}\} \uplus \{f_\sub{BU} f_\sub{UU}^n f_\sub{UA}\}_\Nset \right)
  \end{equation*}
  is well-defined and
  \begin{equation*}
    \ex^U(f) = \Tr_\sub{\ki}^U(f).
  \end{equation*}
\end{lemma} \begin{proof}
  The previous lemma established that the sequence
  \begin{equation*}
    s(n) = \sum_{j=0}^n f_\sub{BU} f_\sub{UU}^j f_\sub{UA}
  \end{equation*}
  is convergent under the claim's assumptions, with \(\lim s \keq \Tr_\sub{\ki}^U(f) - f_\sub{BA}\).
  It is evident that \(s \circ \max \colon \fpset{\Nset} \to \C(A,B)\) is a subnet of \(s\), so
  \begin{equation*}
    \Lim{J \in \fpset{\Nset}} s(\max(J)) \keq \Tr_\sub{\ki}^U(f) - f_\sub{BA}.
  \end{equation*}
  Let \(\sigma \colon \fpset{\Nset} \to \C(A,B)\) be the net of partial sums of the family
  \begin{equation*}
    \{f_\sub{BU} f_\sub{UU}^n f_\sub{UA}\}_\Nset.
  \end{equation*}
  For any \(J \in \fpset{\Nset}\) we may write:
  \begin{equation*}
    \sigma(J) = s(\max(J)) - \{f_\sub{BU} f_\sub{UU}^j f_\sub{UA}\}_{\tilde{J}}
  \end{equation*}
  and, using that \(f\) satisfies~\eqref{eq:sum_of_gaps} along with the fact that \(+\) and the inversion map are continuous, we obtain
  \begin{equation*}
    \Lim{J \in \fpset{\Nset}} \sigma(J) \keq \Lim{J \in \fpset{\Nset}} s(\max(J)) - \Lim{J \in \fpset{\Nset}} \{f_\sub{BU} f_\sub{UU}^j f_\sub{UA}\}_{\tilde{J}} \keq \left(\Tr_\sub{\ki}^U(f) - f_\sub{BA} \right) - 0
  \end{equation*}
  implying that \(\lim \sigma \keq \Tr_\sub{\ki}^U(f) - f_\sub{BA}\).
  Let \(\hat{\sigma}\) be the net of partial sums of the family
  \begin{equation*}
    \{f_\sub{BA}\} \uplus \{f_\sub{BU} f_\sub{UU}^n f_\sub{UA}\}_\Nset
  \end{equation*}
  and recall that, according to the definition of hom-convergence UDC and, in particular, that of the functor \(G \colon \HAG \to \SCat{g}\),
  \begin{equation*}
    \Sigma \left( \{f_\sub{BA}\} \uplus \{f_\sub{BU} f_\sub{UU}^n f_\sub{UA}\}_\Nset \right) \keq \lim \hat{\sigma}
  \end{equation*}
  and, hence, \(\ex^U(f) \keq \lim \hat{\sigma}\).
  Moreover, it is immediate that \(\lim \hat{\sigma} \keq f_\sub{BA} + \lim \sigma\) so that \(\ex^U(f) \keq \Tr_\sub{\ki}^U(f)\), as claimed.
\end{proof}

The conditions imposed in the definition of tail vanishing morphisms are quite strong but, as shown in the next section, they are satisfied by every morphism in \(\FdContraction\). This will be sufficient to prove that the execution formula makes \(\FdContraction\) a traced monoidal category.
Furthermore, it will be shown that the category \(\LSIContraction\) of linear shift invariant quantum processes over discrete time (see Section~\ref{sec:LSI}) is also traced using the corresponding execution formula.

\section{Quantum iterative loops and the execution formula}
\label{sec:quantum_ex_trace}

The focus of this section is on categories of Hilbert spaces and contractions.
Contractions have some useful properties with regards to convergence.
The section begins by introducing some well-known results that will be essential in the following discussions.

\subsection{Properties of contractions in \(\Hilb\) (\emph{Preamble})}

The following results are well-known in the literature of contractions and they will be essential in later discussions within this section.
The results in this subsection hold for arbitrary Hilbert spaces, thus including infinite-dimensional ones.

Every contraction can be seen as a component of an isometry acting on larger Hilbert spaces; such an isometry is commonly known as a \emph{dilation}.
Other dilations of contractions are known (perhaps the most widely known being Sz.-Nagy's dilation theorem: Theorem 4.1 from~\cite{Contraction}); the one established in the proposition below is arguably the simplest dilation and it will be sufficient for our purposes.

\begin{proposition}[Halmos~\cite{Halmos}] \label{prop:Halmos}
  Let \(f \in \Contraction(A,B)\) and define \(D_f = \sqrt{\id - f^\dagger f}\).
  Then, the linear map
  \begin{equation*}
    g = \begin{pmatrix}
      -f^\dagger & D_f \\
      D_{f^\dagger} & f
    \end{pmatrix}
  \end{equation*}
  is an isometry.
\end{proposition} \begin{proof}
  In Section I.3 from~\cite{Contraction} it is shown that \(D_f^\dagger = D_f\) and \(f \circ D_f = D_{f^\dagger} \circ f\).
  Then,
  \begin{align*}
    g^\dagger g &=
    \begin{pmatrix}
      -f & D_{f^\dagger} \\
      D_f & f^\dagger \\
    \end{pmatrix}
    \begin{pmatrix}
      -f^\dagger & D_f \\
      D_{f^\dagger} & f \\
    \end{pmatrix} \\
    &=
    \begin{pmatrix}
      ff^\dagger + D_{f^\dagger}^2 & D_{f^\dagger} f - f D_f \\
      f^\dagger D_{f^\dagger} - D_f f^\dagger & D_f^2 + f^\dagger f
    \end{pmatrix} \\
    &=
    \begin{pmatrix}
      ff^\dagger + \id - ff^\dagger & f D_f - f D_f \\
      f^\dagger D_{f^\dagger} - f^\dagger D_{f^\dagger} & \id - f^\dagger f + f^\dagger f
    \end{pmatrix} \\
    &=
    \begin{pmatrix}
      \id & 0 \\
      0 & \id
    \end{pmatrix}
  \end{align*}
  and, hence, \(g^\dagger g = \id\), implying that \(g\) is an isometry.
\end{proof}

Contractions satisfy a useful property in terms of the sum of the norm of its components, as established below.

\begin{proposition} \label{prop:isometry_column}
  For any morphism \(f \in \Contraction(A \oplus A', B \oplus B')\) and any \(a \in A\),
  \begin{equation*}
    \norm{f_\sub{BA}(a)}^2 + \norm{f_\sub{B'A}(a)}^2 \leq \norm{a}^2
  \end{equation*}
\end{proposition} \begin{proof}
  Suppose there is a vector \(a \in A\) such that
  \begin{equation*}
    \norm{f_\sub{BA}(a)}^2 + \norm{f_\sub{B'A}(a)}^2 > \norm{a}^2.
  \end{equation*}
  Then, the vector \(\iota_A(a) = \vector{a}{0}\) in \(A \oplus A'\) satisfies:
  \begin{equation*}
    \norm{f(\iota_A(a))}^2 = \norm{\vector{f_\sub{BA}(a)}{f_\sub{B'A}(a)}}^2 = \norm{f_\sub{BA}(a)}^2 + \norm{f_\sub{B'A}(a)}^2 > \norm{a}^2 = \norm{\iota_A(a)}^2.
  \end{equation*}
  But \(f\) is a contraction and we have reached a contradiction \(\norm{f(\iota_A(a))} > \norm{\iota_A(a)}\). Therefore, it must be that \(\norm{f_\sub{BA}(a)}^2 + \norm{f_\sub{B'A}(a)}^2 \leq \norm{a}^2\) for all \(a \in A\).
\end{proof}

\begin{corollary} \label{cor:isometry_column_zero}
  For any \(f \in \Contraction(A \oplus A', B \oplus B')\),
  \begin{equation*}
    \forall a \in A,\ \norm{f_\sub{BA}(a)} = \norm{a} \implies f_\sub{B'A} = 0
  \end{equation*}
\end{corollary} \begin{proof}
  The previous proposition establishes that for all \(a \in A\):
  \begin{equation*}
    \norm{f_\sub{B'A}(a)}^2 \leq \norm{a}^2 - \norm{f_\sub{BA}(a)}^2.
  \end{equation*}
  Then, if \(\norm{f_\sub{BA}(a)} = \norm{a}\) we have that \(\norm{f_\sub{B'A}(a)} \leq 0\) for all \(a \in A\) and therefore \(f_\sub{B'A}\) is the zero map \(a \mapsto 0\).
\end{proof}

Contractions have a canonical decomposition into a direct sum of a unitary component and a nonunitary component, as established in Theorem~\ref{thm:Contraction_decomposition}.
The following definitions and proposition build towards this result.

\begin{definition}
  Let \(f \in \Hilb(H,H)\) and let \(H_0 \subseteq H\) be a closed subspace; let \(H_1\) be the orthogonal complement of \(H_0\) so that \(H = H_0 \oplus H_1\).
  If for every \(v \in H_0\) and every \(v' \in H_1\) it is satisfied that \(f(v) \in H_0\) and \(f(v') \in H_1\), we say that \(H_0\) \emph{reduces} \(f\).
  In such a case, \(f\) may be decomposed as follows:
  \begin{equation*}
    f = f_0 \oplus f_1 = \begin{pmatrix}
      f_0 & 0 \\
      0 & f_1
    \end{pmatrix}
  \end{equation*}
  where \(f_0 \colon H_0 \to H_0\) and \(f_1 \colon H_1 \to H_1\).
\end{definition}

Equivalently, reducing subspaces may be defined as subspaces that are invariant both on \(f\) and on its adjoint \(f^\dagger\).

\begin{proposition} \label{prop:reducing_space}
  Let \(f \in \Hilb(H,H)\) and let \(H_0 \subseteq H\) be a closed subspace; \(H_0\) reduces \(f\) iff for all \(v \in H_0\) both \(f(v) \in H_0\) and \(f^\dagger(v) \in H_0\).
\end{proposition} \begin{proof}
  Let \(H_1\) be the orthogonal complement of \(H_0\) and decompose \(f\) as a matrix:
  \begin{equation*}
    f = \begin{pmatrix}
      f_{00} \colon H_0 \to H_0 & f_{01} \colon H_1 \to H_0 \\
      f_{10} \colon H_0 \to H_1 & f_{11} \colon H_1 \to H_1
    \end{pmatrix}.
  \end{equation*}
  Assume \(f(v) \in H_0\) for all \(v \in H_0\), it follows that \(\pi_1 f(H_0) = \{0\}\) so that \(f_{10} = 0\); similarly, if \(f^\dagger(v) \in H_0\) for all \(v \in H_0\) then \(f_{01}^\dagger = 0\) and therefore \(f_{01} = 0\).
  Thus, \(f = f_{00} \oplus f_{11}\) and for all \(v' \in H_1\) it is satisfied that \(f(v') \in H_1\), implying \(H_0\) reduces \(f\).
  Conversely, if \(H_0\) reduces \(f\) it is immediate that \(f^\dagger(v) = f_{00}^\dagger(v) \in H_0\) for all \(v \in H_0\).
\end{proof}

\begin{definition}[Section I.3 from~\cite{Contraction}]
  A contraction \(f \in \Contraction(H,H)\) is \gls{completely nonunitary} if there is no reducing subspace \(H' \subseteq H\) such that the component \(f' \colon H' \to H'\) is unitary.
\end{definition}

\begin{theorem}[Theorem 3.2 from~\cite{Contraction}] \label{thm:Contraction_decomposition}
  Every contraction \(f \in \Contraction(H,H)\) has a reducing subspace \(H_0\) inducing a decomposition \(H = H_0 \oplus H_1\) such that the component \(f_0 \colon H_0 \to H_0\) is unitary and the component \(f_1 \colon H_1 \to H_1\) is completely nonunitary.
  This decomposition is uniquely determined, with
  \begin{equation*}
    H_0 = \{v \in H \mid \forall n \in \Nset,\ \norm{f^n(v)} = \norm{v} = \norm{f^{\dagger n}(v)}\}.
  \end{equation*}
\end{theorem} \begin{proof}
  The proof is paraphrased from~\cite{Contraction} (Theorem 3.2) and included here for completeness.
  For any contraction \(f \in \Contraction(H,H)\) let \(D_f \colon H \to H\) be the linear map \(D_f = \sqrt{\id - f^\dagger f}\) and let \(K_f = \ker(D_f)\), which is evidently a closed subspace of \(H\) and thus \(K_f \in \Hilb\).
  Notice that \(D_f\) is self-adjoint and for all \(v \in H\):
  \begin{equation*}
    \norm{D_f(v)}^2 = \braket{v}{D_f^2(v)} = \braket{v}{(\id - f^\dagger f)(v)} = \braket{v}{v} - \braket{f(v)}{f(v)} = \norm{v}^2 - \norm{f(v)}^2
  \end{equation*}
  and, hence, \(v \in K_f\) if and only if \(\norm{f(v)} = \norm{v}\) so \(K_f\) may be equivalently defined as the Hilbert space:
  \begin{equation} \label{eq:Kf_alternative}
    K_f = \{v \in H \mid \norm{f(v)} = \norm{v}\}.
  \end{equation}
  Moreover, if \(v \in K_f\) then \(D_f^2(v) = (\id - f^\dagger f)(v) = 0\) and we obtain the implication:
  \begin{equation} \label{eq:in_K_implies_isometric}
    \norm{f(v)} = \norm{v} \implies f^\dagger f (v) = v.
  \end{equation}
  Let \(H_0\) be the following subspace:
  \begin{equation*}
    H_0 = \bigcap_{n=0}^\infty (K_{f^n} \cap K_{f^{\dagger n}})
  \end{equation*}
  which is a Hilbert space since any intersection of closed subspaces is closed.
  If we expand the right hand side using~\eqref{eq:Kf_alternative}, we obtain the definition of \(H_0\) given in the claim.
  For any \(v \in H_0\) and all \(n \in \Nset\):
  \begin{align*}
    &\norm{f^nf(v)} = \norm{f^{n+1}(v)} = \norm{v} = \norm{f(v)} \\
    &\norm{f^{\dagger n}f(v)} = \norm{f^{\dagger n-1} f^\dagger f(v)} = \norm{f^{\dagger n-1}(v)} = \norm{v} = \norm{f(v)}
  \end{align*}
  where the third equality of the second equation uses~\eqref{eq:in_K_implies_isometric}.
  Therefore, \(f(v) \in H_0\) and a similar argument shows that \(f^\dagger(v) \in H_0\), implying that \(H_0\) reduces \(f\).

  Let \(H_1\) be the orthogonal complement of \(H_0\) and let \(f = f_0 \oplus f_1\) according to the decomposition \(H = H_0 \oplus H_1\).
  For any \(v \in H_0\), we have that \(v \in K_f\) so implication~\eqref{eq:in_K_implies_isometric} establishes that \(f_0^\dagger f_0 = \id_{H_0}\); conversely, we have that \(v \in K_{f^\dagger}\) so that \(f_0 f_0^\dagger = \id_{H_0}\), implying \(f_0\) is unitary.
  On the other hand, let \(H_2 \subseteq H_1\) be a closed subspace that reduces \(f\) and assume the corresponding \(f_2 \colon H_2 \to H_2\) is unitary.
  Then, for every \(v \in H_2\) and \(n \in \Nset\) we would have \(\norm{f^n(v)} = \norm{v} = \norm{f^{\dagger n}(v)}\) and, hence, \(v \in H_0\); but \(H_2 \subseteq H_1\) and \(H_1\) is orthogonal to \(H_0\), so we reach a contradiction.
  Therefore, we conclude that \(f_2\) cannot be unitary and, hence, \(f_1\) is completely nonunitary, as claimed.
\end{proof}

Finally, the following result will be essential to proving that \((\FdContraction,\oplus,\ex)\) is a totally traced category.

\begin{proposition}[Bartha~\cite{Bartha}] \label{prop:Bartha_isometry}
  Let \(f \in \Isometry(A \oplus U, B \oplus U)\) where \(\id\!-\!f_\sub{UU}\) is invertible, then
  \begin{equation*}
    h = f_\sub{BA} + f_\sub{BU} (\id\!-\!f_\sub{UU})^{-1} f_\sub{UA}
  \end{equation*}
  is an isometry.
\end{proposition} \begin{proof}
  The proof is paraphrased from~\cite{Bartha} and included here for completeness.
  For the sake of brevity, use the following shorthand for the components of \(f\):
  \begin{equation*}
    f = \begin{pmatrix}
      a & b \\
      c & d
    \end{pmatrix}
    = \begin{pmatrix}
      f_\sub{BA} & f_\sub{BU} \\
      f_\sub{UA} & f_\sub{UU}
    \end{pmatrix}
  \end{equation*}
  Considering that \(f\) is an isometry, we have that \(f^\dagger f = \id\); this gives us the following equalities:
  \begin{align}
    a^\dagger a + c^\dagger c &= \id_\sub{A} \label{eq:Bartha_4} \\
    a^\dagger b + c^\dagger d &= 0_\sub{AU} \label{eq:Bartha_5} \\
    b^\dagger a + d^\dagger c &= 0_\sub{UA} \label{eq:Bartha_6} \\
    b^\dagger b + d^\dagger d &= \id_\sub{U}. \label{eq:Bartha_7}
  \end{align}
  To prove that \(h\) is an isometry we show that \(h^\dagger h = \id\), where
  \begin{equation*}
    h^\dagger h = (a + b (\id - d)^{-1} c)^\dagger (a + b (\id - d)^{-1} c).
  \end{equation*}
  Using that \((g \circ f)^\dagger = f^\dagger \circ g^\dagger\) and \(((\id - g)^{-1})^\dagger = (\id - g^\dagger)^{-1}\), \(h^\dagger h\) is equal to:
  \begin{equation*}
    a^\dagger a + a^\dagger b (\id - d)^{-1} c + c^\dagger (\id - d^\dagger)^{-1} b^\dagger a + c^\dagger (\id - d^\dagger)^{-1} b^\dagger b (\id - d)^{-1} c.
  \end{equation*}
  Using~\eqref{eq:Bartha_5} and~\eqref{eq:Bartha_6}, this is equal to:
  \begin{equation*}
    a^\dagger a - c^\dagger d (\id - d)^{-1} c - c^\dagger (\id - d^\dagger)^{-1} d^\dagger c + c^\dagger (\id - d^\dagger)^{-1} b^\dagger b (\id - d)^{-1} c
  \end{equation*}
  and, thus, \(h^\dagger h = a^\dagger a + c^\dagger g c\) where
  \begin{equation*}
    g = (\id - d^\dagger)^{-1} b^\dagger b (\id - d)^{-1} - d (\id - d)^{-1} - (\id - d^\dagger)^{-1} d^\dagger.
  \end{equation*}
  Thanks to~\eqref{eq:Bartha_4}, to prove that \(h^\dagger h = \id\) it is sufficient to show that \(g = \id\). We show \(g = \id\) by applying \((\id - d^\dagger) \circ -\) and \(- \circ (\id - d)\) to both sides of the equation, then rearranging the terms to reach~\eqref{eq:Bartha_7}:
  \begin{align*}
    g = \id &\iff b^\dagger b - (\id - d^\dagger) d - d^\dagger (\id - d) = (\id - d^\dagger) (\id - d) \\
      &\iff b^\dagger b - d + d^\dagger d - d^\dagger + d^\dagger d = \id - d - d^\dagger + d^\dagger d \\
      &\iff b^\dagger b + d^\dagger d = \id
  \end{align*}
  Consequently, \(h^\dagger h = \id\) and \(h\) is an isometry, as claimed.
\end{proof}

\begin{corollary} \label{cor:Bartha_contraction}
  Let \(f \in \Contraction(A \oplus U, B \oplus U)\) where \(\id\!-\!f_\sub{UU}\) is invertible, then
  \begin{equation*}
    h = f_\sub{BA} + f_\sub{BU} (\id\!-\!f_\sub{UU})^{-1} f_\sub{UA}
  \end{equation*}
  is a contraction.
\end{corollary} \begin{proof}
  Proposition~\ref{prop:Halmos} established that for every contraction \(f\) the following linear map:
  \begin{equation*}
    g = \begin{pmatrix}
      -f^\dagger & D_f \\
      D_{f^\dagger} & f
    \end{pmatrix}
  \end{equation*}
  is an isometry, where \(D_f = \sqrt{\id - f^\dagger f}\).
  Next, we decompose \(g \colon B\!\oplus\!U\!\oplus\!A\!\oplus\!U \to A\!\oplus\!U\!\oplus\!B\!\oplus\!A\) as follows:
  \begin{equation*}
    g = \begin{pmatrix}
      a \colon B\!\oplus\!U\!\oplus\!A \to A\!\oplus\!U\!\oplus\!B \quad& b \colon U \to A\!\oplus\!U\!\oplus\!B \\
      c \colon B\!\oplus\!U\!\oplus\!A \to U & f_\sub{UU}
    \end{pmatrix}
  \end{equation*}
  Since \(g\) is an isometry and \(\id\!-\!f_\sub{UU}\) is invertible by assumption, the previous proposition implies that:
  \begin{equation*}
    h' = a + b (\id\!-\!f_\sub{UU})^{-1} c
  \end{equation*}
  is an isometry.
  It is straightforward to check that \(h = \pi_\sub{B} \circ h' \circ \iota_\sub{A}\) where the injection \(\iota_\sub{A}\) is an isometry and the projection \(\pi_\sub{B}\) is a contraction.
  Therefore, \(h\) is a contraction.
\end{proof}

\subsection{The execution formula in \(\FdContraction\)}
\label{sec:ex_FdContraction}

The goal of this subsection is to prove that \((\FdContraction,\oplus,\ex)\) is a (totally) traced monoidal category, where \(\ex\) is the execution formula.
To do so, the notions of hom-convergence UDC and \(U\)-tail vanishing morphisms introduced in Section~\ref{sec:UDC_limit} will be used.
In essence, the strategy is to derive the validity of \(\ex\) as a categorical trace via the kernel-image trace on \(\FdHilb\).
Once the result that \((\FdContraction,\oplus,\ex)\) is a totally traced category has been established, it follows that \(\FdIsometry\) and \(\FdUnitary\) are themselves totally traced; the latter results were previously proven by Bartha~\cite{Bartha} using the framework of matrix iteration theories.

Before the formal discussion begins, it is worth presenting a simple example that will illustrate some of the intricacies of these results.
Let the linear map \(H \colon \Cset \oplus \Cset \to \Cset \oplus \Cset\) be described by the matrix:
\begin{equation} \label{eq:Hadamard}
  H = \begin{pmatrix}
    h_{00} & h_{01} \\
    h_{10} & h_{11} 
  \end{pmatrix} 
   = \frac{1}{\sqrt{2}}\begin{pmatrix}
    1 & 1 \\
    1 & -1    
  \end{pmatrix}.
\end{equation}
Notice that \(H\) is unitary and hence a morphism in \(\FdContraction\).
We may evaluate the execution formula:
\begin{equation*}
  \ex^\Cset(H) = \Sigma \left( \{h_{00}\} \uplus \{h_{01} h_{11}^n h_{10}\}_\Nset \right) = \tfrac{1}{\sqrt{2}} + \sum_{n=0}^\infty \tfrac{1}{\sqrt{2}} \left(-\tfrac{1}{\sqrt{2}} \right)^n \tfrac{1}{\sqrt{2}}
\end{equation*}
and realise that \(\ex^\Cset(H) = 1\), which is a unitary \(\Cset \to \Cset\).
Similarly, \(\ex^\Cset(\sigma H \sigma) = 1\) as well and \(\ex^\Cset(\sigma H) = -1 = \ex^\Cset(H \sigma )\), where \(\sigma\) is the symmetric braiding in \((\FdContraction,\oplus,\{0\})\).
What is peculiar about this situation is that, if we were to sum all elements from \(\ex^\Cset(\sigma H \sigma)\) but \(h_{11}\), the result is larger than \(1\):
\begin{equation*}
  \Sigma \{h_{10} h_{00}^n h_{01}\}_\Nset = \sum_{n=0}^\infty \tfrac{1}{\sqrt{2}} \left(\tfrac{1}{\sqrt{2}}\right)^n \tfrac{1}{\sqrt{2}} = \tfrac{1}{2 - \sqrt{2}} > 1
\end{equation*}
and it is not until we add in \(h_{11} = - \tfrac{1}{\sqrt{2}}\) that the result becomes \(1\).
The key insight to be drawn from this example is that, even when the result of the execution formula is a contraction, its partial sums need not be contractions.

The following lemma establishes that every contraction in \(\Hilb\) satisfies one of the conditions in the definition of tail vanishing morphisms (Definition~\ref{def:tail_convergent}).

\begin{lemma} \label{lem:lim_last_term}
  Let \(f \colon A \oplus U \to B \oplus U\) be a contraction in the hom-convergence UDC \(\Hilb\) defined in terms of the strong operator topology (see Proposition~\ref{prop:Hilb_hom-convergence}).
  Then,
  \begin{equation*}
    \Lim{n \to \infty} f_\sub{BU} f_\sub{UU}^n \keq 0.
  \end{equation*}
\end{lemma} \begin{proof}
  Recall that this limit exists iff for every finite set of vectors \(S \subseteq U\) and every \(\epsilon > 0\) there is some \(n\) such that for all \(k \geq n\) we have that \(f_\sub{BU} f_\sub{UU}^k \in B^0_{S,\epsilon}\),\footnote{As established by Proposition~\ref{prop:base_limit} and implication~\eqref{eq:SOT_ball_in_ball}, it is sufficient to consider only the open sets from the base for \(\tau_\sub{\mathrm{SOT}}\) `centered' at \(0\), \ie{} of the form \(B^0_{S,\epsilon}\) for some \(S\) and \(\epsilon\).} \ie{} we need to show that:
  \begin{equation} \label{eq:lim_last_term_condition}
    k \geq n \implies \forall u \in S,\ \norm{f_\sub{BU} f_\sub{UU}^k(u)} < \epsilon
  \end{equation}
  To do so, the first step will be to show that the following is satisfied:
  \begin{equation} \label{eq:lim_last_term_aux}
    \Lim{n \in \Nset} \left( \norm{f_\sub{UU}^n(u)}^2 - \norm{f_\sub{UU}^{n+1}(u)}^2 \right) \keq 0.
  \end{equation}
  For every \(u \in U\), the collection of real numbers \(S_u = \{\norm{f_\sub{UU}^{n}(u)}^2\}_\Nset\) is bounded from below by \(0\) and, hence, by completeness of the total order \((\Rset,\leq)\), the infimum of the set \(S_u\) exists.
  Thus, for each \(\epsilon > 0\) there is some \(k \in \Nset\) such that
  \begin{equation*}
    \norm{f_\sub{UU}^{k}(u)}^2 < \inf(S_u) + \epsilon
  \end{equation*}
  since otherwise \(\inf(S_u)\! +\! \epsilon\) would be a lower bound to \(S_u\) greater than \(\inf(S_u)\), contradicting the definition of infimum.
  Considering that \(f_\sub{UU}\) is necessarily a contraction, every \(k' \geq k\) will satisfy
  \begin{equation*}
    \norm{f_\sub{UU}^{k'}(u)}^2 - \inf(S_u) \ \leq\ \norm{f_\sub{UU}^{k}(u)}^2 - \inf(S_u)
  \end{equation*}
  and, consequently, \(\norm{f_\sub{UU}^{k'}(u)}^2 - \inf(S_u) < \epsilon\) so that in the standard topology on \(\Rset\) we have that:
  \begin{equation*}
    \Lim{n \in \Nset}\, \norm{f_\sub{UU}^n(u)}^2 \keq \inf(S_u).
  \end{equation*}
  The previous argument is commonly known as the monotone convergence theorem.
  Moreover, it is trivial to check that \(\inf(S_u) = \inf(S_{f_\sub{UU}(u)})\) for every \(u \in U\) and due to addition (and negation) of real numbers being continuous:
  \begin{equation*} 
    \Lim{n \in \Nset} \left( \norm{f_\sub{UU}^n(u)}^2 - \norm{f_\sub{UU}^{n+1}(u)}^2 \right) \keq \inf(S_u) - \inf(S_{f_\sub{UU}(u)}) = 0
  \end{equation*}
  and, hence, equation~\eqref{eq:lim_last_term_aux} is satisfied.

  Proposition~\ref{prop:isometry_column} establishes that for all \(u \in U\):
  \begin{equation*}
    \norm{f_\sub{BU}(u)}^2 \leq \norm{u}^2 - \norm{f_\sub{UU}(u)}^2
  \end{equation*}
  and, in particular, for every \(n \in \Nset\), we have that \(f_\sub{UU}^n(u) \in U\) and
  \begin{equation*}
    \norm{f_\sub{BU}f_\sub{UU}^n(u)}^2 \leq \norm{f_\sub{UU}^n(u)}^2 - \norm{f_\sub{UU}^{n+1}(u)}^2.
  \end{equation*}
  Since the limit of the right hand side is known to be zero~\eqref{eq:lim_last_term_aux} and the sequence of real numbers \(\norm{f_\sub{BU}f_\sub{UU}^n(u)}^2\) is bounded from below by \(0\), it follows that:
  \begin{equation*}
    \Lim{n \to \infty}\, \norm{f_\sub{BU} f_\sub{UU}^n(u)}^2 \keq 0.
  \end{equation*}
  Therefore, for each \(u \in U\) and each \(\epsilon > 0\) we can find an \(n_u \in \Nset\) such that every \(k \geq n_u\) satisfies \(\norm{f_\sub{BU} f_\sub{UU}^k(u)}^2 < \epsilon^2\).
  For any given finite set of vectors \(S \subseteq U\) we can take \(n = \max{\{n_u \mid u \in S\}}\) so that~\eqref{eq:lim_last_term_condition} is satisfied, completing the proof.
\end{proof}

The previous lemma applies to contractions between arbitrary Hilbert spaces, including the case of infinite dimensions.
In contrast, the rest of the results in this subsection have only been established for the case of finite-dimensional spaces and Section~\ref{sec:open_questions} discusses the obstacles in the way of their generalisation.

\begin{lemma} \label{lem:FdContraction_ki}
  Every contraction \(f \colon A \oplus U \to B \oplus U\) in \(\FdHilb\) satisfies \((k,i) \Vdash \Tr^U(f)\) for some morphisms \(k \in \FdHilb(U,B)\) and \(i \in \FdHilb(A,U)\).
\end{lemma} \begin{proof}
  The proof will follow from Corollary~\ref{cor:FdHilb_ki} which establishes that it suffices to show that:
  \begin{equation*}
    \begin{aligned}
      \im(f_\sub{UA}) &\subseteq \im(\id\!-\!f_\sub{UU}) \\
      \ker(\id\!-\!f_\sub{UU}) &\subseteq \ker(f_\sub{BU}).
    \end{aligned}
  \end{equation*}
  Let \(v \in \ker(\id\!-\!f_\sub{UU})\); this implies \(\norm{f_\sub{UU}(v)} = \norm{v}\) and, necessarily, \(f_\sub{BU}(v) = 0\) for \(f\) to be a contraction.
  Consequently, \(\ker(\id\!-\!f_\sub{UU}) \subseteq \ker(f_\sub{BU})\) is satisfied.
  To prove the inclusion of images, we convert it to an statement on kernels as follows:
  \begin{align*}
    \im(f_\sub{UA}) \subseteq \im(\id\!-\!f_\sub{UU}) &\iff \im(\id\!-\!f_\sub{UU})^{\perp} \subseteq \im(f_\sub{UA})^{\perp} &&\text{(Proposition~\ref{prop:perp_inclusion})} \\
      &\iff \ker(\id\!-\!f_\sub{UU}^\dagger) \subseteq \ker(f_\sub{UA}^\dagger) &&\text{(Lemma~\ref{prop:ker_dagger_im_perp})}
  \end{align*}
  where the first step holds because \(\im(\id\!-\!f_\sub{UU})\) is a closed subspace --- in \(\FdHilb\) every subspace is closed.
  Considering that the adjoint of any contraction \(f\) is also a contraction,\footnote{Recall that every bounded linear map \(f \in \Hilb(A,B)\) satisfies that its adjoint \(f^\dagger\) has the same operator norm \(\opnorm{f} = \opnorm{f^\dagger}\). Thus, if \(f\) is a contraction then \(f^\dagger\) is also a contraction.} the inclusion \(\ker(\id\!-\!f_\sub{UU}^\dagger) \subseteq \ker(f_\sub{UA}^\dagger)\) follows from the same argument used for \(\ker(\id\!-\!f_\sub{UU}) \subseteq \ker(f_\sub{BU})\), completing the proof.
\end{proof}

The phrasing of the previous lemma purposely refers to the category \(\FdHilb\) instead of \(\FdContraction\) because, in principle, the morphisms \(k\) and \(i\) need not be contractions themselves.
For instance, in the case of the morphism \(H\) from~\eqref{eq:Hadamard} we have that \((k,i) \Vdash \Tr_\sub{\ki}^\Cset(\sigma H \sigma)\) where \(k = \tfrac{1}{\sqrt{2} - 1} = i\), but \(\abs{\tfrac{1}{\sqrt{2} - 1}} > 1\) and, hence, \(k\) and \(i\) are morphisms in \(\FdHilb\) but not in \(\FdContraction\).

The following two lemmas build towards Lemma~\ref{lem:FdContraction_tail_convergent} where it is established that every contraction \(f \colon A \oplus U \to B \oplus U\) between finite-dimensional Hilbert spaces is \(U\)-tail vanishing.

\begin{lemma} \label{lem:FdContraction_decomposition}
  Let \(f \in \FdContraction(H,H)\) and let \(f = f_0 \oplus f_1\) be its decomposition into a unitary \(f_0\) and a completely nonunitary contraction \(f_1 \colon H_1 \to H_1\) (see Theorem~\ref{thm:Contraction_decomposition}).
  Then,
  \begin{equation*}
    \opnorm{f_1^{\dim{H_1}}} < 1.
  \end{equation*}
\end{lemma} \begin{proof}
  If \(H_1 = \{0\}\) then \(\opnorm{f_1} = 0\) and the proof is trivial.
  Otherwise, recall from~\eqref{eq:Kf_alternative} that for any contraction \(g\) we can define a closed subspace
  \begin{equation*}
    K_g = \{v \in H \mid \norm{g(v)} = \norm{v}\}.
  \end{equation*}
  Then, for every \(k \in \Nset\) we may define a closed subspace of \(H_1\) as follows:
  \begin{equation*}
    V_k = \bigcap_{n=0}^k K_{f_1^n} = \{v \in H_1 \mid \forall n \leq k,\ \norm{f_1^n(v)} = \norm{v}\}.
  \end{equation*}
  Clearly, \(V_0 = H_1\) and \(V_{k+1} \subseteq V_k\) for all \(k \in \Nset\).
  Fix a value of \(k \in \Nset\) such that \(V_k \not= \{0\}\) and suppose --- for the sake of an argument by contradiction --- that \(V_k = V_{k+1}\).
  Then, \(v \in V_k\) implies \(f_1(v) \in V_k\) since for all \(n \leq k\):
  \begin{equation*}
    \norm{f_1^n(f_1(v))} = \norm{f_1^{n+1}(v)} = \norm{v} = \norm{f_1(v)}.
  \end{equation*}
  Moreover, we have that \(\norm{f_1(v)} = \norm{v}\), so the restriction \(f_1\vert_{V_k} \colon V_k \to V_k\) is an isometry.
  In finite dimensions, every isometry is surjective (and, hence, unitary), implying that for each \(v \in V_k\) we also have that \(f_1^\dagger(v) \in V_k\). Then, \(V_k\) reduces \(f_1\) according to Proposition~\ref{prop:reducing_space} but \(f_1\vert_{V_k}\) is unitary, contradicting that \(f_1\) is completely nonunitary.
  Consequently, it must be that for all \(k \in \Nset\):
  \begin{equation*}
    V_k \not= \{0\} \implies V_k \not = V_{k+1}
  \end{equation*}
  and, hence, \(V_{k+1}\) is a proper subspace of \(V_k\) and \(\dim{V_{k+1}} < \dim{V_k}\).
  Then, it is immediate that \(V_{\dim{H_1}} = \{0\}\) which, according to the definition of \(V_k\), implies that for all \(v \in H_1\) (other than \(0\)):
  \begin{equation*}
    \norm{f_1^{\dim{H_1}}(v)} < \norm{v}.
  \end{equation*}
  Since \(H_1\) is finite-dimensional it follows that \(\opnorm{f_1^{\dim{H_1}}} < 1\), as claimed.
\end{proof}

\begin{lemma} \label{lem:FdContraction_bound_to_sum}
  Let \(f \colon A \oplus U \to B \oplus U\) be a contraction in \(\FdHilb\) and decompose \(U = U_0 \oplus U_1\) according to the decomposition of \(f_\sub{UU}\) into its unitary and completely nonunitary components respectively.
  Then, there is a real number \(M\) such that:  
  \begin{equation*}
    \sum_{k=0}^\infty \opnorm{f_\sub{UU}^k f_\sub{UA}} \leq M
  \end{equation*}
  where
  \begin{equation*}
    M = \begin{cases}
      \dim{U_1} \cdot \left(1 - \opnorm{f_\sub{U_1U_1}^{\dim{U_1}}} \right)^{-1} &\ifc \dim{U_1} \not= 0 \\
      0 &\otherwise.
    \end{cases}
  \end{equation*}
\end{lemma} \begin{proof}
  According to Theorem~\ref{thm:Contraction_decomposition}, we may decompose \(f_\sub{UU}\) so that \(f\) may be expressed as:
  \begin{equation*}
    f = \begin{pmatrix}
      f_\sub{BA} & f_\sub{BU_0} & f_\sub{BU_1} \\
      f_\sub{U_0A} & f_\sub{U_0U_0} & f_\sub{U_0U_1} \\
      f_\sub{U_1A} & f_\sub{U_1U_0} & f_\sub{U_1U_1}
    \end{pmatrix}
  \end{equation*}
  where \(f_\sub{U_0U_0}\) is unitary.
  Then, Corollary~\ref{cor:isometry_column_zero} establishes that \(f_\sub{BU_0} = 0\) and \(f_\sub{U_1U_0} = 0\) and, similarly, \(f_\sub{U_0A} = 0\) and \(f_\sub{U_0U_1} = 0\) because \(f_\sub{U_0U_0}^\dagger\) is also an isometry.
  Consequently, for every \(k \in \Nset\) the following is satisfied:
  \begin{equation*}
    f_\sub{UU}^k f_\sub{UA} =
    \begin{pmatrix}
      f_\sub{U_0U_0}^k & 0 \\
      0 & f_\sub{U_1U_1}^k
    \end{pmatrix}
    \begin{pmatrix}
      0 \\ f_\sub{U_1A}
    \end{pmatrix}
    =
    \begin{pmatrix}
      0 \\ f_\sub{U_1U_1}^k f_\sub{U_1A}
    \end{pmatrix}
  \end{equation*}
  and, hence,
  \begin{equation*}
    \opnorm{f_\sub{UU}^k f_\sub{UA}} = \opnorm{f_\sub{U_1U_1}^k f_\sub{U_1A}}.
  \end{equation*}
  Therefore, if \(U_1 = \{0\}\) then \(\opnorm{f_\sub{UU}^k f_\sub{UA}} = 0\) for all \(k \in \Nset\) and it is trivial that \(\sum_{k=0}^\infty \opnorm{f_\sub{UU}^k f_\sub{UA}} = 0\), satisfying the claim in the case where \(\dim{U_1} = 0\).
  Otherwise, if \(U_1 \not= \{0\}\) we may use the well-known inequality \(\opnorm{h \circ g} \leq \opnorm{h} \cdot \opnorm{g}\) to obtain:
  \begin{equation*}
    \sum_{k=0}^\infty \opnorm{f_\sub{UU}^k f_\sub{UA}} \leq \left(\sum_{r=0}^{\dim{U_1}-1} \opnorm{f_\sub{U_1U_1}^r} \right) \left(\sum_{q=0}^\infty \opnorm{f_\sub{U_1U_1}^{q \cdot \dim{U_1}} f_\sub{U_1A}} \right)
  \end{equation*}
  and due to \(f_\sub{U_1U_1}\) and \(f_\sub{U_1A}\) being contractions, we conclude that:
  \begin{equation*}
    \sum_{k=0}^\infty \opnorm{f_\sub{UU}^k f_\sub{UA}} \leq \dim{U_1} \cdot \sum_{q=0}^\infty \left(\opnorm{f_\sub{U_1U_1}^{\dim{U_1}}} \right)^q
  \end{equation*}
  Finally, according to Lemma~\ref{lem:FdContraction_decomposition}, we have that \(\opnorm{f_\sub{U_1U_1}^{\dim{U_1}}} < 1\) so it is follows from the convergence criteria of the geometric series that:
  \begin{equation*}
    \sum_{k=0}^\infty \opnorm{f_\sub{UU}^k f_\sub{UA}} \ \leq \ \dim{U_1} \cdot \left(1 - \opnorm{f_\sub{U_1U_1}^{\dim{U_1}}} \right)^{-1}
  \end{equation*}
  thus proving the claim in the case where \(\dim{U_1} \not= 0\).
\end{proof}

\begin{lemma} \label{lem:FdContraction_tail_convergent}
  Every contraction \(f \colon A \oplus U \to B \oplus U\) in \(\FdHilb\) is \(U\)-tail vanishing.
\end{lemma} \begin{proof}
  Lemma~\ref{lem:lim_last_term} already established that \(\lim f_\sub{BU} f_\sub{UU}^n \keq 0\). It remains to check that
  \begin{equation} \label{eq:FdContraction_sum_of_gaps}
    \Lim{J \in \fpset{\Nset}} \Sigma \{f_\sub{BU} f_\sub{UU}^j f_\sub{UA}\}_{\tilde{J}} \keq 0
  \end{equation}
  is satisfied, where
  \begin{equation*}
    \tilde{J} = \{n \in \Nset \mid n < \max(J) \text{ and } n \not\in J\}.
  \end{equation*}
  This requires us to show that for all \(\epsilon > 0\), there is a \(J \in \fpset{\Nset}\) such that
  \begin{equation*}
    J \subseteq J' \implies \opnorm{\Sigma \{f_\sub{BU} f_\sub{UU}^j f_\sub{UA}\}_{\tilde{J'}}} < \epsilon.
  \end{equation*}
  If \(\tilde{J'} = \varnothing\) it is trivial to prove that \(\Sigma \{f_\sub{BU} f_\sub{UU}^j f_\sub{UA}\}_{\tilde{J'}} = 0\).
  Otherwise, if \(\tilde{J'} \not= \varnothing\) let \(m_J \in \Nset\) be the index of the first `gap' in \(J\), explicitly:
  \begin{equation*}
    m_J = \begin{cases}
      \min \tilde{J} &\ifc \tilde{J} \not= \varnothing \\
      \max{J}+1 &\text{otherwise.}
    \end{cases}
  \end{equation*}
  Notice that whenever \(J \subseteq J'\) we have that \(m_J \leq \min{} \tilde{J'}\) and, hence, \(j-m_J \geq 0\) for all \(j \in \tilde{J'}\), so that we may write:
  \begin{equation*}
    \opnorm{\Sigma \{f_\sub{BU} f_\sub{UU}^j f_\sub{UA}\}_{\tilde{J'}}} \leq \opnorm{f_\sub{BU} f_\sub{UU}^{m_J}} \cdot \opnorm{\Sigma \{f_\sub{UU}^{j-m_J} f_\sub{UA}\}_{\tilde{J'}}}.
  \end{equation*}
  Moreover,
  \begin{equation*}
    \opnorm{\Sigma \{f_\sub{UU}^{j-m_J} f_\sub{UA}\}_{\tilde{J'}}} \leq \sum_{k=0}^\infty \opnorm{f_\sub{UU}^k f_\sub{UA}}
  \end{equation*}
  due to the triangle inequality of norms along with norms being positive numbers.
  Recall from Lemma~\ref{lem:FdContraction_bound_to_sum} that there is some \(M \in \Rset\) such that:
  \begin{equation*}
    \sum_{k=0}^\infty \opnorm{f_\sub{UU}^k f_\sub{UA}} \leq M.
  \end{equation*}
  We conclude that
  \begin{equation*}
    J \subseteq J' \implies \opnorm{\Sigma \{f_\sub{BU} f_\sub{UU}^j f_\sub{UA}\}_{\tilde{J'}}} \leq \opnorm{f_\sub{BU} f_\sub{UU}^{m_J}} \cdot M.
  \end{equation*}
  If \(M = 0\) then \(f_\sub{UU}\) is unitary (see the proof of Lemma~\ref{lem:FdContraction_bound_to_sum} for further details) and hence \(f_\sub{UA} = 0\) according to Corollary~\ref{cor:isometry_column_zero} so that~\eqref{eq:FdContraction_sum_of_gaps} follows trivially.
  Otherwise, if \(M \not= 0\) we may use that \(\lim f_\sub{BU} f_\sub{UU}^n \keq 0\) (established in Lemma~\ref{lem:lim_last_term}) which implies that for every \(\epsilon > 0\) we may choose some \(\widehat{m} \in \Nset\) such that for all \(m' \in \Nset\):
  \begin{equation*}
    \widehat{m} \leq m' \implies \opnorm{f_\sub{BU} f_\sub{UU}^{m'}} < \frac{\epsilon}{M}.
  \end{equation*}
  Then, we simply need to choose some \(J \in \fpset{\Nset}\) such that \(m_J = \widehat{m}\) and it will follow that for all \(J' \in \fpset{\Nset}\):
  \begin{equation*}
    J \subseteq J' \implies \opnorm{\Sigma \{f_\sub{BU} f_\sub{UU}^j f_\sub{UA}\}_{\tilde{J'}}} < \epsilon
  \end{equation*}
  thus, verifying that~\eqref{eq:FdContraction_sum_of_gaps} is satisfied and proving that \(f\) is \(U\)-tail vanishing.
\end{proof}

Finally, we reach the main result of this section.

\begin{theorem} \label{thm:FdContraction_traced}
  The category \((\FdContraction,\oplus,\ex)\) is totally traced, where \(\ex\) is defined on any morphism \(f \colon A \oplus U \to B \oplus U\) as follows:
  \begin{equation*}
    \ex^U(f) = \Sigma \left( \{f_\sub{BA}\} \uplus \{f_\sub{BU} f_\sub{UU}^n f_\sub{UA}\}_\Nset \right).
  \end{equation*}
\end{theorem} \begin{proof}
  Corollary~\ref{cor:FdHilb_hom-convergence} established that \(\FdHilb\) is a hom-convergence UDC and, for every contraction \(f \colon A \oplus U \to B \oplus U\) in \(\FdHilb\), Lemma~\ref{lem:FdContraction_ki} shows that \((k,i) \Vdash \Tr_\sub{\ki}^U(f)\) whereas Lemma~\ref{lem:FdContraction_tail_convergent} shows that \(f\) is \(U\)-tail vanishing.
  According to Lemma~\ref{lem:ex_from_ki} this is sufficient for \(\ex^{U}(f)\) to be defined in \(\FdHilb\) for every contraction \(f\) and every \(U \in \FdHilb\), with \(\ex^{U}(f) = \Tr_\sub{\ki}^U(f)\).
  Since \(\Tr_\sub{\ki}\) is a partial trace in \(\FdHilb\) and the embedding functor \(\FdContraction \into \FdHilb\) is faithful and strict monoidal, Proposition~\ref{prop:induced_trace} implies that \((\FdContraction,\oplus,\widehat{\ex})\) is a partially traced category, where:
  \begin{equation*}
    \widehat{\ex}^U(f) = \begin{cases}
      \ex^U(f) &\ifc \ex^U(f) \text{ is a contraction} \\
      \undefined &\otherwise.
    \end{cases}
  \end{equation*}

  To show that \((\FdContraction,\oplus,\widehat{\ex})\) is a totally traced category we need to show that for every contraction \(f \in \FdContraction(A \oplus U,B \oplus U)\) the result of the execution formula \(\ex^U(f)\) is guaranteed to be a contraction. 
  We prove this by an inductive argument on the dimension of \(U\).
  If \(U = \Cset\) the component \(f_\sub{\Cset\Cset} \colon \Cset \to \Cset\) is just a scalar.
  \begin{itemize}
    \item If \(\abs{f_\sub{\Cset\Cset}} < 1\) then \(1\!-\!f_\sub{\Cset\Cset}\) is invertible and
    \begin{equation*}
      \Sigma \{f_\sub{\Cset\Cset}^n\}_\Nset = (1\!-\!f_\sub{\Cset\Cset})^{-1}.
    \end{equation*}
    Thanks to \(\FdHilb\) being \(\SCat{g}\)-enriched, composition distributes over \(\Sigma\) so that:
    \begin{equation*}
      \ex^\Cset(f) = f_\sub{BA} + f_\sub{B\Cset} (1\!-\!f_\sub{\Cset\Cset})^{-1} f_\sub{\Cset A}.
    \end{equation*}
    Then, according to Corollary~\ref{cor:Bartha_contraction}, \(\ex^\Cset(f)\) is a contraction.
    \item Otherwise, if \(\abs{f_\sub{\Cset\Cset}} = 1\) then Corollary~\ref{cor:isometry_column_zero} implies that \(f_\sub{B\Cset} = 0\) and, hence, \(\ex^{U}(f) = f_\sub{BA}\), which is necessarily a contraction.
  \end{itemize}
  If \(U = \Cset^{\oplus k}\) for some \(k > 1\) it follows from repeated applications of the above that \(\ex^\Cset(\ldots \ex^\Cset(f))\) is a contraction.
  Since \((\FdContraction,\oplus,\widehat{\ex})\) has been established to be partially traced, we may use vanishing II repeatedly to obtain:
  \begin{equation*}
    \widehat{\ex}^\Cset(\ldots \widehat{\ex}^\Cset(f)) \keq \widehat{\ex}^{\Cset^{\oplus k}}(f)
  \end{equation*}
  and, hence, \(\ex^{\Cset^{\oplus k}}(f)\) must be a contraction.
  For any arbitrary \(U\) there is a unitary \(\phi \colon U \to \Cset^{\oplus \dim{U}}\) due to \(U\) being finite-dimensional.
  Clearly, \(f' = (\id_B \oplus \phi) f (\id_A \oplus \phi^\dagger)\) is a contraction and it has already been established that \(\ex^{\Cset^{\oplus\dim{U}}}(f')\) is a contraction; then, using dinaturality we obtain:
  \begin{equation*}
    \widehat{\ex}^{\Cset^{\oplus\dim{U}}}(f') \keq \widehat{\ex}^U(f \circ (\id_A \oplus \phi^\dagger \phi)) = \widehat{\ex}^U(f)
  \end{equation*}
  implying that \(\ex^U(f)\) is a contraction for any arbitrary \(f \in \FdContraction(A\oplus U, B \oplus U)\).
  Consequently, \((\FdContraction,\oplus,\widehat{\ex})\) is a totally traced category, completing the proof.
\end{proof}

The theorem above along with Lemma~\ref{lem:ex_from_ki} establish that both the execution formula and the kernel-image trace coincide in \(\FdContraction\). This is in stark contrast to Proposition 3.14 from~\cite{Malherbe} where it was shown that no partial trace in \((\Vect,\oplus)\) --- nor in \((\FdHilb,\oplus)\) --- could simultaneously coincide with the execution formula and the kernel-image trace. This is no contradiction, though, as the proof of their result is based on a linear map presented as a counterexample, but such a linear map is not a contraction and hence their result need not hold for \(\FdContraction\). Notice that a similar situation occurred in Example~\ref{ex:FdHilb_counterexample} where a linear map in \(\FdHilb\) (but not in \(\FdContraction\)) violated vanishing II.
In retrospective, this is not surprising: notice that the characterisation of tail vanishing morphisms (Definition~\ref{def:tail_convergent}) immediately rules out any linear map whose operator norm is strictly larger than one. Since we heavily relied on the fact that morphisms in \(\FdContraction\) are tail vanishing (Lemma~\ref{lem:FdContraction_tail_convergent}) in our argument leading to the previous theorem (particularly, Lemma~\ref{lem:ex_from_ki}), our argument is exploiting a key difference between \(\FdHilb\) and \(\FdContraction\).

On the other hand, even though the categories \(\FdIsometry\) and \(\FdUnitary\) are not UDCs (because they lack quasi-projections), we may use their faithful strict monoidal embedding into \(\FdContraction\) to induce a total trace in them, calculated via the execution formula.

\begin{corollary}
  The category \((\FdIsometry,\oplus,\widehat{\ex})\) is totally traced, where:
  \begin{equation*}
    \widehat{\ex}^U(f) = \begin{cases}
      \ex^U(f) &\ifc \ex^U(f) \text{ is an isometry} \\
      \undefined &\otherwise.
    \end{cases}
  \end{equation*}
  Similarly, the category \((\FdUnitary,\oplus,\widehat{\ex})\) is totally traced, where \(\ex^U(f)\) is instead required to be unitary.
\end{corollary} \begin{proof}
  The embedding functor \(\FdIsometry \into \FdContraction\) is faithful and strict monoidal and, hence, \((\FdIsometry,\oplus,\widehat{\ex})\) is partially traced according to Proposition~\ref{prop:induced_trace}.
  Next, we must prove that for every isometry \(f \colon A \oplus U \to B \oplus U\) in \(\FdHilb\) \(\ex^U(f)\) is an isometry.
  This follows from a similar argument as the one used in the previous theorem:
  \begin{itemize}
    \item if \(U = \Cset\) then Proposition~\ref{prop:Bartha_isometry} can be used to show that \(\ex^\Cset(f)\) is an isometry;
    \item then, repeated application of this implies that \(\ex^\Cset(\ldots \ex^\Cset(f))\) is an isometry and, by vanishing II, \(\ex^{\Cset^{\oplus k}}(f)\) is an isometry as well;
    \item for every \(U\) there is a unitary \(\phi \colon U \to \Cset^{\oplus \dim{U}}\) due to \(U\) being finite-dimensional, so \(f' = (\id_B \oplus \phi) f (\id_A \oplus \phi^\dagger)\) is an isometry and we may use dinaturality to conclude that \(\ex^U(f)\) is an isometry.
  \end{itemize}
  Consequently, \(\widehat{\ex}\) is a total function and \((\FdIsometry,\oplus,\widehat{\ex})\) is a totally traced category.

  The case of \((\FdUnitary,\oplus,\widehat{\ex})\) can be proven with a similar argument, this time using that, in \(\FdHilb\), a morphism \(f \colon A \to B\) is unitary if and only if \(\dim(A) = \dim(B)\) and \(f\) is an isometry.
\end{proof}

Thus, it has been shown that the execution formula --- which is closely related to iterative loops --- is a well-defined trace in some important categories of quantum processes.
This was previously established in the case of \(\FdIsometry\) and \(\FdUnitary\) by Bartha~\cite{Bartha} using a different approach; a brief comparison between these two approaches appears in Section~\ref{sec:trace_rel_work}.
The proof of some of the results in this section heavily relied on the fact that the objects of \(\FdContraction\) are finite-dimensional Hilbert spaces; some open questions regarding the general case of \(\Contraction\) are discussed in Section~\ref{sec:open_questions}. Furthermore, the following remark discusses the prospects of proving that the execution formula is continuous.

\begin{remark} \label{rmk:continuous_trace}
  We know from Proposition~\ref{prop:preserve_limit} that a function
  \begin{equation*}
    \ex^U \colon  \FdUnitary(A \oplus U, B \oplus U) \to  \FdUnitary(A, B)
  \end{equation*}
  is continuous if and only if for every net \(\alpha \colon D \to \FdUnitary(A \oplus U, B \oplus U)\) the existence of \(\lim \alpha\) implies \(\lim \ex^U \circ \alpha \keq \ex^U(\lim \alpha)\). The topology considered here is the strong operator topology (see~\ref{def:SOT}). Moreover, recall that \(\ex\) is defined in terms of a \(\Sigma\) function that arises from the limit of certain net of finite partial sums (see~\ref{def:topological_Sigma}). Consequently, we may write:
  \begin{equation*}
    \ex^U (\Lim{d \in D} \alpha(d)) \keq \Lim{\fpset{\Nset}} \Lim{d \in D} \sigma_{\alpha(d)}
  \end{equation*}
  where \(\sigma_{\alpha(d)} \colon \fpset{\Nset} \to \FdUnitary(A, B)\) is the net of finite partial sums corresponding to the sum of paths of \(\alpha(d) \colon A \oplus U \to B \oplus U\). Notice that we know that \(\ex^U (\lim \alpha)\) is well defined since \(\lim \alpha\) is unitary and \(\ex\) is total in \(\FdUnitary\).
  Then, we would need to prove that the limit point
  \begin{equation*}
    \Lim{d \in D} (\ex^U  \alpha(d)) \keq \Lim{d \in D} \Lim{\fpset{\Nset}} \sigma_{\alpha(d)}
  \end{equation*}
  exists. Notice that the difference between these expressions is the ordering of the limits so it is not immediate that the existence of \(\ex^U(\lim \alpha)\) implies that of \(\lim \ex^U \circ \alpha\).
  To prove such a limit exchange holds we often need both that:
  \begin{itemize}
    \item for each \(d \in D\) the corresponding limit converges --- in this case, \(\Lim{J \in \fpset{\Nset}} \sigma_{\alpha(d)}\) does exist for all \(d \in D\) since \(\ex\) is total --- and that
    \item the net \(\sigma_{\alpha(-)}\) \emph{uniformly} converges with respect to \(D\) --- in this case meaning that for each open neighbourhood \(U\) we can choose \(d_U \in D\) such that for all \(d \geq d_U\) and all \(J \in \fpset{\Nset}\) it holds that \(\sigma_{\alpha(d)}(J) \in U\).
  \end{itemize}
  Such conditions are sufficient in the case of metric spaces and when limits are taken over sequences (see Moore-Osgood theorem, Theorem 2 from Chapter VII in~\cite{Graves}). However, it is not immediate whether the second condition is satisfied in our case, nor whether these two conditions would be sufficient in our general case.
\end{remark}

The following subsection compares \((\FdContraction,\oplus,\ex)\) and \((\CPTR,\oplus,\ex)\) and shows that their traces cannot be reconciled, hinting at a fundamental difference between quantum and classical iterative loops.

\subsection{Comparing the traces in \(\FdContraction\) and \(\CPTR\)}
\label{sec:comparing_traces}

In Section~\ref{sec:classical_loop}, \((\CPTR,\oplus,\ex)\) was given as an example of a totally traced category capturing the notion of classical iterative loops of quantum processes.
The previous section has established that \((\FdContraction,\oplus,\ex)\) is totally traced as well and, in this case, different execution paths may cancel out --- additive inverses of morphisms exist --- which is a fundamental characteristic of quantum control-flow.
We show that the canonical functor from \(\FdContraction\) to \(\CPTR\) is not a traced monoidal functor, not even when the definition of traced monoidal functor is relaxed.

\begin{definition} \label{def:env_functor}
  Let \(E \colon \FdContraction \to \CPTR\) be the faithful functor that maps Hilbert spaces \(A \in \FdContraction\) to \Cstar-algebras \(B(A) \in \CPTR\) and which maps morphisms \(f \in \FdContraction(A,B)\) to the CPTR map \(E(f)\) defined for all \(\rho \in B(A)\) as follows:
  \begin{equation*}
    E(f)(\rho) = f \circ \rho \circ f^\dagger.
  \end{equation*}
\end{definition}

Notice that this functor is only \emph{lax} monoidal: for every \(A,B \in \FdContraction\) there is a canonical morphism \(\theta_{A,B} \colon E(A) \oplus E(B) \to E(A \oplus B)\) given by:
\begin{equation*}
  \theta_{A,B}(\rho,\rho') = \begin{pmatrix}
    \rho & 0 \\ 0 & \rho'
  \end{pmatrix}
\end{equation*}
for all \((\rho,\rho') \in E(A) \oplus E(B)\), but \(\theta_{A,B}\) is not an isomorphism.
It is straightforward to check that \(\theta\) is a natural transformation and, indeed, \((E,\theta)\) is a lax monoidal functor.
Moreover, considering that \(\theta\) is injective, it has a left inverse \(\phi \colon E(A \oplus B) \to E(A) \oplus E(B)\) acting as follows on every \(\rho \in E(A \oplus B)\):
\begin{equation*}
  \phi \begin{pmatrix}
    \rho_\sub{AA} & \rho_\sub{AB} \\
    \rho_\sub{BA} & \rho_\sub{BB} 
  \end{pmatrix} = (\rho_\sub{AA}, \rho_\sub{BB}).
\end{equation*}
It is immediate to check that, indeed, \(\phi \circ \theta = \id\) but \(\theta \circ \phi \not= \id\).
Notice that \((E,\phi)\) is a colax monoidal functor; in a situation such as this, we may define a weaker notion of traced monoidal functors.

\begin{definition}
  Let \((\C,\oplus_\C,\Tr_\C)\) and \((\D,\oplus_\D,\Tr_\D)\) be totally traced categories and let \(F \colon \C \to \D\) be a functor such that \((F,\theta)\) is lax monoidal, \((F,\phi)\) is colax monoidal and \(\phi \circ \theta = \id\).
  Then, \(F\) is said to be a \emph{lax traced monoidal functor} if for all \(A,B \in \C\) there is a function \(\epsilon_{A,B} \colon \D(F(A),F(B)) \to \D(F(A),F(B))\) such that, for all \(f \in \C(A\oplus U,B \oplus U)\), it maps:
  \begin{equation*}
    F(\Tr_\C^U(f)) \mapsto \Tr_\D^{F(U)}(\phi \circ F(f) \circ \theta).
  \end{equation*}
  Alternatively, \(F\) is said to be a \emph{colax traced monoidal functor} if \(\epsilon_{A,B}\) maps in the other direction:
  \begin{equation*}
    \Tr_\D^{F(U)}(\phi \circ F(f) \circ \theta) \mapsto F(\Tr_\C^U(f)).
  \end{equation*}
\end{definition}

The motivation behind this definition is two-fold: on one hand, the requirement (from the standard definition of traced monoidal functor) that \(F\) is strong monoidal is lifted and, on the other hand, instead of requiring that \(F(\Tr_\C(f)) = \Tr_\D(\phi \circ F(f) \circ \theta)\) we only require that there is some arbitrary mapping between the results.
This is an extremely weak notion of traced monoidal functor and yet, the functor \(E \colon \FdContraction \to \CPTR\) is neither lax traced nor colax traced, making it apparent that the traces in \(\FdContraction\) and \(\CPTR\) are fundamentally different.
This result is shown by counterexample.

Let \(\sigma\) be the symmetric braiding in \(\FdContraction\), let \(i \colon \Cset \to \Cset\) be the morphism in \(\FdContraction\) that multiplies a complex number by the imaginary unit and let \(H \colon \Cset \oplus \Cset \to \Cset \oplus \Cset\) be the Hadamard map:
\begin{equation*}
  H = \frac{1}{\sqrt{2}} \begin{pmatrix}
    1 & 1 \\
    1 & -1
  \end{pmatrix}.
\end{equation*}
Fix the following morphisms of type \(\Cset \oplus \Cset \oplus \Cset \to \Cset \oplus \Cset \oplus \Cset\) in \(\FdContraction\):
\begin{equation*}
  f = \id_\Cset \oplus H \quad\quad g = \id_\Cset \oplus \sigma \quad\quad h = (\id_{\Cset \oplus \Cset} \oplus i) \circ g.
\end{equation*}
Applying \(\ex\) to these morphisms yields:
\begin{equation} \label{eq:counterexample_Contraction}
  \ex^\Cset(f) = \id_{\Cset \oplus \Cset} \quad\quad \ex^\Cset(g) = \id_{\Cset \oplus \Cset} \quad\quad \ex^\Cset(h) = \id_{\Cset} \oplus i
\end{equation}
where the case of \(\ex^\Cset(f)\) follows from the example discussed at the beginning of Section~\ref{sec:ex_FdContraction}, and those of \(\ex^\Cset(g)\) and \(\ex^\Cset(h)\) follow from the yanking axiom of traced categories.
Let \(f' \colon E(\Cset \oplus \Cset) \oplus E(\Cset) \to E(\Cset \oplus \Cset) \oplus E(\Cset)\) be
\begin{equation*}
  f' = \phi \circ E(f) \circ \theta
\end{equation*}
and define \(g'\) and \(h'\) in the same manner.
Let \(f'\) be decomposed as follows:
\begin{equation*}
  f' = \begin{pmatrix}
    f'_{00} \colon E(\Cset \oplus \Cset) \to E(\Cset \oplus \Cset) \quad& f'_{01} \colon E(\Cset) \to E(\Cset \oplus \Cset) \\
    f'_{10} \colon E(\Cset \oplus \Cset) \to E(\Cset) \quad& f'_{11} \colon E(\Cset) \to E(\Cset) \\
  \end{pmatrix}
\end{equation*}
and similarly for \(g'\) and \(h'\).
Let \(\ket{\psi} \in \Cset \oplus \Cset\) be \(\ket{\psi} = \tfrac{1}{\sqrt{2}}\ket{0} + \tfrac{1}{\sqrt{2}}\ket{1}\) and let \(\rho = \ketbra{\psi}{\psi}\); then:
\begin{equation*}
  f'_{00}(\rho) = \frac{1}{2}\begin{pmatrix}
    1 & \tfrac{1}{\sqrt{2}} \\ \tfrac{1}{\sqrt{2}} & \tfrac{1}{2}
  \end{pmatrix}
  \quad\quad
  g'_{00}(\rho) = \frac{1}{2}\begin{pmatrix}
    1 & 0 \\ 0 & 0
  \end{pmatrix}
  \quad\quad
  h'_{00}(\rho) = \frac{1}{2}\begin{pmatrix}
    1 & 0 \\ 0 & 0
  \end{pmatrix}
\end{equation*}
and, for all \(n \in \Nset\),
\begin{align*}
  f'_{01} (f'_{11})^n f'_{10}(\rho) &= \frac{1}{4} \cdot \frac{1}{2^{n+1}}\begin{pmatrix}
    0 & 0 \\ 0 & 1
  \end{pmatrix}
  \\
  g'_{01} (g'_{11})^n g'_{10}(\rho) &= \frac{1}{2} \cdot \frac{1}{2^{n+1}}\begin{pmatrix}
    0 & 0 \\ 0 & 1
  \end{pmatrix}
  \\
  h'_{01} (h'_{11})^n h'_{10}(\rho) &= \frac{1}{2} \cdot \frac{1}{2^{n+1}}\begin{pmatrix}
    0 & 0 \\ 0 & 1
  \end{pmatrix}
\end{align*}
so that, when all of these paths are summed up, we obtain:
\begin{equation} \label{eq:counterexample_CPTR}
  \begin{aligned}
    \ex^{E(\Cset)}(f')(\rho) &= \frac{1}{2}\begin{pmatrix}
      1 & \tfrac{1}{\sqrt{2}} \\ \tfrac{1}{\sqrt{2}} & 1
    \end{pmatrix}
    \\
    \ex^{E(\Cset)}(g')(\rho) &= \frac{1}{2}\begin{pmatrix}
      1 & 0 \\ 0 & 1
    \end{pmatrix}
    \\
    \ex^{E(\Cset)}(h')(\rho) &= \frac{1}{2}\begin{pmatrix}
      1 & 0 \\ 0 & 1
    \end{pmatrix}.
  \end{aligned}
\end{equation}

According to~\eqref{eq:counterexample_Contraction}, \(\ex^\Cset(f) = \ex^\Cset(g)\) and, hence, \(E(\ex^\Cset(f)) = E(\ex^\Cset(g))\); however, according to~\eqref{eq:counterexample_CPTR}, \(\ex^{E(\Cset)}(f') \not= \ex^{E(\Cset)}(g')\) since they are distinct at least at \(\rho\).
Consequently, there exists no function acting as \(E(\ex^\Cset(f)) \mapsto \ex^{E(\Cset)}(f')\) both for \(f\) and \(g\), preventing \(E\) from being a lax traced monoidal functor.
On the other hand, \(\ex^\Cset(g) \not= \ex^\Cset(h)\) and \(E(\ex^\Cset(g)) \not= E(\ex^\Cset(h))\) due to \(E\) being faithful, but \(\ex^{E(\Cset)}(g') = \ex^{E(\Cset)}(h')\) so that a function acting as \(\ex^{E(\Cset)}(g') \mapsto E(\ex^\Cset(g))\) both for \(g\) and \(h\) cannot exist, thus preventing \(E\) from being colax traced monoidal either.

\subsection{Open questions in \(\Contraction\)}
\label{sec:open_questions}

Sections~\ref{sec:UDC} and~\ref{sec:quantum_ex_trace} have been organised in such a way that all of the results whose proofs depend on finite-dimensionality appear in Section~\ref{sec:ex_FdContraction}.
In this section, we discuss the role that finite-dimensionality plays in these proofs; the objective being to shed some light on whether or not \((\Contraction,\oplus,\ex)\) is a partially traced category and, if it were, identify the obstacles in the way of proving so.
\begin{itemize}
  \item Lemma~\ref{lem:FdContraction_ki} establishes that every contraction in \(\FdHilb\) can be traced using the kernel-image trace. Recall that \((\Hilb,\oplus,\Tr_\ki)\) is a partially traced category (see Proposition~\ref{prop:SCat_ki_trace}) and, hence, we may ask whether every contraction in \(\Hilb\) can be traced.
  The proof of Lemma~\ref{lem:FdContraction_ki} relies on the fact that \(\im(\id\!-\!f_\sub{UU})\) is a closed subspace of \(U\), which is trivial to prove since it is a finite-dimensional space.
  Unfortunately, it is unclear whether \(\im(\id\!-\!f_\sub{UU})\) would be closed for an arbitrary morphism \(f \colon A \oplus U \to B \oplus U\) in \(\Contraction\).
  \item Lemma~\ref{lem:FdContraction_decomposition} establishes that every completely nonunitary \(f_1 \colon H_1 \to H_1\) in \(\FdContraction\) satisfies that \(f_1^{\dim{H_1}}\) is a \emph{strict} contraction. This claim cannot be properly phrased in \(\Contraction\) since \(\dim{H_1}\) may be infinite. Moreover, Example~\ref{ex:Hari} provides a contraction \(g\) whose component \(g_\sub{UU}\) is completely nonunitary but satisfies \(\opnorm{g_\sub{UU}^n} = 1\) for all \(n \in \Nset\), giving strong evidence that a result akin to Lemma~\ref{lem:FdContraction_decomposition} cannot be established in \(\Contraction\).
  \item Lemma~\ref{lem:FdContraction_bound_to_sum} establishes an upper bound for \(\sum \opnorm{f_\sub{UU}^k f_\sub{UA}}\). Its proof heavily relies on Lemma~\ref{lem:FdContraction_decomposition} and, hence, it is unlikely its proof can be generalised to \(\Contraction\).
  Furthermore, Example~\ref{ex:Hari} provides a contraction \(g\) for which \(\sum \opnorm{g_\sub{UU}^k g_\sub{UA}}\) cannot be bounded from above; thus, Lemma~\ref{lem:FdContraction_bound_to_sum} does not hold in \(\Contraction\).
\end{itemize}

These three lemmas were used in Section~\ref{sec:ex_FdContraction} to establish that every contraction in \(\FdHilb\) is tail vanishing (Lemma~\ref{lem:FdContraction_tail_convergent}).
Due to the above, the approach from Section~\ref{sec:ex_FdContraction} cannot be used to prove that every contraction in \(\Hilb\) is tail vanishing.
However, recall that Lemma~\ref{lem:lim_last_term} establishes that \(\lim f_\sub{BU} f_\sub{UU}^n \keq 0\) for every contraction \(f \colon A \oplus U \to B \oplus U\) in \(\Hilb\) and, hence, for \(f\) to be \(U\)-tail vanishing we only need to show that:
\begin{equation*}
  \Lim{J \in \fpset{\Nset}} \Sigma \{f_\sub{BU} f_\sub{UU}^j f_\sub{UA}\}_{\tilde{J}} \keq 0
\end{equation*}
where \(\tilde{J}\) is the set of `gaps' in \(J\) (see Definition~\ref{def:tail_convergent}).
Unlike the previous lemmas, this equation does seem to hold for the contraction provided in Example~\ref{ex:Hari} and, if this were the case for all contractions, it would imply that \((\Contraction,\oplus,\ex)\) is a partially traced category, as established in Lemma~\ref{lem:ex_from_ki}.
Consequently, there is some hope that \(\Contraction\) is partially traced with respect to the execution formula, but proving so would require an approach different from the one presented in Section~\ref{sec:ex_FdContraction}.

\begin{example} \label{ex:Hari}
  Let \(U\), \(A\) and \(B\) be separable Hilbert spaces of countably infinite dimension, each with an orthonormal basis given by \(\{u_i\}_{i \in \Nset}\), \(\{a_i\}_{i \in \Nset}\) and \(\{b_i\}_{i \in \Nset}\), respectively.
  Let \(g \colon A \oplus U \to B \oplus U\) be a morphism in \(\Contraction\) where, for all \(n \in \Nset\):
  \begin{equation*}
    g(u_n) = \vector{\sqrt{\tfrac{1}{2^n}}\cdot b_n}{\sqrt{1 - \tfrac{1}{2^n}}\cdot u_{n+1}}
  \end{equation*}
  and \(g(a_n) = u_n\).
  Even though for every \(v \in U\) there is some \(c \in \Rset\) strictly smaller than \(1\) such that \(\norm{g_\sub{UU}(v)} \leq c \cdot \norm{v}\), notice that \(\opnorm{g_\sub{UU}} = 1\); the same applies to \(\opnorm{g_\sub{UU}^n} = 1\) for every \(n \in \Nset\).
  We thank Hari Bercovici for providing us with this example.
\end{example} 

On another note, it may be argued that the trace in \((\FdContraction,\oplus,\ex)\) is not appropriate to capture iterative loops: we would expect that different execution paths --- \eg{} \(f_\sub{BA}\) and \(f_\sub{BU} f_\sub{UU}^2 f_\sub{UA}\) --- on the same input would produce their output at different time-steps. However, \(\FdContraction\) cannot capture the notion of time since this is an infinite domain and \(\FdContraction\) can only deal with finite-dimensional Hilbert spaces.
To this end, the following section introduces a particular category of quantum processes over time on which we the execution formula is a well-defined trace, thus providing a satisfactory formalisation of iterative loops in quantum computing.

\section{Coherent quantum processes over time}
\label{sec:LSI}

When the output of a process depends only on its input at the previous time-step, the time domain may be omitted from the mathematical formalism.
For instance, given any sequence of operations \(U_n \in \Unitary(H,H)\) and any initial state \(\ket{\psi_0} \in H\) the state after \(t\) steps can be easily represented as
\begin{equation*}
  \ket{\psi_t} = U_t \ldots U_2 U_1 \ket{\psi_0}.
\end{equation*}
Thus, the time domain may be abstracted away from the state and represented by the sequence of operations that remains to be applied.
In contrast, if we turn to open quantum walks (see Section~\ref{sec:QW}) we encounter situations as the one depicted in Figure~\ref{fig:time_domain_loop_bis}:\footnote{Figure~\ref{fig:time_domain_loop_bis} is a copy of Figure~\ref{fig:time_domain_loop}, reproduced here for the reader's convenience.} if the walker enters device \(f \colon A \oplus U \to B \oplus U\) through \(A\) at time-step \(t_0\), some of its amplitude will be found in \(B\) at \(t_0+1\) but some will traverse the loop and eventually leave at a later time-step.
Consequently, the output is not only localised at \(t_0+1\) but instead its amplitude spreads throughout \([t_0+1,\infty)\). 
This situation may be captured by functions describing the input and output over time.
For instance, in the case of Figure~\ref{fig:time_domain_loop_bis}, we may use functions \(a \colon \Zset \to \Cset\) and \(b \colon \Zset \to \Cset\) to describe the input (state on \(A\)) and output (state on \(B\)) respectively:
\begin{equation*}
  a(t) = \begin{cases}
    \psi &\ifc t = t_0 \\
    0 &\otherwise
  \end{cases}
  \quad\quad
  b(t) = \begin{cases}
    0 &\ifc t \leq t_0 \\
    f_\sub{BA}(\psi) &\ifc t = t_0+1 \\
    f_\sub{BU} f_\sub{UU}^{t-t_0-2} f_\sub{UA}(\psi) &\otherwise.
  \end{cases}
\end{equation*}

\begin{figure}
  \centering
  \input{Figures/1/time_domain_loop}
  \caption{A state \(\ket{\psi}\) is given as input to a quantum iterative process; two time-steps later, the state is in a superposition \(\ket{\psi_\sub{ba}} + \ket{\psi_\sub{bua}} + \ket{\psi_\sub{uua}}\) where \(\ket{\psi_\sub{ba}} = f_\sub{BA}\ket{\psi}\), \(\ket{\psi_\sub{bua}} = f_\sub{BU}f_\sub{UA}\ket{\psi}\) and \(\ket{\psi_\sub{uua}} = f_\sub{UU}f_\sub{UA}\ket{\psi}\). The term \(\ket{\psi_\sub{ba}}\) `leaves' the system a step earlier than \(\ket{\psi_\sub{bua}}\).}
  \label{fig:time_domain_loop_bis}
\end{figure}

Then, \(\abs{b(t)}^2\) determines the probability of finding the walker in \(B\) if we were to measure at time-step \(t\).
Thus, we would require that the square of the amplitudes sum up to one:
\begin{equation*}
  \sum_{t=-\infty}^\infty \abs{b(t)}^2 = 1
\end{equation*}
but, in order to work with vector spaces, we instead consider states up to normalisation and, hence, impose that functions describing a quantum state must be square-summable:
\begin{equation*}
  \sum_{t=-\infty}^\infty \abs{a(t)}^2 < \infty.
\end{equation*}
We can now define a Hilbert space of square-summable functions of type \(\Zset \to \Cset\).

\begin{definition} \label{def:ltwo}
  Let \GLS{\(\ltwo\)}{ltwo} be the complex Hilbert space defined as follows:\footnote{This is slightly different from the usual \(\ltwo\) space, which is defined as the space of square-summable sequences. In essence, instead of a vector being a sequence (a function \(\Nset \to \Cset\)), we use \(\Zset\) as the domain to describe both directions of time.}
  \begin{itemize}
    \item vectors are functions \(\psi \colon \Zset \to \Cset\) such that
    \begin{equation*}
      \sum_{t=-\infty}^\infty \abs{\psi(t)}^2 < \infty
    \end{equation*}
    \item addition and scalar multiplication are defined pointwise, \ie{} for \(\psi,\phi \in \ltwo\) and \(\alpha \in \Cset\):
    \begin{equation*}
      (\psi + \phi)(t) = \psi(t) + \phi(t) \quad\quad (\alpha \psi)(t) = \alpha \cdot \psi(t)
    \end{equation*}
    \item inner product is defined as follows:
    \begin{equation*}
      \braket{\phi}{\psi} = \sum_{t=-\infty}^\infty \phi(t)^* \cdot \psi(t)
    \end{equation*}
  \end{itemize}
\end{definition}

Proving that the inner product given above is well-defined is a common exercise in courses on functional analysis; it can be shown using the Cauchy-Schwarz inequality in \(\Rset^n\) along with absolute convergence in \(\Cset\).
We should also check that \(\ltwo\) is indeed Cauchy complete but, once again, this is a well-known fact and it is omitted here for the sake of brevity.
The goal of this section is to discuss linear maps \(\ltwo \to \ltwo\) that may be interpreted as operators that act on quantum states over the time domain.
When discussing these, the convolution of functions in \(\ltwo\) will take a fundamental role.

\begin{definition} \label{def:l2_conv}
  For any two functions \(f,g \in \ltwo\) we define \(g \conv f \colon \Zset \to \Cset\) acting on every \(t \in \Zset\) as follows:\footnote{It can be shown that \(f,g \in \ltwo\) implies that \(g \conv f\) is well-defined; the argument is similar to the one used to prove that the inner product in \(\ltwo\) is well-defined, reproducing it for \((g \conv f)(t)\) for every \(t \in \Zset\) and using that \(S_\tau\) is an isometry. However, the function \(g \conv f\) need not be square-summable.}
  \begin{equation*}
    (g \conv f)(t) = \sum_{\tau = -\infty}^\infty g(t - \tau) \cdot f(\tau).
  \end{equation*}
  We refer to \(g \conv f\) as the \gls{convolution} of \(f\) and \(g\).
\end{definition}

It is straightforward to check that convolution is associative and commutative.
Moreover, let \(\delta \in \ltwo\) be the function:
\begin{equation} \label{eq:l2_delta}
  \delta(t) = \begin{cases}
    1 &\ifc t = 0 \\
    0 &\otherwise
  \end{cases}
\end{equation}
then \(\delta \conv f = f\) for any \(f \in \ltwo\) and, due to commutativity, \(f \conv \delta = f\) as well.
Moreover, convolution distributes over pointwise addition of functions, as established below.

\begin{proposition} \label{prop:conv_and_addition}
  Let \(f,g,h \in \ltwo\) and let \((f+g)(t) = f(t) + g(t)\), then:
  \begin{align*}
    h \conv (f+g) &= (h \conv f) + (h \conv g) \\
    (f+g) \conv h &= (f \conv h) + (g \conv h)
  \end{align*}
\end{proposition} \begin{proof}
  By explicit calculation, for all \(t \in \Zset\):
  \begin{align*}
    (h \conv (f+g))(t) &= \sum_{\tau = -\infty}^\infty h(t - \tau) \cdot (f+g)(\tau) \\
      &= \sum_{\tau = -\infty}^\infty h(t - \tau) \cdot (f(\tau) + g(\tau)) \\
      &= \left( \sum_{\tau = -\infty}^\infty h(t - \tau) \cdot f(\tau) \right) + \left( \sum_{\tau = -\infty}^\infty h(t - \tau) \cdot g(\tau) \right) \\
      &= (h \conv f)(t) + (h \conv g)(t) = ((h \conv f) + (h \conv g))(t).
  \end{align*}
  This proves the first identity from the claim; the second one follows from it due to commutativity of convolution.
\end{proof}

The \(\ltwo\) space lets us formalise the notion of quantum states over the discrete-time domain.
Similarly, the definition below describes a Hilbert space \(\Ltwo\) of square-integrable functions, formalising quantum states over continuous time.
An appropriate notion of convolution of functions in \(\Ltwo\) is also provided.

\begin{definition} \label{def:Ltwo}
  Let \GLS{\(\Ltwo\)}{Ltwo} be the complex Hilbert space defined as follows:
  \begin{itemize}
    \item vectors are functions \(\psi \colon \Rset \to \Cset\) such that
    \begin{equation*}
      \int_{-\infty}^\infty \abs{\psi(t)}^2 \,dt < \infty
    \end{equation*}
    \item addition and scalar multiplication are defined pointwise, \ie{} for \(\psi,\phi \in \Ltwo\) and \(\alpha \in \Cset\):
    \begin{equation*}
      (\psi + \phi)(t) = \psi(t) + \phi(t) \quad\quad (\alpha \psi)(t) = \alpha \cdot \psi(t)
    \end{equation*}
    \item inner product is defined as:
    \begin{equation*}
      \braket{\phi}{\psi} = \int_{-\infty}^\infty \phi(t)^* \cdot \psi(t) \,dt.
    \end{equation*}
  \end{itemize}
\end{definition}

\begin{definition} \label{def:L2_conv}
  For any two functions \(f,g \in \Ltwo\) we define \(g \conv f \colon \Rset \to \Cset\) acting on every \(t \in \Rset\) as follows:
  \begin{equation*}
    (g \conv f)(t) = \int_{-\infty}^\infty g(t - \tau) \cdot f(\tau) \,d\tau.
  \end{equation*}
\end{definition}

\begin{remark} \label{rmk:Dirac_delta}
  Convolution in \(\Ltwo\) is associative and commutative, but it has no unit. 
  Physicists often circumvent this by defining \(\delta\) to be the Dirac delta: a map that is zero everywhere except at \(0\), and whose square-integral is equal to \(1\); but, evidently, no such a function exists.
  More formally, we may consider a sequence of functions \(\delta_n\) in \(\Ltwo\) whose square-integral equals \(1\) and whose support is centered at zero, becoming narrower as \(n\) tends to infinity.
  Such a sequence does not converge, but we can nonetheless choose an \(n\) large enough so that the support of \(\delta_n\) is sufficiently narrow for the approximation \(\delta_n \conv f \approx f\) to be satisfactory for practical purposes.
  In what follows, \(\delta_\bullet \in \Ltwo\) will be used to denote a choice of \(\delta_n\) where \(n\) is some arbitrary large number.
\end{remark}

\subsection{The discrete-time Fourier transform (\emph{Preamble})}

A fundamental result in physics and engineering establishes that the convolution of two functions in \(\Ltwo\) can be calculated by pointwise multiplication of the Fourier transformed functions.\footnote{See Section 9.2 from~\cite{Aubin} to learn about the theory around the Fourier transform and Proposition 9.2.3 for this particular result).}
The main goal of this subsection is to introduce its \(\ltwo\) version, known as the discrete-time Fourier transform.
Both the discrete-time Fourier transform and its convolution theorem (Theorem~\ref{thm:l2_convolution}) will be key ingredients in subsequent discussions of quantum processes over discrete time.

\begin{definition} \label{def:DTFT}
  The \gls{discrete-time Fourier transform (DTFT)} is a linear map \(\cF_d\) that sends any \(f \in \ltwo\) to the function \(\cF_d[f] \colon \Rset \to \Cset\) defined below:
  \begin{equation*}
    \cF_d[f](\omega) = \sum_{t = -\infty}^\infty f(t) \, e^{-i\omega t}.
  \end{equation*}
\end{definition}

It is well-known (although not immediate) that if \(f \not = 0\) then \(\cF_d[f]\) is a periodic function with period \(2\pi\); consequently, \(\cF_d[f]\) cannot possibly be square-integrable, \(\cF_d[f] \not\in \Ltwo\).
Nevertheless, the DTFT has a left inverse obtained by integrating over a full \(2\pi\) period:
\begin{equation} \label{eq:DTFT_linv}
  \cF_d^l[g] = \frac{1}{2\pi} \int_{0}^{2\pi} g(\omega) \, e^{i\omega t} \,d\omega
\end{equation}
so that \((\cF_d^l \circ \cF_d)[f] = f\) for any \(f \in \ltwo\).
This can be derived from the following identity, which is satisfied for integers \(n \in \Zset\):
\begin{equation} \label{eq:delta_as_integral}
  \int_0^{2\pi} e^{i\omega n} \,d\omega = 2\pi \cdot \delta(n)
\end{equation}
where, for \(n=0\) the integral is just \(\int_0^{2\pi} d\omega = 2\pi\) and for any other integer \(n \in \Zset\), the interval of integration contains \(\abs{n}\) full periods of the function \(e^{i\omega n}\) and, hence, for each \(\omega \in (0,\pi)\) the value of \(e^{i\omega}\) cancels out with \(e^{i(\omega + \pi)}\) and the integral evaluates to \(0\).
Notice that the existence of this left inverse implies that \(\cF_d\) is an injective map.

\begin{remark} \label{rmk:FT_plane_waves}
  When functions \(\psi \in \ltwo\) are interpreted as (non-normalised) quantum states whose amplitude spreads over discrete time, the DTFT is understood to yield a function \(\cF_d[\psi] \colon \Rset \to \Cset\) describing the same quantum state over the frequency domain.
  This interpretation becomes clearer after examining the identity \(\psi = (\cF_d^l \circ \cF_d)[\psi]\):
  \begin{equation*}
    \psi(t) = \frac{1}{2\pi} \int_{0}^{2\pi} \cF_d[\psi](\omega) \cdot \gamma_\omega(t) \,d\omega
  \end{equation*}
  where \(\gamma_\omega(t) = e^{i\omega t}\) is the \emph{plane wave} of angular frequency \(\omega\).
  Consequently, the state \(\psi\) is described as a superposition (\ie{} a linear combination) of plane waves \(\gamma_\omega\) for all frequencies \(\omega \in \Rset\), where each complex number \(\cF[\psi](\omega)\) determines the amplitude (and phase) of each plane wave.
\end{remark}

The convolution theorem for functions in \(\ltwo\) is presented below.

\begin{theorem} \label{thm:l2_convolution}
  For any \(\omega \in \Rset\) and any two functions \(f,g \in \ltwo\):
  \begin{equation*}
    \cF_d[g \conv f](\omega) = \cF_d[g](\omega) \cdot \cF_d[f](\omega).
  \end{equation*}
\end{theorem} \begin{proof}
  The theorem follows from a simple rearrangement of terms:
  \begin{align*}
    \cF_d[g](\omega)& \cdot \cF_d[f](\omega) = && \\
    &= \left( \sum_{t' = -\infty}^\infty g(t')\, e^{-i\omega t'} \right) \left( \sum_{\tau = -\infty}^\infty f(\tau)\, e^{-i\omega\tau} \right) && \text{(def\@. \(\cF_d\))} \\
    &= \sum_{t' = -\infty}^\infty \sum_{\tau = -\infty}^\infty g(t')\, f(\tau)\, e^{-i\omega (t' + \tau)} &&\text{(rearrangement)} \\
    &= \sum_{t = -\infty}^\infty \sum_{\tau = -\infty}^\infty g(t-\tau)\, f(\tau)\, e^{-i\omega t} &&\text{(substitute \(t' = t - \tau\))} \\
    &= \sum_{t = -\infty}^\infty (g \conv f)(t)\, e^{-i\omega t} &&\text{(def\@. convolution)} \\
    &= \cF[g \conv f](\omega). && \text{(def\@. \(\cF_d\))} \qedhere
  \end{align*}
\end{proof}

The Fourier transform \(\cF \colon \Ltwo \to \Ltwo\) is unitary, implying that \(\norm{\cF[\psi]} = \norm{\psi}\) for any function \(\psi \in \Ltwo\) (see Theorem 9.2.1 from~\cite{Aubin}).
In contrast, the function \(\cF_d[\psi]\) is not even square-integrable.
Nevertheless, a similar result can be obtained for the DTFT by integrating the values of \(\abs{\cF_d[\psi](\omega)}^2\) over a full period, as established below.

\begin{proposition} \label{prop:norm_from_Fourier_entry}
  For any \(\psi \in \ltwo\) the following identity is satisfied:
  \begin{equation*}
    \norm{\psi}^2 = \frac{1}{2\pi} \int_0^{2\pi} \abs{\cF_d[\psi](\omega)}^2 \,d\omega.
  \end{equation*}
\end{proposition} \begin{proof}
  The norm of any \(\psi \in \ltwo\) is:
  \begin{equation*}
    \norm{\psi} = \sqrt{ \sum_{t=-\infty}^\infty \psi(t)^* \cdot \psi(t) }
  \end{equation*}
  as determined by the inner product in \(\ltwo\).
  Then:
  \begin{align*}
    \frac{1}{2\pi} &\int_0^{2\pi} \abs{\cF_d[\psi](\omega)}^2 \,d\omega = && \\
     &= \frac{1}{2\pi} \int_0^{2\pi} \abs{\sum_{t=-\infty}^\infty \psi(t) \, e^{-i\omega t}}^2 \,d\omega &&\text{(def.\@ \(\cF_d\))} \\
     &= \frac{1}{2\pi} \int_0^{2\pi} \left(\sum_{\tau=-\infty}^\infty \psi(\tau)^*\, \, e^{i\omega\tau}\right)\left(\sum_{t=-\infty}^\infty \psi(t) \, e^{-i\omega t}\right) \,d\omega &&\text{(def.\@ \(\abs{\alpha}^2\) for \(\alpha \in \Cset\))} \\
     &= \frac{1}{2\pi} \sum_{t=-\infty}^\infty \sum_{\tau=-\infty}^\infty \psi(\tau)^* \cdot \psi(t) \cdot\! \int_0^{2\pi} e^{i\omega (\tau - t)} \,d\omega &&\text{(rearrange)} \\
     &= \sum_{t=-\infty}^\infty \sum_{\tau=-\infty}^\infty \psi(\tau)^* \cdot \psi(t) \cdot \delta(\tau - t) &&\text{(identity~\eqref{eq:delta_as_integral})} \\
     &= \sum_{t=-\infty}^\infty  \psi(t)^* \cdot \psi(t) \ =\ \norm{\psi}^2. &&\text{(def.\@ \(\delta\) and \(\norm{\psi}\))} \qedhere
  \end{align*}
\end{proof}

\subsection{Linear shift invariant maps}

Physical processes that act on states over an infinite domain are often \emph{shift invariant}. 
The input-output behaviour of a shift invariant device \(D \colon \ltwo \to \ltwo\) acting on quantum states is independent of the time-step the input reaches the device.
To formalise this notion we define a \emph{shift operator} \(S_\tau \colon \ltwo \to \ltwo\) for all \(\tau \in \Zset\):
\begin{equation*}
  S_\tau[\psi](t) = \psi(t-\tau)
\end{equation*}
which captures the passing of an amount \(\tau\) of time.
Then, a \gls{linear shift invariant (LSI)} map \(D \colon \ltwo \to \ltwo\) is a linear map that commutes with every shift operator:
\begin{equation*}
  S_\tau D = D S_\tau
\end{equation*}
for all \(\tau \in \Zset\).
Crucially, any operator \(\ltwo \to \ltwo\) that is both linear and shift invariant can be fully characterised by a function in \(\ltwo\), as shown below.

\begin{proposition} \label{prop:l2_lsi}
  For any LSI map \(D \colon \ltwo \to \ltwo\), there is a unique function \(\chi^D \in \ltwo\) such that for all \(\psi \in \ltwo\):
  \begin{equation*}
    D[\psi] = \chi^D \conv \psi.
  \end{equation*}
  Such a function is precisely \(\chi^D = D[\delta]\) where \(\delta\) is the unit of the convolution~\eqref{eq:l2_delta}. 
\end{proposition} \begin{proof}
  Considering \(\delta\) is the unit of convolution, we have that:
  \begin{equation*}
    \psi(t) = (\delta \conv \psi)(t) = \sum_{\tau = -\infty}^\infty \delta(t-\tau) \cdot \psi(\tau) = \sum_{\tau = -\infty}^\infty S_\tau[\delta](t) \cdot \psi(\tau).
  \end{equation*}
  Then, due to \(D\) being linear and shift invariant we can compute \(D[\psi]\) as follows:
  \begin{equation*}
    D[\psi] = \sum_{\tau = -\infty}^\infty D[S_\tau[\delta]] \cdot \psi(\tau) = \sum_{\tau = -\infty}^\infty S_\tau[D[\delta]] \cdot \psi(\tau) = D[\delta] \conv \psi.
  \end{equation*}
  Therefore, if we let \(\chi^D = D[\delta]\) then \(D[\psi] = \chi^D \conv \psi\) as claimed.
  To prove uniqueness, assume there is a different function \(\xi^D\) such that \(D[\psi] = \xi^D \conv \psi\) for all \(\psi \in \ltwo\). 
  It then follows that \(D[\delta] = \xi^D \conv \delta = \xi^D\) due to \(\delta\) being the unit of convolution and, hence, we conclude that \(\xi^D = \chi^D\).
\end{proof}

Due to convolution being associative, it is apparent that composition of LSI maps \(D,D' \colon \ltwo \to \ltwo\) corresponds to convolution of their characteristic functions:
\begin{equation*}
  (D' \circ D)[\psi] = (\chi^{D'} \conv \chi^D) \conv \psi.
\end{equation*}
Notice that \(D' \circ D\) is again a LSI map \(\ltwo \to \ltwo\) and, hence, \(\chi^{D'} \conv \chi^D\) is necessarily in \(\ltwo\).
A similar result holds for LSI maps \(D \colon \Ltwo \to \Ltwo\); however, in this case the characterisation of \(D\) as a function in \(\Ltwo\) is only an approximation. 
This is due to the lack of a unit of convolution in \(\Ltwo\) (see Remark~\ref{rmk:Dirac_delta}).

\begin{proposition} \label{prop:L2_lsi}
  For any LSI map \(D \colon \Ltwo \to \Ltwo\), let \(\chi^D \in \Ltwo\) be \(\chi^D = D[\delta_\bullet]\); then, for all \(\psi \in \Ltwo\):
  \begin{equation*}
    D[\psi] \approx \chi^D \conv \psi.
  \end{equation*}
\end{proposition} \begin{proof}
  The claim follows from the same argument as that for \(\ltwo\) (Proposition~\ref{prop:l2_lsi}), provided that we are dealing with approximations, as discussed in Remark~\ref{rmk:Dirac_delta}.
\end{proof}

For any finite set \(A\), the vector space \(\oplus_A\,\ltwo\) over the field \(\Cset\) has collections
\begin{equation*}
  \Psi = \{\psi_a \in \ltwo\}_{a \in A}
\end{equation*}
as vectors, with addition and scalar multiplication defined index-wise:
\begin{equation*}
  \Phi + \Psi = \{\phi_a + \psi_a\}_{a \in A}
\end{equation*}
\begin{equation*}
  \alpha \Psi = \{\alpha\psi_a\}_{a \in A}
\end{equation*}
The inner product in \(\oplus_A\,\ltwo\) is defined in terms of the inner product in \(\ltwo\):
\begin{equation*}
  \braket{\Phi}{\Psi} = \sum_{a \in A} \braket{\phi_a}{\psi_a}
\end{equation*}
Considering that \(A\) is a finite set, it follows that \(\oplus_A\,\ltwo\) is complete and, hence, a Hilbert space.
For each \(a \in A\) there is a projection \(\pi_a \colon \oplus_A\,\ltwo \to \ltwo\) that yields the component at index \(a\), \(\pi_a[\Psi] = \psi_a\), and an injection \(\iota_a \colon \ltwo \to \oplus_A\,\ltwo\) that maps \(\phi \in \ltwo\) to a collection \(\{\phi_a\}_{a \in A}\) where \(\phi_a = \phi\) and all other \(\phi_{a'}\) are the constant zero function.
It is trivial to check that for all \(\Psi \in \oplus_A\,\ltwo\) the following identity is satisfied:
\begin{equation} \label{eq:lsi_iota_pi_sum}
  \Psi = \sum_{a \in A} (\iota_a \pi_a)[\Psi] = \sum_{a \in A} \iota_a[\psi_a].
\end{equation}

We can then consider linear maps \(D \colon \oplus_A\,\ltwo \to \oplus_B\,\ltwo\) where both \(A\) and \(B\) are finite sets, and say that these are shift invariant iff for all \((b,a) \in B \times A\) the map \(\pi_b D \iota_a \colon \ltwo \to \ltwo\) is shift invariant.
Once again, convolution lets us describe these maps in terms of functions in \(\ltwo\), as established in the following proposition.

\begin{proposition} \label{prop:oplusl2_lsi}
  Let \(A\) and \(B\) be two arbitrary finite sets. For any LSI map \(D \colon \oplus_A\,\ltwo \to \oplus_B\,\ltwo\) and for each \((b,a) \in B \times A\) let \(\chi^D_{ba} \in \ltwo\) be
  \begin{equation*}
    \chi^D_{ba} = (\pi_b D \iota_a)[\delta]
  \end{equation*}
  Then, for all \(\Psi \in \oplus_A\,\ltwo\):
  \begin{equation*}
    D[\Psi] = \left\{\sum_{a \in A} \chi^D_{ba} \conv \psi_a \right\}_{b \in B}.
  \end{equation*}
\end{proposition} \begin{proof}
  Since for each \((b,a) \in B \times A\) the map \(\pi_b D \iota_a\) is time shift invariant map, we may use Proposition~\ref{prop:l2_lsi} to compute the \(B\)-indexed family \(D[\Psi]\) in terms of convolution with the each of characteristic functions \(\chi^D_{ba}\):
  \begin{align*}
    D[\Psi] &=  \sum_{b \in B} (\iota_b \pi_b) [D[\Psi]] &&\text{(equation~\eqref{eq:lsi_iota_pi_sum})}\\
      &= \sum_{b \in B} (\iota_b \pi_b) \left[D\left[\sum_{a \in A} \iota_a[\psi_a]\right]\right] &&\text{(equation~\eqref{eq:lsi_iota_pi_sum})}\\
      &= \sum_{b \in B} \iota_b \left[\sum_{a \in A} (\pi_b D \iota_a)[\psi_a]\right] &&\text{(linearity)}\\
      &= \sum_{b \in B} \iota_b \left[\sum_{a \in A} \chi^D_{ba} \conv \psi_a \right] &&\text{(Proposition~\ref{prop:l2_lsi})}\\
      &= \left\{\sum_{a \in A} \chi^D_{ba} \conv \psi_a \right\}_{b \in B}. &&\text{(equation~\eqref{eq:lsi_iota_pi_sum})} \qedhere
  \end{align*}
\end{proof}

Therefore, an LSI map \(D \colon \oplus_A\,\ltwo \to \oplus_B\,\ltwo\) is characterised by a collection of \(\ltwo\) functions: \(\{\chi^D_{ba}\}_{(b,a) \in B \times A}\).
Moreover, if we arrange this collection into a \(\abs{B} \times \abs{A}\) matrix and, similarly, represent any \(\Psi \in \oplus_A\,\ltwo\) as an \(\abs{A} \times 1\) matrix of its \(\psi_a \in \ltwo\) components, the previous proposition establishes that:
\begin{equation*}
  D[\Psi] = 
  \begin{pmatrix}
    \chi^D_{ba} & \dots & \chi^D_{ba'} \\
    \vdots & \ddots & \vdots \\
    \chi^D_{b'a} & \dots & \chi^D_{b'a'}
  \end{pmatrix}
  \circledast
  \begin{pmatrix}
    \psi_a \\ \vdots \\ \psi_a'
  \end{pmatrix}
\end{equation*}
where \(\circledast\) denotes formal matrix multiplication whose entry-wise multiplication is replaced by convolution.
Furthermore, considering that the DTFT is an injective map --- due to it having a left inverse~\eqref{eq:DTFT_linv} --- any \(\Psi\) and \(D\) can be uniquely determined by the collection of their Fourier transformed components:
\begin{align*}
  \widehat{\Psi} &= \{\cF_d[\psi_a]\}_{a \in A} \\
  \widehat{D} &= \{\cF_d[\chi^D_{ba}]\}_{(b,a) \in B \times A}
\end{align*}
where the hat notation is used to indicate that \(\widehat{D}\) and \(\widehat{\Psi}\) are representations in the frequency domain (see Remark~\ref{rmk:FT_plane_waves}).
For each \(\omega \in \Rset\) let:
\begin{align*}
  \widehat{\Psi}_\omega &= \{\cF_d[\psi_a](\omega)\}_{a \in A} \\
  \widehat{D}_\omega &= \{\cF_d[\chi^D_{ba}](\omega)\}_{(b,a) \in B \times A}
\end{align*}
and notice that the elements of these collections are complex numbers.
Then, \(\widehat{D}\) can be identified with a collection of linear maps \(\widehat{D}_\omega \colon \Cset^{\abs{A}} \to \Cset^{\abs{B}}\) indexed by \(\omega \in \Rset\).

The importance of these remarks is due to the convolution theorem (Theorem~\ref{thm:l2_convolution}), which turns the convolution of entries used in \(\circledast\) into standard multiplication, so that for each \(\omega \in \Rset\):
\begin{equation} \label{eq:DPsiomega_identity}
  \widehat{(D[\Psi])}_\omega = \widehat{D}_\omega \cdot \widehat{\Psi}_\omega
\end{equation}
where \(- \cdot -\) corresponds to standard multiplication of matrices.
Considering that all of the operations involved in this discussion are associative, composition of LSI maps turns into index-wise composition of the linear maps that comprise them:
\begin{equation*}
  \widehat{(D' \circ D)}_\omega = \widehat{D'}_\omega \circ \widehat{D}_\omega
\end{equation*}
In particular, it unravels a deep connection between the algebra of LSI maps and the algebra of linear maps on finite-dimensional Hilbert spaces.
This motivates the definition of a category of LSI processes.

\begin{definition} \label{def:LSI}
  Let \(\LSI\) be the category whose objects are finite sets and whose morphisms \(f \colon A \to B\) are LSI maps \(\oplus_A\,\ltwo \to \oplus_B\,\ltwo\) represented as their \(\Rset\)-indexed collections:
  \begin{equation*}
    f = \{\widehat{f}_\omega \colon \Cset^{\abs{A}} \to \Cset^{\abs{B}}\}_{\omega \in \Rset}.
  \end{equation*}
  Composition is given by index-wise composition of the linear maps:
  \begin{equation*}
    g \circ f = \{\widehat{g}_\omega \circ \widehat{f}_\omega\}_{\omega \in \Rset}
  \end{equation*}
  and, for any finite set \(A\), its identity morphism is the \(\Rset\)-indexed family of copies of the identity map \(\widehat{\id}_\omega \colon \Cset^{\abs{A}} \to \Cset^{\abs{A}}\).
\end{definition}

\begin{remark}
  Unfortunately, the definition of \(\LSI\) cannot be reproduced for linear shift invariant maps over \(\Ltwo\) since, according to Proposition~\ref{prop:L2_lsi}, \(D[\psi]\) is only approximately equal to \(\chi^D \conv \psi\).
  It would perhaps be possible to define a category where equality is replaced by approximation up to certain factor; unfortunately, every time we compose we would reduce the precision of the approximation.
\end{remark}

The following results establish how the \(\widehat{D}_\omega\) components of a LSI map interact with the norm in \(\oplus_A\,\ltwo\). The ultimate goal is to define a category of LSI contractions.

\begin{lemma} \label{lem:norm_from_Fourier}
  For any finite set \(A\) and any vector \(\Psi \in \oplus_A\,\ltwo\) the following identity is satisfied:
  \begin{equation*}
    \norm{\Psi}^2 = \frac{1}{2\pi} \int_0^{2\pi} \norm{\widehat{\Psi}_\omega}^2 \,d\omega.
  \end{equation*}
\end{lemma} \begin{proof}
  The claim follows from Proposition~\ref{prop:norm_from_Fourier_entry}, as shown below:
  \begin{align*}
    \norm{\Psi}^2 &= \sum_{a \in A} \norm{\psi_a}^2 &&\text{(norm in \(\oplus_A\,\ltwo\))}\\
      &= \sum_{a \in A} \frac{1}{2\pi} \int_0^{2\pi} \abs{\cF_d[\psi_a](\omega)}^2 \,d\omega &&\text{(Proposition~\ref{prop:norm_from_Fourier_entry})}\\
      &= \frac{1}{2\pi} \int_0^{2\pi} \sum_{a \in A} \abs{\cF_d[\psi_a](\omega)}^2 \,d\omega &&\text{(rearrangement)}\\
      &= \frac{1}{2\pi} \int_0^{2\pi} \norm{\widehat{\Psi}_\omega}^2 \,d\omega &&\text{(norm in \(\Cset^{\abs{A}}\))} \qedhere
  \end{align*}
\end{proof}

\begin{theorem} \label{thm:LSI_sufficient_contraction}
  Let \(D \colon \oplus_A\,\ltwo \to \oplus_B\,\ltwo\) be an arbitrary LSI map represented by a collection of linear maps \(\{\widehat{D}_\omega \colon \Cset^{\abs{A}} \to \Cset^{\abs{B}}\}_{\omega \in \Rset}\).
  If for all \(\omega \in \Rset\) the linear map \(\widehat{D}_\omega\) is a contraction, then \(D\) is itself a contraction.
\end{theorem} \begin{proof}  
  According to the previous lemma, for each vector \(\Psi \in \oplus_A\,\ltwo\) we may calculate the norm of \(D[\Psi]\) as follows:
  \begin{equation*}
    \norm{D[\Psi]} = \sqrt{\frac{1}{2\pi} \int_0^{2\pi} \norm{\widehat{(D[\Psi])}_\omega}^2 \,d\omega}
  \end{equation*}
  and, due to equation~\eqref{eq:DPsiomega_identity} and the assumption that \(\widehat{D}_\omega\) is a contraction,
  \begin{equation*}
    \norm{\widehat{(D[\Psi])}_\omega} = \norm{\widehat{D}_\omega \cdot \widehat{\Psi}_\omega} \leq \norm{\widehat{\Psi}_\omega}.
  \end{equation*}
  It then follows that:
  \begin{equation*}
    \norm{D[\Psi]} \leq \sqrt{\frac{1}{2\pi} \int_0^{2\pi} \norm{\widehat{\Psi}_\omega}^2 \,d\omega}
  \end{equation*}
  and, according to the previous lemma, the right hand side is equal to \(\norm{\Psi}\), so that \(\norm{D[\Psi]} \leq \norm{\Psi}\).
  This holds for every vector \(\Psi \in \oplus_A\,\ltwo\), implying that \(D\) is a contraction, as claimed.
\end{proof}

This result motivates the definition of a category of linear shift invariant maps that are \emph{guaranteed} to be contractions.

\begin{definition} \label{def:LSIContraction}
  Let \(\LSIContraction\) be the subcategory of \(\LSI\) obtained by imposing that morphisms \(f \colon A \to B\)
  \begin{equation*}
    f = \{\widehat{f}_\omega \colon \Cset^{\abs{A}} \to \Cset^{\abs{B}}\}_{\omega \in \Rset}
  \end{equation*}
  must satisfy that each \(\widehat{f}_\omega\) is a morphism in \(\FdContraction(\Cset^{\abs{A}},\Cset^{\abs{B}})\).
\end{definition}

\begin{remark} \label{rmk:LSIContraction}  
  The hom-set \(\LSIContraction(A,B)\) may not contain every LSI contraction \(\oplus_A\,\ltwo \to \oplus_B\,\ltwo\) since the condition that each linear map \(\widehat{D}_\omega\) is a contraction has only been established to be sufficient (Theorem~\ref{thm:LSI_sufficient_contraction}).
  Nevertheless, it is reasonable to expect this condition is also necessary: if for some \(\omega_0 \in (0,2\pi)\) there is a vector \(v \in \Cset^{\abs{A}}\) such that \(\norm{\widehat{D}_{\omega_\sub{0}}(v)} > \norm{v}\), then we may engineer a state \(\Psi \in \oplus_A\,\ltwo\) that, for each \(\omega \in \Rset\), is characterised by:
  \begin{equation*}
    \widehat{\Psi}_\omega = v \cdot \delta_\bullet(\omega - \omega_0)
  \end{equation*}
  where \(\delta_\bullet \colon \Rset \to \Cset\) approximates the Dirac delta (see Remark~\ref{rmk:Dirac_delta}).
  Then, using Lemma~\ref{lem:norm_from_Fourier} we find that \(\norm{\Psi}^2 \approx \tfrac{1}{2\pi} \norm{v}^2\) and \(\norm{D[\Psi]}^2 \approx \tfrac{1}{2\pi} \norm{\widehat{D}_{\omega_\sub{0}}(v)}^2\) so that we may expect \(\norm{D[\Psi]} > \norm{\Psi}\), preventing \(D\) from being a contraction.
  Thus, it is likely that the condition established in Theorem~\ref{thm:LSI_sufficient_contraction} is not only sufficient, but also necessary.
\end{remark}

Importantly, morphisms in \(\LSIContraction\) can be seen as \(\Rset\)-indexed collections of morphisms from \(\FdContraction\).
Then, the execution formula from \(\FdContraction\) induces an execution formula in \(\LSIContraction\).
To this end, the following proposition lets us discuss addition of LSI maps in terms of the index-wise addition of its component linear maps.

\begin{proposition} \label{prop:LSI_addition}
  For any two morphisms \(f,g \in \LSI(A,B)\), let their addition \(f+g\) be the following \(\Rset\)-indexed collection:
  \begin{equation*}
    f+g = \{\widehat{f}_\omega + \widehat{g}_\omega\}_{\omega \in \Rset}.
  \end{equation*}
  Then, \(f+g \in \LSI(A,B)\).
\end{proposition} \begin{proof}
  First, realise that if we see \(f\) and \(g\) as linear shift invariant maps \(\oplus_A\,\ltwo \to \oplus_B\,\ltwo\), their pointwise addition is again a linear shift invariant map --- it will commute with shift operators.
  The challenge is to prove that such a map is characterised by the collection \(\{\widehat{f}_\omega + \widehat{g}_\omega\}_{\omega \in \Rset}\) or, more explicitly, we must show that:
  \begin{equation} \label{eq:LSI_addition_aux}
    \widehat{(f+g)}_\omega = \widehat{f}_\omega + \widehat{g}_\omega
  \end{equation}
  for all \(\omega \in \Rset\).
  Recall that for any LSI map \(h \colon \oplus_A\,\ltwo \to \oplus_B\,\ltwo\) each \(\widehat{h}_\omega \colon \Cset^{\abs{A}} \to \Cset^{\abs{B}}\) is the linear map given by the matrix:
  \begin{equation*}
    \widehat{h}_\omega = \begin{pmatrix}
      \cF_d[\chi^h_{ba}](\omega) & \dots & \cF_d[\chi^h_{ba'}](\omega) \\
      \vdots & \ddots & \vdots \\
      \cF_d[\chi^h_{b'a}](\omega) & \dots & \cF_d[\chi^h_{b'a'}](\omega) \\
    \end{pmatrix}
  \end{equation*}
  where \(\chi^h_{ba} \in \ltwo\) is the unique function satisfying \((\pi_b h \iota_a)[\psi] = \chi^h_{ba} \conv \psi\) for all \(\psi \in \ltwo\).
  Considering that addition of linear maps is defined pointwise and that \(\pi_b\) is linear:
  \begin{equation*}
    (\pi_b (f+g) \iota_a)[\psi] = (\pi_b f \iota_a)[\psi] + (\pi_b g \iota_a)[\psi]
  \end{equation*}
  for all \(\psi \in \ltwo\).
  Then, thanks to convolution distributing over addition (see Proposition~\ref{prop:conv_and_addition}) we have that:
  \begin{equation*}
    (\pi_b (f+g) \iota_a)[\psi] = (\chi^f_{ba} \conv \psi) + (\chi^g_{ba} \conv \psi) = (\chi^f_{ba} + \chi^g_{ba}) \conv \psi
  \end{equation*}
  and, due to uniqueness of the characteristic function, we conclude that \(\chi^{f+g}_{ba} = \chi^f_{ba} + \chi^g_{ba}\).
  This, along with linearity of \(\cF_d\), implies that \(\cF_d[\chi^{f+g}_{ba}] = \cF_d[\chi^f_{ba}] + \cF_d[\chi^g_{ba}]\) and it is then immediate that \(\widehat{(f+g)}_\omega = \widehat{f}_\omega + \widehat{g}_\omega\) as required by equation~\eqref{eq:LSI_addition_aux}, completing the proof.
\end{proof}

Evidently, it is generally not the case that for two morphisms \(f,g \in \LSIContraction(A,B)\) their addition \(f+g\) is again in \(\LSIContraction\), since addition of contractions need not result in a contraction.
In the following subsection we prove that \(\LSIContraction\) is a totally traced category with the execution formula.

\subsection{The execution formula in \(\LSIContraction\)}
\label{sec:LSI_trace}

Since morphisms in \(\LSI\) are \(\Rset\)-indexed collections of morphisms in \(\FdHilb\) and composition and addition in \(\LSI\) are defined index-wise, it is not surprising that \(\LSI\) is a \(\SCat{g}\)-UDC.
Similarly, \(\LSIContraction\) is a \(\SCat{w}\)-UDC.

\begin{proposition} \label{prop:LSI_UDC}
  \(\LSI\) is a \(\SCat{g}\)-UDC.
\end{proposition} \begin{proof}
  The \(\Sigma\)-group structure on every hom-set \(\LSI(A,B)\) is defined in terms of the \(\Sigma\)-group structure on \(\FdHilb(\Cset^{\abs{A}},\Cset^{\abs{B}})\), \ie{} for each family \(\{f_i\}_I \in \LSI(A,B)^*\) define: 
  \begin{equation*}
    \Sigma \{f_i\}_I = \begin{cases}
      \{\widehat{g}_\omega\}_{\omega \in \Rset} &\ifc \forall \omega \in \Rset,\, \Sigma \{\widehat{(f_i)}_\omega\}_I \keq \widehat{g}_\omega \\
      \undefined &\otherwise.
    \end{cases}
  \end{equation*}
  Since \(\FdHilb(\Cset^{\abs{A}},\Cset^{\abs{B}})\) is a \(\Sigma\)-group and \(\Sigma\) on \(\LSI(A,B)\) is defined index-wise, it is immediate that \(\LSI(A,B)\) is a \(\Sigma\)-group.
  Composition in \(\LSI\) is \(\Sigma\)-bilinear since both composition and addition in \(\LSI\) are defined index-wise and composition in \(\FdHilb\) is \(\Sigma\)-bilinear; thus, \(\LSI\) is a \(\SCat{g}\)-enriched category.
  Moreover, \((\LSI,\oplus,\varnothing)\) is a symmetric monoidal category where \(\oplus\) corresponds to disjoint union of sets on objects and, on any two morphisms \(f \in \LSI(A,B)\) and \(g \in \LSI(C,D)\), it acts as index-wise direct sum:
  \begin{equation*}
    f \oplus g = \{\widehat{f}_\omega \oplus \widehat{g}_\omega \colon \Cset^{\abs{A}} \oplus \Cset^{\abs{C}} \to \Cset^{\abs{B}} \oplus \Cset^{\abs{D}}\}_{\omega \in \Rset}.
  \end{equation*}
  This is well-defined thanks to the isomorphism \(\Cset^{\abs{A \uplus B}} \iso \Cset^{\abs{A}} \oplus \Cset^{\abs{B}}\).
  The monoidal unit \(\varnothing\) is clearly a zero object since \(\Cset^{\abs{\varnothing}} = \{0\}\), so that \(f \in \LSI(\varnothing,A)\) and \(g \in \LSI(A,\varnothing)\) are \(\Rset\)-indexed collections of copies of the zero map in \(\FdHilb\).
  Verifying the identity \(\Sigma \{\id_A \oplus 0, 0 \oplus \id_B\} \keq \id_{A \uplus B}\) is conceptually simple since both \(\oplus\) and \(\Sigma\) in \(\LSI\) are defined index-wise on morphisms, and the identity is satisfied in \(\FdHilb\).
  Consequently, \(\LSI\) is a \(\SCat{g}\)-UDC, as claimed.
\end{proof}

\begin{corollary} \label{prop:LSIContraction_UDC}
  \(\LSIContraction\) is a \(\SCat{w}\)-UDC.
\end{corollary} \begin{proof}
  By definition, \(\LSIContraction\) is a subcategory of \(\LSI\) and we may define a symmetric monoidal structure on it in the same manner we did for \(\LSI\) in the previous proposition.
  Moreover, we may also provide a \(\Sigma\)-monoid structure on each hom-set in \(\LSIContraction\) via the same index-wise definition --- this time using the \(\Sigma\)-monoid structure of hom-sets in \(\FdContraction\).
  Recall that \(\FdContraction\) is only \(\SCat{w}\)-enriched, since a pair of contractions does not generally add up to a contraction; thus, the enrichment in \(\LSIContraction\) is over \(\SCat{w}\) as well.
  Consequently, the same proof strategy from the previous proposition establishes that \(\LSIContraction\) is a \(\SCat{w}\)-UDC.
\end{proof}

It follows from the fact that \((\FdContraction,\oplus,\ex)\) is a totally traced category that \((\LSIContraction,\oplus,\ex)\) is totally traced as well.

\begin{theorem} \label{thm:LSIContraction_traced}
  The category \((\LSIContraction,\oplus,\ex)\) is totally traced.
\end{theorem} \begin{proof}
  Recall that on any \(\SCat{*}\)-UDC the execution formula \(\ex\) is defined as follows:
  \begin{equation*}
    \ex^U(f) = \Sigma \{f_\sub{BA}\} \uplus \{f_\sub{BU} f_\sub{UU}^n f_\sub{UA}\}_\Nset.
  \end{equation*}
  Since the morphisms in \(\LSIContraction\) are, by definition, \(\Rset\)-indexed collections of morphisms in \(\FdContraction\) and \(\Sigma\), \(\oplus\) and composition in \(\LSIContraction\) are all defined index-wise, we have that \(\ex^U(f)\) corresponds to:
  \begin{equation*}
     \ex^U(f) = \{\ex^{\Cset^\dsub{\abs{U}}}(\widehat{f}_\omega)\}_{\omega \in \Rset}
  \end{equation*}
  where \(\ex^{\Cset^\dsub{\abs{U}}}\) is the execution formula in \(\FdContraction\).
  Theorem~\ref{thm:FdContraction_traced} establishes that \((\FdContraction,\oplus,\ex)\) is a totally traced category and, hence, for all \(\omega \in \Rset\) the linear map \(\ex^{\Cset^\dsub{\abs{U}}}(\widehat{f}_\omega)\) is a well-defined contraction.
  According to Theorem~\ref{thm:LSI_sufficient_contraction} this guarantees that the LSI map \(\ex^U(f)\) is itself a contraction and, hence, the execution formula is a total function in \(\LSIContraction\).
  Moreover, \((\LSIContraction,\oplus,\ex)\) is a traced monoidal category since the execution formula in \(\FdContraction\) is a valid trace, which guarantees that the axioms of traced monoidal categories are satisfied for every index \(\omega \in \Rset\).
  Thus, we conclude that \((\LSIContraction,\oplus,\ex)\) is a totally traced category, as claimed.
\end{proof}

As discussed at the beginning of this section, the physical interpretation of the Hilbert space \(\oplus_A\,\ltwo\) is that each vector is a quantum state that spreads over a finite set of locations \(A\) and over the discrete-time domain \(\Zset\).
LSI maps \(\oplus_A\,\ltwo \to \oplus_B\,\ltwo\) are interpreted as devices with a finite number of input and output ports, labelled by the elements in \(A\) and \(B\) respectively, and whose input-output behaviour is independent of the time-step the input reaches the device.
Then, \(\ex^U(f)\) may be interpreted as the device obtained after connecting the output ports labelled by \(U\) to the input ports with the same label.

Considering the prevalence of the notion of time --- an infinite domain --- in the definition of \(\LSI\), it is perhaps surprising that we may prove \(\LSIContraction\) is a traced category by means of an argument in finite dimensions, \ie{} using \(\FdContraction\).
The key to this result were the discrete-time Fourier transform (Definition~\ref{def:DTFT}), together with its convolution theorem (Theorem~\ref{thm:l2_convolution}) and the unique characterisation of LSI maps (Proposition~\ref{prop:oplusl2_lsi}).
In short, the Fourier transform lets us uniquely characterise time-shift invariant maps according to their action on plane waves \(\gamma_\omega(t) = e^{i\omega t}\) (see Remark~\ref{rmk:FT_plane_waves}) and the convolution theorem establishes that the action of these maps at different frequencies \(\omega\) do not interact with each other.
Consequently, in order to describe the input-output behaviour of a time-shift invariant device, we simply need to figure out how it affects the phase and amplitude of each plane wave \(\gamma_\omega\) for every angular frequency \(\omega \in \Rset\); this is precisely what each map \(\widehat{f}_\omega \colon \Cset^{\abs{A}} \to \Cset^{\abs{B}}\) represents.
In some sense, the fact that \((\LSIContraction,\oplus,\ex)\) is a traced category tells us that we may calculate the behaviour of a device with feedback by an iterative method, adding up the contribution of each of the possible paths from ports in \(A\) to ports in \(B\).
Moreover, the fact that \(\ex\) is total indicates that this iterative method always converges.
These results are not novel: physicists and engineers have been using these facts about LSI maps and the Fourier transform for quite a long time; however, there is some value in providing a formal description of this situation in the language of traced monoidal categories, as it may enable computer scientists to use these notions in a more abstract setting. 

The category \(\LSI\) captures the notion of quantum processes over a discrete domain \(\Zset\); such domain has been interpreted to describe the line of time, but it could instead be interpreted to correspond to discrete one-dimensional space, or we may extend it to a domain \(\Zset^{3}\) so that the appropriate version of \(\LSI\) would capture maps that are invariant with respect to translations on three-dimensional space.
A natural question is whether a similar approach could be used to discuss LSI maps acting on states over a \emph{continuous} domain \(\Rset\).
To do so, we would need to work with the Hilbert space \(\Ltwo\) but, as discussed in Remark~\ref{rmk:Dirac_delta}, convolution in \(\Ltwo\) does not have a unit and, consequently, the characterisation of LSI via an \(\Rset\)-indexed collection of linear maps is only valid up to approximation (see Proposition~\ref{prop:L2_lsi}).
In practice, physicists and engineers tend to disregard these approximations and treat them as equalities (see Remark~\ref{rmk:Dirac_delta}); however, the objective of this section was to discuss feedback in the formal setting of traced monoidal categories and, to that extent, the case of the continuous-time domain is left as an open question.

\section{Related work}
\label{sec:trace_rel_work}

In this section, the contributions of this chapter are put into perspective and compared to previous works in the field.

\paragraph*{Hoshino~\cite{RTUDC}.} Hoshino proved that every \(\SCat{s}\)-UDC is partially traced using the execution formula.
The same result is established in Theorem~\ref{thm:classical_trace} of this chapter, although the proof strategy differs from Hoshino's.
In~\cite{RTUDC}, Hoshino proves that every \(\SCat{s}\)-UDC can be faithfully embedded in a category with countable biproducts and enriched over \emph{totally} defined strong \(\Sigma\)-monoids --- \ie{} every family of morphisms is summable.
It is straightforward to show that these categories are totally traced using the execution formula and, thus, the faithful embedding provided by Hoshino immediately induces a partial trace on the original category as described in Proposition~\ref{prop:induced_trace}.
Unfortunately, such a result cannot be reproduced in the general case of \(\SCat{w}\)-UDCs: there are examples of such categories that are not partially traced using the execution formula, \eg{} \((\FdHilb,\oplus,\{0\})\), see Example~\ref{ex:FdHilb_counterexample}.
Our interest in \(\SCat{w}\)-enriched categories comes from the fact that the categories of quantum processes we are interested in (\eg{} \(\Unitary\) and \(\LSI\)) have additive inverses of morphisms; these capture destructive interference in quantum computing but are not available in a \(\SCat{s}\)-enrichment.

\paragraph*{The Barr functor.} The Barr functor \(\PInj \to \Contraction\) originally proposed in~\cite{Barr} induces a \(\SCat{s}\)-enrichment in \(\Contraction\) using that of \(\PInj\).
The details of the induced \(\SCat{s}\)-enrichment in \(\Contraction\) appear in Section 5.5 from~\cite{Haghverdi}; in essence, families of contractions are deemed summable if and only if the domain and codomain of every pair of contractions come from orthogonal subspaces and, hence, no quantum interference may occur between them.
Such a situation can be perceived as a quantum computer simulating a classical iterative loop, which is not of interest to this thesis.

\paragraph*{Malherbe, Scott \& Selinger~\cite{Malherbe}.} The kernel-image trace plays an essential role in our proof that \((\FdContraction,\oplus,\ex)\) is a totally traced category.
We can draw a loose analogy between the kernel-image trace of a morphism \(f \colon A \oplus U \to B \oplus U\) and the existence of a pseudo-inverse \((\id\!-\!f_\sub{UU})^+\) which may be used to calculate \(\Tr_\sub{\ki} \approx f_\sub{BA} + f_\sub{BU} (\id\!-\!f_\sub{UU})^+ f_\sub{UA}\).
Under this analogy, the intuition behind our proof that \((\FdContraction,\oplus,\ex)\) is totally traced loosely corresponds to the identity in the real numbers \((1-r)^{-1} = \sum_{n = 0}^\infty r^n\), which only holds when both sides are defined and \(r < 1\).
More formally, Lemma~\ref{lem:ex_from_ki} identifies sufficient conditions for a morphism \(f\) to satisfy that the well-definedness of \(\Tr_\sub{\ki}(f)\) implies that of \(\ex(f)\), with their results matching.
Then, we take advantage of the kernel-image trace being a valid categorical trace in \(\Hilb\) to prove that the execution formula is a valid trace of its subcategory \(\FdContraction\).

\paragraph*{Bartha~\cite{Bartha}.} Bartha had shown that both \(\FdIsometry\) and \(\FdUnitary\) are totally traced categories using the execution formula.
Bartha's proof makes use of the kernel-image trace in a similar way as we do and, in particular, our proof that \(\FdContraction\) is totally traced (instead of simply partially traced) relies on the approach Bartha used to prove the same result in the case of \(\FdIsometry\) (see Proposition~\ref{prop:Bartha_isometry}).
Even though Bartha did not explicitly show that \(\FdContraction\) is totally traced, this can be easily inferred from the result that \(\FdIsometry\) is totally traced and, hence, we do not claim this to be a novel result.
Instead, the approach we have used to reach this result is novel: whereas Bartha relies on matrix iteration theories, our approach relies on the more general theory of Hausdorff abelian groups.
Moreover, our approach is self-contained, whereas Bartha defers the details of a key proof (that of Theorem 8 from~\cite{Bartha}) for the reader to check, assuming prior (nontrivial) knowledge on matrix iteration theories.

Our original goal was to answer an open question posed by Bartha: is \(\Isometry\) partially traced?
We rephrased this question in terms of \(\Contraction\) and, unfortunately, the question remains open.
However, we have been able to narrow down the search of an answer to the open problems discussed in Section~\ref{sec:open_questions}; in particular, we have shown that \(\Hilb\) is a hom-convergence UDC and, hence, even if \(\Contraction\) were not partially traced, we have opened up a path to prove that certain subcategories of \(\Contraction\) may be partially traced.
In particular, \(\LSIContraction\) is an example of a category that faithfully embeds in \(\Contraction\) and that is totally traced, as established in Section~\ref{sec:LSI}.

\paragraph*{Pseudo-traces with delays.} Previous works such as~\cite{DelayedTrace} have formalised discrete-time quantum iterative loops in an ad-hoc manner by introducing a notion of categorical pseudo-trace similar to the execution formula but that explicitly includes a `delay' operator in its definition.
Doing so implies that the category no longer satisfies the yanking axiom, as intuitively shown in the diagram below
\[\input{Figures/3/delayed_trace}.\]

In contrast, the category \(\LSI\) captures such `delays' in the definition of its morphisms and, hence, an appropriate notion of discrete-time quantum iteration can be achieved in the standard framework of traced monoidal categories and the execution formula.
To the best of our knowledge, our result that \((\LSIContraction,\oplus,\ex)\) (Theorem~\ref{thm:LSIContraction_traced}) is a totally traced category is the first instance of a formalisation of discrete-time quantum iterative loops as a categorical trace.

%% file: Figures/3/pic_trace.tex
\begin{tikzpicture}
  \node[rectangle,draw=black,thick,minimum width=6mm,minimum height=10mm] (f) {\(f\)};
  \coordinate[below=3mm of f.west] (A);
  \coordinate[left=7mm of A] (Ad);
  \draw (Ad) -- node[below] {\tiny \(A\)} (A);
  \coordinate[below=3mm of f.east] (B);
  \coordinate[right=7mm of B] (Bd);
  \draw (Bd) -- node[below] {\tiny \(B\)} (B);
  \coordinate[above=3mm of f.west] (Ui);
  \coordinate[left=7mm of Ui] (Uid);
  \draw (Uid) -- node[above] {\tiny \(U\)} (Ui);
  \coordinate[above=3mm of f.east] (Uo);
  \coordinate[right=7mm of Uo] (Uod);
  \draw (Uod) -- node[above] {\tiny \(U\)} (Uo);
\end{tikzpicture}
\quad\quad \mapsto \quad\quad
\begin{tikzpicture}
  \node[rectangle,draw=black,thick,minimum width=6mm,minimum height=10mm] (f) {\(f\)};
  \coordinate[below=3mm of f.west] (A);
  \coordinate[left=7mm of A] (Ad);
  \draw (Ad) -- node[below] {\tiny \(A\)} (A);
  \coordinate[below=3mm of f.east] (B);
  \coordinate[right=7mm of B] (Bd);
  \draw (Bd) -- node[below] {\tiny \(B\)} (B);
  \coordinate[above=3mm of f.west] (Ui);
  \coordinate[left=3mm of Ui] (Uid);
  \draw (Uid) -- node[above] {\tiny \(U\)} (Ui);
  \coordinate[above=3mm of f.east] (Uo);
  \coordinate[right=3mm of Uo] (Uod);
  \draw (Uod) -- node[above] {\tiny \(U\)} (Uo);
  \coordinate[above=6mm of Uid] (Uiu);
  \coordinate[above=6mm of Uod] (Uou);
  \draw (Uiu) -- (Uou);
  \draw (Uid) edge[out=180,in=180,looseness=1.5] (Uiu);
  \draw (Uod) edge[out=0,in=0,looseness=1.5] (Uou);
\end{tikzpicture}

%% file: Figures/3/naturality.tex
\begin{tikzpicture}
  \node (margin) {
    \begin{tikzpicture}
      \node[rectangle,draw=black,thick,minimum width=6mm,minimum height=10mm] (f) {\(f\)};
      \coordinate[below=3mm of f.west] (A);
      \node[left=10mm of A,rectangle,draw=black,thick,minimum size=4mm] (g) {\(g\)};
      \draw (g.east) -- node[below] {\tiny \(A\)} (A);
      \coordinate[left=5mm of g.west] (Ad);
      \draw (g.west) -- node[below] {\tiny \(A'\)} (Ad);
      \coordinate[below=3mm of f.east] (B);
      \node[right=10mm of B,rectangle,draw=black,thick,minimum size=4mm] (h) {\(h\)};
      \draw (h) -- node[below] {\tiny \(B\)} (B);
      \coordinate[right=5mm of h.east] (Bd);
      \draw (h.east) -- node[below] {\tiny \(B'\)} (Bd);
      \coordinate[above=3mm of f.west] (Ui);
      \coordinate[left=3mm of Ui] (Uid);
      \draw (Uid) -- node[above] {\tiny \(U\)} (Ui);
      \coordinate[above=3mm of f.east] (Uo);
      \coordinate[right=3mm of Uo] (Uod);
      \draw (Uod) -- node[above] {\tiny \(U\)} (Uo);
      \coordinate[above=6mm of Uid] (Uiu);
      \coordinate[above=6mm of Uod] (Uou);
      \draw (Uiu) -- (Uou);
      \draw (Uid) edge[out=180,in=180,looseness=1.5] (Uiu);
      \draw (Uod) edge[out=0,in=0,looseness=1.5] (Uou);
      \coordinate[below=2mm of f.south] (td);
      \coordinate[above=6mm of f.north] (tu);
      \coordinate[left=7mm of f.west] (tl);
      \coordinate[right=7mm of f.east] (tr);
      \draw[dashed,draw=black!60] (td) -| (tl) |- (tu) -| (tr) |- (td);
    \end{tikzpicture}
  };
\end{tikzpicture}
\quad\quad &= \quad\quad
\begin{tikzpicture}
  \node (margin) {
    \begin{tikzpicture}
      \node[rectangle,draw=black,thick,minimum width=6mm,minimum height=10mm] (f) {\(f\)};
      \coordinate[below=3mm of f.west] (A);
      \node[left=3mm of A,rectangle,draw=black,thick,minimum size=4mm] (g) {\(g\)};
      \draw (g.east) -- node[below] {\tiny \(A\)} (A);
      \coordinate[left=8mm of g.west] (Ad);
      \draw (g.west) -- node[below] {\tiny \(A'\)} (Ad);
      \coordinate[below=3mm of f.east] (B);
      \node[right=3mm of B,rectangle,draw=black,thick,minimum size=4mm] (h) {\(h\)};
      \draw (h) -- node[below] {\tiny \(B\)} (B);
      \coordinate[right=8mm of h.east] (Bd);
      \draw (h.east) -- node[below] {\tiny \(B'\)} (Bd);
      \coordinate[above=3mm of f.west] (Ui);
      \coordinate[left=10mm of Ui] (Uid);
      \draw (Uid) -- node[above] {\tiny \(U\)} (Ui);
      \coordinate[above=3mm of f.east] (Uo);
      \coordinate[right=10mm of Uo] (Uod);
      \draw (Uod) -- node[above] {\tiny \(U\)} (Uo);
      \coordinate[above=6mm of Uid] (Uiu);
      \coordinate[above=6mm of Uod] (Uou);
      \draw (Uiu) -- (Uou);
      \draw (Uid) edge[out=180,in=180,looseness=1.5] (Uiu);
      \draw (Uod) edge[out=0,in=0,looseness=1.5] (Uou);
      \coordinate[below=2mm of f.south] (td);
      \coordinate[above=6mm of f.north] (tu);
      \coordinate[left=14mm of f.west] (tl);
      \coordinate[right=14mm of f.east] (tr);
      \draw[dashed,draw=black!60] (td) -| (tl) |- (tu) -| (tr) |- (td);
    \end{tikzpicture}
  };
\end{tikzpicture}

%% file: Figures/3/dinaturality.tex
\begin{tikzpicture}
  \node (margin) {
    \begin{tikzpicture}
      \node[rectangle,draw=black,thick,minimum width=6mm,minimum height=10mm] (f) {\(f\)};
      \coordinate[below=3mm of f.west] (A);
      \coordinate[left=10mm of A] (Ad);
      \draw (A) -- node[below] {\tiny \(A\)} (Ad);
      \coordinate[below=3mm of f.east] (B);
      \coordinate[right=19mm of B] (Bd);
      \draw (B) -- node[below] {\tiny \(B\)} (Bd);
      \coordinate[above=3mm of f.west] (Ui);
      \coordinate[left=3mm of Ui] (Uid);
      \draw (Uid) -- node[above] {\tiny \(U\)} (Ui);
      \coordinate[above=3mm of f.east] (Uo);
      \node[right=4mm of Uo,rectangle,draw=black,thick,minimum size=4mm] (g) {\(g\)};
      \coordinate[right=3mm of g] (Uod);
      \draw (Uo) -- node[above] {\tiny \(U'\)} (g.west);
      \draw (Uod) -- node[above] {\tiny \(U\)} (g.east);
      \coordinate[above=6mm of Uid] (Uiu);
      \coordinate[above=6mm of Uod] (Uou);
      \draw (Uiu) -- (Uou);
      \draw (Uid) edge[out=180,in=180,looseness=1.5] (Uiu);
      \draw (Uod) edge[out=0,in=0,looseness=1.5] (Uou);
      \coordinate[below=2mm of f.south] (td);
      \coordinate[above=6mm of f.north] (tu);
      \coordinate[left=7mm of f.west] (tl);
      \coordinate[right=16mm of f.east] (tr);
      \draw[dashed,draw=black!60] (td) -| (tl) |- (tu) -| (tr) |- (td);
    \end{tikzpicture}
  };
\end{tikzpicture}
\quad\quad &= \quad\quad
\begin{tikzpicture}
  \node (margin) {
    \begin{tikzpicture}
      \node[rectangle,draw=black,thick,minimum width=6mm,minimum height=10mm] (f) {\(f\)};
      \coordinate[below=3mm of f.west] (A);
      \coordinate[left=19mm of A] (Ad);
      \draw (A) -- node[below] {\tiny \(A\)} (Ad);
      \coordinate[below=3mm of f.east] (B);
      \coordinate[right=10mm of B] (Bd);
      \draw (B) -- node[below] {\tiny \(B\)} (Bd);
      \coordinate[above=3mm of f.west] (Ui);
      \node[left=4mm of Ui,rectangle,draw=black,thick,minimum size=4mm] (g) {\(g\)};
      \coordinate[left=3mm of g] (Uid);
      \draw (Ui) -- node[above] {\tiny \(U\)} (g.east);
      \draw (Uid) -- node[above] {\tiny \(U'\)} (g.west);
      \coordinate[above=3mm of f.east] (Uo);
      \coordinate[right=3mm of Uo] (Uod);
      \draw (Uod) -- node[above] {\tiny \(U'\)} (Uo);
      \coordinate[above=6mm of Uid] (Uiu);
      \coordinate[above=6mm of Uod] (Uou);
      \draw (Uiu) -- (Uou);
      \draw (Uid) edge[out=180,in=180,looseness=1.5] (Uiu);
      \draw (Uod) edge[out=0,in=0,looseness=1.5] (Uou);
      \coordinate[below=2mm of f.south] (td);
      \coordinate[above=6mm of f.north] (tu);
      \coordinate[left=16mm of f.west] (tl);
      \coordinate[right=7mm of f.east] (tr);
      \draw[dashed,draw=black!60] (td) -| (tl) |- (tu) -| (tr) |- (td);
    \end{tikzpicture}
  };
\end{tikzpicture}

%% file: Figures/3/superposing.tex
\begin{tikzpicture}
  \node (margin) {
    \begin{tikzpicture}
      \node[rectangle,draw=black,thick,minimum width=6mm,minimum height=10mm] (f) {\(f\)};
      \coordinate[below=3mm of f.west] (A);
      \coordinate[left=10mm of A] (Ad);
      \draw (A) -- node[below] {\tiny \(A\)} (Ad);
      \coordinate[below=3mm of f.east] (B);
      \coordinate[right=10mm of B] (Bd);
      \draw (B) -- node[below] {\tiny \(B\)} (Bd);
      \coordinate[above=3mm of f.west] (Ui);
      \coordinate[left=3mm of Ui] (Uid);
      \draw (Uid) -- node[above] {\tiny \(U\)} (Ui);
      \coordinate[above=3mm of f.east] (Uo);
      \coordinate[right=3mm of Uo] (Uod);
      \draw (Uod) -- node[above] {\tiny \(U\)} (Uo);
      \coordinate[above=6mm of Uid] (Uiu);
      \coordinate[above=6mm of Uod] (Uou);
      \draw (Uiu) -- (Uou);
      \draw (Uid) edge[out=180,in=180,looseness=1.5] (Uiu);
      \draw (Uod) edge[out=0,in=0,looseness=1.5] (Uou);
      \node[below=4mm of f,rectangle,draw=black,thick,minimum size=6mm] (g) {\(g\)};
      \coordinate[left=10mm of g.west] (Cd);
      \coordinate[right=10mm of g.east] (Dd);
      \draw (Cd) -- node[below] {\tiny \(C\)} (g.west);
      \draw (Dd) -- node[below] {\tiny \(D\)} (g.east);
      \coordinate[below=2mm of f.south] (td);
      \coordinate[above=6mm of f.north] (tu);
      \coordinate[left=7mm of f.west] (tl);
      \coordinate[right=7mm of f.east] (tr);
      \draw[dashed,draw=black!60] (td) -| (tl) |- (tu) -| (tr) |- (td);
    \end{tikzpicture}
  };
\end{tikzpicture}
\quad\quad &= \quad\quad
\begin{tikzpicture}
  \node (margin) {
    \begin{tikzpicture}
      \node[rectangle,draw=black,thick,minimum width=6mm,minimum height=10mm] (f) {\(f\)};
      \coordinate[below=3mm of f.west] (A);
      \coordinate[left=10mm of A] (Ad);
      \draw (A) -- node[below] {\tiny \(A\)} (Ad);
      \coordinate[below=3mm of f.east] (B);
      \coordinate[right=10mm of B] (Bd);
      \draw (B) -- node[below] {\tiny \(B\)} (Bd);
      \coordinate[above=3mm of f.west] (Ui);
      \coordinate[left=3mm of Ui] (Uid);
      \draw (Uid) -- node[above] {\tiny \(U\)} (Ui);
      \coordinate[above=3mm of f.east] (Uo);
      \coordinate[right=3mm of Uo] (Uod);
      \draw (Uod) -- node[above] {\tiny \(U\)} (Uo);
      \coordinate[above=6mm of Uid] (Uiu);
      \coordinate[above=6mm of Uod] (Uou);
      \draw (Uiu) -- (Uou);
      \draw (Uid) edge[out=180,in=180,looseness=1.5] (Uiu);
      \draw (Uod) edge[out=0,in=0,looseness=1.5] (Uou);
      \node[below=4mm of f,rectangle,draw=black,thick,minimum size=6mm] (g) {\(g\)};
      \coordinate[left=10mm of g.west] (Cd);
      \coordinate[right=10mm of g.east] (Dd);
      \draw (Cd) -- node[below] {\tiny \(C\)} (g.west);
      \draw (Dd) -- node[below] {\tiny \(D\)} (g.east);
      \coordinate[below=2mm of g.south] (td);
      \coordinate[above=6mm of f.north] (tu);
      \coordinate[left=7mm of f.west] (tl);
      \coordinate[right=7mm of f.east] (tr);
      \draw[dashed,draw=black!60] (td) -| (tl) |- (tu) -| (tr) |- (td);
    \end{tikzpicture}
  };
\end{tikzpicture}

%% file: Figures/3/vanishingI.tex
\begin{tikzpicture}
  \node (margin) {
    \begin{tikzpicture}
      \node[rectangle,draw=black,thick,minimum width=6mm,minimum height=10mm] (f) {\(f\)};
      \coordinate[below=3mm of f.west] (A);
      \coordinate[left=10mm of A] (Ad);
      \draw (A) -- node[below] {\tiny \(A\)} (Ad);
      \coordinate[below=3mm of f.east] (B);
      \coordinate[right=10mm of B] (Bd);
      \draw (B) -- node[below] {\tiny \(B\)} (Bd);
      \coordinate[above=3mm of f.west] (Ui);
      \coordinate[left=3mm of Ui] (Uid);
      \draw (Uid) -- node[above] {\tiny \(I\)} (Ui);
      \coordinate[above=3mm of f.east] (Uo);
      \coordinate[right=3mm of Uo] (Uod);
      \draw (Uod) -- node[above] {\tiny \(I\)} (Uo);
      \coordinate[above=6mm of Uid] (Uiu);
      \coordinate[above=6mm of Uod] (Uou);
      \draw (Uiu) -- (Uou);
      \draw (Uid) edge[out=180,in=180,looseness=1.5] (Uiu);
      \draw (Uod) edge[out=0,in=0,looseness=1.5] (Uou);
      \coordinate[below=2mm of f.south] (td);
      \coordinate[above=6mm of f.north] (tu);
      \coordinate[left=7mm of f.west] (tl);
      \coordinate[right=7mm of f.east] (tr);
      \draw[dashed,draw=black!60] (td) -| (tl) |- (tu) -| (tr) |- (td);
    \end{tikzpicture}
  };
\end{tikzpicture}
\quad\quad &= \quad\quad
\begin{tikzpicture}
  \node (margin) {
    \begin{tikzpicture}
      \node[rectangle,draw=black,thick,minimum width=6mm,minimum height=10mm] (f) {\(f\)};
      \coordinate[below=3mm of f.west] (A);
      \coordinate[left=7mm of A] (Ad);
      \draw (A) -- node[below] {\tiny \(A\)} (Ad);
      \coordinate[below=3mm of f.east] (B);
      \coordinate[right=7mm of B] (Bd);
      \draw (B) -- node[below] {\tiny \(B\)} (Bd);
      \coordinate[above=3mm of f.west] (Ui);
      \node[left=3mm of Ui,circle,fill=black,scale=0.3] (Uid) {};  
      \draw (Ui) -- node[above] {\tiny \(I\)} (Uid);
      \coordinate[above=3mm of f.east] (Uo);
      \node[right=3mm of Uo,circle,fill=black,scale=0.3] (Uod) {};  
      \draw (Uo) -- node[above] {\tiny \(I\)} (Uod);
    \end{tikzpicture}
  };
\end{tikzpicture}

%% file: Figures/3/vanishingII.tex
\begin{tikzpicture}
  \node (margin) {
    \begin{tikzpicture}
      \node[rectangle,draw=black,thick,minimum width=6mm,minimum height=14mm] (f) {\(f\)};
      \coordinate[below=5mm of f.west] (A);
      \coordinate[left=6mm of A] (Ad);
      \coordinate[left=17mm of A] (Add);
      \draw (A) -- node[below] {\tiny \(A\)} (Ad);
      \draw (Ad) -- (Add);
      \coordinate[below=5mm of f.east] (B);
      \coordinate[right=6mm of B] (Bd);
      \coordinate[right=17mm of B] (Bdd);
      \draw (B) -- node[below] {\tiny \(B\)} (Bd);
      \draw (Bd) -- (Bdd);
      \coordinate[above=5mm of f.west] (Vi);
      \coordinate[left=2mm of Vi] (Vid);
      \draw (Vid) -- node[above] {\tiny \(V\)} (Vi);
      \coordinate[above=5mm of f.east] (Vo);
      \coordinate[right=2mm of Vo] (Vod);
      \draw (Vod) -- node[above] {\tiny \(V\)} (Vo);
      \coordinate[above=6mm of Vid] (Viu);
      \coordinate[above=6mm of Vod] (Vou);
      \draw (Viu) -- (Vou);
      \draw (Vid) edge[out=180,in=180,looseness=1.5] (Viu);
      \draw (Vod) edge[out=0,in=0,looseness=1.5] (Vou);
      \coordinate[left=6mm of f.west] (Uid);
      \draw (Uid) -- node[above] {\tiny \(U\)} (f.west);
      \coordinate[right=6mm of f.east] (Uod);
      \draw (Uod) -- node[above] {\tiny \(U\)} (f.east);
      \coordinate[above=16mm of Uid] (Uiu);
      \coordinate[above=16mm of Uod] (Uou);
      \draw (Uiu) -- (Uou);
      \draw (Uid) edge[out=180,in=180,looseness=1.5] (Uiu);
      \draw (Uod) edge[out=0,in=0,looseness=1.5] (Uou);
      \coordinate[below=2mm of f.south] (td);
      \coordinate[above=6mm of f.north] (tu);
      \coordinate[left=6mm of f.west] (tl);
      \coordinate[right=6mm of f.east] (tr);
      \draw[dashed,draw=black!60] (td) -| (tl) |- (tu) -| (tr) |- (td);
      \coordinate[below=4mm of f.south] (ttd);
      \coordinate[above=11mm of f.north] (ttu);
      \coordinate[left=14.5mm of f.west] (ttl);
      \coordinate[right=14.5mm of f.east] (ttr);
      \draw[dashed,draw=black!60] (ttd) -| (ttl) |- (ttu) -| (ttr) |- (ttd);
    \end{tikzpicture}
  };
\end{tikzpicture}
\quad\quad &= \quad\quad
\begin{tikzpicture}
  \node (margin) {
    \begin{tikzpicture}
      \node[rectangle,draw=black,thick,minimum width=6mm,minimum height=10mm] (f) {\(f\)};
      \coordinate[below=3mm of f.west] (A);
      \coordinate[left=15mm of A] (Ad);
      \draw (A) -- node[below] {\tiny \(A\)} (Ad);
      \coordinate[below=3mm of f.east] (B);
      \coordinate[right=15mm of B] (Bd);
      \draw (B) -- node[below] {\tiny \(B\)} (Bd);
      \coordinate[above=3mm of f.west] (Ui);
      \coordinate[left=8mm of Ui] (Uid);
      \draw (Uid) -- node[above] {\tiny \(U \otimes V\)} (Ui);
      \coordinate[above=3mm of f.east] (Uo);
      \coordinate[right=8mm of Uo] (Uod);
      \draw (Uod) -- node[above] {\tiny \(U \otimes V\)} (Uo);
      \coordinate[above=9mm of Uid] (Uiu);
      \coordinate[above=9mm of Uod] (Uou);
      \draw (Uiu) -- (Uou);
      \draw (Uid) edge[out=180,in=180,looseness=1.5] (Uiu);
      \draw (Uod) edge[out=0,in=0,looseness=1.5] (Uou);
      \coordinate[below=2mm of f.south] (td);
      \coordinate[above=9mm of f.north] (tu);
      \coordinate[left=13mm of f.west] (tl);
      \coordinate[right=13mm of f.east] (tr);
      \draw[dashed,draw=black!60] (td) -| (tl) |- (tu) -| (tr) |- (td);
    \end{tikzpicture}
  };
\end{tikzpicture}

%% file: Figures/3/yanking.tex
\begin{tikzpicture}
  \node (margin) {
    \begin{tikzpicture}
      \node[rectangle,minimum width=8mm,minimum height=10mm] (f) {};
      \coordinate[below=3mm of f.west] (A);
      \coordinate[left=10mm of A] (Ad);
      \draw (A) -- node[below] {\tiny \(U\)} (Ad);
      \coordinate[below=3mm of f.east] (B);
      \coordinate[right=10mm of B] (Bd);
      \draw (B) -- node[below] {\tiny \(U\)} (Bd);
      \coordinate[above=3mm of f.west] (Ui);
      \coordinate[left=3mm of Ui] (Uid);
      \draw (Uid) -- node[above] {\tiny \(U\)} (Ui);
      \coordinate[above=3mm of f.east] (Uo);
      \coordinate[right=3mm of Uo] (Uod);
      \draw (Uod) -- node[above] {\tiny \(U\)} (Uo);
      \draw (A) edge[out=0,in=180] (Uo);
      \draw (Ui) edge[out=0,in=180] (B);
      \coordinate[above=6mm of Uid] (Uiu);
      \coordinate[above=6mm of Uod] (Uou);
      \draw (Uiu) -- (Uou);
      \draw (Uid) edge[out=180,in=180,looseness=1.5] (Uiu);
      \draw (Uod) edge[out=0,in=0,looseness=1.5] (Uou);
      \coordinate[below=2mm of f.south] (td);
      \coordinate[above=6mm of f.north] (tu);
      \coordinate[left=7mm of f.west] (tl);
      \coordinate[right=7mm of f.east] (tr);
      \draw[dashed,draw=black!60] (td) -| (tl) |- (tu) -| (tr) |- (td);
    \end{tikzpicture}
  };
\end{tikzpicture}
\quad\quad &= \quad\quad
\begin{tikzpicture}
  \node (margin) {
    \begin{tikzpicture}
      \coordinate (O);
      \coordinate[left=5mm of O] (Ui);
      \coordinate[left=8mm of O] (Uid);
      \coordinate[right=5mm of O] (Uo);
      \coordinate[right=8mm of O] (Uod);
      \draw (Uid) -- node[above] {\tiny \(U\)} (Ui);
      \draw (Ui) -- (Uo);
      \draw (Uod) -- node[above] {\tiny \(U\)} (Uo);
    \end{tikzpicture}
  };
\end{tikzpicture}

%% file: Figures/3/delayed_trace.tex
\tr(f) \ = \ 
\begin{tikzpicture}
  \node[rectangle,draw=black,thick,minimum width=6mm,minimum height=10mm] (f) {\(f\)};
  \coordinate[below=3mm of f.west] (A);
  \coordinate[left=7mm of A] (Ad);
  \draw (Ad) -- (A);
  \coordinate[below=3mm of f.east] (B);  
  \coordinate[right=12mm of B] (Bd);
  \draw (Bd) -- (B);
  \coordinate[above=3mm of f.west] (Ui);
  \coordinate[left=3mm of Ui] (Uid);
  \draw (Uid) -- (Ui);
  \coordinate[above=3mm of f.east] (Uo);
  \node[circle,draw=black,thick,minimum size=4mm, inner sep=0,right=3mm of Uo] (delta) {\footnotesize \(\delta\)};
  \coordinate[right=1mm of delta] (Uod);
  \draw (Uo) -- (delta);
  \draw (Uod) -- (delta);
  \coordinate[above=6mm of Uid] (Uiu);
  \coordinate[above=6mm of Uod] (Uou);
  \draw (Uiu) -- (Uou);
  \draw (Uid) edge[out=180,in=180,looseness=1.5] (Uiu);
  \draw (Uod) edge[out=0,in=0,looseness=1.5] (Uou);
\end{tikzpicture}
\quad\quad;\quad\quad
\tr(\sigma) \ = \ 
\begin{tikzpicture}
  \node[rectangle,minimum width=8mm,minimum height=10mm] (f) {};
  \coordinate[below=3mm of f.west] (A);
  \coordinate[left=7mm of A] (Ad);
  \draw (Ad) -- (A);
  \coordinate[below=3mm of f.east] (B);  
  \coordinate[right=12mm of B] (Bd);
  \draw (Bd) -- (B);
  \coordinate[above=3mm of f.west] (Ui);
  \coordinate[left=3mm of Ui] (Uid);
  \draw (Uid) -- (Ui);
  \coordinate[above=3mm of f.east] (Uo);
  \node[circle,draw=black,thick,minimum size=4mm, inner sep=0,right=3mm of Uo] (delta) {\footnotesize \(\delta\)};
  \coordinate[right=1mm of delta] (Uod);
  \draw (A) edge[out=0,in=180] (Uo);
  \draw (Ui) edge[out=0,in=180] (B);
  \draw (Uo) -- (delta);
  \draw (Uod) -- (delta);  \coordinate[above=6mm of Uid] (Uiu);
  \coordinate[above=6mm of Uod] (Uou);
  \draw (Uiu) -- (Uou);
  \draw (Uid) edge[out=180,in=180,looseness=1.5] (Uiu);
  \draw (Uod) edge[out=0,in=0,looseness=1.5] (Uou);
\end{tikzpicture}
\ = \ 
\begin{tikzpicture}
  \node[circle,draw=black,thick,minimum size=4mm, inner sep=0] (delta) {\footnotesize \(\delta\)};
  \coordinate[left=3mm of delta] (Uid);
  \coordinate[right=3mm of delta] (Uod);
  \draw (Uid) -- (delta);
  \draw (delta) -- (Uod);
\end{tikzpicture}
\ \not= \ \id

%% file: Chapters/4.tex
\chapter{Weakly measured while loops}
\label{chap:while}

This chapter reproduces the contents of a publication due to Andres-Martinez and Heunen~\cite{WeakWhileLoop}; the text and notation has been revised to better fit this thesis.
The focus of this chapter is on classical control flow of quantum programs and, hence, measurements are fundamental. As such, the framework of completely positive weak trace reducing maps (CPTR) will be used (see Section~\ref{sec:CPTR} for a brief introduction).

A while loop tests a termination condition on every iteration but, on a quantum computer, such measurements perturb the evolution of the algorithm.
In this chapter we define a while loop primitive using weak measurements, offering a trade-off between the perturbation caused and the amount of information gained per iteration.
This trade-off is adjusted with a parameter set by the programmer.
We provide sufficient conditions that let us determine, with arbitrarily high probability, a worst-case estimate of the number of iterations the loop will run for.
As an example, we solve Grover's search problem using a while loop and prove the quadratic quantum speed-up is maintained.

This chapter is structured as follows: Section~\ref{sec:weak_meas} gives a brief introduction to weak measurements and defines the notion of \(\kappa\)-measurement; Section~\ref{sec:general} contains the main contribution of this chapter: the proposal of \(\kappa\)-while loops and the study of its properties; Section~\ref{sec:Grover} shows that Grover's algorithm may be implemented using a \(\kappa\)-while loop, and Section~\ref{sec:4_discussion} concludes the chapter with discussion of related work and open questions.

\section{Weak measurements (\emph{Preamble})}
\label{sec:weak_meas}

Roughly speaking, a weak measurement is ``a measurement which gives very little information about the system on average, but also disturbs the state very little''~\cite{ToddTutorial}.
For instance, the field of quantum feedback control uses weak measurements (often, continuous measurements) to monitor a state; the stream of measurement outcomes is used to control the strength of a Hamiltonian that corrects the system (see~\cite{QuantumFeedbackSurvey} for a survey).
Our approach is inspired by these ideas, contextualised for their application to algorithm design.
We restrict ourselves to the discrete-time setting and define a particular kind of parametrised measurement, the \(\kappa\)-measurement, that behaves as a weak measurement when \(\kappa\) is small.

Let \(H\) be a Hilbert space and let \(B\) be an orthonormal basis of \(H\); we wish to apply a measurement to test whether a state \(\psi \in H\) satisfies a predicate \(Q \colon B \to \{0,1\}\).
Assume the existence of a unitary \(O_Q \colon H \otimes \CTwo \to H \otimes \CTwo\) acting as the \emph{oracle} of predicate \(Q\):
\begin{equation*}
  O_Q \ket{x,p} = \ket{x,p \oplus Q(x)}.
\end{equation*}
Fix a value of parameter \(\kappa \in [0,1]\) and let \(P = \Span\{\bot,\top\}\) be an auxiliary space known as the \emph{probe}. We define a unitary \(E_{\kappa,Q}\) as follows:
\begin{align*}
  E_{\kappa,Q} &= (O_Q^\dagger \otimes \id_P) \, (\id_H \otimes \Lambda(R_\kappa)) \, (O_Q \otimes \id_P) \\
  \Lambda(R_\kappa) &= \ketbra{0}{0} \otimes \id_P + \ketbra{1}{1} \otimes R_\kappa \\
  R_\kappa &= \begin{pmatrix}
               \sqrt{1 - \kappa} &  \sqrt{\kappa} \\
               \sqrt{\kappa}     & -\sqrt{1 - \kappa} \\
              \end{pmatrix},
\end{align*}
The purpose of \(R_\kappa\) is to map \(\ket{\bot}\) to \(\alpha \ket{\bot} + \beta \ket{\top}\) so that \(\lvert \beta \rvert^2 = \kappa\).\footnote{Such a behaviour is not unique to this definition of \(R_\kappa\). For instance, \(Z(\theta) R_\kappa Z(\theta')\) for any \(\theta\) and \(\theta'\) is an equally valid choice for this unitary, where \(Z(\theta)\) is a \(Z\)-rotation of angle \(\theta\).}
In particular, if \(\kappa = 0\) then \(R_0 = \id\) and if \(\kappa = 1\) then \(R_1 \ket{\bot} = \ket{\top}\).
\(E_{\kappa,Q}\) is a quantumly controlled version of \(R_\kappa\) where the outcome from the oracle \(O_Q\) acts as the control.
Notice that the auxiliary qubit the oracle \(O_Q\) acts upon is restored to its initial state by \(O_Q^\dagger\) so it may be reused and we need not keep track of it. We initialise it to \(\ket{0}\) and omit it in further discussions.

\begin{definition} \label{def:kappa_measurement}
	A \GLS{\(\kappa\)-measurement}{k-measurement} of predicate \(Q\) may be applied on any density matrix \(\rho \in B(H)\) by the following procedure:
	\begin{itemize}
	  \item apply the unitary \(E_{\kappa,Q} \colon H \otimes P \to H \otimes P\) on the state \(\rho \otimes \ketbra{\bot}{\bot}\),
	  \item apply a measurement to discern between the orthogonal subspaces \(H \otimes \Span \{\bot\}\) and \(H \otimes \Span \{\top\}\).
	\end{itemize}
\end{definition}

Since \(E_{\kappa,Q}\) entangles the probe with the result of the oracle, the outcome of this measurement provides some information about the original state in \(H\).
If the outcome of the \(\kappa\)-measurement is \(\bot\) and \(\kappa < 1\), we obtain no definitive information on whether the state in \(H\) satisfies \(Q\) or not: for any state \(\rho \in B(H)\), the probability of outcome \(\top\) is
\begin{equation*}
	p_\top(\rho) = \kappa \cdot p_Q(\rho) 
\end{equation*}
where \(p_Q(\rho)\) is the probability of predicate \(Q\) being satisfied by \(\rho\), 
\begin{equation} \label{eq:pQ}
  p_Q(\rho) =  \tr \left( (\id_H \otimes \bra{1})\, O_Q \,(\rho \otimes \ketbra{0}{0})\, O_Q^\dagger \,(\id_H \otimes \ket{1}) \right).
\end{equation} 
The smaller \(\kappa\) is, the less likely it is that we read outcome \(\top\). 
On the other hand, when the outcome of the \(\kappa\)-measurement is \(\top\), no matter the value of \(\kappa\) we are \emph{certain} that the state left in space \(H\) satisfies predicate \(Q\).
In contrast, the map \(\cW \colon B(H) \to B(H)\) corresponding to outcome \(\bot\) does \emph{not} project the state to the subspace of \(H\) where \(Q\) is unsatisfiable:
\begin{equation} \label{eq:action_W}
  \cW(\rho) = \frac{W_\bot \,(\rho \otimes \ketbra{\bot}{\bot}) \,W_\bot^\dagger}{1 - p_\top(\rho)}
\end{equation}
where \(W_\bot = (\id_{H} \otimes \ketbra{\bot}{\bot})\,E_{\kappa,Q}\).

Intuitively, a \(\kappa\)-measurement models the behaviour of a device with internal state space \(P\) used to measure a property \(Q\) of a system \(H\), with \(\kappa\) indicating how strongly the probing device couples with the system.

\begin{remark} \label{rmk:known_probe_outcome}
After applying a measurement on a pure state, we may describe the result as a mixed state on the \Cstar-algebra \(B(H \otimes P)\). However, in doing so are omitting the information we have obtained from the classical outcome of the measurement: we in fact know whether the state is in \(H \otimes \Span \{\bot\}\) or in \(H \otimes \Span \{\top\}\).
For our purposes, it is more elucidating to describe the two possible outcomes separately.
Notice that when the state in \(H\) prior to measurement is a pure state both possible outcomes yield a pure state; this fact will be used in Section~\ref{sec:Grover} to simplify our discussions.
\end{remark}

\section{The \(\kappa\)-while loop}
\label{sec:general}

This section contains the main contribution of this chapter: we define \(\kappa\)-while loops and introduce the properties of \(\cA\)-guarantee and robustness which let us estimate the worst-case runtime of certain algorithms using \(\kappa\)-while loops. In short, a \(\kappa\)-while loop is a classically controlled while loop where the test of the termination condition is realised by a \(\kappa\)-measurement.

\subsection{Motivation}
\label{sec:motivation}

Fix a predicate \(Q\), and let \(\cC\) be a completely positive (weak) trace reducing (CPTR) map acting on \(B(H)\); we will refer to \(\cC\) as the \emph{body} of the loop.

\begin{definition}
Let \(\sigma \in B(H)\) be an arbitrary density matrix. Define a predicate \(\cA_{Q,\sigma} \colon \Nset \to \Bset\) as follows:
\begin{equation*}
  \cA_{Q,\sigma}(n) \iff p_Q(\cC^n(\sigma)) > \frac{1}{2}
\end{equation*}
where \(p_Q\) is given in~\eqref{eq:pQ}.
If \(\cA_{Q,\sigma}(n)\) is satisfied, we say \(n\) is an \(\cA\)-iteration --- read as \emph{active iteration}.
\end{definition}

If we can somehow identify an \(\cA\)-iteration \(m\), we may use a simple approach to find a state satisfying \(Q\): apply \(\cC^m(\sigma)\) and then perform a projective measurement; if the outcome does not satisfy \(Q\), repeat the process from the initial state \(\sigma\). Thanks to \(\cA_{Q,\sigma}(m)\) being satisfied, the probability of succeeding at some point within \(k\) restarts is \(1 - 1/2^k\) so the probability of success quickly approaches \(1\) as \(k\) increases.
This is an efficient approach whenever a small \(\cA\)-iteration \(m\) is known. 
In fact, this is how the standard Grover's algorithm works, choosing an iteration \(m\) where \(p_Q\) is maximised.

However, this measure-restart strategy can only be implemented if we already know when \(\cA\)-iterations will occur.
The distribution of \(\cA\)-iterations may become unpredictable as soon as some randomness is introduced in the body of the loop.
For instance, consider the standard Grover's algorithm being implemented on a faulty machine where, without notice, the quantum memory may reset to its initial state. 
If such an event occurs mid-computation the evolution is effectively restarted, but the algorithm -- oblivious to the memory reset -- continues for only the remaining fixed number of iterations.
In contrast, a version of Grover's algorithm that uses a \(\kappa\)-while loop (such as the one we describe in Section~\ref{sec:Grover}) will, by definition, keep iterating until it succeeds to find the target state.
Although a contrived example, this illustrates how a while loop could provide reliability against unpredictable behaviour.
More realistic scenarios that would lead to an unpredictable distribution of \(\cA\)-iterations may come from algorithms that, by design, explore the state space in an unpredictable manner; Section~\ref{sec:beyond_quantum_search} proposes the field of quantum walk-based algorithms as an area were such examples may perhaps be found.

In this section we present the concept of \(\kappa\)-while loops as an abstract quantum programming construct.
In Section~\ref{sec:Grover} we present a version of Grover's algorithm that uses a \(\kappa\)-while loop and maintains the quadratic quantum speed-up.
In the future, we hope to apply \(\kappa\)-while loops to practical problems where the distribution of \(\cA\)-iterations cannot be predicted.
To this end, we provide sufficient conditions that let us determine, with arbitrarily high probability, a worst-case estimate of the number of iterations the loop will run for.
Remarkably, these conditions do not require us to know the precise distribution of \(\cA\)-iterations, but only a guarantee of their proportion throughout the algorithm in the worst-case scenario.

\subsection{Definition and properties}
\label{sec:main}

In Figure~\ref{fig:kappa_while} we propose the syntax for the \(\kappa\)-while loop and its implementation in a quantum programming language with classical control flow. The syntax of the \(\kappa\)-while loop is meant to be read as ``repeat \(\cC\) while it is \emph{not certain} that \(Q\) is satisfied''.
On each iteration, the value of the state \(\rho\) will be updated to \(\cC(\rho)\), followed by a \(\kappa\)-measurement of predicate \(Q\) (see Section~\ref{sec:weak_meas}). If the outcome of the \(\kappa\)-measurement is \(\bot\), the loop keeps iterating; the state \(\rho\) becomes \(\cW(\rho)\) where \(\cW\) is given in~\eqref{eq:action_W}. Otherwise, outcome \(\top\) causes the loop to halt and we succeed in obtaining a state \(\rho\) that satisfies \(Q\).

\begin{figure}
	\input{Figures/4/kappa_while}
	\caption{\emph{Left:} the syntax we use to represent a \(\kappa\)-while loop; \(\kappa \in [0,1]\) is a parameter set by the programmer and \(Q\) is the predicate to be measured. \emph{Right:} the pseudocode that implements the \(\kappa\)-while loop on a programming language with classical control flow, following the notation from~\cite{YingBook}; \(E_{\kappa,Q}\) is defined in Section~\ref{sec:weak_meas} and the condition \(M[q] = \bot\) indicates that \(q\) is measured to distinguish between \(H \otimes \Span \{\bot\}\) and \(H \otimes \Span \{\top\}\), with the condition being satisfied if \(q\) collapses to the former.}
	\label{fig:kappa_while}
\end{figure}

\begin{definition} \label{def:guarantee}
An isotone function \(f \colon \Nset \to \Nset\) is a \(\cA_{Q,\sigma}\)-\emph{guarantee} if for all \(n \in \Nset\):
	\begin{equation} \label{eq:guarantee}
	  n \ \leq \ \abs{\{ k \in \Nset \mid k \leq f(n),\, \cA_{Q,\sigma}(k) \}}
	\end{equation}
\end{definition}

If such a function \(f\) exists we are promised that there will be at least \(n\) active iterations within the first \(f(n)\) applications of \(\cC\). In principle, this need not be a tight bound. 
For instance, we may only know that within the first thousand iterations there are at least ten \(\cA\)-iterations; a valid \(\cA_{Q,\sigma}\)-guarantee in this case would be a function \(f\) such that \(f(1)=f(2)=\dots=f(10) = 1000\). This gives us little information about when any of these \(\cA\)-iterations actually occur.

\begin{definition} \label{def:robustness}
	The evolution induced by a CPTR map \(\cC\) on a state \(\sigma\) is said to be \(\varepsilon\)-\emph{robust} to \(\kappa\)-measurements of predicate \(Q\) if there is an isotone function \(g \colon \Nset \to \Nset\) such that:
	\begin{equation}
	  \forall n \in \Nset, \;\; \exists m \leq g(n) \colon \quad \abs{p_Q(\cC^n(\sigma)) - p_Q((\cW\cC)^m(\sigma))} \leq \varepsilon
	\end{equation}
	where \(\varepsilon < 1/2\) and \(\cW\) is given in~\eqref{eq:action_W}. The function \(g\) is said to be a \emph{witness} of the robustness.
\end{definition}

In practice, if the evolution is \(\varepsilon\)-robust this means that the weakly-measured evolution will provide a similar curve of the success probability \(p_Q\) throughout its iterations, but certain intervals may span a larger number of iterations. Then, the witness function \(g\) provides a bound on how much slower the weakly-measured evolution may be.

\begin{lemma}
\label{lem:comp_guarantee}
If \(f\) is a \(\cA_{Q,\sigma}\)-guarantee and \(g\) witnesses that \(\cC\) is \(\varepsilon\)-robust to \(\kappa\)-measurements of predicate \(Q\), then \(g \circ f\) is a \(\cA'_{Q,\sigma}\)-guarantee where \(\cA'_{Q,\sigma}\) is the predicate satisfying:
\begin{equation*}
  \cA'_{Q,\sigma}(n) \iff p_Q \left( (\cW\cC)^n(\sigma) \right) > \frac{1}{2} - \varepsilon
\end{equation*}
\end{lemma} 
\begin{proof}
	Since \(f\) is a \(\cA_{Q,\sigma}\)-guarantee, at least \(n\) instances of \(\cA_{Q,\sigma}\)-iterations will occur within the first \(f(n)\) applications of \(\cC\). If we switch to weakly measured iterations \(\cW\cC\), robustness tells us that within the first \(gf(n)\) iterations of the loop we will find those \(n\) active iterations again, but the probability \(p_Q\) might differ by \(\varepsilon\), so it is at least \(p_Q > 1/2 - \varepsilon\). Thus, for every \(n \in \Nset\) it follows that:
	\begin{equation*}
		n \ \leq \ \abs{\{ k \in \Nset \mid k \leq gf(n),\, \cA'_{Q,\sigma}(k) \}}
	\end{equation*}
	This concludes the proof.
\end{proof}

Lemma~\ref{lem:comp_guarantee} guarantees that at least \(n\) of the first \(gf(n)\) iterations of the \(\kappa\)-while loop will be \(\cA'\)-iterations. For each of these \(\cA'\)-iterations, the probability of outcome \(\top\) is at least \(\kappa(1/2 - \varepsilon)\) since \(p_\top = \kappa \cdot p_Q\). Within \(N\) instances of \(\cA'\)-iterations the probability of loop termination is:
\begin{equation*}
	P_{succ} > 1 - \big(1 - \kappa(1/2 - \varepsilon)\big)^N
\end{equation*}
Since for all \(x \in (0,1)\) we have that \((1 - x)^\frac{1}{x} < \tfrac{1}{e}\), if the algorithm is allowed to run for at least \(N = \frac{2}{\kappa(1-2\varepsilon)}\) \(\cA'\)-iterations, the probability of success will be:
\begin{equation*}
  P_{succ} > 1 - \tfrac{1}{e} > \tfrac{1}{2}
\end{equation*}
Therefore, with probability higher than \(1/2\), our \(\kappa\)-while loop will halt within its first 
\begin{equation} \label{eq:T_bound}
  T = gf\left(\frac{2}{\kappa(1-2\varepsilon)}\right)
\end{equation}
iterations.
Thus, equation~\eqref{eq:T_bound} estimates the median of the number of iterations the \(\kappa\)-while loop runs for before halting.
More generally, for any \(c \in \Nset\), the probability of success after \(c N\) instances of \(\cA'\)-iterations is greater than \(1 - 1/e^c\). Since \(1/e^c\) quickly approaches \(0\) as \(c\) increases, we may say that, with arbitrarily high probability, the loop halts successfully within \(T_c = gf(cN)\) iterations for small \(c\). Hence, we may claim that the time complexity of the loop is asymptotically bounded by \(gf(cN)\) or, in other words, its time complexity is
\begin{equation*}
	\bigO{gf\left(\frac{2c}{\kappa(1-2\varepsilon)}\right)}.
\end{equation*}
Hence, if we manage to choose appropriate \(\kappa\), \(\varepsilon\), \(f\) and \(g\), we can prove statements about quantum speed-ups --- under the assumption that the time complexity of the best classical algorithm is known. Notice that if we do not manage to find witness functions \(f\) and \(g\) to prove certain quantum speed-up this does not rule out its existence.

\begin{remark}
Definition~\ref{def:guarantee} can be weakened so that instead of \(f\) being a `deterministic' guarantee of \(\cA_{Q,\sigma}\), we only impose that \(f(n)\) satisfies~\eqref{eq:guarantee} with probability higher than \(\eta\).
In such a case, after \(T_c\) iterations we can achieve a success probability arbitrarily close to \(\eta\).
\end{remark}

\section{An example: Grover's algorithm}
\label{sec:Grover}

In Grover's search problem~\cite{Grover}, we are given an unsorted set of elements \(B\), about which we know no structure or heuristics, and a predicate \(\chi \colon B \to \{0,1\}\) that satisfies \(\chi(\star) = 1\) for only one element \(\star \in B\).
We are tasked with finding this marked element \(\star\).
The standard Grover's algorithm defines an iteration operator \(G\) using an oracle of \(\chi\) and applies it a fixed number of times \(K \approx \sqrt{\abs{B}}\) on an initial state \(\ket{\psi}\).
Afterwards, a measurement on the basis \(B\) is applied, finding the marked element with high probability.
Thus, the algorithm achieves a quadratic quantum speed-up: it requires \(\bigO{\sqrt{\abs{B}}}\) iterations whereas a classical algorithm would require \(\bigO{\abs{B}}\) iterations since there is no available heuristic to guide the search.
In this section we discuss a different approach to Grover's search problem where a \(\kappa\)-while loop is used instead.

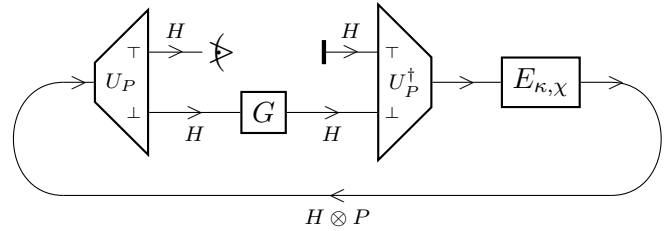
\begin{figure}
	\input{Figures/4/code_walk}
	\caption{\emph{Left:} pseudocode describing our algorithm. The while loop is controlled by a \(\kappa\)-measurement as described in Section~\ref{sec:main}. \emph{Right:} the algorithm's information flow; \(E_{\kappa,\chi}\) is given in Section~\ref{sec:weak_meas}, \(U_P\) applies the canonical isomorphism \(H \otimes P \cong H \oplus H\), separating with respect to the orthogonal basis \(\{\ket{\bot},\ket{\top}\}\) of \(P\).}
	\label{fig:code_walk}
\end{figure}

The algorithm's pseudocode is given in Figure~\ref{fig:code_walk}. When the \(\kappa\)-while loop halts, we are \emph{certain} that the state is \(\ket{\star}\).
We do not fix the number of times \(G\) is applied; instead, we fix the measurement strength \(\kappa\).
We will see that for \(\kappa \approx \abs{B}^{-1/2}\) the loop terminates within \(c\sqrt{\abs{B}}\) iterations with arbitrarily high probability, where \(c\) is a small constant.
Thus, the quadratic quantum speed-up of Grover's algorithm is preserved.

It is important to remark that, in the case of Grover's iterator \(G\), we can easily predict when the \(\cA\)-iterations occur (see Section~\ref{sec:motivation}).
Therefore, \(\kappa\)-while loops do not provide an algorithmic advantage on this problem.
Instead, we present Grover's search problem as a simple proof of concept showing how classically controlled while loops may be to monitor a quantum algorithm without destroying the quantum speed-up. 
This approach to Grover's problem was first proposed by Mizel~\cite{Mizel}.
This section reformulates Mizel's results in the broader framework of Section~\ref{sec:main}.

\subsection{Standard Grover's algorithm (\emph{Preamble})}
\label{sec:std_Grover}

We first give a brief description of the standard Grovers algorithm and introduce our notation.
The standard Grover's algorithm acts on a quantum state in \(H = \Span{B}\), starting from the uniform superposition state:
\begin{equation}
	\ket{\psi} = \tfrac{1}{\sqrt{\abs{B}}} \sum_{b \in B} \ket{b}.
\end{equation}
Let \(\ket{\psi_1} = \ket{\star}\) be the target state and 
\begin{equation*}
  \ket{\psi_0} = \tfrac{1}{\sqrt{\abs{B}-1}} \left( \ket{\psi} - \sqrt{\abs{B}} \cdot \ket{\psi_1} \right)
\end{equation*}
For every angle \(a \in [0,2\pi)\) define the state \(\ket{a}\) as follows:
\begin{equation} \label{eq:angle_state}
  \ket{a} = \cos{a}\,\ket{\psi_0} + \sin{a}\,\ket{\psi_1}.
\end{equation} 
Let \(\alpha = \arcsin{\braket{\star}{\psi}} = \arcsin{\abs{B}^{-1/2}}\); it is customary to assume that \(\abs{B}\) is large enough so that the approximation \(\alpha \approx \abs{B}^{-1/2}\) is reasonable. Notice that \(\ket{\alpha} = \ket{\psi}\).

For every state \(\ket{\varphi} \in H\) let:
\begin{equation*}
  S_\varphi = 2\ketbra{\varphi}{\varphi} - \id_H
\end{equation*}
Grover's iteration is given by the operator \(G = S_\psi \, S_{\psi_0}\) where \(S_{\psi_0}\) may be implemented using an oracle of \(\chi\).

On a state \(\ket{a}\), the action of Grover's iteration \(G\) is:
\begin{equation} \label{eq:action_G}
  G\ket{a} = S_\psi \, S_{\psi_0}\ket{a} = S_\psi\ket{-a} = \ket{a+2\alpha}
\end{equation}
and after \(k\) iterations:
\begin{equation*}
  G^k\ket{\psi} = \cos{(\alpha + 2k\alpha)}\ket{\psi_0} + \sin{(\alpha + 2k\alpha)}\ket{\psi_1}
\end{equation*}
The standard Grover's algorithm applies \(G\) a total of \(K = \floor{(\pi\sqrt{\abs{B}})/4}\) times on \(\ket{\psi}\) so that the amplitude of \(\ket{\psi_1}\) is maximised. Then, the state is measured on the basis \(B\), finding the marked element \(\star\) with high probability.

\subsection{While loop approach}
\label{sec:weak_Grover}

In this section we discuss the algorithm given in Figure~\ref{fig:code_walk}. To show that it retains the quadratic quantum speed-up, we provide an \(\cA\)-guarantee and prove that \(G\) is robust to \(\kappa\)-measurements of \(\chi\), as described in Section~\ref{sec:main}.
Our framework is general enough to deal with mixed states and CPTR maps but, in the case of Grover's algorithm, the body of the \(\kappa\)-while loop \(G\) is unitary and both the initial state \(\ket{\psi}\) and the target \(\ket{\star}\) are pure states.
Therefore, as discussed in Remark~\ref{rmk:known_probe_outcome}, all states involved in our analysis will be pure states.

\begin{lemma} \label{lem:G_guarantee}
	The function \(f \colon \Nset \to \Nset\) given by \(f(n) = 2n+K\) where \(K = \floor{\tfrac{\pi\sqrt{\abs{B}}}{4}}\) is a \(\cA_{\chi,\ket{\psi}}\)-guarantee.
\end{lemma} \begin{proof}
	By definition~\eqref{eq:angle_state} of \(\ket{a}\), 
	\begin{equation*}
	  p_\chi(\ket{a}) = \abs{\braket{\star}{a}}^2 = \sin^2 a.
	\end{equation*} 
	Whenever \(a \in (\tfrac{\pi}{4},\tfrac{3\pi}{4})\) or \(a \in (\tfrac{5\pi}{4},\tfrac{7\pi}{4})\), we have \(p_\chi(\ket{a}) > \tfrac{1}{2}\).
	Moreover, \(G\ket{a} = \ket{a+2\alpha}\) according to~\eqref{eq:action_G}, so the angle is increased a constant amount on each iteration.
	Thus, within \(m \in \Nset\) iterations we know that the number of \(\cA\)-iterations is at least \(\tfrac{m - K}{2}\) where \(K\) (defined in the claim) is the number of applications of \(G\) it takes to traverse a \(\tfrac{\pi}{2}\) angle, which is the longest interval of consecutive iterations not satisfying \(\cA_{\chi,\ket{\psi}}\).
	Assuming the worst case scenario we find that within \(m = 2n + K\) iterations it is guaranteed that \(n\) of them are \(\cA_{\chi,\ket{\psi}}\)-iterations, \ie{} for all \(n \in \Nset\)
	\begin{equation*}
		n \ \leq \ \abs{\{ k \in \Nset \mid k \leq f(n),\, \cA_{\chi,\ket{\psi}}(k) \}}
	\end{equation*}
	thus, proving the claim.
\end{proof}

Consider applying a \(\kappa\)-measurement at iteration \(n\): if the outcome is \(\top\), the state becomes \(\ket{\star}\), otherwise it suffers a collapse towards \(\ket{\psi_0}\) determined by the projector \(W_\bot\) from Section~\ref{sec:weak_meas}:
\begin{equation} \label{eq:collapse}
\begin{aligned}
  W_\bot \ket{a,\bot} &=  \cos{a}\,W_\bot\ket{\psi_0,\bot} + \sin{a}\,W_\bot\ket{\psi_1,\bot} \\
   &= \cos{a}\,\ket{\psi_0,\bot} + \sin{a}\,(\id_H \otimes \ketbra{\bot}{\bot})\,E_{\kappa,\chi}\ket{\psi_1,\bot} \\
   &= \cos{a}\,\ket{\psi_0,\bot} + \sin{a}\,\sqrt{1 - \kappa}\,\ket{\psi_1,\bot}
\end{aligned}
\end{equation}
Thus, the output is still in a superposition between \(\ket{\psi_0}\) and \(\ket{\psi_1}\).
Notice that \(W_\bot\) returns an unnormalised state; the actual CPTP map \(\cW\) from~\eqref{eq:action_W} takes care of the normalization.
Then, we may describe the action of \(\cW\) as \(\ket{a} \xmapsto{\cW} \ket{a - \theta(a,\kappa)}\) on every angle \(a \in [0,2\pi)\) where the value of \(\theta(a,\kappa)\) may be calculated from~\eqref{eq:collapse} following the geometric construction of Figure~\ref{fig:theta_triangle}.
We now find an upper bound of \(\theta(a,\kappa)\); this will be used to prove robustness of \(G\).
To reduce clutter, we use the shorthand:
\begin{equation} \label{eq:xi}
  \xi = \sqrt{1 - \kappa}
\end{equation}

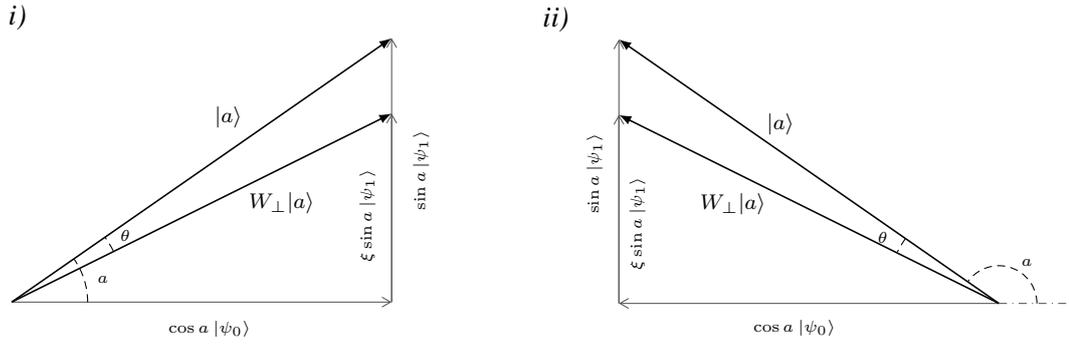
\begin{figure}
	\centering
	\input{Figures/4/theta_triangle}
	\caption{Geometric relation between the angle \(a\) before weak measurement and the offset \(\theta(a,\kappa)\) after measuring outcome \(\bot\). The construction is provided when \emph{i)} \(a \in [0,\tfrac{\pi}{2})\)and when \emph{ii)} \(a \in [\tfrac{\pi}{2},\pi)\).}
	\label{fig:theta_triangle}
\end{figure}

\begin{lemma} \label{lem:theta_bound}
For any \(\kappa\) and angle \(a\)
\begin{equation*}
  \abs{\theta(a,\kappa)} \leq \arcsin \left( \frac{1-\xi}{1+\xi} \right).
\end{equation*}
\end{lemma} \begin{proof}
	Suppose \(a \in \left[ 0,\tfrac{\pi}{2} \right)\) and use the properties of sines on the triangle from Figure~\ref{fig:theta_triangle}i to obtain:
	\begin{equation*}
	  \frac{\sin (a - \theta)}{\xi \sin a} = \frac{\sin (\pi/2 - a + \theta)}{\cos a}
	\end{equation*}
	This simplifies as follows:
	\begin{equation*}
	 \begin{aligned}
	  && \frac{\sin (a - \theta)}{\xi \sin a} &= \frac{\cos (a - \theta)}{\cos a} \\
	  \iff&& \tan (a - \theta) &= \xi \tan a \\
	  \iff&& \frac{\tan a - \tan \theta}{1 + \tan a \tan \theta} &= \xi \tan a \\
	 \end{aligned}
	\end{equation*}
	Solving for \(\theta\) gives:
	\begin{equation} \label{eq:theta}
	  \theta(a,\kappa) = \arctan \left( \frac{(1 - \xi) \tan a}{1 + \xi \tan^2 a} \right)
	\end{equation}

	If \(a \in \left[\tfrac{\pi}{2},\pi\right)\) instead, a similar analysis yields:
	\begin{equation}\label{eq:thetaNeg}
	  \theta(a,\kappa) = -\arctan \left( \frac{(1 - \xi) \tan a}{1 + \xi \tan^2 a} \right)
	\end{equation}
	which only differs from equation~\eqref{eq:theta} in the sign. If \(a\) is in the third quadrant, then we obtain equation~\eqref{eq:theta} and, if \(a\) is in the fourth quadrant, we get~\eqref{eq:thetaNeg}.
	These sign changes are convenient: the geometric analysis in Figure~\ref{fig:theta_triangle} shows that in the first (and third) quadrant, the resulting angle is \(a - \abs{\theta}\), whereas in the second (and fourth) quadrant it is \(a + \abs{\theta}\). Thus the resulting angle is \(a - \theta\) regardless of which quadrant \(a\) lies in.

	Next, we determine the maximum value of \(\theta(a,\kappa)\). To do so, fix \(\kappa\) and find the critical points of \(\theta\) as a function on \(a\). These happen where the derivative
	\begin{equation*}
	  \frac{\mathrm{d}\theta}{\mathrm{d}a} = 1 - \frac{\xi}{\cos^2 a + \xi^2\sin^2 a}
	\end{equation*}
	vanishes. Critical points occur periodically (once in every quadrant), but all of them reach the same absolute value. The critical point within the first quadrant happens at:
	\begin{equation*}
	  a = \arccos \left( \sqrt{\frac{\xi}{\xi + 1}} \right)
	\end{equation*}
	When applied to equation~\eqref{eq:theta} this identity gives a tight upper bound:
	\begin{equation*}
	  \abs{\theta(a,\kappa)} \leq \arctan \left( \frac{1-\xi}{2\sqrt{\xi}} \right)
	\end{equation*}
	Using the equality \(\sin(\arctan(x)) = \frac{x}{\sqrt{x^2 + 1}}\), the fact that \(\xi \in [0,1]\) and some basic algebra we reach:
	\begin{equation*}
	  \sin \abs{\theta(a,\kappa)} \leq \frac{1-\xi}{1+\xi}
	\end{equation*}
	thus, proving the claim.
\end{proof}

\begin{corollary} \label{cor:theta_leq_kappa}
For every angle \(a \in [0,2\pi)\) the bound 
\begin{equation*}
	\abs{\theta(a,\kappa)} \leq \arcsin \kappa
\end{equation*}
holds.
\end{corollary}\begin{proof}
	By definition, \(\kappa \in [0,1]\). Using this fact along with~\eqref{eq:xi} and simple algebra we can prove that:
	\begin{equation*} 
	  0 \leq \frac{1-\xi}{1+\xi} \leq \kappa.
	\end{equation*}
	Then, the corollary follows from Lemma~\ref{lem:theta_bound}.
\end{proof}

According to these results, if \(\kappa \approx 1\) the collapse \(\theta(a,\kappa)\) can be up to \(\tfrac{\pi}{2}\), thus always sending the state back to \(\ket{\psi_0}\), preventing a gradual evolution of the angle. In such a case, the algorithm would lose its quantum speed-up. 
Interestingly, Corollary~\ref{cor:theta_leq_kappa} shows we can keep \(\theta(a,\kappa)\) small by bounding \(\kappa\), so that the action of \(G\) overcomes the effect of the collapse.

\begin{lemma} \label{lem:G_robustness}
	The evolution of the unitary map \(G\) on the uniform superposition \(\ket{\psi}\) is \(\varepsilon\)-robust to \(\kappa\)-measurements of \(\chi\) for
	\begin{equation*}
	  \kappa \leq \tfrac{1}{\sqrt{\abs{B}}} \quad\quad \text{ and } \quad\quad \varepsilon = \sin 3\alpha
	\end{equation*}
	witnessed by \(g(n) = 2n\).
\end{lemma} 
\begin{proof}
	Let \(\{a_n\}\) be the following sequence of angles:
	\begin{equation*}
		\begin{aligned}
		  a_{n+1} &= a_n + 2\alpha \\
		  a_0     &= \alpha
		\end{aligned}
	\end{equation*}
	Recall from Section~\ref{sec:std_Grover} that \(G^n \ket{\psi} = \ket{a_n}\).
	Similarly, we define a sequence \(\{b_n\}\) satisfying \((\cW G)^n\ket{\psi} = \ket{b_n}\):
	\begin{equation} \label{eq:recurrence}
		\begin{aligned}
		  b_{n+1} &= b_n - \theta(b_n,\kappa) + 2\alpha \\
		  b_0     &= \alpha.
		\end{aligned}
	\end{equation}
	Remember that \(\sin \alpha = \abs{B}^{-1/2}\), so the bound imposed on \(\kappa\) in the claim can be rephrased as \(\kappa \leq \sin \alpha\). 
	Thanks to Corollary~\ref{cor:theta_leq_kappa}, this implies \(\abs{\theta(b_n,\kappa)} \leq \alpha\) on every iteration.
	Then \(\alpha \leq b_{n+1} - b_n \leq 3\alpha\) for every \(n \in \Nset\), whereas \(a_{n+1} - a_n = 2\alpha\).
	Therefore, for some \(m \leq 2n\) we have that \(b_m \leq a_n \leq b_{m+1}\) and \(a_n - b_m \leq 3\alpha\).
	Since \(p_\chi(G^n\ket{\psi}) = p_\chi(\ket{a_n}) = \sin^2 a_n\) and \(p_\chi((\cW G)^m\ket{\psi}) = p_\chi(\ket{b_m}) = \sin^2 b_m\), it follows that:
	\begin{equation*}
		\abs{p_\chi(G^n\ket{\psi}) - p_\chi((\cW G)^m\ket{\psi})} \ = \ \abs{\sin^2 a_n - \sin^2 b_m}.
	\end{equation*}
	If we now restrict to \(a_n \in \left[ 0, \tfrac{\pi}{2} \right)\) and assume that \(\alpha\) is small enough so that \(a_n - 3\alpha \geq 0\) we find that \(\sin^2 a_n - \sin^2 b_m \leq \sin^2 a_n - \sin^2\, (a_n - 3\alpha)\); then, using the identity \(\sin^2 x - \sin^2(x-y) = \sin(y) \sin(2 x - y)\) (which holds for all \(x,y \in \Rset\)) we conclude that:
	\begin{equation*}
		\sin^2 a_n - \sin^2 b_m \leq \sin 3\alpha - \sin\, (2a_n - 3\alpha) \leq \sin 3\alpha.
	\end{equation*}
	A similar argument follows if \(a_n \in \left[ \tfrac{\pi}{2}, \pi \right)\) or if \(a_n\) is in the third or fourth quadrant so that, for every \(n \in \Nset\):
	\begin{equation*}
	  \exists m \leq g(n) \colon \quad \abs{p_\chi(G^n\ket{\psi}) - p_\chi((\cW G)^m\ket{\psi})} \ \leq \ \sin 3\alpha
	\end{equation*}
	thus, proving the claim.
\end{proof}

From the general result of Lemma~\ref{lem:comp_guarantee} and the discussion following it along with the assumption that \(\alpha\) is small --- so that \(\sin 3\alpha < 1/2\) --- we conclude that, with arbitrarily high probability, the algorithm in Figure~\ref{fig:code_walk} will succeed in finding the marked element \(\ket{\star}\) within
\begin{equation*}
  T_c = gf\left(\frac{2c}{\kappa(1-2\varepsilon)}\right)
\end{equation*}
iterations for small \(c \in \Nset\).
Combining this with Lemmas~\ref{lem:G_guarantee} and~\ref{lem:G_robustness} we find that, for \(\kappa \leq \abs{B}^{-1/2}\), the number of iterations is within \(\bigO{\sqrt{\abs{B}}}\). 

\begin{remark} \label{rmk:constant_factor}
Our analysis yields an arguably large constant factor
\begin{equation*}
  T_c \approx (8c + \tfrac{\pi}{2})\sqrt{\abs{B}}.
\end{equation*}
However, it is important to remark that this factor is an overestimation due to the simplifications applied when finding \(f\) and \(g\) in this section. With these simplifications we intended to prioritise clarity of our proof stategy. Tighter bounds would yield a more accurate factor. In fact, Figure~\ref{fig:sampling} suggests the average number of iterations before success is approximately \(2\sqrt{\abs{B}}\).

Moreover, our choice of \(\kappa = \abs{B}^{-1/2}\) is not optimal.
To prove Lemma~\ref{lem:G_robustness} we only needed that the collapse on each iteration was smaller than the increment of the angle due to \(G\), \ie{} for every angle \(a\) we required that \(\theta(a,\kappa) < 2\alpha\). If we set \(\kappa = 5\abs{B}^{-1/2}\) and use the bound of \(\theta(a,\kappa)\) given by Lemma~\ref{lem:theta_bound}, we can verify that for an instance of \(\abs{B} = 10^6\) we have the bound \(\theta(a,\kappa) < 1.3\alpha\), so the collapse is appropriately bounded below \(2\alpha\).
In this particular case, numerical analysis shows that the average number of iterations required is slightly smaller than \(K = \tfrac{\pi}{4}\sqrt{\abs{B}}\) where \(\floor{K}\) is the number of iterations the standard Grover's algorithm runs for.
The optimal value of \(\kappa\) depends on the parameters of the problem --- \ie{} it may be different for different values of \(\abs{B}\).
We conjecture that such an optimal value of \(\kappa\) exists for every instance of Grover's problem so that our approach and the standard one coincide in their expected number of iterations.

The realisation that for a fine-tuned value of \(\kappa\) our algorithm would perform as well as the standard Grover's algorithm was first discussed in a recent paper~\cite{Differentiable}. The paper provides an automatic procedure based on differentiable programming that looks for optimal values of the parameters of a quantum program. The authors discuss our algorithm as a use case of their approach and, for small values of \(\abs{B}\), they find optimal values for \(\kappa\) satisfying the conjecture above.
\end{remark}

\begin{figure}
	\centering
	\includegraphics[scale=0.43]{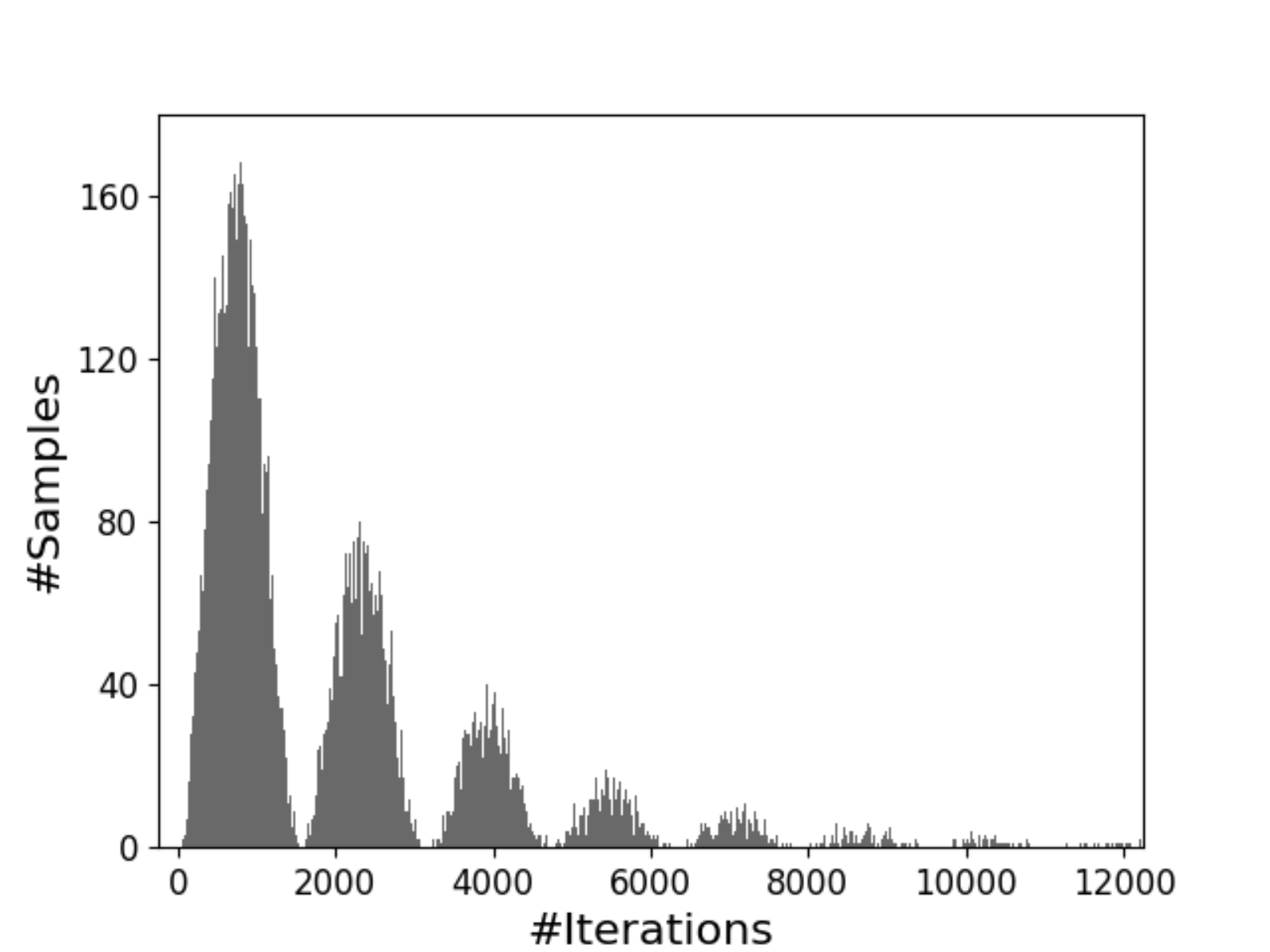}
	\includegraphics[scale=0.43]{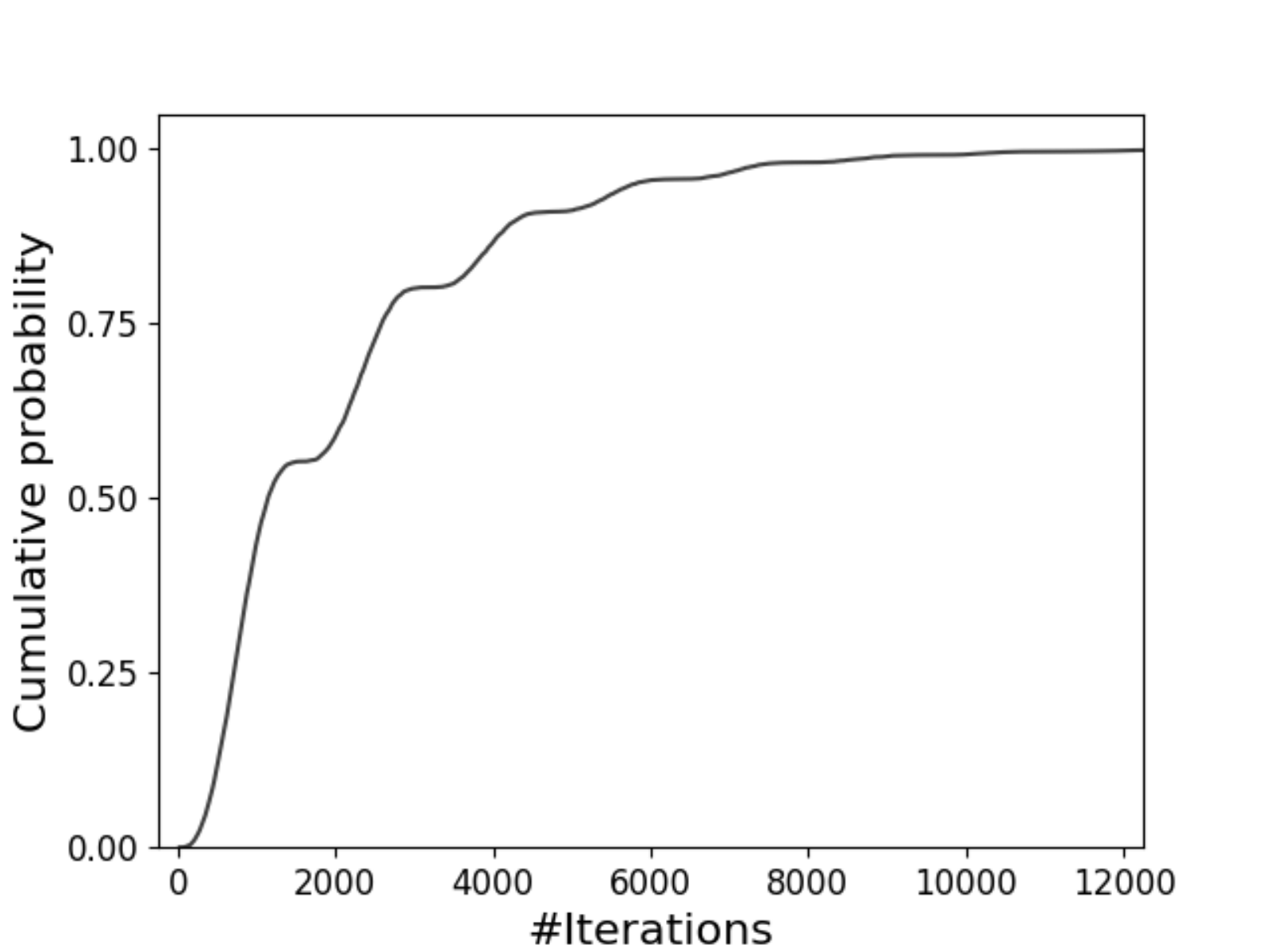}
	\caption{Distribution of the number of iterations before success for parameters \(\abs{B} = 10^6\) and \(\kappa = \abs{B}^{-1/2}\); \emph{left:} histogram, \emph{right:} cumulative probability distribution. Drawn from \(10000\) samples, obtained by sampling the success probability \(p_\top = \kappa\sin^2{b_n}\) where \(b_n\) is given by~\eqref{eq:recurrence}. The median is approximately \(1000\) iterations, and mean \(2000\) iterations, \ie{} approximately \(\tfrac{1}{\kappa}\) and \(\tfrac{2}{\kappa}\) respectively.}
	\label{fig:sampling}
\end{figure}

Figure~\ref{fig:sampling} helps us visualise the behaviour of the algorithm: as the angle \(b_n\) monotonically increases throughout the iterations, the instantaneous probability of success \(p_\top = \kappa\sin^2{b_n}\) changes periodically. The peaks of the histogram correspond to intervals of \(\cA\)-iterations where the probability of success approaches its maximum (that is, \(p_\top \approx \kappa\)); these alternate with troughs where the probability of halting approaches \(0\).
These periodic stages of unlikely termination are the cause of the plateaus in the second plot from Figure~\ref{fig:sampling} where the cumulative probability of halting --- \ie{} for each \(n \in \Nset\), the probability of halting before \(n\) iterations --- is graphed.
Apart from these plateaus, the curve resembles that of the cumulative probability of a geometric distribution, which is to be expected since these model experiments that ``continue until success".
The data points used to draw the figures were obtained by sampling from the success probability \(p_\top = \kappa\sin^2{b_n}\) throughout iterations until success, where \(b_n\) is given by~\eqref{eq:recurrence}; no simulation of the quantum algorithm is required since we are not running the computation itself.

\section{Discussion}
\label{sec:4_discussion}

The \(\kappa\)-while loop is a programming construct that offers a promising abstraction: a way to `peek' at quantum states while computing without sacrificing asymptotic quantum speed-up. The key challenge to its use is the need to strike a balance --- by tuning the measurement strength \(\kappa\) --- so that the collapse per iteration is low enough and the information gained is sufficient. The body of a \(\kappa\)-while loop may be any CPTR map, hence \(\kappa\)-while loops may be nested inside each other. As with classical while loops, termination implies the satisfaction of the predicate, which is a useful feature for the analysis of program correctness.
We emphasise that these \(\kappa\)-while loops can be realised in any quantum programming language with classical control flow, as indicated in Figure~\ref{fig:kappa_while}.
More precisely, \(\kappa\)-while loops are a particular instance of the classically controlled while loops described in~\cite{YingLoop}.
Hence, we expect it would be immediate to apply to \(\kappa\)-while loops any findings from the literature on classical control flow of quantum programs; for instance, verification of program correctness~\cite{YingHoare}, study of loop termination~\cite{YingTermination,YingLoop} and semantics~\cite{SelingerPL,YingBook}.
We conclude this chapter discussing previous works and proposing new avenues of research.

\subsection{Related work}

The key component of the \(\kappa\)-while loop is a weak measurement parametrised by \(\kappa\). Weak measurements are not new to quantum computer scientists, as they are at the heart of quantum feedback control~\cite{QuantumFeedbackSurvey}. Our \(\kappa\)-while loop can be seen as an example of quantum feedback control where the weak measurements are applied at discrete time steps.
Discrete time feedback control has been used to protect a single qubit from decoherence~\cite{FeedbackQECC,FeedbackQECCExperiment}, while similar notions of weak measurement (over continuous time) have been proposed to monitor and drive complex evolutions~\cite{ContinuousFeedback}.
However, weak measurements are rarely used in the design of quantum algorithms; Mizel's work~\cite{Mizel} being the only case we are aware of.
By defining the \(\kappa\)-while loop programming construct we are attempting to introduce weak measurements to  algorithm designers and programming language experts in a language that is familiar to them.

We have shown that a \(\kappa\)-while loop may be used to implement Grover's algorithm; however, there is no algorithmic advantage with respect to the standard approach.
The reason is that the distribution of \(\cA\)-iterations throughout Grover's algorithm is easy to predict, and thus we can fix the number of iterations we should run it for in advance, instead of using a \(\kappa\)-while loop.
Our approach trivially extends to the more general setting of amplitude amplification originally proposed in~\cite{AAA}.

Mizel~\cite{Mizel} proposed a version of Grover's algorithm that is, in essence, the same as ours.
We learned about Mizel's work while writing the first version of our publication~\cite{WeakWhileLoop}.
Our work can be seen as a generalisation of Mizel's where our study of Grover's algorithm is meant as a proof of concept; our main contribution is the proposal of a programming construct --- the \(\kappa\)-while loop --- and the introduction of tools to study the worst-case complexity of algorithms using it.
Mizel presents their algorithm as a fixed point routine, where a mixed state gradually converges to the target state.
However, this fixed point behaviour is an artefact of disregarding the outcome of the weak measurement (see Remark~\ref{rmk:known_probe_outcome}): we know the loop will eventually halt and, when it does, the result will be the target state; the fixed point behaviour Mizel observes is that of the cumulative probability of halting converging to \(1\) as the number of iterations tends to infinity (see Figure~\ref{fig:sampling}) and, thus, it is an statement about the stochastic behaviour of multiple runs of the algorithm.
In contrast, we have focused on describing the evolution of the quantum state in a single run, and we have shown that the state keeps evolving at approximately the same rate throughout the algorithm (\ie{} the angle increases an amount between \(\alpha\) and \(3\alpha\) per iteration) until a \(\top\) measurement outcome occurs, in which case the state collapses onto the target subspace.
The field of quantum trajectories~\cite{ToddTutorial} focuses on understanding these kind of stochastic dynamics in the presence of weak measurements.

It is worth mentioning that algorithms for Grover's search using a fixed point approach do exist~\cite{GroverFixPoint,YoderFixPoint}.
In these, no measurement is applied during execution; instead, each iteration applies a unitary parametrised by a value that is gradually reduced throughout the algorithm.
Intuitively, each iteration moves the state closer to the target, but each time the step is smaller to avoid overshooting.
The drawback of this fixed point approach is that, unless we can implement a circuit where the unitary's parameter can be changed during runtime, each iteration requires a different circuit.

\subsection{Applications beyond quantum search}
\label{sec:beyond_quantum_search}

We have presented a \(\kappa\)-while loop version of Grover's algorithm as a proof of concept.
However, the pragmatic value of \(\kappa\)-while loops cannot be confirmed unless we find examples of quantum algorithms where the distribution of \(\cA\)-iterations is unknown or it would be costly to predict.
We believe such examples may be found in the field of quantum walk based algorithms.
The literature on quantum walks provides multiple results where quantum computers exhibit an advantage in the study of Markov processes~\cite{Szegedy,VertexFinding}.
The unitary evolution of these quantum walks is defined in terms of the transition matrix of the Markov process, which is often only required to be symmetric and ergodic.
Thus, the literature from this field offers a diverse class of examples of quantum algorithms whose evolution may be arbitrarily complex.
We argue there may be Markov processes whose quantum walk can be proven to be robust to \(\kappa\)-measurements (as per Definition~\ref{def:robustness}) while, at the same time, being complex enough to prevent an accurate estimation of when a desired state would be reached.
The argument goes as follows: on one hand, the choice of \(\kappa\) may be determined by a lower bound of the rate at which the walker traverses the graph (so that \(\kappa\)-measurements are not so strong that they would prevent the walker from reaching certain regions of the graph). This rate may be estimated as the infimum of local rates, using notions such as conductance and effective resistances~\cite{Belovs}.
On the other hand, in order to predict \emph{when} the walker reaches a particular state, we would need to take into account the global behaviour of the walk, which may be a more daunting task.
Thus, in such situations choosing the value of the measurement strength \(\kappa\) may be an easier task than that of deciding the number of iterations in advance.

Another possible avenue of research is the study of continuously measured while loops.
The notion of weak measurement used in this chapter is discrete, but continuous measurements have been studied extensively in the literature of quantum feedback control~\cite{ContinuousMeasurement}.
We can entertain the idea of a while loop whose body describes an infinitesimal step of the evolution --- for instance, by providing a Hamiltonian --- thus extending \(\kappa\)-while loops to the continuous-time setting.
This may be valuable from an experimentalist perspective; for instance, it hints at implementations of algorithms (for instance, Grover's search) where the probe is continuously measured throughout the execution, thus simplifying control: we no longer need to know when each iteration finishes and the next one starts.

%% file: Figures/4/kappa_while.tex
\begin{tikzpicture}
  \node (code_alias) {
    \begin{algorithm}
        while $!_\kappa$ $Q[\rho]$ do
          $\rho$ $\gets$ $\cC(\rho)$
        end
    \end{algorithm}
  };
  \node[right=5mm of code_alias] (code_def) {
    \begin{algorithm}
        $q$ $\gets$ $\ketbra{\bot}{\bot}$
        while $M[q] = \bot$ do
          $\rho$ $\gets$ $\cC(\rho)$
          $\rho,q$ $\gets$ $E_{\kappa,Q}(\rho \otimes q)$
        end
    \end{algorithm}
  };
\end{tikzpicture}

%% file: Figures/4/code_walk.tex
\begin{tikzpicture}
\node (walk) {
  \begin{tikzpicture}
    \node[trapezium,draw=black,thick,trapezium angle=45,minimum width=15mm,shape border rotate=90] (BS) {\scriptsize \(U_P\)};
    \coordinate[below=4mm of BS.east] (BS-0); \node[left=-0.5mm of BS-0] {\tiny \(\bot\)};
    \coordinate[above=4mm of BS.east] (BS-1); \node[left=-0.5mm of BS-1] {\tiny \(\top\)};
    \node[trapezium,draw=black,thick,trapezium angle=45,minimum width=15mm,shape border rotate=270,right=30mm of BS] (BSd) {\scriptsize \(U_P^{\dagger}\)};
    \coordinate[below=4mm of BSd.west] (BSd-0); \node[right=-0.5mm of BSd-0] {\tiny \(\bot\)};
    \coordinate[above=4mm of BSd.west] (BSd-1); \node[right=-0.5mm of BSd-1] {\tiny \(\top\)};
    \node[right=7mm of BS-1] (meas) {\large \(\eye\)};
    \coordinate[left=7mm of BSd-1] (0);
    \coordinate[above=1.6mm of 0] (0up);
    \coordinate[below=1.6mm of 0] (0down);
    \node[rectangle,draw=black,thick,right=12mm of BS-0] (G) {\small \(G\)};
    \node[rectangle,draw=black,thick,right=9mm of BSd] (E) {\small \(E_{\kappa,\chi}\)};
    \coordinate[right=4mm of E] (rightE);
    \coordinate[below=15mm of rightE] (belowE);
    \coordinate[left=4mm of BS] (leftBS);
    \coordinate[below=15mm of leftBS] (belowBS);
    \draw (0up) edge[ultra thick] (0down);
    \draw (0) --node[midway,pin={[pin distance=-2mm, pin edge={opacity=0}] above : \scriptsize\(H\)}] {\scriptsize \(>\)} (BSd-1);
    \draw (BS-1) --node[midway,pin={[pin distance=-2mm, pin edge={opacity=0}] above : \scriptsize\(H\)}] {\scriptsize \(>\)} (meas);
    \draw (BS-0) --node[midway,pin={[pin distance=-2mm, pin edge={opacity=0}] below : \scriptsize\(H\)}] {\scriptsize \(>\)} (G.west);  
    \draw (G.east) --node[midway,pin={[pin distance=-2mm, pin edge={opacity=0}] below : \scriptsize\(H\)}] {\scriptsize \(>\)} (BSd-0); 
    \draw (BSd.east) --node[midway] {\scriptsize \(>\)} (E.west); 
    \draw (E.east) --node[right] {\scriptsize \(>\)} (rightE); 
    \draw (rightE) edge[out=0,in=0,looseness=1.5] (belowE);
    \draw (belowE) --node[midway,pin={[pin distance=-2mm, pin edge={opacity=0}] below : \scriptsize\(H \otimes P\)}] {\scriptsize \(<\)} (belowBS);
    \draw (belowBS) edge[out=180,in=180,looseness=1.5] (leftBS);
    \draw (leftBS) --node[midway] {\scriptsize \(>\)} (BS.west);
  \end{tikzpicture}
};
\node[left=5mm of walk] (code) {
  \begin{algorithm}
  input: $\chi$, $\ket{\psi}$
  output: $q$
  begin
    $q$ $\gets$ $\ket{\psi}$
    while $!_\kappa$ $\chi[q]$ do
      $q$ $\gets$ $G(q)$
    end
  end
  \end{algorithm}
};
\end{tikzpicture}

%% file: Figures/4/theta_triangle.tex
\begin{tikzpicture}
  \node (fstQuad) {
    \begin{tikzpicture}
    \draw[-angle 60,black!65] (0,0) -- (10,0) node[midway,black,xshift=3pt,yshift=-10pt] {\tiny \(\cos{a} \ \ket{\psi_0}\)};
    \draw[-angle 60,black!65] (10,0) -- (10,7) node[midway,black,rotate=90,yshift=-11pt] {\tiny \(\sin{a} \ \ket{\psi_1}\)};
    \draw[-angle 60,black!65] (10,0) -- (10,5) node[midway,black,rotate=90,xshift=-3pt,yshift=9pt] {\tiny \(\xi \sin{a} \ \ket{\psi_1}\)};
    \draw[-latex,semithick] (0,0) -- (10,7) node[midway,yshift=20pt,xshift=10pt] {\scriptsize \(\ket{a}\)};
    \draw[-latex,semithick] (0,0) -- (10,5) node[midway,yshift=2pt,xshift=30pt] {\scriptsize \(W_\bot\ket{a}\)};
    \draw[densely dashed] (2,0) arc (0:35:2) node[midway,xshift=7pt] {\tiny \(a\)};
    \draw[densely dashed] (27:3) arc (27:35:3) node[midway,xshift=6pt,yshift=3pt] {\tiny \(\theta\)};
    \end{tikzpicture}
  };
  \node[right=15mm of fstQuad] (sndQuad) {
    \begin{tikzpicture}
    \draw[-angle 60,black!65] (0,0) -- (-10,0) node[midway,black,xshift=-25pt,yshift=-9pt] {\tiny \(\cos{a} \ \ket{\psi_0}\)};
    \draw[-angle 60,black!65] (-10,0) -- (-10,7) node[midway,black,rotate=90,xshift=-20pt,yshift=9pt] {\tiny \(\sin{a} \ \ket{\psi_1}\)};
    \draw[-angle 60,black!65] (-10,0) -- (-10,5) node[midway,black,rotate=90,xshift=-27pt,yshift=-7pt] {\tiny \(\xi \sin{a} \ \ket{\psi_1}\)};
    \draw[-latex,semithick] (0,0) -- (-10,7) node[midway,yshift=18pt,xshift=-20pt] {\scriptsize \(\ket{a}\)};
    \draw[-latex,semithick] (0,0) -- (-10,5) node[midway,yshift=2pt,xshift=-45pt] {\scriptsize \(W_\bot\ket{a}\)};
    \draw[dashdotted,black!65] (0,0) -- (2,0);
    \draw[densely dashed] (1,0) arc (0:145:1) node[midway,yshift=2pt] {\tiny \(a\)};
    \draw[densely dashed] (145:3) arc (145:153:3) node[midway,xshift=-13pt,yshift=3pt] {\tiny \(\theta\)};
    \end{tikzpicture}
  };
  \node[above left=-2mm and -5mm of fstQuad] (a) {\small \emph{i)}};
  \node[above left=-2mm and -2mm of sndQuad] (b) {\small \emph{ii)}};
\end{tikzpicture}

%% file: Chapters/5.tex
\chapter{Final remarks}
\label{chap:remarks}

The driving motivation of this thesis has been the study of control flow in quantum programs, with particular focus on unbounded iterative loops.
We have studied both the case of quantum control flow and the case of classical control flow.
Along the way we have presented multiple theoretical results whose relevance was discussed within their corresponding chapters.
In this final chapter we wish to highlight the two contributions of major pragmatic value.

Chapter~\ref{chap:trace} culminates with Theorem~\ref{thm:LSIContraction_traced}, establishing that \((\LSIContraction, \oplus, \ex)\) is a totally traced category.
The morphisms in \(\LSIContraction\) describe quantum processes over discrete time and the execution formula in this category represents quantum iteration.
This result opens up the possibility of providing categorical semantics for quantum programming languages using \emph{quantum} control flow.
Indeed, in Section~\ref{sec:semantics} we sketch a toy programming language with quantum control flow and we provide its denotational semantics in the category \(\LSIContraction\).

In Chapter~\ref{chap:while} we have proposed a novel programming primitive: the \(\kappa\)-while loop.
These \(\kappa\)-while loops are an instance of \emph{classical} control flow since the decision-making process is dependent on the outcome of a measurement.
This primitive may be introduced into any quantum programming language supporting classical control; we discuss how this may be achieved in Section~\ref{sec:kappa_CPTR} and provide denotational semantics in \(\CPTR\) for the \(\kappa\)-while loop.

\section{A toy language with quantum control flow}
\label{sec:semantics}

In this section we sketch the bare bones of a programming language \qwhile{} supporting quantum control flow and provide its denotational semantics using the results from Chapter~\ref{chap:trace}.
Previous works~\cite{YingFlow,DelayedTrace,PBS} have proposed programming languages and mathematical frameworks aiming to describe quantum computation with quantum control.
The shortcomings of the approach proposed by Ying, Yu and Feng~\cite{YingFlow} were discussed in Section~\ref{sec:control_flow}, whereas Section~\ref{sec:trace_rel_work} discussed the connection between our work and that of~\cite{DelayedTrace}.
The authors of~\cite{PBS} proposed the PBS-calculus, a diagrammatic language used to describe and reason about quantum computation with quantum control flow; at the end of this section we compare \qwhile{} with their approach.

\begin{definition}
	We use finite sets \(A,B\) to specify the input and output types of programs \(P \colon A \to B\) in \qwhile{}.
	Every program \(P \colon A \to B\) in \qwhile{} is inductively defined using the following syntax:
	\begin{equation*}
		U \mid \Delta_t \mid P;P' \mid P\!\oplus\!P' \mid \dowhilePX
	\end{equation*}
	\begin{itemize}
		\item there is a \emph{unitary primitive} \(U \colon A \to A\) for every finite set \(A\) and every unitary \(U \colon \Cset^{\abs{A}} \to \Cset^{\abs{A}}\);
		\item there is a \emph{delay primitive} \(\Delta_t \colon \{\bullet\} \to \{\bullet\}\) for every \(t \in \Nset\) and every element \(\bullet\);
		\item there is a \emph{sequential combinator} yielding \(P; P' \colon A \to C\) for every two programs \(P \colon A \to B\) and \(P' \colon B \to C\);
		\item there is a \emph{parallel combinator} yielding \(P \oplus P' \colon A \uplus A' \to B \uplus B'\) for any two programs \(P \colon A \to B\) and \(P' \colon A' \to B'\);
		\item for every program \(P \colon A \uplus X \to B \uplus X\) there is a \emph{while loop combinator} yielding \((\dowhilePX) \colon A \to B\).
	\end{itemize}
\end{definition}

We choose to write \((\dowhilePX)\) instead of the more usual syntax \((\text{while } X \text{ do } P)\) since its semantics are easier to describe.
The semantics of \((\dowhilePX)\) will be given by the execution formula where the shortest execution path is \(P_\sub{BA} \colon A \to B\) which can be understood as a single application of \(P\).
Thus, we want our syntax to suggest that at least one iteration of \(P\) will always be applied, hence, we `do' \(P\) before we test the loop condition.

The denotational semantics of \qwhile{} is given by a mapping \(\sem{-}\) from programs \(P \colon A \to B\) to morphisms \(\sem{P} \in \LSIContraction(A,B)\). 
Recall from Section~\ref{sec:LSI} that a morphism \(f \colon A \to B\) in \(\LSIContraction\) corresponds to a linear shift invariant contraction \(\oplus_A\, \ltwo \to \oplus_B\, \ltwo\). According to Definition~\ref{def:LSIContraction}, a morphism \(f \in \LSIContraction(A,B)\) is represented as an \(\Rset\)-indexed collection of contractions:
\begin{equation*}
	f = \{\widehat{f}_\omega \in \FdContraction(\Cset^{\abs{A}},\Cset^{\abs{B}})\}_{\omega \in \Rset}.
\end{equation*}

\begin{definition}
	The denotational semantics of a program \(P \colon A \to B\) in \qwhile{} is the morphism \(\sem{P} \in \LSIContraction(A,B)\) inductively defined as follows:
	\begin{itemize}
		\item for each \emph{unitary primitive} \(U \colon A \to A\), let \(\sem{U}\) be the \(\Rset\)-indexed collection where \(\widehat{\sem{U}}_\omega = U\) for all \(\omega \in \Rset\);
		\item for each \emph{delay primitive} \(\Delta_t \colon \{\bullet\} \to \{\bullet\}\), let \(\sem{\Delta_t}\) be the \(\Rset\)-indexed collection where \(\widehat{\sem{\Delta_t}}_\omega = e^{-i \omega t}\) for all \(\omega \in \Rset\);
		\item for each program of the form \(P; P'\), let \(\sem{P; P'}\) be the \(\Rset\)-indexed collection where \(\widehat{\sem{P; P'}}_\omega = \widehat{\sem{P'}}_\omega \circ \widehat{\sem{P}}_\omega\);
		\item for each program of the form \(P \oplus P'\), let \(\sem{P \oplus P'}\) be the \(\Rset\)-indexed collection where \(\widehat{\sem{P \oplus P'}}_\omega = \widehat{\sem{P}}_\omega \oplus \widehat{\sem{P'}}_\omega\);
		\item for each program of the form \((\dowhilePX)\) where \(P \colon A \uplus X \to B \uplus X\), let
		\begin{equation*}
			\widehat{\sem{\dowhilePX}}_\omega = \ex^X(\sem{P}_\omega)
		\end{equation*}
		where \(\ex\) is the execution formula in \(\LSIContraction\).
	\end{itemize}
\end{definition}

The definition of \(\sem{-}\) on each syntax primitive follows from the discussion on linear shift invariant maps in Section~\ref{sec:LSI}. In particular, recall that the elements of the \(\Rset\)-indexed collection representing a linear time-shift invariant map \(f\) may be understood as the action --- via multiplication --- of \(f\) on inputs that are periodic signals of angular frequency \(\omega \in \Rset\) (see Remark~\ref{rmk:FT_plane_waves}). 
Since unitary primitives are instantaneous, these act in the same manner on all input signals, whereas a delay primitive simply changes the phase of the signal.
The correctness of these definitions can be checked by obtaining the characteristic function \(\chi^U = U\!\cdot\!\delta\) and \(\chi^\Delta_t = S_t[\delta]\) of these LSI maps (see Proposition~\ref{prop:l2_lsi}), then applying the discrete-time Fourier transform to them (Definition~\ref{def:DTFT}).
We established in Section~\ref{sec:LSI} that, thanks to the convolution theorem~\ref{thm:l2_convolution}, composition of LSI maps can be described in terms of index-wise composition of their corresponding \(\Rset\)-indexed collections; a similar result trivially follows for parallel composition (\ie{} monoidal product).
Finally, we provide semantics for while loops using the execution formula; well-definedness is guaranteed since \((\LSIContraction,\oplus,\ex)\) is a totally traced category, as established in Theorem~\ref{thm:LSIContraction_traced}.

\paragraph*{Comparison with PBS-calculus~\cite{PBS}.} The authors of the PBS-calculus consider two separate state spaces: one being the internal state (holding `data') and another --- referred to as the `polarisation' --- which determines the execution path the computation will take.
In the case of \qwhile{} we only capture the `polarisation' space since the state space assigned to each input port \(a \in A\) of a program \(P \colon A \to B\) is the \(\ltwo\) space of square-summable functions \(\Zset \to \Cset\), \ie{} we only capture amplitude over time on each port \(a \in A\).
It would not be conceptually difficult to extend \qwhile{} so that the inputs and outputs ports may carry an `internal state'; it is a matter of increasing dimensions.
However, there is a fundamental conceptual difference between \qwhile{} and the PBS-calculus: the PBS-calculus does not support any primitive that changes the polarisation state during computation.
This restriction reduces the expressivity of loops dramatically: the variables used in conditional statements are, in fact, constants defined at the program's input.
The toy programming language \qwhile{} does not exhibit this restriction.

Notice that the introduction of the time delay primitive \(\Delta_t\) to the syntax of \qwhile{} is not necessary for the denotational semantics to be well-defined. In fact, if we remove \(\Delta_t\) from the syntax we would be able to define the semantics of \qwhile{} within \((\FdContraction, \oplus, \ex)\) --- or even \(\FdUnitary\) --- which we have shown to be a totally traced category as well (Theorem~\ref{thm:FdContraction_traced}).
However, in such a case the semantics would suggest that the different execution paths of the program may interact with each other as if all of them produced their output in the same instant --- this is not a physically sound proposal.
Thus, the delay primitive \(\Delta_t\) is added for pragmatic reasons rather than theoretical limitations.
An interesting consequence is that we may perceive \qwhile{} as a language for manipulating wave functions since states are signals rather than qubits.
However, recall that \(\LSIContraction\) only captures processes over \emph{discrete} time, whereas wave functions are defined over continuous time.

\section{The \(\kappa\)-while loop in existing programming languages}
\label{sec:kappa_CPTR}

The \(\kappa\)-while loop primitive may be introduced into any quantum programming language supporting classical control flow and standard while loops.
Figure~\ref{fig:kappa_while_bis} shows how this may be achieved:\footnote{Figure~\ref{fig:kappa_while_bis} is a copy of Figure~\ref{fig:kappa_while}, reproduced here for the reader's convenience.} a \(\kappa\)-while loop is nothing more than a standard while loop controlled by a \(\kappa\)-measurement, so it suffices to unpack the definition of \(\kappa\)-measurement (Definition~\ref{def:kappa_measurement}) within the body of the loop.
This may be done by a parser that scans the code and replaces any occurrence of a \(\kappa\)-while loop with the equivalent code from Figure~\ref{fig:kappa_while_bis}.

\begin{figure}
	\input{Figures/4/kappa_while}
	\caption{\emph{Left:} the syntax we use to represent a \(\kappa\)-while loop; \(\kappa \in [0,1]\) is a parameter set by the programmer and \(Q\) is the predicate to be measured. \emph{Right:} the pseudocode that implements the \(\kappa\)-while loop on a programming language with classical control flow, following the notation from~\cite{YingBook}; \(E_{\kappa,Q}\) is defined in Section~\ref{sec:weak_meas} and the condition \(M[q] = \bot\) indicates that \(q\) is measured to distinguish between \(H \otimes \Span \{\bot\}\) and \(H \otimes \Span \{\top\}\), with the condition being satisfied if \(q\) collapses to the former.}
	\label{fig:kappa_while_bis}
\end{figure}

On a similar note, the semantics of the \(\kappa\)-while loop may also be defined in terms of the semantics of the base language.
In fact, it is as simple as evaluating the semantics of the equivalent code provided in the right hand side of Figure~\ref{fig:kappa_while_bis} which, being comprised only of primitives of the base programming language, should already have well-defined semantics.
Nevertheless, and for the purpose of completeness, we construct the morphism in \(\CPTR\) that ought to describe the denotational semantics of the \(\kappa\)-while loop.
Let \(Q\) be the halting predicate and let \(\kappa\) be its measurement strength; we may construct a function 
\begin{equation*}
	\mathrm{Wk}_{\kappa,Q} \colon \CPTR(B(H),B(H)) \to \CPTR(B(H)\!\oplus\!B(H), B(H)\!\oplus\!B(H))
\end{equation*} 
that maps a morphism \(\cC \in \CPTR(B(H),B(H))\) describing the body of a \(\kappa\)-while loop to its corresponding weakly measured counterpart.
To do so, use the isomorphism \(h \colon H \otimes \Span \{\bot, \top\} \to H \oplus H\) and the functor \(F \colon \FdContraction \to \CPTR\) given in Definition~\ref{def:env_functor} along with the canonical inclusion \(\theta \colon B(H) \oplus B(H) \to B(H \oplus H)\) and its left inverse \(\phi \colon B(H \oplus H) \to B(H) \oplus B(H)\) (see Section~\ref{sec:comparing_traces}).\footnote{Notice that \(\theta\) may be understood as \emph{state preparation} and \(\phi\) as \emph{measurement}~\cite{Robin}.} Then, for every \(\cC \in \CPTR(B(H),B(H))\) we may define \(\mathrm{Wk}_{\kappa,Q}(\cC)\) as follows:
\begin{equation*}
	\mathrm{Wk}_{\kappa,Q}(\cC) =  \phi \circ F(h \circ E_{\kappa,Q}) \circ (\cC \otimes \id) \circ F(h^{-1}) \circ \theta
\end{equation*}
where \(E_{\kappa,Q}\) is a unitary defined in Section~\ref{sec:weak_meas}.
Since \(\CPTR\) is a \(\SCat{s}\)-UDC it has canonical quasi-injections \(\iota_\bot,\iota_\top \colon B(H) \to B(H) \oplus B(H)\) and quasi-projections \(\pi_\bot,\pi_\top \colon B(H) \oplus B(H) \to B(H)\) distinguishing the measurement outcomes \(\bot\) and \(\top\). Then, we may define \(\mathrm{Wk}_{\kappa,Q}\) in terms of its unique decomposition:
\begin{equation*}
	\mathrm{Wk}_{\kappa,Q}(\cC) = 
	\begin{pmatrix}
		\pi_\top \circ \mathrm{Wk}_{\kappa,Q}(\cC) \circ \iota_\top & \pi_\top \circ \mathrm{Wk}_{\kappa,Q}(\cC) \circ \iota_\bot \\
		\pi_\bot \circ \mathrm{Wk}_{\kappa,Q}(\cC) \circ \iota_\top & \pi_\bot \circ \mathrm{Wk}_{\kappa,Q}(\cC) \circ \iota_\bot
	\end{pmatrix}
\end{equation*}
so that \(\ex^{B(H)}(\mathrm{Wk}_{\kappa,Q})(\cC)\) corresponds to the morphism in \(\CPTR\) describing the \(\kappa\)-while loop of body \(\cC\) and halting predicate \(Q\), where the execution formula \(\ex\) in \(\CPTR\) is the one discussed in Section~\ref{sec:classical_loop}.

\section{Conclusion}

With this thesis we are attempting to put forward new perspectives on the field of control flow of quantum programs and, in particular, the study of unbounded iteration. On one hand, we have revisited Bartha's result~\cite{Bartha} that quantum processes with instantaneous coherent quantum feedback are well-defined: more formally, that \((\FdContraction,\oplus,\ex)\) is a totally traced category. A physical interpretation of this result is that, instead of assuming inputs and outputs come in the form of particles interacting with a device, these are described in terms of plane waves (see Remark~\ref{rmk:FT_plane_waves}) which are, in some sense, everywhere at once; thus, feedback in such a situation does not require a notion of `delay' to be well-defined. Following this intuition, we have used the Fourier transform to decompose any input signal to its plane wave components and used the properties of time-shift invariant maps to prove the novel result that \((\LSIContraction,\oplus,\ex)\) is totally traced.
Other novel technical results on category theory and, in particular, on categories of \(\Sigma\)-monoids and UDCs have been introduced while building towards our objective.
On the other hand, we have introduced a novel programming primitive --- the \(\kappa\)-while loop --- that in certain cases may be used to classically control the termination condition of a quantum subroutine without sacrificing the quantum speed-up. 

Quantum computer scientists tend to dismiss unbounded iteration in quantum algorithms as an uninteresting field of work: in the case of classical control flow due to the assumption that testing a termination condition on every iteration would destroy any achievable quantum speed-up and, in the case of quantum control flow, due to the technical obstacles that interference and the possibility of infinitely many execution paths would entail --- these obstacles were discussed in Section~\ref{sec:control_flow}. However, prior to this thesis there already were examples of both classically controlled unbounded loops --- such as Mizel's proposal of a measurement-controlled Grover algorithm~\cite{Grover} --- and instances of quantumly controlled unbounded loops --- in the form of a wide variety of quantum-walk based algorithms. In this thesis we have established that unbounded iteration in quantum computing is mathematically sound. It is my hope that further study on this field will yield quantum programming languages supporting quantum control flow and new algorithms that make use of unbounded iteration.